\documentclass[a4paper]{msc}
\usepackage[pass]{geometry}

\usepackage{latexsym,xspace,calc,amsthm}
\usepackage{amsmath,amssymb-edited} %
\usepackage{nicefrac} %
\usepackage{quoting}
\usepackage{tocloft}

\usepackage{tgtermes}
\usepackage{fix-cm}

\usepackage{cancel}
\usepackage{graphbox} %

\usepackage{tikz}
\usepackage{tikz-cd}
\usetikzlibrary{cd,decorations.markings}
\tikzset{
    dharrow/.style={
        <->,
        postaction={decorate,-},
        }
}
\tikzset{
    dhdashedarrow/.style={
        <->,
        dashed,
        postaction={decorate,-},
        }
    }
\tikzset{
    lrharpoonarrow/.style={
        <[harpoon]->[harpoon],
        postaction={decorate,-},
        }
}
\tikzset{
    lrharpoondashedarrow/.style={
        <[harpoon]->[harpoon],
        dashed, %
        postaction={decorate,-},
        }
}
\usetikzlibrary{arrows}
\usetikzlibrary {arrows.meta} 
\usetikzlibrary{calc}
\usetikzlibrary{positioning}
\usetikzlibrary{snakes,automata,chains}
\usetikzlibrary{graphs}

\usepackage{binarytree}

\usepackage{mdwlist} %
\usepackage{float}   %
\usepackage{centernot} %

\newtheoremstyle{jamiestyle}%
  {4pt}%
  {0pt}%
  {\it}%
  {0pt}%
  {\sc}%
  {.}%
  { }%
  {}%
\theoremstyle{jamiestyle}
\newtheorem{thrm}{Theorem}[subsection]
\newtheorem{prop}[thrm]{Proposition}
\newtheorem{lemm}[thrm]{Lemma}
\newtheorem{corr}[thrm]{Corollary}

\newtheoremstyle{jamienfstyle}%
  {4pt}%
  {0pt}%
  {\normalfont}%
  {0pt}%
  {\sc}%
  {.}%
  { }%
  {}%
\theoremstyle{jamienfstyle}
\newtheorem{nttn}[thrm]{Notation}
\newtheorem{defn}[thrm]{Definition}
\newtheorem{xmpl}[thrm]{Example}
\newtheorem{rmrk}[thrm]{Remark}

\newcommand\jamiesection[1]{\section{#1}}
\newcommand\jamiesubsection[1]{\subsection{#1}}
\newcommand\jamiesubsubsection[1]{\subsubsection{#1}}

\usepackage{color}
\definecolor{mygreen}{rgb}{0,0.6,0}
\definecolor{mygray}{rgb}{0.5,0.5,0.5}
\definecolor{mymauve}{rgb}{0.58,0,0.82}

\usepackage{listings}
 
\definecolor{gray}{RGB}{128, 128, 128}
\definecolor{lightgray}{RGB}{200, 200, 200}
\definecolor{cyan}{RGB}{0, 255, 255}
\definecolor{blue}{RGB}{0, 0, 255}
\definecolor{red}{RGB}{255, 0, 0}
\definecolor{pink}{RGB}{255, 128, 128}
\definecolor{green}{RGB}{0, 128, 0}
\definecolor{lightyellow}{RGB}{255, 255, 200}
\definecolor{purple}{RGB}{128, 0, 128}

\lstdefinestyle{all}
    {basicstyle=\ttfamily\scriptsize,
     keywordstyle=\color{blue}\ttfamily\scriptsize,
     commentstyle=\color{green}\ttfamily\scriptsize,
     stringstyle=\color{red}\ttfamily\scriptsize}

\lstdefinelanguage{hask}{%
    frame=none,
    xleftmargin=2pt,
    belowcaptionskip=\bigskipamount,
    captionpos=b,
    tabsize=2,
    numbers=left,
    numberstyle=\tiny\color{gray},
    emphstyle={\bf},
	morecomment=[s][\color{green}]{\{-}{-\}},
    stringstyle=\mdseries\rmfamily,
    commentstyle=\color{green},
    keywords={},
    keywords=[1]{case, of, data, if, then, else, where, let, in, do},
    keywords=[2]{Chip, Config, CurrencySymbol, TokenName, PubKeyHash, Integer, Value, State, Action, Text, Maybe, Void, TxConstraints,  Contract},
    keywords=[3]{HasNative},
    keywords=[4]{=>},
    keywords=[5]{Just, Nothing, MkChip, MkConfig, SetPrice, Buy},
    keywordstyle=[1]\mdseries\sffamily\color{red},
    keywordstyle=[2]\mdseries\sffamily\color{blue},
    keywordstyle=[3]\mdseries\sffamily\color{green},
    keywordstyle=[4]\mdseries\sffamily,
    keywordstyle=[5]\mdseries\sffamily\color{purple},
    columns=flexible,
    basicstyle=\small\sffamily,
    showstringspaces=false,
    breaklines=false,
    showspaces=false,
    escapeinside={--}{\^^M},escapebegin={\color{green}--},escapeend={},
    literate= {+}{{$+$}}1 {/}{{$/$}}1 {*}{{$*$}}1 {=}{{$=$}}1
              {>}{{$>$}}1 {<}{{$<$}}1 {\\}{{$\lambda$}}1
              {\\\\}{{\char`\\\char`\\}}1
              {->}{{$\rightarrow$}}2 {>=}{{$\geq$}}2 {<-}{{$\leftarrow$}}2
              {<=}{{$\leq$}}2 {=>}{{$\Rightarrow$}}2
              {\ .}{{$\circ$}}2 {\ .\ }{{$\circ$}}2
              {>>}{{>>}}2 {>>=}{{>>=}}2
              {|}{{$\mid$}}1
              {\_}{{\underline{\hspace{2mm}}}}2
}

\lstdefinelanguage{solidity}{%
    frame=none,
    xleftmargin=2pt,
    belowcaptionskip=\bigskipamount,
    captionpos=b,
    tabsize=2,
    numbers=left,
    numberstyle=\tiny\color{gray},
    emphstyle={\bf},
	morecomment=[s][\color{green}]{\{-}{-\}},
    stringstyle=\mdseries\rmfamily,
    commentstyle=\color{green},
    keywords={},
    keywords=[1]{pragma, solidity, contract, event, constructor, require, function, return, emit},
    keywords=[2]{address, uint, mapping},
    keywords=[3]{public, payable, external, view, returns},
    keywords=[4]{=>, +=, -=, =, <=, ==},
    keywords=[5]{msg, sender, transfer, value},
    keywordstyle=[1]\mdseries\sffamily\color{red},
    keywordstyle=[2]\mdseries\sffamily\color{blue},
    keywordstyle=[3]\mdseries\sffamily\color{green},
    keywordstyle=[4]\mdseries\sffamily,
    keywordstyle=[5]\mdseries\sffamily\color{purple},
    columns=flexible,
    basicstyle=\small\sffamily,
    showstringspaces=false,
    breaklines=false,
    showspaces=false,
    escapeinside={--}{\^^M},escapebegin={\color{green}--},escapeend={},
    literate= {+}{{$+$}}1 {/}{{$/$}}1 {*}{{$*$}}1 {=}{{$=$}}1
              {>}{{$>$}}1 {<}{{$<$}}1 {\\}{{$\lambda$}}1
              {\\\\}{{\char`\\\char`\\}}1
              {->}{{$\rightarrow$}}2 {>=}{{$\geq$}}2 {<-}{{$\leftarrow$}}2
              {<=}{{$\leq$}}2 {=>}{{$\Rightarrow$}}2
              {\ .}{{$\circ$}}2 {\ .\ }{{$\circ$}}2
              {>>}{{>>}}2 {>>=}{{>>=}}2
              {|}{{$\mid$}}1
              {\_}{{\underline{\hspace{2mm}}}}2
}

\makeatletter
\newcommand\hpn[2][]{%
  \ext@arrow 9999{\hpnfill@}{#1}{#2}}
\newcommand\hpnfill@{%
  \arrowfill@\leftharpoonup\relbar\rightharpoondown}
\makeatother

\newcommand\id{\f{id}}

\newcommand\oldin{\in}

\NewCommandCopy{\oldsetminus}{\setminus}
\renewcommand\setminus{{{\hspace{1pt}{\oldsetminus}\hspace{1pt}}}}

\newcommand\THREE{{\mathbf 3}}

\newcommand\xor{\mathbin{\mathsf{\small xor}}}

\newcommand\atopen{T}
\newcommand\afilter{F}
\newcommand\apoint{P}

\newcommand\leftopeninterval[1]{(#1]}
\newcommand\rightopeninterval[1]{[#1)}
\newcommand\openinterval[1]{(#1)}

\newcommand\opens{\tf{Open}}
\newcommand\regularOpens{\tf{Open}_{\f{reg}}}
\newcommand\topens{\tf{Topen}}
\newcommand\closed{\tf{Closed}}
\newcommand\regularClosed{\tf{Closed}_{\f{reg}}}

\newcommand\closure[1]{|#1|}

\newcommand{\dotarrow}{%
   \mathrel{\ooalign{\hss\raise.85ex\hbox{\scalebox{1.25}{\normalfont .}}%
   \kern0.35ex\hss\cr$\rightarrow$}}}

\newcommand\overlaps{{\rlap{$>$}\,{<}}}
\newcommand\notbetween{\mathbin{\cancel{\between}}}
\newcommand\notintertwinedwith{\mathrel{\notbetween}}

\newcommand\notintersectswith{\notbetween}

\newcommand\intertwined[1]{#1_{\between}}
\newcommand\intertwinedwith{\mathrel{\between}}

\newcommand\topind{\mathrel{\mathring{=}}} %

\newcommand\cti{\leq}  %
\newcommand\nbhd[0]{\f{nbhd}}
\newcommand\interior[0]{\f{interior}}

\newcommand\community[0]{\f{K}}
\newcommand\framecommunity[0]{\f{k}}

\newcommand\cast[1]{#1^{\ast c}}
\newcommand\cclo[1]{#1^c}

\makeatletter
\newcommand\@deffont[2][]{{\bfseries #2}\index{#1}}
\newcommand\deffont{\@dblarg\@deffont}
\makeatother
\newcommand\powerset{\f{pow}}

\newcommand\f[1]{\mathit{#1}}
\newcommand\tf[1]{\mathsf{#1}}
\newcommand\ns[1]{\mathsf{#1}}

\newcommand\liff{\Longleftrightarrow}
\newcommand\limp{\Longrightarrow}

\newcommand\minus{\text{-}}
\newcommand\plus{{+}}
\newcommand\Forall[1]{\forall #1.}
\newcommand\Exists[1]{\exists #1.}

\newcommand\cent{\vdash}

\newcommand\mone{{\text{-}1}}

\makeatletter
\DeclareRobustCommand{\barcent}{\mathbin{\mathpalette\barcent@@\relax}}
\newcommand{\barcent@@}[2]{%
  \vbox{\offinterlineskip
    \sbox\z@{$\m@th#1\cent$}%
    \ialign{%
      \hfil##\hfil\cr
      $\m@th#1{}_{\minus}\kern-\scriptspace$\cr
      \noalign{\kern-.3\ht\z@}
      \box\z@\cr
    }%
  }%
}
\makeatother

\makeatletter
\def\pmb@#1#2{\setbox8\hbox{$\m@th#1{#2}$}%
  \setboxz@h{$\m@th#1\mkern-.1mu$}\pmbraise@\wdz@
  \binrel@{#2}%
  \dimen@-\wd8 %
  \binrel@@{%
    \mkern-.1mu\copy8 %
    \kern\dimen@\mkern-.2mu\copy8 %
    \kern\dimen@\mkern-.3mu\copy8 %
    \kern\dimen@\mkern-.4mu\copy8 %
    \kern\dimen@\mkern.1mu\copy8 %
    \kern\dimen@\mkern.2mu\copy8 %
    \kern\dimen@\mkern.3mu\copy8 %
    \kern\dimen@\mkern.0mu\raise\pmbraise@\copy8 %
    \kern\dimen@\mkern.4mu\box8 %
           }%
}
\makeatother

\newcommand\ttop{{\pmb\top}}
\newcommand\tbot{{\pmb\bot}}

\newcommand\tand{{\pmb\wedge}}
\newcommand\tor{{\pmb\vee}}

\makeatletter
\newcommand{\circlearrow}{}%
\DeclareRobustCommand{\circlearrow}{%
  \mathrel{\vphantom{\shortrightarrow}\mathpalette\circle@arrow\relax}%
}
\newcommand{\circle@arrow}[2]{%
  \m@th
  \ooalign{%
    \hidewidth$#1\circ\mkern1mu$\hidewidth\cr
    $#1\longrightarrow$\cr}%
}
\makeatother

\makeatletter
\newcommand*\bigcdot{\mathpalette\bigcdot@{.5}}
\newcommand*\bigcdot@[2]{\mathbin{\vcenter{\hbox{\scalebox{#2}{$\m@th#1\bullet$}}}}}
\makeatother

\usepackage{datetime}
\yyyymmdddate

\begin{document}
\title{Semiframes: the algebra of semitopologies and actionable coalitions} %
\lefttitle{Semiframes}
\righttitle{Murdoch J. Gabbay}
\begin{authgrp}
\author{Murdoch J. Gabbay}
\affiliation{MACS, Heriot-Watt University \\ Edinburgh, UK
        \url{www.gabbay.org.uk}}
\end{authgrp}

\begin{abstract}
We introduce semiframes (an algebraic structure) and investigate their duality with semitopologies (a topological one).
Both semitopologies and semiframes are relatively recent developments, arising from a novel application of topological ideas to study decentralised computing systems.

Semitopologies generalise topology by removing the condition that intersections of open sets are necessarily open.
The motivation comes from identifying the notion of an \emph{actionable coalition} in a distributed system --- a set of participants with sufficient resources for its members to collaborate to take some action --- with \emph{open set}; since just because two sets are actionable (have the resources to act) does not necessarily mean that their intersection is. 

We define notions of category and morphism and prove a categorical duality between (sober) semiframes and (spatial) semitopologies, and we investigate how key well-behavedness properties that are relevant to understanding decentralised systems, transfer (or do not transfer) across the duality.

\end{abstract}

\maketitle

\tableofcontents

\section{Introduction}
\label{sect.intro}

This paper is about a duality between semiframes (an algebraic structure) and semitopologies (a topological one).
We will define semiframes, study their duality with semitopologies, and also motivate and define relevant well-behavedness conditions on semitopologies and investigate how these conditions dualise to algebraic conditions on semiframes.

Both semitopologies and semiframes are relatively recent developments, arising from a novel application of topological ideas to study decentralised computing systems.
So while this paper is about pure mathematics, and it can be read as mathematics for its own sake, it would be nice to understand the motivation for why we looked at these structures in the first place.

\subsection{The challenge of decentralised systems}

Consider a computing system whose operation is distributed over many collaborating participants, like a blockchain. 
At a high level, how can we design algorithms that participants can run such that the overall system will act coherently in some appropriate sense? 

We could just assume a central controller; a \emph{dictator participant}.
In this case, `collaborating' just means doing as the dictator instructs, and `truth' just means whatever the dictator says.

This is indeed how many distributed systems actually work, but it introduces tradeoffs of latency, scalability, availability, and trust which may or may not be acceptable depending on the use case.
Do we trust the dictator?  What if it crashes, or gets hacked?  What if the network is slow?  What if it is just busy?
Sometimes, we need participants to be able, to a greater or lesser extent, to collaborate independently of central control.

This is the topic of \emph{decentralised} systems, which are distributed systems that do \emph{not} adopt a central controller.
Decentralised systems are hard to design and implement, because control itself is spread out,\footnote{It is actually much worse than this: we may also have to worry about messages being lost or delayed, and about some participants being faulty, and not following algorithms and protocols correctly.  But even in the ideal case where communication is reliable and all participants are correct, decentralised algorithms are \emph{still} hard.} but often this is the only architecture that is practical at scale and speed.

While blockchains are an example of decentralised systems,\footnote{Pedant point: some blockchains have a central control, but this is a special case: most are fully decentralised.} %
these design issues arise repeatedly: e.g. if we want to manage a system of drones, or satellites, or low-power devices.
An ant colony is an interesting example of a (semi-)decentralised system in nature: while a queen ant exists, she does not dictate every action of every ant; for much of the time, individual ants work autonomously, yet they still manage to do so coherently and for the good of the colony. 
The human brain is another important example of a decentralised system. %
These systems are all inherently decentralised: they maintain some clear notion of global coherence (e.g. there is an identifiable thing called `an ant colony' or `my sense of self') all while their parts collaborate to perform significant actions independently of any explicit central control.

This leads us to consider the concept of what in~\cite{gabbay:semtad,gabbay:semdca} is called an \emph{actionable coalition}: a set of participants in a decentralised system with sufficient resources to act (if the participants in the coalition choose to do so), without the involvement --- i.e. without the help, the permission, or even the knowledge --- of any other participants.\footnote{The empty set has no resources but is still an actionable coalition because it has no members, so all of those members can act as they like.}

Let us return to our ant colony and consider a column of foraging ants:
what are the actionable coalitions here, where the ants' task is to bring food back to the colony?
When the column discovers some food, a relevant notion of actionable coalition is \emph{any collection of ants with the physical strength to heave the food back to the hive}.
Note that the ants do not need anybody to organise them into a group, or to tell them what to do: they swarm over the food until some group of them manage to move the food --- aha; an actionable coalition has assembled and agreed on their action! --- and then that coalition gets on with heaving while the rest of the ants move on.

When we describe an actionable coalition as ``a group of participants with the resources to act'', note that --- this being a decentralised system --- the action concerns just participants in that coalition and not the entire set of participants.
Contrast with e.g. a national voting system, which is distributed but still centralised in the sense that the winning candidate wins for \emph{everybody}, not just for the people who voted for that candidate.\footnote{How to design an algorithm that can run on a decentralised system and get coherent and sensible behaviour from the overall system, even though any actionable coalition can (within the rules of the algorithm) potentially agree on some action and carry it out, without waiting for participants not in that coalition to participate, approve, or even necessarily know that an action has been taken --- is \emph{the} challenge of designing decentralised algorithms.  It is not easy!  Explaining the magic of practical decentralised algorithms is out of scope for this paper; but it \emph{can be done}, and notions of actionable coalition are often key to such algorithms.  The interested reader can find more information in~\cite{lamport2001paxos,cachinbook}.}

A union of actionable coalitions is an actionable coalition.\footnote{Our ants, being mathematical idealisations, are massless and occupy zero volume.  However, even if the ants had mass and volume, some of them could just stand around drinking the ant equivalent of tea and encouraging the ones who are actually doing the heavy lifting; the mathematics of actionable coalitions makes no practical judgments about which participants are doing useful work in the actual implementation.}
This makes the set of all actionable coalitions look a bit like a topology, except that an intersection of actionable coalitions need not be actionable (just because sets $S$ and $S'$ are each collectively strong enough to act, does not mean that $S\cap S'$ is; e.g. $S\cap S'$ may consist of just a single ant).

This brings us to a mathematical abstraction: \emph{point-set semitopology}, which generalises point-set topology by removing the condition that intersections of open sets are necessarily open: 

\subsection{Point-set semitopologies}

\begin{nttn}
\label{nttn.powerset}
Suppose $\ns P$ is a set.
Write $\powerset(\ns P)$ for the powerset of $\ns P$ (the set of subsets of $\ns P$). %
\end{nttn}

\begin{defn}
\label{defn.semitopology}
A \deffont{semitopological space}, or \deffont{semitopology} for short, consists of a pair $(\ns P, \opens(\ns P))$ of 
a (possibly empty) set $\ns P$ of \deffont{points}, and 
a set $\opens(\ns P)\subseteq\powerset(\ns P)$ of \deffont{open sets}, 
such that:
\begin{enumerate*}
\item\label{semitopology.empty.and.universe}
$\varnothing\in\opens(\ns P)$ and $\ns P\in\opens(\ns P)$.
\item\label{semitopology.unions}
If $X\subseteq\opens(\ns P)$ then $\bigcup X\in\opens(\ns P)$.\footnote{There is a little overlap between this clause and the first one: if $X=\varnothing$ then by convention $\bigcup X=\varnothing$.  Thus, $\varnothing\in\opens(\ns P)$ follows from both clause~1 and clause~2.  If desired, the reader can just remove the condition $\varnothing\in\opens(\ns P)$ from clause~1, and no harm would come of it.} 
\end{enumerate*}
We may write $\opens(\ns P)$ just as $\opens$, if $\ns P$ is irrelevant or understood, and we may write $\opens_{\neq\varnothing}$ for the set of nonempty open sets.
\end{defn}

\begin{rmrk}[Justification for semitopologies]
\label{rmrk.collaboration}
A classic text~\cite{vickers:topvl} justifies topology as follows: %
\begin{enumerate*}
\item
Logically, open sets model \emph{affirmations}:\footnote{%
\href{https://www.wordnik.com/words/affirmation}{\emph{Affirmation:}}\ Something declared to be true; a positive statement or judgment (\url{https://www.wordnik.com/words/affirmation},\ permalink: \url{https://web.archive.org/web/20230608073651/https://www.wordnik.com/words/affirmation}).
}
an open set $O$ corresponds to an affirmation (of $O$)~\cite[page~10]{vickers:topvl}.
\item
Computationally, open sets model \emph{semidecidable properties}.  See the first page of the preface in~\cite{vickers:topvl}.
\end{enumerate*}
We can justify semitopologies in similar style as follows:
\begin{enumerate*}
\item
Logically, open sets model \emph{actionable} affirmations: an open set $O$ corresponds to an outcome that is agreed on \emph{and can be actioned} within $O$.
\item
Computationally, open sets model \emph{actionable} outcomes.
\end{enumerate*}
So the distinction between topology and semitopology is this: topology is about things that can be positively \emph{said}; whereas semitopology is about things that can be positively \emph{done}. 
\end{rmrk}

Identifying actionable coalitions as a topological concept is new; but as is often the case for key concepts, with hindsight we can see them everywhere.
The notion corresponds that of a \emph{quorum} in the classic distributed systems literature~\cite{lamport_part-time_1998}.
Social choice theorists have a similar notion called a \emph{winning coalition}~\cite[Item~5, page~40]{riker:thepc}.
In blockchains, the XRP Ledger~\cite{schwartz_ripple_2014} and Stellar network~\cite{lokhafa:fassgp} have explicit notions of actionable coalition, in the sense that a participant must specify information about their desired actionable coalitions when they sign up to the system.
In Ethereum and other proof-of-stake blockchains, an actionable coalition is (roughly speaking) any group of participants with more than half of the voting power of the system. 

\subsection{Well-behavedness conditions on semitopologies}

It turns out that we can understand a surprising amount about a decentralised system, just by analysing its actionable coalitions; i.e. by viewing it as a semitopology.
And when we do, the most relevant and interesting semitopologies turn out to be ill-behaved from the usual topological viewpoint, because they are almost never Hausdorff; indeed the most well-behaved semitopologies satisfy an \emph{anti-Hausdorff} property of \emph{being intertwined}, that all of their open sets intersect.

This is because if all actionable coalitions intersect, then any two actions that are taken must be compatible, in the sense that there must exist at least one participant who was able to take \emph{both} actions.
This implies a practically useful form of coherence for the overall system: while it is not the case that participants must act in synchrony (since that would not be decentralised), any actions that are taken must be compatible on some nonempty overlap; we make this formal in Theorem~\ref{thrm.correlated} and Corollary~\ref{corr.correlated.intersect}.
In practice, it turns out that this property is often enough to derive important coherence properties for the overall system.
 
Notions of \emph{transitive open set} (topen) and \emph{regular point} become central to the theory, where topens are sets of pairwise intertwined participants (all their open neighbourhoods intersect), 
and points are called `regular' when they have a topen neighbourhood.
Briefly and in a nutshell: `good' semitopologies have lots of regular points, and the best ones consist entirely of regular points.

The question we should then as is: what \emph{are} semitopologies, and what \emph{are} well-behavedness properties like having topens and being regular?

One answer is that things are what the definitions define them to be --- e.g. semitopologies are Definition~\ref{defn.semitopology} --- but the reader may know that we can use algebra to give better answers. 
In particular, looking at Definition~\ref{defn.semitopology} it might appear that semitopologies are just complete join-semilattices in sets.
Well, this is not quite true, as we shall see:

\subsection{Semiframes: the algebraic dual to semitopologies}

There is well-known recipe for understanding sets-based structures: we \emph{dualise} them.
That is: we decide what the morphisms should be (answer: continuous functions), we build a dual category of \emph{semiframes}, and we study that.

It will turn out that a semiframe is a \emph{compatible complete semilattice}.
What `compatibility' means is given in Definition~\ref{defn.semiframe} and it has to do with an algebraic abstraction of nonempty intersection of open sets. 
What we can observe here is that the devil will be in the details (as always): the definition `compatible complete semilattice' is not immediately obvious, and the proofs require work, and the duals to the well-behavedness conditions on semitopologies translate to algebraic structure in the world of semiframes.

Studying the dualities between semitopologies and semiframes, and how well-behavedness properties on points translate (or do not translate!) to their algebraic duals, is the topic of the paper.

\begin{rmrk}[Semitopologies are compatible complete join-semilattices in sets]
\label{rmrk.amazing}
Something interesting will happen when we dualise semitopologies: we will find ourselves obliged to introduce the notion of a \emph{compatibility relation} $\ast$ (Definition~\ref{defn.compatibility.relation}).
This arises naturally twice:
\begin{itemize*}
\item
as a key technicality in our duality result, and also 
\item
as an algebraic analogue of the intertwinedness property $\between$ which we need to express well-behavedness properties such as regularity (Section~\ref{sect.regular.points}).
\end{itemize*}
It is interesting that $\ast$/$\between$ appears for two reasons that on the face of things have quite distinct origins: algebraically in the duality; and as a pragmatically motivated well-behavedness conditions on sets.

So a message of our study of duality is to corroborate that intertwinedness is both mathematically fundamental and practically useful, and that (as noted above) semitopologies are not just complete join semilattices in sets, but something a bit richer; namely \emph{compatible} complete join-semilattices in sets.
\end{rmrk}

\subsection{Map of the paper}

\begin{enumerate}
\item
Section~\ref{sect.intro} is the Introduction.  You Are Here.
\item
In Section~\ref{sect.semitopology} we show how \textbf{continuity corresponds to local agreement} (Definition~\ref{defn.semitopology} and Lemma~\ref{lemm.open.lc}).
\item
In Section~\ref{sect.transitive.sets} we discuss \textbf{transitive sets}, \textbf{topens}, and \textbf{intertwined points}.
These are all anti-separation well-behavedness properties; e.g. Transitive and topen sets are guaranteed to be in agreement, in a sense made precise in Theorem~\ref{thrm.correlated} and Corollary~\ref{corr.correlated.intersect}.
\item
In Section~\ref{sect.regular.points} we classify points in more detail, introducing notions of \textbf{regular}, \textbf{weakly regular}, and \textbf{quasiregular} points (Definition~\ref{defn.tn}).\footnote{The other main classification is \textbf{conflicted} points, in Definition~\ref{defn.conflicted}.  These properties are connected by an equation: regular = weakly regular + conflicted; see Theorem~\ref{thrm.r=wr+uc}.}
 
Regular points are points contained in a topen set, and as per Theorem~\ref{thrm.correlated} they display good behaviour.
\item
In Section~\ref{sect.closed.sets} we characterise regularity as \textbf{regular = weakly regular + unconflicted}: see Theorem~\ref{thrm.r=wr+uc}. 
\item
In Section~\ref{sect.r=qr+ht} we characterise regularity as \textbf{regular = quasiregular + hypertransitive}: see Theorem~\ref{thrm.regular=qr+sc}. 

This completes our treatment of point-set semitopologies, and next we dualise, as follows: 
\item
In Section~\ref{sect.semiframes} we introduce \textbf{semiframes}.
These are the algebraic version of semitopologies, and they are to semitopologies as frames are to topologies.

We discover that semiframes are not just join-semilattices; \textbf{semiframes are \emph{compatible} semilattices}, which include a \emph{compatibility relation} $\ast$ to abstract the property of sets intersection $\between$ (see Remark~\ref{rmrk.amazing}).
\item
In Section~\ref{sect.semifilters.and.points} we introduce \textbf{semifilters}.
These play a similar role as filters do in topologies, except that semifilters have a \textbf{compatibility condition} instead of closure under finite meets.
We develop the notion of abstract points (completely prime semifilters), and show how to build a semitopology out of the abstract points of a semiframe. 
\item
In Section~\ref{sect.spatial.and.sober} we introduce \textbf{sober semitopologies} and \textbf{spatial semiframes}.
The reader familiar with categorical duality may be familiar with these conditions, though some details are significantly different from the topological case (see for instance the discussion in Subsection~\ref{subsect.sober.top.contrast}).
\item
In Section~\ref{sect.duality} we consider the \textbf{duality} between suitable categories of (sober) semitopologies and (spatial) semiframes.
\item
In Section~\ref{sect.closer.look.at.semifilters} we \textbf{dualise the well-behavedness conditions} from Section~\ref{sect.transitive.sets} to algebraic versions.
The correspondence is good (Proposition~\ref{prop.regular.match.up}) but also imperfect in some interesting ways (Remark~\ref{rmrk.summary.of.sc}).
\item
In Section~\ref{sect.conclusions} we conclude and discuss related and future work.
\end{enumerate}

\subsection{Comments on the context of this research}

\begin{rmrk}[Algebraic topology $\neq$ semitopology and semiframes]
Algebraic topology has been applied to the solvability of distributed-computing tasks in computational models (e.g. the impossibility of $k$-set consensus and the Asynchronous Computability Theorem~\cite{herlihy_asynchronous_1993,borowsky_generalized_1993,saks_wait-free_1993}; see~\cite{herlihy_distributed_2013} for a survey).

This paper is not that!
Algebra and topology are versatile tools: this paper is algebraic and topological, but in different senses and to different ends. 
\end{rmrk}

\begin{rmrk}[Where the interesting properties are]
Topology often studies spaces with strong separability properties between points, like Hausdorff separability.
For our applications, the well-behavedness properties of semitopologies, and corresponding properties in semiframes, centre on clusters of points that \emph{cannot} be separated.

For example, we state and discuss an `anti-Hausdorff' anti-separation property which we call \emph{being intertwined} (see Definition~\ref{defn.intertwined.points} and Remark~\ref{rmrk.not.hausdorff}).
Within an intertwined set, continuity implies agreement in a particularly strong sense (see Corollary~\ref{corr.intertwined.correlated}).
This leads us to study classes of semitopologies and semiframes with various anti-separation well-behavedness conditions;
see most notably \emph{regularity} (Sections~\ref{sect.regular.points} and~\ref{sect.closer.look.at.semifilters}).
\end{rmrk}

\begin{rmrk}
As the reader may know, \emph{frames} and \emph{locales} are the same thing: the category of locales is just the categorical opposite of the category of frames.
So every time we write `semiframe', the reader can safely read `semilocale'; these are two names for essentially the same structure up to reversing arrows.

The literature on frames and locales is huge, as indeed is the literature on topology.
Classic texts are~\cite{johnstone:stos,maclane:sheglf}.
More recent and very readable presentations are in \cite{picado:fraltw,picado:seppft}.
This literature is a rich source of ideas for things to do with semiframes, with respect to which we cannot possibly be comprehensive in this single paper: there are many things of interest that we simply have not done yet, or perhaps have not even (yet) realised could be done, and this is a feature, not a bug, since it reflects a wide scope for possible future research.

A partial list of possible future work is in Subsection~\ref{subsect.future.work}; and lists of properties and non-properties of semiframes/semitopologies vs. frames/topologies are in Subsections~\ref{subsect.familiar}, \ref{subsect.things.that.are.different}, \ref{subsect.sober.top.contrast}, and~\ref{subsect.vs}.
\end{rmrk} 

\begin{rmrk}
This paper expands on selected material from a technical monograph~\cite{gabbay:semdca},
and it continues where a journal paper on point-set topologies~\cite{gabbay:semtad} left off.

For this document to be meaningful and self-contained, we include some material on point-set semitopologies %
--- just as a treatment of frames would define point-set topologies. 
The material has been organised, edited, and streamlined specifically to provide a direct route to semiframes and their duality. 

Our goal for this paper is to offer an accessible, dedicated, and self-contained journal exposition on what is currently understood about semiframes and their relation to point-set semitopologies, in the context of current developments in decentralised systems.
\end{rmrk}

\begin{rmrk}[Who should read this paper?]
This work is aimed at the advanced researcher interested in a new kind of point-free topology, albeit one that is grounded in the practical needs of modern decentralised systems, and at the advanced practitioner interested in seeing what they may already be doing in a new (and more abstract) mathematical light.

I hope this paper will stimulate interest in compatible complete join-semilattices, and related structures, as an elegant abstraction of decentralised systems.
\end{rmrk}

\section{Semitopology}
\label{sect.semitopology}

\jamiesubsection{Definition and examples}

\subsubsection{Semitopologies are not topologies}

Recall from Definition~\ref{defn.semitopology} the definition of a semitopology.
As a sets structure, a semitopology on $\ns P$ is like a \emph{topology} on $\ns P$, but without the condition that the intersection of two open sets be an open set.
This gives semitopologies their own distinct character, and we can see this before we have even developed much of the theory, as we discuss in Remark~\ref{rmrk.closing.under.intersections} and Lemma~\ref{lemm.two.min}: 
 
\begin{rmrk}
\label{rmrk.closing.under.intersections}
Every semitopology $(\ns P,\opens)$ gives rise to a topology just by closing opens under intersections.
But semitopologies are far more than just subbases for a corresponding topology:
\begin{enumerate*}
\item
Completing a semitopology to a topology by closing under intersections, loses information.

For example: the `many', `all-but-one', and `more-than-one' semitopologies in Example~\ref{xmpl.semitopologies} express three distinct notions of quorum, yet all three yield the discrete semitopology (Definition~\ref{defn.value.assignment}) if we close under intersections and $\ns P$ is infinite.
See also the overview in Subsection~\ref{subsect.vs}. 
\item
We are \emph{explicitly interested} in situations where intersections of open sets need not be open, due to our motivating interpretation of open sets as actionable coalitions (as discussed in Section~\ref{sect.intro}).
\end{enumerate*}
\end{rmrk}

Much of the difference in flavour between topologies and semitopologies comes down to this:
\begin{lemm}
\label{lemm.two.min}
\leavevmode
\begin{enumerate*}
\item
In topologies, if a point $p$ has a minimal open neighbourhood then it is least (= unique minimal).
\item
In semitopologies, a point may have multiple distinct minimal open neighbourhoods.
\end{enumerate*}
\end{lemm}
\begin{proof}
To see that in a topology every minimal open neighbourhood is least, just note that if $p\in A$ and $p\in B$ then $p\in A\cap B$.
So if $A$ and $B$ are two minimal open neighbourhoods then $A\cap B$ is contained in both and by minimality is equal to both.

To see that in a semitopology a minimal open neighbourhood need not be least, it suffices to provide an example.
Consider $(\ns P,\opens)$ defined as follows, as illustrated in Figure~\ref{fig.two.min}:
\begin{itemize*}
\item
$\ns P=\{0,1,2\}$
\item
$\opens = \bigl\{ \varnothing,\ \{0,1\},\ \{1,2\},\ \{0,1,2\} \bigr\}$
\end{itemize*}
Note that $1$ has two minimal open neighbourhoods: $\{0,1\}$ and $\{1,2\}$. 
\end{proof}

\begin{figure}
\centering
\includegraphics[align=c,width=0.4\columnwidth,trim={50 120 50 120},clip]{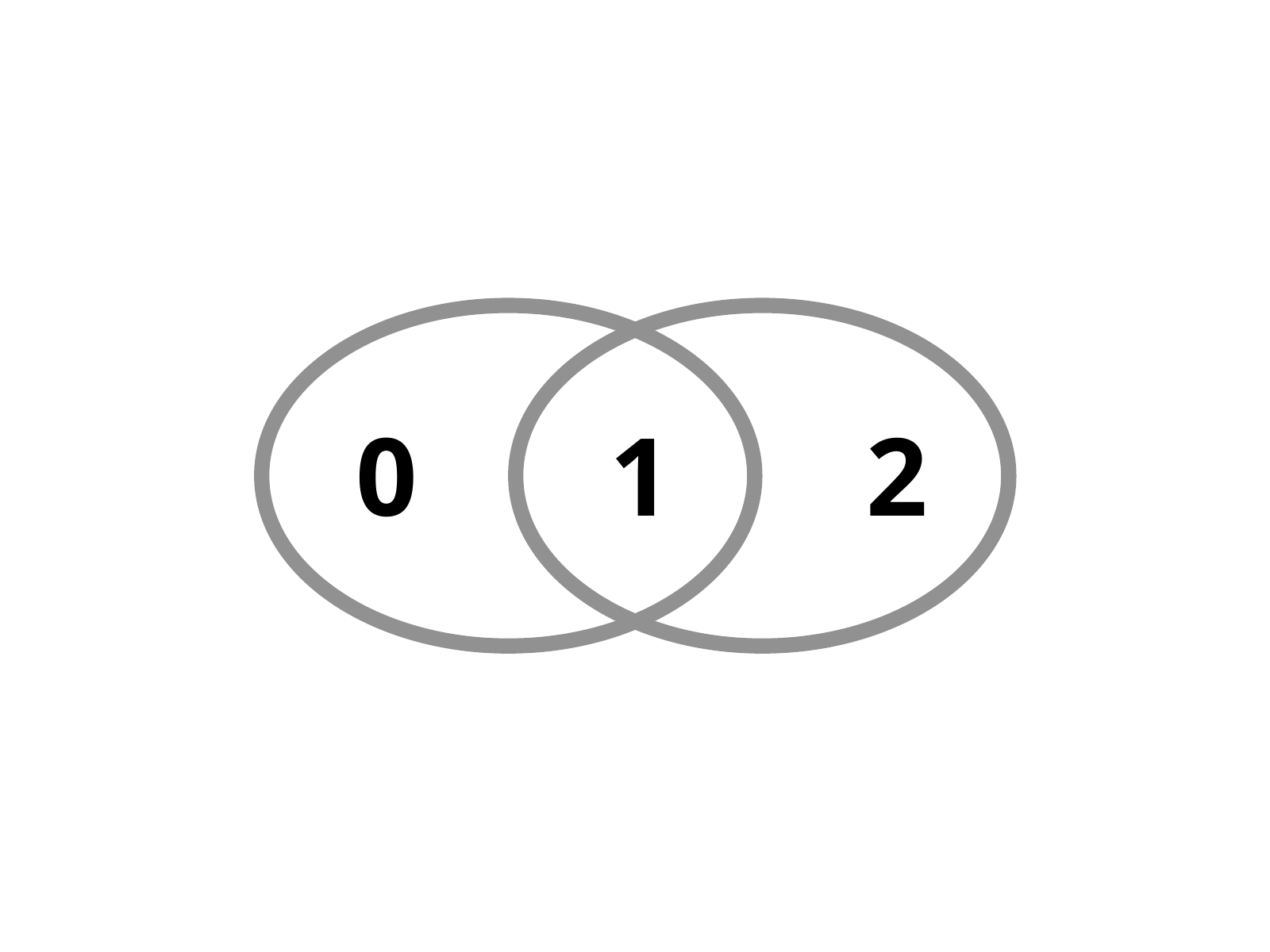}
\\[2ex]
\caption{An example of a point with two minimal open neighbourhoods (Lemma~\ref{lemm.two.min})}
\label{fig.two.min}
\end{figure}

\jamiesubsubsection{Examples}

As standard, we can make any set $\tf{Val}$ into a semitopology (indeed, it is also a topology) just by letting open sets be the powerset: 
\begin{defn}
\label{defn.value.assignment}
\leavevmode
\begin{enumerate*}
\item\label{item.discrete.semitopology}
Call $(\ns P,\powerset(\ns P))$ the \deffont{discrete semitopology on $\ns P$}.
 
We may call a set with the discrete semitopology a \deffont{semitopology of values}, and when we do we will usually call it $\tf{Val}$.
We may identify $\tf{Val}$-the-set and $\tf{Val}$-the-discrete-semitopology; meaning will always be clear.
\item\label{item.value.assignment}
When $(\ns P,\opens)$ is a semitopology and $\tf{Val}$ is a semitopology of values, we may call a function $f:\ns P\to\tf{Val}$ a \deffont[value assignment $f:\ns P\to\tf{Val}$]{value assignment}.

Note that a value just assigns values to points, and in particular we do not assume \emph{a priori} that it is continuous, where continuity is defined just as for topologies (see Definition~\ref{defn.continuity}).
\end{enumerate*} 
\end{defn}

\begin{xmpl}
\label{xmpl.semitopologies}
We consider further examples of semitopologies:
\begin{enumerate}
\item
Every topology is also a semitopology; intersections of open sets are allowed to be open in a semitopology, they are just not constrained to be open.
In particular, the discrete topology is also a discrete semitopology (Definition~\ref{defn.value.assignment}(\ref{item.discrete.semitopology})).
\item
The \deffont{initial semitopology} $(\varnothing,\{\varnothing\})$ and the \deffont{final semitopology} $(\{\ast\},\{\varnothing,\{\ast\}\})$ are semitopologies. 
\item\label{item.boolean.discrete}
An important discrete semitopological space is 
$$
\mathbb B=\{\bot,\top\}
\quad\text{with the discrete semitopology}\quad
\opens(\mathbb B)=\{\varnothing, \{\bot\},\{\top\},\{\bot,\top\}\}.
$$
We may silently treat $\mathbb B$ as a (discrete) semitopological space henceforth.
\item\label{item.trivial.topology}
Take $\ns P$ to be any nonempty set.
Let the \deffont[trivial semitopology]{trivial semitopology} (this is also a topology) on $\ns P$ have 
$$
\opens =\{\varnothing, \ns P\}.
$$
So (as usual) there are only two open sets: the one containing nothing, and the one containing every point.\footnote{According to Wikipedia, this space is also called \emph{indiscrete}, \emph{anti-discrete}, \emph{concrete}, and \emph{codiscrete} (\url{https://en.wikipedia.org/wiki/Trivial_topology}).}

The only nonempty open is $\ns P$ itself, reflecting a notion of actionable coalition that requires unanimous agreement. 
\item
Suppose $\ns P$ is a set and $\mathcal F\subseteq\powerset(\ns P)$ is nonempty and up-closed (so if $P\in\mathcal F$ and $P\subseteq P'\subseteq\ns P$ then $P'\in\mathcal F$, then $(\ns P,\mathcal F)$ is a semitopology.
This is not necessarily a topology, because we do not insist that $\mathcal F$ is a filter (i.e. is closed under intersections).

We give four sub-examples for different choices of $\mathcal P\subseteq\powerset(\ns P)$:
\begin{enumerate}
\item\label{item.supermajority}
Take $\ns P$ to be any finite nonempty set.
Let the \deffont{supermajority semitopology} have 
$$
\opens =\{\varnothing\}\cup\{O\subseteq\ns P \mid \f{cardinality}(O)\geq \nicefrac{2}{3}*\f{cardinality}(\ns P)\}.
$$
So $O$ is open when it contains at least two-thirds of the points.

Two-thirds is a typical threshold used for making progress in consensus algorithms.
\item
Take $\ns P$ to be any nonempty set.
Let the \deffont{many semitopology} have
$$
\opens = \{\varnothing\}\cup\{O\subseteq\ns P \mid \f{cardinality}(O)=\f{cardinality}(\ns P)\} .
$$
For example, if $\ns P=\mathbb N$ then open sets include $\f{evens}=\{2*n \mid n\in\mathbb N\}$ and $\f{odds}=\{2*n\plus 1 \mid n\in\mathbb N\}$.

Its notion of open set captures an idea that an actionable coalition is a set that may not be all of $\ns P$, but does at least biject with it.
\item\label{item.counterexample.X-x}
Take $\ns P$ to be any nonempty set.
Let the \deffont{all-but-one semitopology} have
$$
\opens = \{\varnothing,\ \ns P\}\cup\{\ns P\setminus \{p\}\mid p\in\ns P\} .
$$
This semitopology is not a topology.

The notion of actionable coalition here is that there may be at most one objector (but not two).
\item\label{item.counterexample.more-than-one}
Take $\ns P$ to be any set with cardinality at least $2$.
Let the \deffont{more-than-one semitopology} have
$$
\opens = \{\varnothing\}\cup\{O\subseteq\ns P \mid \f{cardinality}(O) \geq 2\} .
$$
This semitopology is not a topology.

This notion of actionable coalition reflects a security principle in banking and accounting (and elsewhere) of \emph{separation of duties}, that functional responsibilities be separated such that at least two people are required to complete an action --- so that errors (or worse) cannot be made without being discovered by another person.
\end{enumerate}
\item
Take $\ns P=\mathbb R$ (the set of real numbers) and let open sets be generated by intervals of the form $\rightopeninterval{0,r}$ or $\leftopeninterval{\minus r,0}$ for any strictly positive real number $r>0$.

This semitopology is not a topology, since (for example) $\leftopeninterval{1,0}$ and $\rightopeninterval{0,1}$ are open, but their intersection $\{0\}$ is not open.
\item\label{item.quorum.system}
In~\cite{naor:loacaq} a notion of \emph{quorum system} is discussed, defined as any collection of pairwise intersecting sets.
Quorum systems are key to the theory and practical implementation of consensus algorithms.

Every quorum system gives rise naturally to the least semitopology that contains it, just by closing under arbitrary unions.

To give one specific example of a quorum system from~\cite{naor:loacaq}, consider $n\times n$ grid of cells with quorums being sets consisting of any full row and a full column; note that any two quorums must intersect in at least two points.
We obtain a semitopology just by closing under arbitrary unions.
\end{enumerate}
\end{xmpl}

\jamiesubsection{Continuity, and its interpretation as agreement}
\label{subsect.continuity}

\begin{defn}
\label{defn.continuity}
We import standard topological notions of inverse image and continuity:
\begin{enumerate}
\item
Suppose $\ns P$ and $\ns P'$ are any sets and $f:\ns P\to\ns P'$ is a function.
Suppose $O'\subseteq\ns P'$.
Then write $f^\mone(O')$ for the \deffont[inverse image $f^\mone(O')$]{inverse image} or \deffont[preimage $f^\mone(O')$]{preimage} of $O'$, defined by
$$
f^\mone(O')=\{p{\in}\ns P \mid f(p)\in O'\} . 
$$
\item\label{item.continuous.function}
Suppose $(\ns P,\opens)$ and $(\ns P',\opens')$ are semitopological spaces (Definition~\ref{defn.semitopology}).
Call a function $f:\ns P\to\ns P'$ \deffont[continuous function]{continuous} when the inverse image of an open set is open.
In symbols:
$$
\Forall{O'\in\opens'} f^\mone(O')\oldin\opens .
$$
\item\label{item.continuous.function.at.p}
Call a function $f:\ns P\to\ns P'$ \deffont[continuous function at a point]{continuous at $p\in\ns P$} when
$$
\Forall{O'{\in}\opens'}f(p)\in O'\limp \Exists{O_{p,O'}{\in}\opens}p\in O_{p,O'}\land O_{p,O'}\subseteq f^\mone(O') .
$$
In words: $f$ is continuous at $p$ when the inverse image of every open neighbourhood of $f(p)$ contains an open neighbourhood of $p$.
\item
Call a function $f:\ns P\to\ns P'$ \deffont[continuous function on a set]{continuous on $P\subseteq\ns P$} when $f$ is continuous at every $p\in P$.
\end{enumerate}
\end{defn}

\begin{lemm}
\label{lemm.alternative.cont}
Suppose $(\ns P,\opens)$ and $(\ns P',\opens')$ are semitopological spaces (Definition~\ref{defn.semitopology}) and suppose $f:\ns P\to\ns P'$ is a function.
Then the following are equivalent:
\begin{enumerate*}
\item
$f$ is continuous (Definition~\ref{defn.continuity}(\ref{item.continuous.function})).
\item
$f$ is continuous at every $p\in\ns P$ (Definition~\ref{defn.continuity}(\ref{item.continuous.function.at.p})).
\end{enumerate*}
\end{lemm}
\begin{proof}
The top-down implication is immediate, taking $O=f^\mone(O')$.

For the bottom-up implication, given $p$ and an open neighbourhood $O'\ni f(p)$, we write
$$
O=\bigcup\{O_{p,O'}\in\opens \mid p\in\ns P,\ f(p)\in O'\}.
$$
Above, $O_{p,O'}$ is the open neighbourhood of $p$ in the preimage of $O'$, which we know exists by Definition~\ref{defn.continuity}(\ref{item.continuous.function.at.p}).

It is routine to check that $O= f^\mone(O')$, and since this is a union of open sets, it is open. 
\end{proof}

\begin{defn}
\label{defn.locally.constant}
Suppose that:
\begin{itemize*}
\item
$(\ns P,\opens)$ is a semitopology and 
\item
$\tf{Val}$ is a semitopology of values (Definition~\ref{defn.value.assignment}(\ref{item.discrete.semitopology})) and 
\item
$f:\ns P\to \tf{Val}$ is a value assignment (Definition~\ref{defn.value.assignment}(\ref{item.value.assignment}); an assignment of a value to each element in $\ns P$).
\end{itemize*}
Then:
\begin{enumerate*}
\item
Call $f$ \deffont[locally constant at a point]{locally constant at $p\in\ns P$} when there exists $p\in O_p\in\opens$ such that 
$$
\Forall{p'{\in}O_p}f(p)=f(p').
$$
So $f$ is locally constant at $p$ when it is constant on some open neighbourhood $O_p$ of $p$.
\item
Call $f$ \deffont[locally constant on a set]{locally constant} when it is locally constant at every $p\in\ns P$.
\end{enumerate*} 
\end{defn}

\begin{lemm}
\label{lemm.open.lc}
Suppose $(\ns P,\opens)$ is a semitopology and $\tf{Val}$ is a semitopology of values and $f:\ns P\to\tf{Val}$ is a value assignment.
Then the following are equivalent:
\begin{itemize*}
\item
$f$ is locally constant / locally constant at $p\in\ns P$ (Definition~\ref{defn.locally.constant}).
\item
$f$ is continuous / continuous at $p\in\ns P$ (Definition~\ref{defn.continuity}). 
\end{itemize*}
\end{lemm}
\begin{proof}
This is just by pushing around definitions, but we spell it out:
\begin{itemize}
\item
Suppose $f$ is continuous, consider $p\in\ns P$, and write $v=f(p)$.
By our assumptions we know that $f^\mone(v)$ is open, and $p\in f^\mone(v)$.
This is an open neighbourhood $O_p$ on which $f$ is constant, so we are done.
\item
Suppose $f$ is locally constant, consider $p\in\ns P$, and write $v=f(p)$.
By assumption we can find $p\in O_p\in\opens$ on which $f$ is constant, so that $O_p\subseteq f^\mone(v)$.
\qedhere\end{itemize}
\end{proof}

\begin{rmrk}[Continuity = agreement]
\label{rmrk.continuity=consensus}
Lemma~\ref{lemm.open.lc} tells us that
we can view the problem of obtaining agreement across an actionable coalition (as discussed in Section~\ref{sect.intro}) as being the same as obtaining a value assignment that is continuous (at least) on that coalition. 
\end{rmrk}

\jamiesubsection{Neighbourhoods of a point}

Definition~\ref{defn.open.neighbourhood} is standard from topology, and Lemma~\ref{lemm.open.is.open} is a (standard) characterisation of openness, which will be useful later: 
\begin{defn}
\label{defn.open.neighbourhood}
\label{defn.nbhd}
\label{defn.nbhd.system}
Suppose $(\ns P,\opens)$ is a semitopology and $p\in\ns P$ and $O\in\opens$.
Then:
\begin{enumerate*}
\item
call $O$ an \deffont{open neighbourhood} of $p$ when $p\in O$.
\item
Define $\nbhd(p)\subseteq\opens$ the \deffont{neighbourhood system} of $p$ by
$$
\nbhd(p)=\{O\in\opens \mid p\in \opens\} .
$$ 
\end{enumerate*}
\end{defn}

\begin{lemm}
\label{lemm.open.is.open}
Suppose $(\ns P,\opens)$ is a semitopology and suppose $P\subseteq\ns P$ is any set of points.
Then the following are equivalent:
\begin{itemize*}
\item
$P\in\opens$.
\item
Every point $p$ in $P$ has an open neighbourhood in $P$. 
\end{itemize*}
In symbols we can write:
$$
\Forall{p{\in}P}\Exists{O{\in}\opens}(p\in O\land O\subseteq P)
\quad\text{if and only if}\quad
P\in\opens
$$
\end{lemm}
\begin{proof}
If $P$ is open then $P$ itself is an open neighbourhood for every point that it contains. 

Conversely, if every $p\in P$ contains some open neighbourhood $p\in O_p \subseteq P$ then $P=\bigcup\{O_p\mid p\in P\}$ and this is open by condition~\ref{semitopology.unions} of Definition~\ref{defn.semitopology}.
\end{proof}

\begin{rmrk}
An initial inspiration for modelling collaborative action using semitopologies, came from noting that the standard topological property described above in Lemma~\ref{lemm.open.is.open}, corresponds to the \emph{quorum sharing} property in \cite[Property~1]{losa:stecbi}; the connection to topological ideas had not been noticed in~\cite{losa:stecbi}.
\end{rmrk}

\jamiesection{Transitive sets \& topens}
\label{sect.transitive.sets}

\jamiesubsection{Some background on sets intersection}

Some notation will be convenient:
\begin{nttn}
\label{nttn.between}
Suppose $X$, $Y$, and $Z$ are sets.
\begin{enumerate*}
\item\label{item.between}
Write 
$$
X\between Y
\quad\text{when}\quad 
X\cap Y\neq\varnothing.
$$
When $X\between Y$ holds then we say (as standard) that $X$ and $Y$ \deffont[intersecting sets $X\between Y$]{intersect}.\index{$X\between Y$ (intersection of sets)}
\item
We may chain the $\between$ notation, writing for example 
$$
X\between Y\between Z
\quad\text{for}\quad
X\between Y\ \land \  Y\between Z
$$
\item
We may write $X\notbetween Y$ for $\neg(X\between Y)$, thus $X\notbetween Y$ when $X\cap Y=\varnothing$.
\end{enumerate*}
\end{nttn}

\begin{rmrk}
\emph{Note on design in Notation~\ref{nttn.between}:}
It is uncontroversial that if $X\neq\varnothing$ and $Y\neq\varnothing$ then $X\between Y$ should hold precisely when $X\cap Y\neq\varnothing$ --- but there is an edge case! 
What truth-value should $X\between Y$ return when $X$ or $Y$ is empty?
\begin{enumerate*}
\item
It might be nice if $X\subseteq Y$ would imply $X\between Y$.
This argues for setting 
$$
(X=\varnothing\lor Y=\varnothing)\limp X\between Y .
$$
\item
It might be nice if $X\between Y$ were monotone on both arguments (i.e. if $X\between Y$ and $X\subseteq X'$ then $X'\between Y$).
This argues for setting 
$$
(X=\varnothing\lor Y=\varnothing)\limp X\notbetween Y .
$$
\item
It might be nice if $X\between X$ always --- after all, should a set \emph{not} intersect itself? --- and this argues for setting 
$$
\varnothing\between\varnothing ,
$$ 
even if we also set $\varnothing\notbetween Y$ for nonempty $Y$. 
\end{enumerate*}
All three choices are defensible, and they are consistent with the following nice property:
$$
X\between Y \limp (X\between X \lor Y\between Y) . 
$$
We choose the second --- if $X$ or $Y$ is empty then $X\notbetween Y$ --- because it gives the simplest definition that $X\between Y$ precisely when $X\cap Y\neq\varnothing$.
\end{rmrk}

We list some elementary properties of $\between$ from Notation~\ref{nttn.between}(\ref{item.between}):
\begin{lemm}
\label{lemm.between.elementary}
\leavevmode
\begin{enumerate*}
\item\label{item.between.nonempty}
$X\between X$ if and only if $X\neq\varnothing$.
\item\label{item.between.symmetric}
$X\between Y$ if and only if $Y\between X$.
\item\label{between.elementary.either.or}
$X\between (Y\cup Z)$ if and only if $(X\between Y) \lor (X\between Z)$.
\item\label{between.subset}
If $X\subseteq X'$ and $X\neq\varnothing$ then $X\between X'$.
\item\label{between.monotone}
Suppose $X\between Y$.
Then $X\subseteq X'$ implies $X'\between Y$, and $Y\subseteq Y'$ implies $X\between Y'$. 
\item\label{between.nonempty}
If $X\between Y$ then $X\neq\varnothing$ and $Y\neq\varnothing$.
\end{enumerate*}
\end{lemm}
\begin{proof}
By facts of sets intersection.
\end{proof}

\jamiesubsection{Transitive open sets and value assignments}

\begin{rmrk}[Taking stock of topens]
\label{rmrk.some.intuition.on.topens}
Transitive sets are of interest because values of continuous functions are strongly correlated on them.
This is Theorem~\ref{thrm.correlated}, especially part~2 of Theorem~\ref{thrm.correlated}.

A transitive \emph{open} set --- a \emph{topen} --- is even more important, because an open set corresponds in our semitopological model to a \emph{quorum} (a collection of participants that can make progress), so a transitive open set is a collection of participants that can make progress and are guaranteed to do so in consensus, where algorithms succeed.

For this and other reasons, we very much care about finding topens and understanding when points are associated with topen sets (e.g. by having topen neighbourhoods).
As we develop the maths, this will then lead us on to consider various regularity properties (Definition~\ref{defn.tn}).
But first, we start with transitive sets and topens: 
\end{rmrk}

\begin{defn}
\label{defn.transitive}
Suppose $(\ns P,\opens)$ is a semitopology.
Suppose $\atopen\subseteq\ns P$ is any set of points.
\begin{enumerate*}
\item\label{transitive.transitive}
Call $\atopen$ \deffont[transitive set]{transitive} when 
$$
\Forall{O,O'{\in}\opens} O\between \atopen \between O' \limp O\between O'. 
$$
\item\label{transitive.cc}
Call $\atopen$ \deffont[topen set]{topen} when $\atopen$ is nonempty transitive and open.\footnote{%
The empty set is trivially transitive and open, so it would make sense to admit it as a (degenerate) topen.  However, it turns out that we mostly need the notion of `topen' to refer to certain kinds of neighbourhoods of points (we will call them \emph{communities}; see Definition~\ref{defn.tn}).  It is therefore convenient to exclude the empty set from being topen, because while it is the neighbourhood of every point that it contains, it is not a neighbourhood of any point.} 

We may write 
$$
\topens=\{ \atopen\in\opens_{\neq\varnothing} \mid \atopen\text{ is transitive}\} .
$$
\item\label{transitive.max.cc}
Call $S$ a \deffont[maximal topen set]{maximal topen} when $S$ is a topen that is not a subset of any strictly larger topen.\footnote{`Transitive open' $\to$ `topen', like `closed and open' $\to$ `clopen'.
}
\end{enumerate*}
\end{defn}

Theorem~\ref{thrm.correlated} clarifies why transitivity is interesting: continuous value assignments are constant --- if we think of points as participants, `constant function' here means `in agreement' --- across transitive sets.
\begin{thrm}
\label{thrm.correlated}
Suppose that:
\begin{itemize*}
\item
$(\ns P,\opens)$ is a semitopology.
\item
$\tf{Val}$ is a semitopology of values (a nonempty set with the discrete semitopology; see Definition~\ref{defn.value.assignment}(\ref{item.discrete.semitopology})). 
\item
$f:\ns P\to\tf{Val}$ is a value assignment (Definition~\ref{defn.value.assignment}(\ref{item.value.assignment})). 
\item
$T\subseteq\ns P$ is a transitive set (Definition~\ref{defn.transitive}) --- in particular this will hold if $\atopen$ is topen --- and $p,p'\in T$.
\end{itemize*} 
Then:
\begin{enumerate*}
\item\label{item.correlated.1}
If $f$ is continuous at $p$ and $p'$ then $f(p)=f(p')$.
\item\label{item.correlated.2}
As a corollary, if $f$ is continuous on $\atopen$, then $f$ is constant on $\atopen$.
\end{enumerate*}
In words we can say: 
\begin{quoting}
Continuous value assignments are constant across transitive sets.
\end{quoting}
\end{thrm}
\begin{proof}
Part~\ref{item.correlated.2} follows from part~\ref{item.correlated.1} since if $f(p)=f(p')$ for \emph{any} $p,p'\in T$, then by definition $f$ is constant on $\atopen$.
So we now just need to prove part~\ref{item.correlated.1} of this result.

Consider $p,p'\in T$.
By continuity on $\atopen$, there exist open neighbourhoods $p\in O\subseteq f^\mone(f(p))$ and $p'\in O'\subseteq f^\mone(f(p'))$.
By construction $O\between \atopen \between O'$ (because $p\in O\cap T$ and $p'\in T\cap O'$).
By transitivity of $\atopen$ it follows that $O\between O'$. 
Thus, there exists $p''\in O\cap O'$, and by construction $f(p) = f(p'') = f(p')$.
\end{proof}

Corollary~\ref{corr.correlated.intersect} is an easy and useful consequence of Theorem~\ref{thrm.correlated}:
\begin{corr}
\label{corr.correlated.intersect}
Suppose that:
\begin{itemize*}
\item
$(\ns P,\opens)$ is a semitopology. 
\item
$f:\ns P\to \tf{Val}$ is a value assignment to some set of values $\tf{Val}$ (Definition~\ref{defn.value.assignment}). 
\item
$f$ is continuous on topen sets $\atopen, \atopen'\in\topens$.
\end{itemize*}
Then 
$$
\atopen\between \atopen'
\quad\text{implies}\quad 
\Forall{p\in\atopen,p'\in\atopen'} f(p)=f(p').
$$
\end{corr}
\begin{proof}
By Theorem~\ref{thrm.correlated} $f$ is constant on $\atopen$ and $\atopen'$.
We assumed that $\atopen$ and $\atopen'$ intersect, and the result follows.
\end{proof}

A converse to Theorem~\ref{thrm.correlated} also holds:
\begin{prop}
\label{prop.correlated.converse}
Suppose that:
\begin{itemize*}
\item
$(\ns P,\opens)$ is a semitopology.
\item
$\tf{Val}$ is a semitopology of values with at least two elements (to exclude a degenerate case that no functions exist, or they exist but there is only one because there is only one value to map to).
\item
$T\subseteq\ns P$ is any set. 
\end{itemize*} 
Then 
\begin{itemize*}
\item
\emph{if} for every $p,p'\in T$ and every value assignment $f:\ns P\to\tf{Val}$, $f$ continuous at $p$ and $p'$ implies $f(p)=f(p')$, 
\item
\emph{then} $\atopen$ is transitive.
\end{itemize*}
\end{prop}
\begin{proof}
We prove the contrapositive. 
Suppose $\atopen$ is not transitive, so there exist $O,O'\in\opens$ such that $O\between \atopen\between O'$ and yet $O\cap O'=\varnothing$.
We choose two distinct values $v\neq v'\in\tf{Val}$ and define $f$ to map any point in $O$ to $v$ and any point in $\ns P\setminus O$ to $v'$.

Choose some $p\in O$ and $p'\in O'$.
It does not matter which, and some such $p$ and $p'$ exist, because $O$ and $O'$ are nonempty by Lemma~\ref{lemm.between.elementary}(\ref{between.nonempty}), since $O\between\atopen$ and $O'\between\atopen$).

We note that $f(p)=v$ and $f(p')=v'$ and $f$ is continuous at $p\in O$ and $p'\in O'\subseteq\ns P\setminus O$, yet $f(p)\neq f(p')$.
\end{proof}

We can sum up what Theorem~\ref{thrm.correlated} and Proposition~\ref{prop.correlated.converse} mean, as follows:
\begin{rmrk}
\label{rmrk.transitive.correlated}
Suppose $(\ns P,\opens)$ is a semitopology and $\tf{Val}$ is a semitopology of values with at least two elements.
Say that a value assignment $f:\ns P\to\tf{Val}$ \deffont[splits (value assignment splits a set)]{splits} a set $T\subseteq\ns P$ when there exist $p,p'\in T$ such that $f$ is continuous at $p$ and $p'$ and $f(p)\neq f(p')$. 
Then Theorem~\ref{thrm.correlated} and Proposition~\ref{prop.correlated.converse} together say in words that: 
\begin{quoting}
$T\subseteq\ns P$ is transitive if and only if it cannot be split by a value assignment that is continuous on $T$. 
\end{quoting}
Intuitively, transitive sets characterise areas of guaranteed agreement.
\end{rmrk}

\jamiesubsection{Examples and discussion of transitive sets and topens}

We may routinely order sets by subset inclusion; including open sets, topens, closed sets, and so on, and we may talk about maximal, minimal, greatest, and least elements.
We include the (standard) definition for reference: 
\begin{nttn}
\label{nttn.min.max}
Suppose $(\ns P,\leq)$ is a poset.
Then:
\begin{enumerate*}
\item
Call $p\in\ns P$ \deffont[maximal element (in poset)]{maximal} when $\Forall{p'}p{\leq}p'\limp p'=p$ and \deffont[minimal element (in poset)]{minimal} when $\Forall{p'}p'{\leq}p\limp p'=p$.
\item
Call $p\in\ns P$ \deffont[greatest element (in poset)]{greatest} when $\Forall{p}p'\leq p$ and \deffont[least element (in poset)]{least} when $\Forall{p'}p\leq p'$.
\end{enumerate*}
\end{nttn}

\begin{xmpl}[Examples of transitive sets]
\label{xmpl.singleton.transitive}
\leavevmode
\begin{enumerate*}
\item\label{item.singleton.transitive}
$\{p\}$ is transitive, for any single point $p\in\ns P$. 
\item
The empty set $\varnothing$ is (trivially) transitive.
It is not topen because we insist in Definition~\ref{defn.transitive}(\ref{transitive.cc}) that topens are nonempty.
\item
Call a set $P\subseteq\ns P$ \emph{topologically indistinguishable} when (using Notation~\ref{nttn.between}) for every open set $O$, 
$$
P\between O\liff P\subseteq O .
$$ 
It is easy to check that if $P$ is topologically indistinguishable, then it is transitive.
\end{enumerate*} 
\end{xmpl}

\begin{xmpl}[Examples of topens]
\label{xmpl.cc}
\leavevmode
\begin{enumerate*}
\item\label{item.cc.two.regular}
Take $\ns P=\{0, 1, 2\}$, with open sets $\varnothing$, $\ns P$, $\{0\}$, and $\{2\}$. 
This has two maximal topens $\{0\}$ and $\{2\}$  as illustrated in Figure~\ref{fig.012} (top-left diagram). 
\item\label{item.cc.two.regular.b}
Take $\ns P=\{0, 1, 2\}$, with open sets $\varnothing$, $\ns P$, $\{0\}$, $\{0, 1\}$, $\{2\}$, $\{1,2\}$, and $\{0,2\}$. 
This has two maximal topens $\{0\}$ and $\{2\}$, as illustrated in Figure~\ref{fig.012} (top-right diagram). 
\item\label{item.xmpl.cc.3}
Take $\ns P=\{0,1,2,3,4\}$, with open sets generated by $\{0, 1\}$, $\{1\}$, $\{3\}$, and $\{3,4\}$.
This has two maximal topens $\{0,1\}$ and $\{3,4\}$, as illustrated in Figure~\ref{fig.012} (lower-left diagram). 
\item\label{item.xmpl.cc.4}
Take $\ns P=\{0,1,2,\ast\}$, with open sets generated by $\{0\}$, $\{1\}$, $\{2\}$, $\{0, 1,\ast\}$, and $\{1,2,\ast\}$.
This has three maximal topens $\{0\}$, $\{1\}$, and $\{2\}$, as illustrated in Figure~\ref{fig.012} (lower-right diagram). 
\item
Take the all-but-one semitopology from Example~\ref{xmpl.semitopologies}(\ref{item.counterexample.X-x}) on $\mathbb N$: so $\ns P=\mathbb N$ with opens $\varnothing$, $\mathbb N$, and $\mathbb N\setminus \{x\}$ for every $x\in\mathbb N$.
This has a single maximal topen $\mathbb N$.
\item
The semitopology in Figure~\ref{fig.square.diagram} has no topen sets at all ($\varnothing$ is transitive and open, but by definition in Definition~\ref{defn.transitive}(\ref{transitive.cc}) topens have to be nonempty).
\end{enumerate*}
\end{xmpl}

\begin{figure}
\centering
\includegraphics[align=c,width=0.4\columnwidth,trim={50 60 50 120},clip]{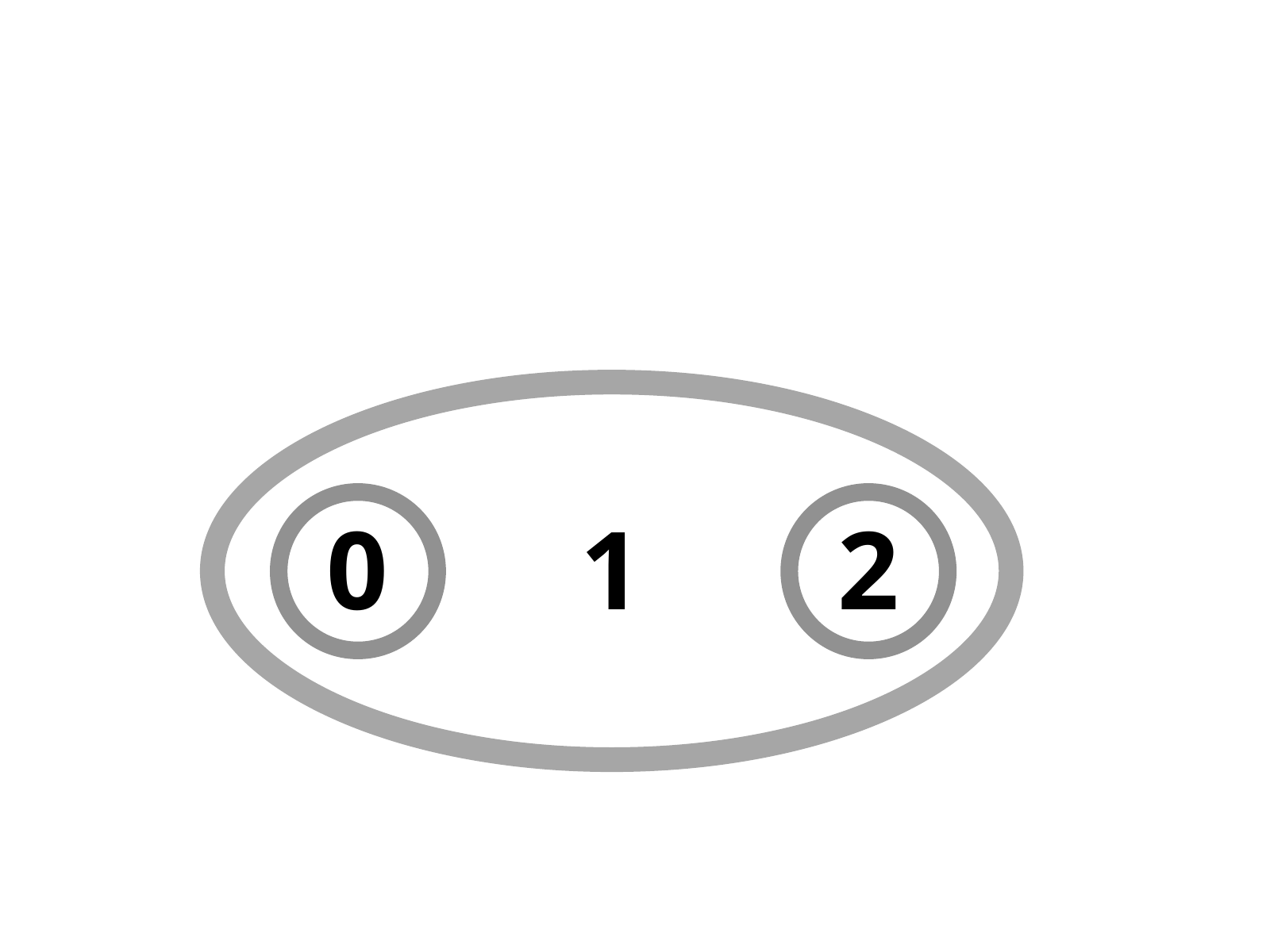}
\includegraphics[align=c,width=0.4\columnwidth,trim={50 60 50 120},clip]{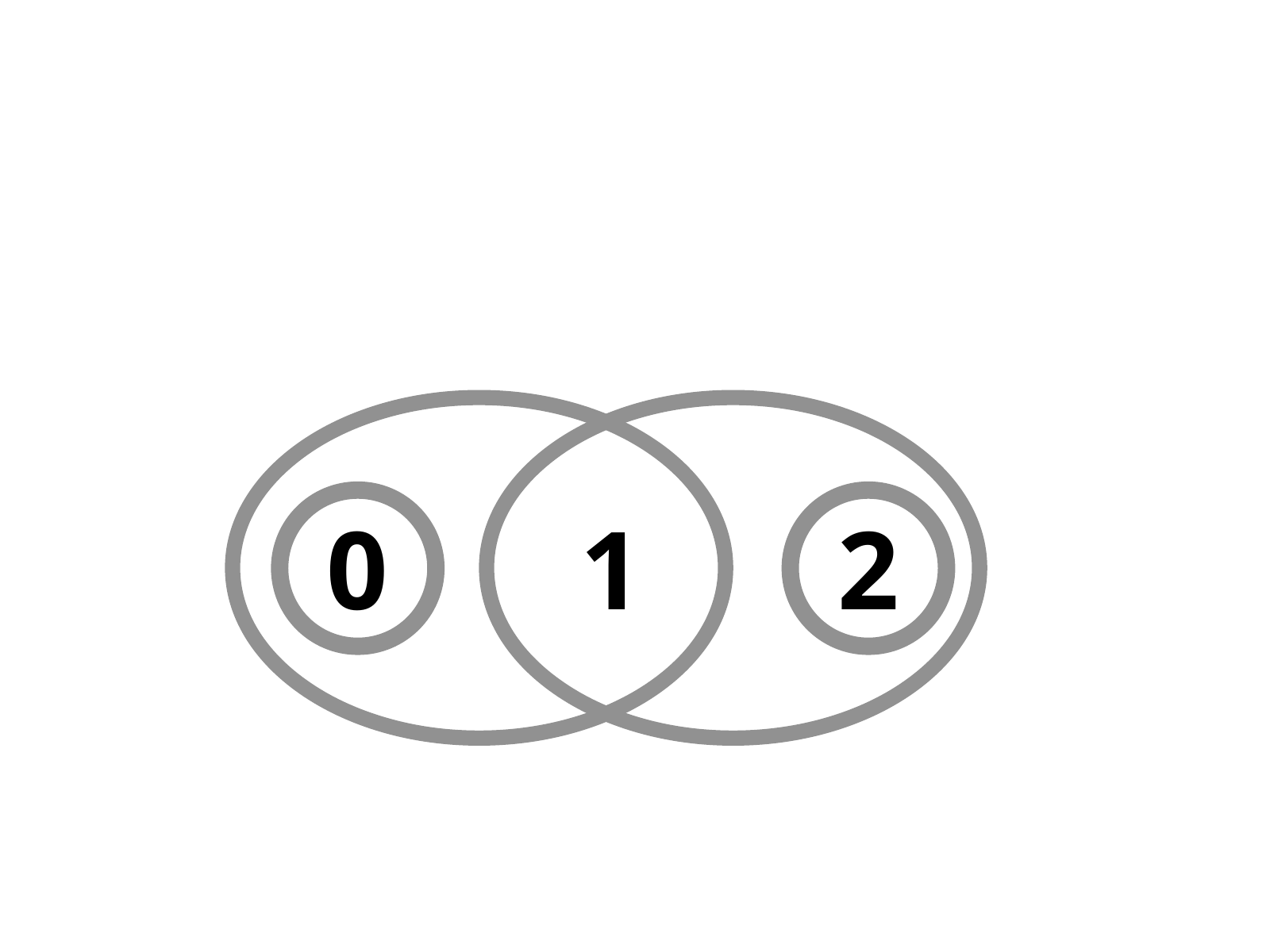}
\\
\includegraphics[align=c,width=0.35\columnwidth,trim={20 20 20 20},clip]{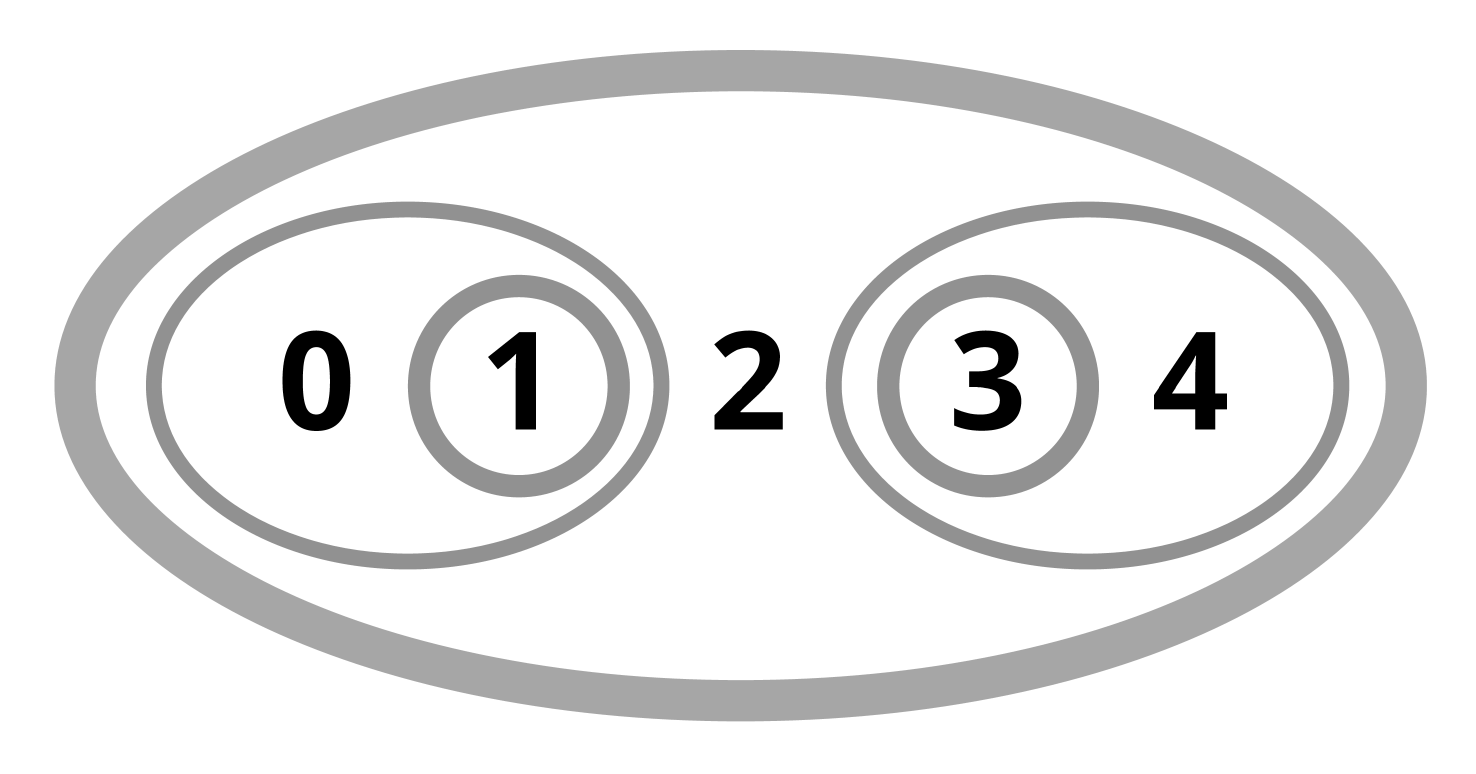}
\quad\  
\includegraphics[align=c,width=0.35\columnwidth,trim={50 20 50 20},clip]{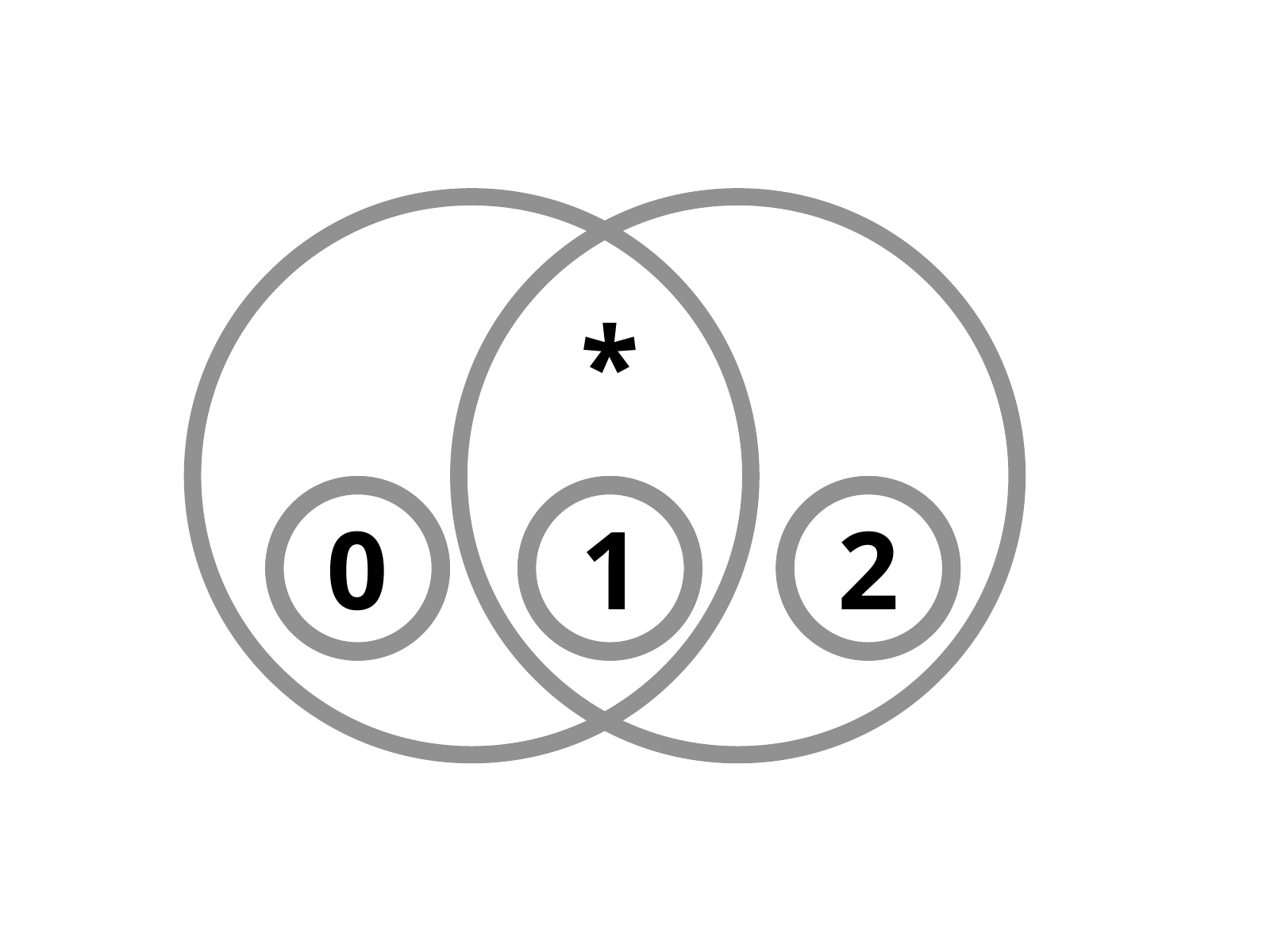}
\\
\includegraphics[width=0.35\columnwidth,trim={50 50 50 20},clip]{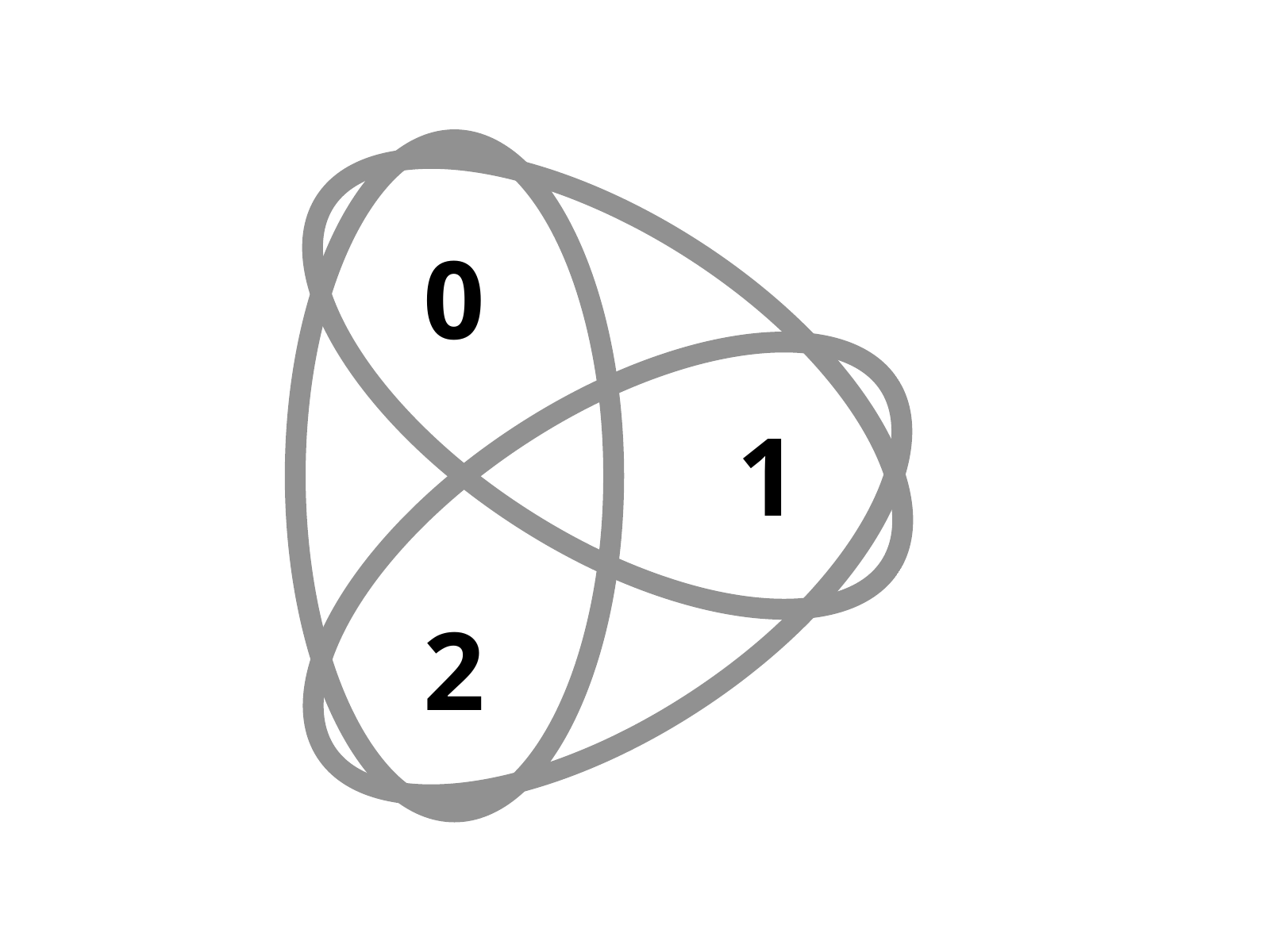}

\begin{flushleft}
\noindent\emph{Here and elsewhere, we might omit open sets that are unions of open sets that are illustrated.  
For example, we explicitly draw the universal open set in the left-hand diagrams above, but not in the right-hand and bottom diagrams above.
Meaning is clear and we get cleaner diagrams.
}
\end{flushleft}
\ \\[2ex]
\caption{Examples of topens (Example~\ref{xmpl.cc})}
\label{fig.not-strong-topen}
\label{fig.012}
\end{figure}

\jamiesubsection{Closure properties of transitive sets and topens}
\label{subsect.closure.properties.of.tt}

\begin{rmrk}
Transitive sets have some nice closure properties which we treat in this Subsection --- here we mean `closure' in the sense of ``the set of transitive sets is closed under various operations'', and not in the topological sense of `closed sets'.
Topens --- nonempty transitive \emph{open} sets --- have even better closure properties, which emanate from the requirement in Lemma~\ref{lemm.transitive.transitive} that at least one of the transitive sets $\atopen$ or $\atopen'$ is open. 
\end{rmrk}

\begin{lemm}
\label{lemm.transitive.subset}
Suppose $(\ns P,\opens)$ is a semitopology and $\atopen\subseteq \ns P$. 
Then:
\begin{enumerate*}
\item\label{item.transitive.subset.1}
If $\atopen$ is transitive and $\atopen'\subseteq \atopen$, then $\atopen'$ is transitive.
\item\label{item.transitive.subset.2}
If $\atopen$ is topen and $\varnothing\neq \atopen'\subseteq \atopen$ is nonempty and open, then $\atopen'$ is topen.
\end{enumerate*}
\end{lemm}
\begin{proof}
\leavevmode
\begin{enumerate}
\item
By Definition~\ref{defn.transitive} it suffices to consider open sets $O$ and $O'$ such that $O\between \atopen'\between O'$, and prove that $O\between O'$.
But this is simple: by Lemma~\ref{lemm.between.elementary}(\ref{between.monotone}) $O\between \atopen\between O'$, so $O\between O'$ follows by transitivity of $\atopen$. 
\item
Direct from part~\ref{item.transitive.subset.1} of this result and Definition~\ref{defn.transitive}(\ref{transitive.cc}).
\qedhere\end{enumerate}
\end{proof}

\begin{lemm}
\label{lemm.transitive.transitive}
Suppose 
$(\ns P,\opens)$ is a semitopology and $\atopen,\atopen'\subseteq\ns P$ are transitive, and suppose that at least one of $\atopen$ and $\atopen'$ is open.
Then
$$
\Forall{O,O'\in\opens}O\between \atopen \between \atopen'\between O' \limp O\between O'. 
$$
\end{lemm}
\begin{proof}
We simplify using Definition~\ref{defn.transitive} and our assumption that one of $\atopen$ and $\atopen'$ is open.
We consider the case that $\atopen'$ is open: 
$$
\begin{array}{r@{\ }l@{\qquad}l}
O\between \atopen\between \atopen'\between O'
\limp&
O\between \atopen' \between O'
&\text{$\atopen$ transitive, $\atopen'$ open}
\\
\limp&
O\between O'
&\text{$\atopen'$ transitive}.
\end{array}
$$
The argument for when $\atopen$ is open, is precisely similar.
\end{proof}

\begin{prop}
\label{prop.cc.unions}
Suppose $(\ns P,\opens)$ is a semitopology.
Then if 
$\mathcal \atopen$ is a set of pairwise intersecting topens, then $\bigcup\mathcal \atopen$ is topen. 
\end{prop}
\begin{proof}
$\bigcup\mathcal \atopen$ is open by Definition~\ref{defn.semitopology}(\ref{semitopology.unions}).
Also, if $O\between\bigcup\mathcal \atopen\between O'$ then there exist $\atopen,\atopen'\in\mathcal \atopen$ such that $O\between \atopen$ and $\atopen'\between O'$.
We assumed $\atopen\between \atopen'$, so by Lemma~\ref{lemm.transitive.transitive} (since $\atopen$ and $\atopen'$ are open) we have $O\between O'$ as required. 
\end{proof}

\begin{corr}
\label{corr.max.cc}
Suppose $(\ns P,\opens)$ is a semitopology.
Then every topen $\atopen$ is contained in a unique maximal topen.
\end{corr}
\begin{proof}
Consider $\mathcal \atopen = \{\atopen\cup \atopen' \mid \atopen'\text{ topen}\land \atopen\between \atopen'\}$.
By Proposition~\ref{prop.cc.unions} this is a set of topens.
By construction they all contain $\atopen$, so by our assumption that $\atopen\neq\varnothing$, they pairwise intersect, and by Proposition~\ref{prop.cc.unions} again $\bigcup\mathcal \atopen$ is topen.
It is easy to check that this is the unique maximal topen that contains $\atopen$. 
\end{proof}

\jamiesubsection{Intertwined points} 
\label{subsect.intertwined.points}

\jamiesubsubsection{The basic definition}

\begin{defn}
\label{defn.intertwined.points}
Suppose $(\ns P,\opens)$ is a semitopology and $p,p'\in\ns P$.
\begin{enumerate*}
\item\label{item.p.intertwinedwith.p'}
Call $p$ and $p'$ \deffont[intertwined (two points $p\intertwinedwith p'$)]{intertwined} when $\{p,p'\}$ is transitive.\index{$p\intertwinedwith p'$ (two intertwined points)}
Unpacking Definition~\ref{defn.transitive} this means:
$$
\Forall{O,O'{\in}\opens} (p\in O\land p'\in O') \limp O\between O' .
$$ 
By a mild abuse of notation, write 
$$
p\intertwinedwith p' \quad \text{when}\quad \text{$p$ and $p'$ are intertwined}.
$$
\item\label{intertwined.defn}
Define $\intertwined{p}$\index{intertwined of $p$ ($\intertwined{p}$)}\index{$\intertwined{p}$ (points intertwined with a point $p$)} (read `intertwined of $p$') to be the set of points intertwined with $p$.
In symbols: 
$$
\intertwined{p}=\{p'\in\ns P \mid p\intertwinedwith p'\} .
$$
\end{enumerate*}
\end{defn}

\begin{xmpl}
\label{xmpl.how.different?}
We return to the examples in Example~\ref{xmpl.cc}.  
There we note that:
\begin{enumerate*}
\item
$\intertwined{1}=\{0,1,2\}$ and $\intertwined{0}=\{0,1\}$ and $\intertwined{2}=\{1,2\}$.
\item
$\intertwined{1}=\{1\}$ and $\intertwined{0}=\{0\}$ and $\intertwined{2}=\{2\}$.
\item
$\intertwined{x}=\ns P$ for every $x$. 
\end{enumerate*}
\end{xmpl}

\begin{lemm}
\label{lemm.intertwined.not.transitive}
Suppose $(\ns P,\opens)$ is a semitopology.
Then the `is intertwined' relation $\between$ is not necessarily transitive.
That is: $p\intertwinedwith p'\intertwinedwith p''$ does not necessarily imply $p\intertwinedwith p''$.
\end{lemm}
\begin{proof}
It suffices to provide a counterexample.
The semitopology from Example~\ref{xmpl.cc}(\ref{item.cc.two.regular}) (illustrated in Figure~\ref{fig.012}, top-left diagram) will do.
So take 
$\ns P=\{0,1,2\}$ and
$\opens=\{\varnothing,\ns P,\{0\},\{2\}\}$.
Then 
$$
0\between 1
\ \ \text{and}\ \ 1\between 2,
\quad\text{but}\quad
\neg(0\between 2).
\qedhere$$
\end{proof}

\jamiesubsubsection{Pointwise characterisation of transitive sets}

\begin{lemm}
\label{lemm.three.transitive}
Suppose $(\ns P,\opens)$ is a semitopology and $\atopen\subseteq\ns P$.
Then the following are equivalent:
\begin{enumerate*}
\item\label{item.three.transitive.1}
$\atopen$ is transitive.
\item\label{item.three.transitive.2}
$p\intertwinedwith p'$ (meaning by Definition~\ref{defn.intertwined.points} that $\{p,p'\}$ is transitive)
for every $p,p'\in \atopen$.
\end{enumerate*}
\end{lemm}
\begin{proof}
Suppose $\atopen$ is transitive.
Then by Lemma~\ref{lemm.transitive.subset}(\ref{item.transitive.subset.1}), $\{p,p'\}$ is transitive for every $p,p'\in \atopen$.

Suppose $\{p,p'\}$ is transitive for every $p,p'\in \atopen$.
Consider open sets $O$ and $O'$ such that $O\between \atopen\between O'$. 
Choose $p\in O\cap \atopen$ and $p'\in O\cap \atopen'$.
By construction $\{p,p'\}\subseteq \atopen$ so this is transitive.
It follows that $O\between O'$ as required.
\end{proof}

\begin{thrm}
\label{thrm.cc.char}
Suppose $(\ns P,\opens)$ is a semitopology and $\atopen\subseteq\ns P$.
Then the following are equivalent:
\begin{enumerate*}
\item
$\atopen$ is topen.
\item
$\atopen\in\opens_{\neq\varnothing}$ and $\Forall{p,p'{\in}\atopen}p\intertwinedwith p'$.
\end{enumerate*}
In words we can say:
\begin{quoting}
A topen is a nonempty open set of intertwined points.
\end{quoting}
\end{thrm}
\begin{proof}
By Definition~\ref{defn.transitive}(\ref{transitive.cc}), $\atopen$ is topen when it is nonempty, open, and transitive. 
By Lemma~\ref{lemm.three.transitive} this last condition is equivalent to $p\intertwinedwith p'$ for every $p,p'\in \atopen$. 
\end{proof}

A value assignment is constant on a pair of intertwined points, where it is continuous:
\begin{corr}
\label{corr.intertwined.correlated}
Suppose $\tf{Val}$ is a semitopology of values and $f:\ns P\to\tf{Val}$ is a value assignment (Definition~\ref{defn.value.assignment})
and $p,p'\in\ns P$ and $p\between p'$.
Then if $f$ continuous at $p$ and $p'$ then $f(p)=f(p')$.
\end{corr}
\begin{proof}
$\{p,p'\}$ is transitive by Theorem~\ref{thrm.cc.char};
we use Theorem~\ref{thrm.correlated}.
\end{proof}

\begin{rmrk}[Intertwined as `non-Hausdorff']
\label{rmrk.not.hausdorff}
\leavevmode
\\
\noindent Recall that we call a topological space $(\ns P,\opens)$ \deffont[Hausdorff space]{Hausdorff} (or \deffont[$T_2$ space (Hausdorff condition)]{$T_2$}) when any two points can be separated by pairwise disjoint open sets.
Using the $\between$ symbol from Notation~\ref{nttn.between}, we rephrase the Hausdorff condition as
$$
\Forall{p,p'}p\neq p'\limp \Exists{O,O'}(p\in O\land p'\in O'\land \neg (O\between O')) , 
$$
we can simplify to 
$$
\Forall{p,p'}p\neq p'\limp p\notintertwinedwith p' ,
$$
and thus we simplify the Hausdorff condition just to
\begin{equation}
\label{eq.hausdorff}
\Forall{p}\intertwined{p}=\{p\}.
\end{equation}
Note how distinct $p$ and $p'$ being intertwined is the \emph{opposite} of being Hausdorff: $p\intertwinedwith p'$ when $p'\in\intertwined{p}$, and they \emph{cannot} be separated by pairwise disjoint open sets.
Thus the assertion $p\intertwinedwith p'$ in Theorem~\ref{thrm.cc.char} is a negation to the Hausdorff property:
$
\Exists{p}\intertwined{p}\neq\{p\} .
$
\end{rmrk}

\jamiesection{Interiors, communities \& regular points}
\label{sect.regular.points}

\jamiesubsection{Community of a (regular) point}

Definition~\ref{defn.interior} is standard:
\begin{defn}[Open interior]
\label{defn.interior}
Suppose $(\ns P,\opens)$ is a semitopology and $P\subseteq\ns P$.
Define $\interior(P)$ the \deffont{(open) interior of $P$}\index{$\interior(P)$ (open interior)} by
$$
\interior(P)=\bigcup\{ O\in\opens \mid O\subseteq P\} .
$$
\end{defn}

\begin{lemm}
\label{lemm.interior.open}
Suppose $(\ns P,\opens)$ is a semitopology and $P\subseteq\ns P$.
Then $\interior(P)$ from Definition~\ref{defn.interior} is the greatest open subset of $P$.
\end{lemm}
\begin{proof}
Routine by the construction in Definition~\ref{defn.interior} and closure of open sets under unions (Definition~\ref{defn.semitopology}(\ref{semitopology.unions})).
\end{proof}

\begin{corr}
\label{corr.interior.monotone}
Suppose $(\ns P,\opens)$ is a semitopology and $P,P'\subseteq\ns P$.
Then if $P\subseteq P'$ then $\interior(P)\subseteq\interior(P')$.
\end{corr}
\begin{proof}
Routine using Lemma~\ref{lemm.interior.open}.
\end{proof}

\begin{defn}[Community of a point, and regularity]
\label{defn.tn}
Suppose $(\ns P,\opens)$ is a semitopology and $p\in\ns P$.
Then:
\begin{enumerate*}
\item\label{item.tn}
Define $\community(p)$ the \deffont[community of $p$ ($\community(p)$)]{community of $p$}\index{$\community(p)$ (community of a point)} by 
$$
\community(p)=\interior(\intertwined{p}) .
$$
The community of $p$ is always an open set by Lemma~\ref{lemm.interior.open}.
\item\label{item.community.P}
Extend $\community$ to subsets $P\subseteq\ns P$ by taking a sets union:
$$
\community(P) = \bigcup\{\community(p) \mid p\in P\} .
$$
\item\label{item.regular.point}
Call $p$ a \deffont{regular point} when its community is a topen neighbourhood of $p$.
In symbols:
$$
p\text{ is regular}\quad\text{when}\quad p\in\community(p)\in\topens .
$$
\item\label{item.weakly.regular.point}
Call $p$ a \deffont{weakly regular point} when its community is an open (but not necessarily topen) neighbourhood of $p$.
In symbols:
$$
p\text{ is weakly regular}\quad\text{when}\quad p\in\community(p)\in\opens .
$$
\item\label{item.quasiregular.point}
Call $p$ a \deffont{quasiregular point} when its community is nonempty.
In symbols:
$$
p\text{ is quasiregular}\quad\text{when}\quad \varnothing\neq\community(p)\in\opens .
$$
\item\label{item.irregular.point}
If $p$ is not regular then we may call it an \deffont{irregular point}, or just say that it is not regular.
\item\label{item.regular.S}
If $P\subseteq\ns P$ and every $p\in P$ is regular/weakly regular/quasiregular/irregular then we may call $P$ a \deffont{regular/weakly regular/quasiregular/irregular set} respectively (see also Definition~\ref{defn.conflicted}(\ref{item.unconflicted})).
\qedhere\end{enumerate*}
\end{defn}

\begin{rmrk}
The notion of \emph{regular point} in Definition~\ref{defn.tn} is a key well-behavedness property.
Let's remember why it matters:

A topen set is a transitive open set.
We care about transitivity because it implies agreement, as per Theorem~\ref{thrm.correlated} (continuous value assignments are constant on transitive sets).
We care about being open, because we understand this as `being actionable'.
Thus, a regular point is interesting because it is a participant in a topen and thus is capable of safely making progress in algorithms we write on top of the underlying semitopology. 
For convenient reference, the semiframe characterisation of regularity is in Subsection~\ref{subsect.semiframe.regularity}.
\end{rmrk} 

Lemma~\ref{lemm.wr.r} gives an initial overview of the relationships between the properties in Definition~\ref{defn.tn}.
A more detailed treatment follows, which repeats these main points and expands on them and puts them in a detailed context. 
\begin{lemm}
\label{lemm.wr.r}
Suppose $(\ns P,\opens)$ is a semitopology and $p\in\ns P$.
Then:
\begin{enumerate*}
\item\label{item.r.implies.wr}
If $p$ is regular, then $p$ is weakly regular.
\item\label{item.wr.implies.qr}
If $p$ is weakly regular, then $p$ is quasiregular.
\item\label{item.wr.r.no.converse}
The converse implications need not hold. %
\item\label{item.wr.r.not.quasiregular}
Furthermore, it is possible for a point $p$ to not be quasiregular.
\end{enumerate*}
\end{lemm}
\begin{proof}
We consider each part in turn:
\begin{enumerate}
\item
If $p$ is regular then by Definition~\ref{defn.tn}(\ref{item.regular.point}) $p\in\community(p)\in\topens$, so certainly $p\in\community(p)$ and by Definition~\ref{defn.tn}(\ref{item.weakly.regular.point}) $p$ is weakly regular.
\item
If $p$ is weakly regular then by Definition~\ref{defn.tn}(\ref{item.weakly.regular.point}) $p\in\community(p)\in\opens$, so certainly $\community(p)\neq\varnothing$ and by Definition~\ref{defn.tn}(\ref{item.quasiregular.point}) $p$ is quasiregular.
\item
To see that the converse implications need not hold, note that:
\begin{itemize*}
\item
Point $1$ in Example~\ref{xmpl.wr}(\ref{item.wr.2}) (illustrated in Figure~\ref{fig.012}, top-left diagram) is weakly regular ($\community(1)=\{0,1,2\}$) but not regular ($\community(1)$ is open but not topen).
\item
Point $\ast$ in Example~\ref{xmpl.wr}(\ref{item.qr.2}) (illustrated in Figure~\ref{fig.012}, lower-right diagram) is quasiregular ($\community(\ast)=\{1\}$ is nonempty but does not contain $\ast$).
\end{itemize*}
\item
To see that $p$ may not even be quasiregular, take $\ns P=\mathbb R$ (real numbers), with its usual topology (which is also a semitopology).
Then $\intertwined{x}=\{x\}$ and $\community(x)=\varnothing$ for every $x\in\mathbb R$.
\qedhere
\end{enumerate}
\end{proof}

\begin{xmpl}
\label{xmpl.wr}
\leavevmode
\begin{enumerate*}
\item
In Figure~\ref{fig.not-strong-topen} (bottom diagram),\ $0$, $1$, and $2$ are three intertwined points and the entire space $\{0,1,2\}$ consists of a single topen set.
It follows that $0$, $1$, and $2$ are all regular and their community is $\{0,1,2\}$.
\item\label{item.wr.2}
In Figure~\ref{fig.012} (top-left diagram),\ $0$ and $2$ are regular and $1$ is weakly regular but not regular ($1\in\community(1)=\{0,1,2\}$ but $\{0,1,2\}$ is not topen). 
\item\label{item.qr.2}
In Figure~\ref{fig.012} (lower-right diagram),\ $0$, $1$, and $2$ are regular and $\ast$ is quasiregular ($\community(\ast)=\{1\}$).
\item
In Figure~\ref{fig.012} (top-right diagram),\ $0$ and $2$ are regular and $1$ is neither regular, weakly regular, nor quasiregular ($\community(1)=\varnothing$).
\item
In a semitopology of values $(\tf{Val},\powerset(\tf{Val}))$ (Definition~\ref{defn.value.assignment}) every value $v\in\tf{Val}$ is regular, weakly regular, and unconflicted.
\item\label{item.wr.6}
In $\mathbb R$ with its usual topology (which is also a semitopology), every point is unconflicted because the topology is Hausdorff and by Equation~\ref{eq.hausdorff} in Remark~\ref{rmrk.not.hausdorff} this means precisely that $\intertwined{p}=\{p\}$ so $p$ is intertwined just with itself.
Furthermore $p$ is not (quasi/weakly)regular, because $\community(p)=\interior(\intertwined{p})=\varnothing$.
\end{enumerate*} 
\end{xmpl}

\subsection{Further exploration of (quasi-/weak) regularity and topen sets}

\begin{rmrk}
\label{rmrk.T0-T2}
Recall three common separation axioms from topology:
\begin{enumerate*}
\item
$T_0$: if $p_1\neq p_2$ then there exists some $O\in\opens$ such that $(p_1\in O)\xor (p_2\in O)$, where $\xor$ denotes \emph{exclusive or}.
\item
$T_1$: if $p_1\neq p_2$ then there exist $O_1,O_2\in\opens$ such that $p_i\in O_j \liff i=j$ for $i,j\in\{1,2\}$.
\item
$T_2$, or the \emph{Hausdorff condition}: if $p_1\neq p_2$ then there exist $O_1,O_2\in\opens$ such that $p_i\in O_j \liff i=j$ for $i,j\in\{1,2\}$, and $O_1\cap O_2=\varnothing$.
Cf. the discussion in Remark~\ref{rmrk.not.hausdorff}.
\end{enumerate*}
Even the weakest of the well-behavedness property for semitopologies that we consider in Definition~\ref{defn.tn} --- quasiregularity --- is in some sense strongly opposed to the space being Hausdorff/$T_2$ (though not to being $T_1$), as Lemma~\ref{lemm.quasiregular.hausdorff} makes formal.
\end{rmrk}

\begin{lemm}
\label{lemm.quasiregular.hausdorff}
\leavevmode
\begin{enumerate*}
\item
Every quasiregular Hausdorff semitopology is discrete.

In more detail: if $(\ns P,\opens)$ is a semitopology that is quasiregular (Definition~\ref{defn.tn}(\ref{item.quasiregular.point})) and Hausdorff (equation~\ref{eq.hausdorff} in Remark~\ref{rmrk.not.hausdorff}), then $\opens=\powerset(\ns P)$. 
\item
There exists a (quasi)regular $T_1$ semitopology that is not discrete.
\end{enumerate*} 
\end{lemm}
\begin{proof}
We consider each part in turn:
\begin{enumerate}
\item
By the Hausdorff property, $\intertwined{p}=\{p\}$.
By the quasiregularity property, $\community(p)\neq\varnothing$.
It follows that $\community(p)=\{p\}$.
But by construction in Definition~\ref{defn.tn}(\ref{item.tn}), $\community(p)$ is an open interior.
Thus $\{p\}\in\opens$.
The result follows.
\item
It suffices to provide an example.
We use the bottom semitopology in Figure~\ref{fig.not-strong-topen}.
Thus $\ns P=\{0,1,2\}$ and $\opens$ is generated by $\{0,1\}$, $\{1,2\}$, and $\{2,0\}$.
The reader can check that this is regular (since all three points are intertwined) and $T_1$. 
\qedhere\end{enumerate}
\end{proof}

So what is $\community(p)$?
We start by characterising $\community(p)$ as the \emph{greatest} topen neighbourhood of $p$, if this exists:
\begin{lemm}
\label{lemm.intertwined.is.the.greatest}
\label{lemm.max.cc.intertwined}
Suppose $(\ns P,\opens)$ is a semitopology and recall from Definition~\ref{defn.tn}(\ref{item.regular.point}) that $p$ is regular when $\community(p)$ is a topen neighbourhood of $p$.
\begin{enumerate*}
\item\label{item.intertwined.is.the.greatest.1}
If $\community(p)$ is a topen neighbourhood of $p$ (i.e. if $p$ is regular) then $\community(p)$ is a maximal topen.
\item\label{item.intertwined.is.the.greatest.2}
If $p\in \atopen\in\topens$ is a maximal topen neighbourhood of $p$ then $\atopen=\community(p)$.
\end{enumerate*}
\end{lemm}
\begin{proof}
\leavevmode
\begin{enumerate}
\item
Since $p$ is regular, by definition, $\community(p)$ is topen and is a neighbourhood of $p$.
It remains to show that $\community(p)$ is a maximal topen.

Suppose $\atopen$ is a topen neighbourhood of $p$; we wish to prove $\atopen\subseteq \community(p)=\interior(\intertwined{p})$.
Since $\atopen$ is open it would suffice to show that $\atopen\subseteq\intertwined{p}$.
By Theorem~\ref{thrm.cc.char} $p\intertwinedwith p'$ for every $p'\in \atopen$, and it follows immediately that $\atopen\subseteq\intertwined{p}$.
\item
Suppose $\atopen$ is a maximal topen neighbourhood of $p$.

First, note that $\atopen$ is open, and by Theorem~\ref{thrm.cc.char} $\atopen\subseteq\intertwined{p}$, so $\atopen\subseteq\community(p)$.

Now consider any open $O\subseteq\intertwined{p}$.
Note that $\atopen\cup O$ is an open subset of $\intertwined{p}$, so by Theorem~\ref{thrm.cc.char} $\atopen\cup O$ is topen, and by maximality $\atopen\cup O\subseteq \atopen$ and thus $O\subseteq \atopen$.
It follows that $\community(p)\subseteq \atopen$.
\qedhere\end{enumerate}
\end{proof}

\begin{rmrk}
\label{rmrk.how.regularity}
We can use Lemma~\ref{lemm.max.cc.intertwined} to characterise regularity in five equivalent ways: see Theorem~\ref{thrm.max.cc.char} and Corollary~\ref{corr.regular.is.regular}.
\end{rmrk}

\begin{thrm}
\label{thrm.max.cc.char}
Suppose $(\ns P,\opens)$ is a semitopology and $p\in \ns P$.
Then the following are equivalent:
\begin{enumerate*}
\item\label{char.p.regular}
$p$ is regular, or in full: $p\in\community(p)\in\tf{Topen}$.
\item\label{char.Kp.greatest.topen}
$\community(p)$ is the greatest topen neighbourhood of $p$.
\item\label{char.Kp.max.topen}
$\community(p)$ is a maximal topen neighbourhood of $p$.
\item\label{char.max.topen}
$p$ has a maximal topen neighbourhood. 
\item\label{char.some.topen}
$p$ has some topen neighbourhood.
\end{enumerate*}
\end{thrm}
\begin{proof}
We prove a cycle of implications:
\begin{enumerate}
\item
If $\community(p)$ is a topen neighbourhood of $p$ then it is maximal by Lemma~\ref{lemm.intertwined.is.the.greatest}(\ref{item.intertwined.is.the.greatest.1}).
Furthermore this maximal topen neighbourhood of $p$ is necessarily greatest, since if we have two maximal topen neighbourhoods of $p$ then their union is a larger topen neighbourhood of $p$ by Proposition~\ref{prop.cc.unions}. 
\item
If $\intertwined{p}$ is the greatest topen neighbourhood of $p$, then certainly it is maximal.
\item
If $\intertwined{p}$ is a maximal topen neighbourhood of $p$, then certainly $p$ has a maximal topen neighbourhood.
\item
If $p$ has a maximal topen neighbourhood then certainly $p$ has a topen neighbourhood.
\item
Suppose $p$ has a topen neighbourhood $\atopen$.
By Corollary~\ref{corr.max.cc} we may assume without loss of generality that $\atopen$ is a maximal topen.
We use Lemma~\ref{lemm.max.cc.intertwined}(\ref{item.intertwined.is.the.greatest.2}).
\qedhere\end{enumerate}
\end{proof}

\begin{corr}
\label{corr.regular.is.regular}
Suppose $(\ns P,\opens)$ is a semitopology and $p\in\ns P$.
Then the following are equivalent:
\begin{enumerate*}
\item
$p$ is regular.
\item
$p$ is weakly regular and $\community(p)=\community(p')$ for every $p'\in\community(p)$.
\end{enumerate*} 
It might be useful to look at Example~\ref{xmpl.cc}(\ref{item.cc.two.regular.b}) and Figure~\ref{fig.012} (top-right diagram).
In that example the point $1$ is \emph{not} regular, and its community $\{0,1,2\}$ is not a community for $0$ or $2$.
\end{corr}
\begin{proof}
We prove two implications, using Theorem~\ref{thrm.max.cc.char}:
\begin{itemize}
\item
Suppose $p$ is regular.
By Lemma~\ref{lemm.wr.r}(\ref{item.r.implies.wr}) $p$ is weakly regular.
Now consider $p'\in\community(p)$.
By Theorem~\ref{thrm.max.cc.char} $\community(p)$ is topen, so it is a topen neighbourhood of $p'$. 
By Theorem~\ref{thrm.max.cc.char} $\community(p')$ is a greatest topen neighbourhood of $p'$. 
But by Theorem~\ref{thrm.max.cc.char} $\community(p)$ is also a greatest topen neighbourhood of $p$, and $\community(p)\between\community(p')$ since they both contain $p'$.
By Proposition~\ref{prop.cc.unions} and maximality, they are equal.
\item
Suppose $p$ is weakly regular and suppose $\community(p)=\community(p')$ for every $p'\in\community(p)$, and consider $p',p''\in\community(p)$.
Then $p'\intertwinedwith p''$ holds, since $p''\in\community(p')=\community(p)$.
By Theorem~\ref{thrm.cc.char} $\community(p)$ is topen, and by weak regularity $p\in\community(p)$, so by Theorem~\ref{thrm.max.cc.char} $p$ is regular as required. 
\qedhere\end{itemize}
\end{proof}

\begin{corr}
\label{corr.p.p'.regular.community}
Suppose $(\ns P,\opens)$ is a semitopology and $p,p'\in\ns P$.
Then if $p$ is regular and $p'\in\community(p)$ then $p'$ is regular and has the same community.
\end{corr}
\begin{proof}
Suppose $p$ is regular --- so by Definition~\ref{defn.tn}(\ref{item.regular.point}) $p\in\community(p)\in\topens$ --- and suppose $p'\in\community(p)$.
Then by Corollary~\ref{corr.regular.is.regular} $\community(p)=\community(p')$, so $p'\in\community(p')\in\topens$ and by Theorem~\ref{thrm.max.cc.char} $p'$ is regular. 
\end{proof}

\jamiesubsection{Intersection and partition properties of regular spaces}
\label{subsect.topen.partitions}

Proposition~\ref{prop.topen.intersect.subset} is useful for consensus in practice.
Suppose we are a regular point $q$ and we have reached consensus with some topen neighbourhood $O\ni q$.
Suppose further that our topen neighbourhood $O$ intersects with the maximal topen neighbourhood $\community(p)$ of some other regular point $p$.
Then Proposition~\ref{prop.topen.intersect.subset} tells us that we were inside $\community(p)$ all along.
\begin{prop}
\label{prop.topen.intersect.subset}
Suppose $(\ns P,\opens)$ is a semitopology and $p\in\ns P$ is regular and $O\in\topens$ is topen.
Then 
$$
O\between\community(p)
\quad\text{if and only if}\quad
O\subseteq\community(p).
$$
\end{prop}
\begin{proof} 
The right-to-left implication is immediate from Notation~\ref{nttn.between}(\ref{item.between}), given that 
topens are nonempty by Definition~\ref{defn.transitive}(\ref{transitive.cc}).

For the left-to-right implication, suppose $O\between\community(p)$.
By Theorem~\ref{thrm.max.cc.char} $\community(p)$ is a maximal topen, and by Proposition~\ref{prop.cc.unions} $O\cup\community(p)$ is topen.
Then $O\subseteq\community(p)$ follows by maximality.
\end{proof}

\begin{prop}
\label{prop.community.partition}
Suppose $(\ns P,\opens)$ is a semitopology and suppose $p,p'\in\ns P$ are regular.
Then
$$
\community(p)\between\community(p')
\quad\liff\quad
\community(p)=\community(p')
$$
\end{prop}
\begin{proof}
We prove two implications.
\begin{itemize}
\item
Suppose there exists $p''\in\community(p)\cap\community(p')$.
By Corollary~\ref{corr.p.p'.regular.community} ($p''$ is regular and) $\community(p)=\community(p'')=\community(p')$.
\item
Suppose $\community(p)=\community(p')$.
By assumption $p\in\community(p)$, so $p\in\community(p')$.
Thus $p\in\community(p)\cap\community(p')$.
\qedhere\end{itemize}
\end{proof}

Corollary~\ref{corr.topen.partition.char} is a simple characterisation of regular semitopological spaces:
\begin{corr}
\label{corr.topen.partition.char}
Suppose $(\ns P,\opens)$ is a semitopology.
Then the following are equivalent:
\begin{enumerate*}
\item\label{item.topen.partition.char.1}
$(\ns P,\opens)$ is regular.
\item\label{item.topen.partition.char.2}
$\ns P$ partitions into topen sets: there exists some set of topen sets $\mathcal T$ such that $\atopen\notbetween\atopen'$ for every $\atopen,\atopen'\in\mathcal T$ and $\ns P=\bigcup\mathcal T$.
\item\label{item.topen.partition.char.3}
Every $X\subseteq\ns P$ has a cover of topen sets: there exists some set of topen sets $\mathcal T$ such that $X\subseteq\bigcup\mathcal T$.
\end{enumerate*}
\end{corr}
\begin{proof}
The proof is routine from the machinery that we already have.
We prove equivalence of parts~\ref{item.topen.partition.char.1} and~\ref{item.topen.partition.char.2}:
\begin{enumerate}
\item
Suppose $(\ns P,\opens)$ is regular, meaning by Definition~\ref{defn.tn}(\ref{item.regular.S}\&\ref{item.regular.point}) that $p\in\community(p)\in\topens$ for every $p\in\ns P$.
We set $\mathcal T=\{\community(p) \mid p\in\ns P\}$.
By assumption this covers $\ns P$ in topens, and by Proposition~\ref{prop.community.partition} the cover is a partition. 
\item
Suppose $\mathcal T$ is a topen partition of $\ns P$.
By definition for every point $p$ there exists $T\in\mathcal T$ such that $p\in T$ and so $p$ has a topen neighbourhood.
By Theorem~\ref{thrm.max.cc.char}(\ref{char.some.topen}\&\ref{char.p.regular}) $p$ is regular.
\end{enumerate}
We prove equivalence of parts~\ref{item.topen.partition.char.2} and~\ref{item.topen.partition.char.3}:
\begin{enumerate}
\item
Suppose $\mathcal T$ is a topen partition of $\ns P$, and suppose $X\subseteq\mathcal P$.
Then trivially $X\subseteq\bigcup\mathcal T$.
\item
Suppose every $X\subseteq\ns P$ has a cover of topen sets.
Then $\ns P$ has a cover of topen sets; write it $\mathcal T$.
By Corollary~\ref{corr.max.cc} we may assume without loss of generality that $\mathcal T$ is a partition, and we are done.
\qedhere\end{enumerate} 
\end{proof}

\begin{nttn}
Call a semitopology $(\ns P,\opens)$ \deffont{singular} when it contains a single maximal topen subset. 
\end{nttn}

\begin{rmrk}
\label{rmrk.the.moral}
The moral we take from the results and examples above (and those to follow) is that the world we are entering has rather different well-behavedness criteria than those familiar from the study of typical Hausdorff topologies like $\mathbb R$:
\begin{enumerate*}
\item
`Bad' spaces are spaces that are not regular.

$\mathbb R$ with its usual topology (which is also a semitopology) is an example of a `bad' semitopology; it is not even quasiregular.
\item
`Good' spaces are spaces that are regular.

The supermajority and all-but-one semitopologies from Example~\ref{xmpl.semitopologies}(\ref{item.supermajority}\&\ref{item.counterexample.X-x}) are typical examples of `good' semitopologies; both are singular regular spaces.
\item
Corollary~\ref{corr.topen.partition.char} shows that the `good' spaces are just the (disjoint, possibly infinite) unions of singular regular spaces.
\end{enumerate*}
So to sum this up: modulo disjoint unions, the study of consensus behaviour is the study of semitopological spaces that consist of a single topen set of points that are all intertwined with one another. 
\end{rmrk}

\jamiesection{Regular = weakly regular + unconflicted}
\label{sect.closed.sets}

\jamiesubsection{Closed sets}
\label{subsect.closed.sets.basics}

In Subsection~\ref{subsect.closed.sets.basics} we check that some familiar properties of closures carry over from topologies to semitopologies.
There are no technical surprises, but this in itself is a mathematical result that needs checked. 
Then, we will use this to study the relation between closures and sets of intertwined points.

\begin{defn}
\label{defn.closure}
Suppose $(\ns P,\opens)$ is a semitopology and suppose $p\in\ns P$ and $P\subseteq\ns P$.
Then:
\begin{enumerate*}
\item\label{item.closure}
Define $\closure{P}\subseteq\ns P$ the \deffont{closure of $P$} to be the set of points $p$ such that every open neighbourhood of $p$ intersects $P$.
In symbols using Notation~\ref{nttn.between}: 
$$
\closure{P} = \{ p'\in\ns P \mid \Forall{O{\in}\opens} p'\in O \limp P\between O\} .
$$
\item\label{item.closure.p}
As is standard, we may write $\closure{p}$ for $\closure{\{p\}}$.
Unpacking definitions for reference:
$$
\closure{p} = \{ p'\in\ns P \mid \Forall{O{\in}\opens} p'\in O \limp p\in O\} .
$$
\end{enumerate*}
\end{defn}

\begin{lemm}
\label{lemm.closure.monotone}
Suppose $(\ns P,\opens)$ is a semitopology and suppose $P,P'\subseteq\ns P$.
Then taking the closure of a set is: 
\begin{enumerate*}
\item\label{closure.monotone}
\emph{Monotone:}\quad If $P\subseteq P'$ then $\closure{P}\subseteq\closure{P'}$.
\item\label{closure.increasing}
\emph{Increasing:}\quad $P\subseteq\closure{P}$.
\item\label{closure.idempotent}
\emph{Idempotent:}\quad $\closure{P}=\closure{\closure{P}}$.
\end{enumerate*}
\end{lemm}
\begin{proof}
By routine calculations from Definition~\ref{defn.closure}.
\end{proof}

\begin{lemm}
\label{lemm.closure.open.char}
Suppose $(\ns P,\opens)$ is a semitopology and $P\subseteq\ns P$ and $O\in\opens$.
Then 
$$
P\between O
\quad\text{if and only if}\quad 
\closure{P}\between O.
$$
\end{lemm}
\begin{proof}
Suppose $P\between O$.
Then $\closure{P}\between O$ using Lemma~\ref{lemm.closure.monotone}(\ref{closure.increasing}).

Suppose $\closure{P}\between O$.
Pick $p\in \closure{P}\cap O$.
By construction of $\closure{P}$ in Definition~\ref{defn.closure} $p\in O\limp P\between O$.
It follows that $P\between O$ as required.
\end{proof}

\begin{defn}
\label{defn.closed}
Suppose $(\ns P,\opens)$ is a semitopology and suppose $C\subseteq\ns P$.
\begin{enumerate*}
\item\label{item.closed.set}
Call $C$ a \deffont{closed set} when $C=\closure{C}$.
\item
Call $C$ a \deffont{clopen set} when $C$ is closed and open.
\item
Write $\closed$ for the set of \deffont[closed sets $\closed$]{closed sets} (as we wrote $\opens$ for the open sets; the ambient semitopology will always be clear or understood).
\end{enumerate*}
\end{defn}

\begin{lemm}
\label{lemm.closure.closed}
Suppose $(\ns P,\opens)$ is a semitopology and suppose $P\subseteq\ns P$.
Then $\closure{P}$ is closed and contains $P$.
In symbols:
$$
P\subseteq \closure{P}\in\closed .
$$ 
\end{lemm}
\begin{proof}
From Definition~\ref{defn.closed}(\ref{item.closed.set}) and Lemma~\ref{lemm.closure.monotone}(\ref{closure.increasing} \& \ref{closure.idempotent}).
\end{proof}

\begin{xmpl}\leavevmode
\begin{enumerate}
\item
Take $\ns P=\{0,1\}$ and $\opens=\{\varnothing, \{0\}, \{0,1\}\}$.
Then the reader can verify that:
\begin{itemize*}
\item
$\{0\}$ is open.
\item
The closure of $\{1\}$ is $\{1\}$ and $\{1\}$ is closed.
\item
The closure of $\{0\}$ is $\{0,1\}$.
\item
$\varnothing$ and $\{0,1\}$ are the only clopen sets.
\end{itemize*}
\item
Now take $\ns P=\{0,1\}$ and $\opens=\{\varnothing, \{0\}, \{1\}, \{0,1\}\}$.\footnote{Following Definition~\ref{defn.value.assignment} and Example~\ref{xmpl.semitopologies}(\ref{item.boolean.discrete}), this is just $\{0,1\}$ with the \emph{discrete semitopology}.}
Then the reader can verify that:
\begin{itemize*}
\item
Every set is clopen.
\item
The closure of every set is itself.
\end{itemize*}
\end{enumerate}
\end{xmpl}

\begin{rmrk}
There are two standard definitions for when a set is closed: when it is equal to its closure (as per Definition~\ref{defn.closed}(\ref{item.closed.set})), and when it is the complement of an open set.
In topology these are equivalent.
We do need to check that the same holds in semitopology, but as it turns out the proof is routine:
\end{rmrk}

\begin{lemm}
\label{lemm.closed.complement.open}
Suppose $(\ns P,\opens)$ is a semitopology.
Then:
\begin{enumerate*}
\item\label{item.closed.complement.open.1}
Suppose $C\in\closed$ is closed (by Definition~\ref{defn.closed}: $C=\closure{C}$).
Then $\ns P\setminus C$ is open.
\item\label{item.closed.complement.open.2}
Suppose $O\in\opens$ is open.
Then $\ns P\setminus O$ is closed (by Definition~\ref{defn.closed}: $\closure{\ns P\setminus O}=\ns P\setminus O$).
\end{enumerate*}
\end{lemm}
\begin{proof}
\leavevmode
\begin{enumerate}
\item
Suppose $p\in \ns P\setminus C$.
Since $C=\closure{C}$, we have $p\in\ns P\setminus\closure{C}$.
Unpacking Definition~\ref{defn.closure}, this means precisely that there exists $O_p\in\opens$ with $p\in O_p \notbetween C$.
We use Lemma~\ref{lemm.open.is.open}. 
\item
Suppose $O\in\opens$.
Combining Lemma~\ref{lemm.open.is.open} with Definition~\ref{defn.closure} 
it follows that $O\notbetween \closure{\ns P\setminus O}$ so that $\closure{\ns P\setminus O}\subseteq\ns P\setminus O$.
Furthermore, by Lemma~\ref{lemm.closure.monotone}(\ref{closure.increasing}) $\ns P\setminus O\subseteq\closure{\ns P\setminus O}$.
\qedhere\end{enumerate}
\end{proof}

The usual duality between forming closures and interiors, remains valid in semitopologies:
\begin{lemm}
\label{lemm.closure.interior}
Suppose $(\ns P,\opens)$ is a semitopology and $O\in\opens$ and $C\in\closed$.
Then:
\begin{enumerate*}
\item\label{item.closure.interior.open}
$O\subseteq\interior(\closure{O})$.  The inclusion may be strict.
\item\label{item.closure.interior.closed}
$\closure{\interior(C)}\subseteq C$.  The inclusion may be strict.
\item\label{item.closure.interior.complement.closure}
$\interior(\ns P\setminus O)=\ns P\setminus\closure{O}$.
\item\label{item.closure.interior.complement.interior}
$\closure{\ns P\setminus C}=\ns P\setminus\interior(C)$. 
\end{enumerate*}
\end{lemm}
\begin{proof}
The reasoning is just as for topologies, but we spell out the details:
\begin{enumerate}
\item
By Lemma~\ref{lemm.closure.monotone}(\ref{closure.increasing}) $O\subseteq\closure{O}$.
By Corollary~\ref{corr.interior.monotone} $\interior(O)\subseteq\interior(\closure{O})$.
By Lemma~\ref{lemm.interior.open} $O=\interior(O)$, so we are done.

For an example of the strict inclusion, consider $\mathbb R$ with the usual topology (which is also a semitopology) and take $O=(0,1)\cup(1,2)$.
Then $O\subsetneq\interior(\closure{O})=(0,2)$.
\item
By Lemma~\ref{lemm.interior.open} $\interior(C)\subseteq C$.
By Lemma~\ref{lemm.closure.monotone}(\ref{closure.monotone}) $\closure{\interior(C)}\subseteq\closure{C}$.
By Definition~\ref{defn.closed}(\ref{item.closed.set}) (since we assumed $C\in\closed$) $\closure{C}=C$, so we are done.

For an example of the strict inclusion, consider $\mathbb R$ with the usual topology and take $C=\{0\}$.
Then $\closure{\interior(C)}=\varnothing\subsetneq C$.
\item
Consider some $p'\in\ns P$.
By Definition~\ref{defn.interior} $p'\in \interior(\ns P\setminus O)$ when there exists some $O'\in\opens$ such that $p'\in O'\notbetween O$.
By definition in Definition~\ref{defn.closure}(\ref{item.closure}) this happens precisely when $p'\notin\closure{O}$. 
\item
By Definition~\ref{defn.closure}(\ref{item.closure}), $p'\notin \closure{\ns P\setminus C}$ precisely when there exists some $O'\in\opens$ such that $p'\in O'\notbetween \ns P\setminus C$.
By facts of sets this means precisely that $p'\in O'\subseteq C$.
By Definition~\ref{defn.interior} this means precisely that $p'\in\interior(C)$.
\qedhere\end{enumerate}
\end{proof}

\jamiesubsection{Closed neighbourhoods and intertwined points}
\label{subsect.closed.neighbourhoods}

\begin{defn}
\label{defn.cn}
Suppose $(\ns P,\opens)$ is a semitopology.
We generalise Definition~\ref{defn.open.neighbourhood} as follows:
\begin{enumerate*}
\item\label{item.neighbourhood.of.p}
Call $P\subseteq\ns P$ a \deffont{neighbourhood} when it contains an open set (i.e. when $\interior(P)\neq\varnothing$), and call $P$ a \deffont{neighbourhood of $p$} when $p\in\ns P$ and $P$ contains an open neighbourhood of $p$ (i.e. when $p\in\interior(P)$).
In particular:
\item\label{item.closed.neighbourhood.of.p}
$C\subseteq\ns P$ is a \deffont{closed neighbourhood of $p\in\ns P$} when $C$ is closed and $p\in\interior(C)$.
\item\label{item.closed.neighbourhood}
$C\subseteq\ns P$ is a \deffont{closed neighbourhood} when $C$ is closed and $\interior(C)\neq\varnothing$.
\end{enumerate*} 
\end{defn}

\begin{rmrk}
\label{rmrk.cn.interior}
\leavevmode
\begin{enumerate}
\item
If $C$ is a closed neighbourhood of $p$ in the sense of Definition~\ref{defn.cn}(\ref{item.closed.neighbourhood.of.p}) then $C$ is a closed neighbourhood in the sense of Definition~\ref{defn.cn}(\ref{item.closed.neighbourhood}), just because if $p\in\interior(C)$ then $\interior(C)\neq\varnothing$. 
\item
$p\in C$ is not enough for $C$ to be a closed neighbourhood of $p$;
we require the stronger condition $p\in\interior(C)$.
For instance, consider the \deffont{Sierpi\'nski space} 
$$
\f{Sk}=(\{0,1\},\ \{\varnothing,\{1\},\{0,1\}\}) .
$$
is a topology and so also a semitopology (cf. Remark~\ref{rmrk.representations}),
and consider $p=0$ and $C=\{0\}$.
Then $p\in C$ but $p\not\oldin\interior(C)=\varnothing$, so that $C$ is not a closed neighbourhood of $p$. 
\end{enumerate}
\end{rmrk}

\begin{rmrk}
\label{rmrk.re-read.closure}
Recall the definitions of $\intertwined{p}$ and $\closure{p}$:
\begin{itemize*}
\item
The set $\closure{p}$ is the \emph{closure} of $p$.

By Definition~\ref{defn.closure} this is the set of $p'$ such that every open neighbourhood $O'\ni p'$ intersects with $\{p\}$.
\item
The set $\intertwined{p}$ is the set of points \emph{intertwined} with $p$.

By Definition~\ref{defn.intertwined.points}(\ref{intertwined.defn}) this is the set of $p'$ such that every open neighbourhood $O'\ni p'$ intersects with every open neighbourhood $O \ni p$. 
\end{itemize*}
\end{rmrk}

Lemma~\ref{lemm.char.not.intertwined} rephrases Remark~\ref{rmrk.re-read.closure} more precisely by looking at it through sets complements.
We will use it in Lemma~\ref{lemm.cast.comp}(\ref{item.cast.comp.nbhd}):
\begin{lemm}
\label{lemm.char.not.intertwined}
Suppose $(\ns P,\opens)$ is a semitopology and $p\in\ns P$.
Then:
\begin{enumerate*}
\item
$\ns P\setminus\closure{p} = \bigcup \{O\in\opens \mid p\notin O\}\oldin\opens$.
\item\label{item.intertwined.open.avoid}
$\ns P\setminus\intertwined{p} = \bigcup\{O'\in\opens \mid \Exists{O{\in}\opens} p\in O\land O'\notbetween O\}\oldin\opens$.
\item
$\ns P\setminus\intertwined{p} = \bigcup\{O\in\opens \mid p\notin \closure{O}\}\oldin\opens$.
\end{enumerate*}
In words, we can say: $\ns P\setminus\closure{p}$ is the union of the open sets such that $p$ avoids them, and $\ns P\setminus\intertwined{p}$ is the union of the open sets such that $p$ avoids their closures.
\end{lemm} 
\begin{proof}
\leavevmode
\begin{enumerate*}
\item
Immediate from Definitions~\ref{defn.intertwined.points} and~\ref{defn.closure}. %
Openness is from Definition~\ref{defn.semitopology}(\ref{semitopology.unions}).
\item
By a routine argument direct from Definition~\ref{defn.intertwined.points}. 
Openness is from Definition~\ref{defn.semitopology}(\ref{semitopology.unions}).
\item
Rephrasing part~\ref{item.intertwined.open.avoid} of this result using Definition~\ref{defn.closure}(\ref{item.closure}).
\qedhere\end{enumerate*}
\end{proof}

\begin{rmrk}
\label{rmrk.cluster.convergence.2}
We can relate Lemma~\ref{lemm.char.not.intertwined}
to a concept from topology. 
Following standard terminology (\cite[Definition~2, page~69]{bourbaki:gent1} or \cite[page~52]{engelking:gent}), a \deffont{cluster point} $p\in\ns P$ of $\mathcal O\subseteq\opens$ is one such that every open neighbourhood of $p$ intersects every $O\in\mathcal O$.
Then Lemma~\ref{lemm.char.not.intertwined}(\ref{item.intertwined.open.avoid})
identifies $\intertwined{p}$ as the set of cluster points of 
$\nbhd(p)\subseteq\opens$ from Definition~\ref{defn.nbhd}.
\end{rmrk}

\jamiesubsection{Regular = weakly regular + unconflicted}
\label{subsect.r=wr+uc}
\label{subsect.reg.tra.int}

In Lemma~\ref{lemm.intertwined.not.transitive} we asked whether the `is intertwined with' relation $\intertwinedwith$ from Definition~\ref{defn.intertwined.points}(1) is transitive --- answer: not necessarily.

Transitivity of $\intertwinedwith$ is a natural condition.
We now have enough machinery to study it in more detail, and this will help us gain a deeper understanding of the properties of not-necessarily-regular points.

\begin{defn}
\label{defn.conflicted}
Suppose $(\ns P,\opens)$ is a semitopology.
\begin{enumerate*}
\item
Call a point $p$ \deffont{conflicted} when there exist $p'$ and $p''$ such that $p'\intertwinedwith p$ and $p\intertwinedwith p''$ yet $\neg(p'\intertwinedwith p'')$.
\item\label{item.unconflicted}
If $p'\intertwinedwith p\intertwinedwith p''$ implies $p'\intertwinedwith p''$ always, then call $p$ \deffont{unconflicted}.
\item
Continuing Definition~\ref{defn.tn}(\ref{item.regular.S}), if $P\subseteq\ns P$ and every $p\in P$ is conflicted/unconflicted, then we may call $P$ \deffont{conflicted/unconflicted} respectively. 
\end{enumerate*}
\end{defn}

\begin{xmpl}
\label{xmpl.conflicted.points}
We consider some examples:
\begin{enumerate*}
\item\label{item.example.of.conflicted.point}
In Figure~\ref{fig.012} top-left diagram, $0$ and $2$ are unconflicted and intertwined with themselves, and $1$ is conflicted (being intertwined with $0$, $1$, and $2$).

If the reader wants to know what a conflicted point looks like: it looks like $1$. 
\item 
In Figure~\ref{fig.012} top-right diagram, $0$ and $2$ are unconflicted and intertwined with themselves, and $1$ is conflicted (being intertwined with $0$, $1$, and $2$).
\item
In Figure~\ref{fig.012} lower-left diagram, $0$ and $1$ are unconflicted and intertwined with themselves, and $3$ and $4$ are unconflicted and intertwined with themselves, and $2$ is conflicted (being intertwined with $0$, $1$, $2$, $3$, and $4$).
\item
In Figure~\ref{fig.012} lower-right diagram, all points are unconflicted, and $0$ and $2$ are intertwined just with themselves, and $1$ and $\ast$ are intertwined with one another.
\item
In Figure~\ref{fig.square.diagram}, all points are unconflicted and intertwined only with themselves.
(This semitopology is useful for counterexamples and we will use it again.)
\end{enumerate*}
\end{xmpl}

\begin{figure}
\centering
\includegraphics[width=0.4\columnwidth]{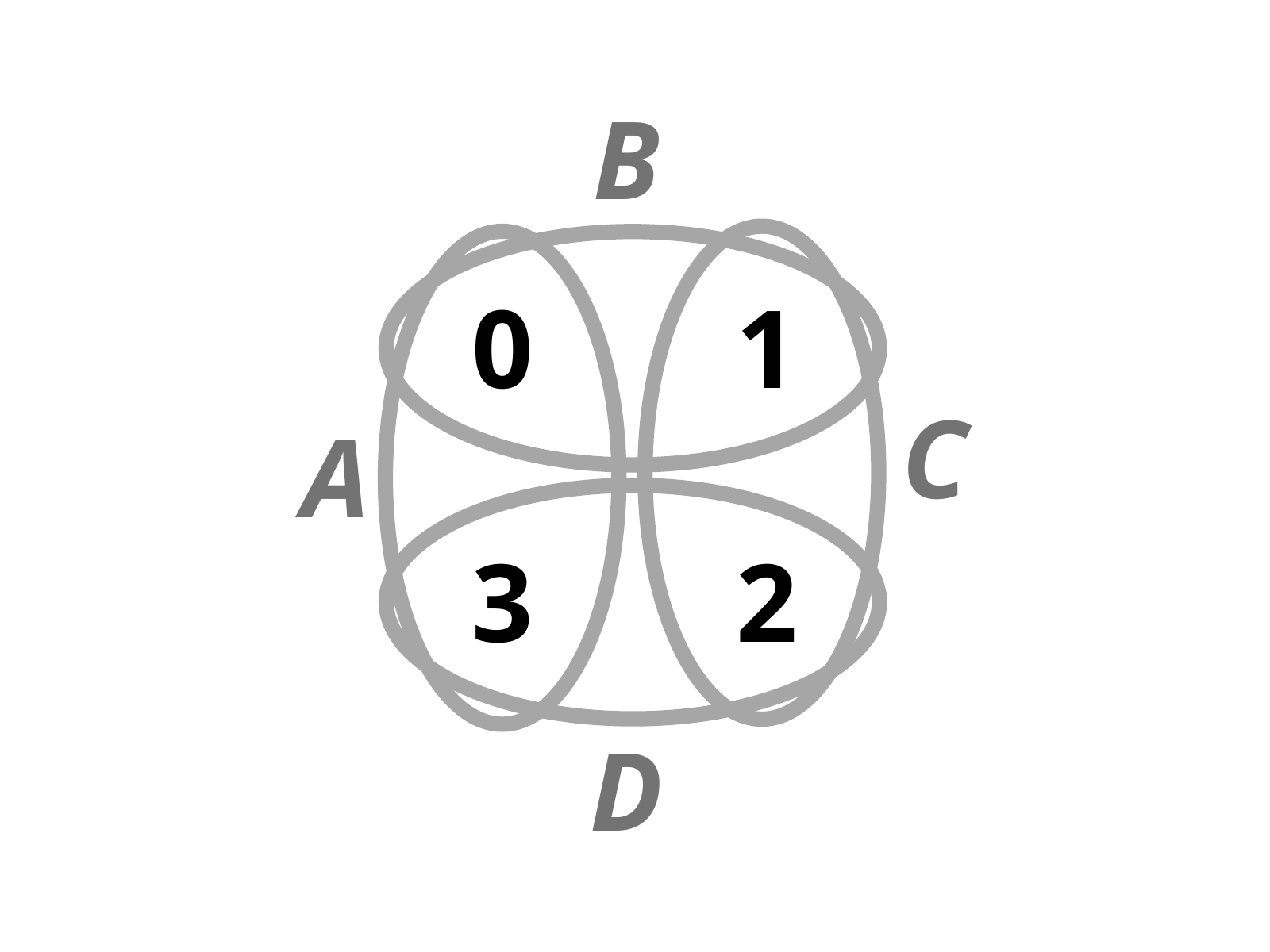}
\\[4ex]
\caption{An unconflicted, irregular space (Proposition~\ref{prop.unconflicted.irregular}) in which every point is intertwined only with itself (Example~\ref{xmpl.conflicted.points})}
\label{fig.square.diagram}
\end{figure}

\begin{prop}
\label{prop.unconflicted.irregular}
Suppose $(\ns P,\opens)$ is a semitopology and $p\in\ns P$.
Then:
\begin{enumerate*}
\item\label{item.reg.implies.unconflicted}
If $p$ is regular then it is unconflicted.

Equivalently by the contrapositive: if $p$ is conflicted then it is not regular.
\item
$p$ may be unconflicted and neither quasiregular, weakly regular, nor regular.
\end{enumerate*}
\end{prop}
\begin{proof}
We consider each part in turn:
\begin{enumerate}
\item
So consider $q\intertwinedwith p \intertwinedwith q'$.
We must show that $q\intertwinedwith q'$, so consider open neighbourhoods $Q\ni q$ and $Q'\ni q'$.
By assumption $p$ is regular, so unpacking Definition~\ref{defn.tn}(\ref{item.regular.point}) $\community(p)$ is a topen (transitive and open) neighbourhood of $p$.
By assumption $Q\between \community(p)\between Q'$, and by transitivity of $\community(p)$ (Definition~\ref{defn.transitive}(\ref{transitive.transitive})) we have $Q\between Q'$ as required.
\item
Consider the semitopology illustrated in Figure~\ref{fig.square.diagram}.
Note that the point $0$ is not conflicted (because it is not intertwined with any other point), but it is also neither quasiregular, weakly regular, nor regular, because its community is the empty set. 
\qedhere\end{enumerate}
\end{proof}

We can combine Proposition~\ref{prop.unconflicted.irregular} with a previous result Lemma~\ref{lemm.wr.r} to get a precise and attractive relation between being 
\begin{itemize*}
\item
regular (Definition~\ref{defn.tn}(\ref{item.regular.point})), 
\item
weakly regular (Definition~\ref{defn.tn}(\ref{item.weakly.regular.point})), and 
\item
unconflicted (Definition~\ref{defn.conflicted}), 
\end{itemize*}
as follows:
\begin{thrm}
\label{thrm.r=wr+uc}
Suppose $(\ns P,\opens)$ is a semitopology and $p\in\ns P$.
Then the following are equivalent:
\begin{itemize*}
\item
$p$ is regular.
\item
$p$ is weakly regular and unconflicted.
\end{itemize*}
More succinctly we can write: \emph{regular = weakly regular + unconflicted}.\footnote{See also Theorem~\ref{thrm.r=wr+sc}, which does something similar for semiframes.}
\end{thrm}
\begin{proof}
We prove two implications:
\begin{itemize}
\item
If $p$ is regular then it is weakly regular by Lemma~\ref{lemm.wr.r} and unconflicted by Proposition~\ref{prop.unconflicted.irregular}(\ref{item.reg.implies.unconflicted}). 
\item
Suppose $p$ is weakly regular and unconflicted.
By Definition~\ref{defn.tn}(\ref{item.weakly.regular.point}) $p\in\community(p)$ and by Lemma~\ref{lemm.three.transitive} it would suffice to show that $q\intertwinedwith q'$ for any $q,q'\in\community(p)$.

So consider $q,q'\in\community(p)$.
Now by Definition~\ref{defn.tn}(\ref{item.tn}) $\community(p)=\interior(\intertwined{p})$ so in particular $q,q'\in\intertwined{p}$.
Thus $q\intertwinedwith p\intertwinedwith q'$, and since $p$ is unconflicted $q\intertwinedwith q'$ as required.
\qedhere\end{itemize}
\end{proof}

\jamiesection{Regular = quasiregular + hypertransitive}
\label{sect.r=qr+ht}

\jamiesubsection{Regular open/closed sets}

We recall some simple results:

\begin{defn}
\label{defn.regular.open.set}
Suppose $(\ns P,\opens)$ is a semitopology.
Recall some standard terminology from topology~\cite[Exercise~3D, page~29]{willard:gent}:
\begin{enumerate*}
\item
We call an open set $O\in\opens$ a \deffont{regular open set} when $O=\interior(\closure{O})$.
\item
We call a closed set $C\in\closed$ a \deffont{regular closed set} when $C=\closure{\interior(C)}$.
\item
Write $\regularOpens$ and $\regularClosed$ for the sets of regular open and regular closed sets respectively.
\end{enumerate*}
\end{defn}

\begin{lemm}
\label{lemm.ic.ci}
Suppose $(\ns P,\opens)$ is a semitopology and 
$O\in\opens$ and $C\in\closed$.
Then:
\begin{enumerate*}
\item
$\closure{O} = \closure{\interior(\closure{O})}$. 
\item
$\interior(C)=\interior(\closure{\interior(C)})$.
\end{enumerate*}
\end{lemm}
\begin{proof}
We use Lemma~\ref{lemm.closure.interior}(\ref{item.closure.interior.open}\&\ref{item.closure.interior.complement.closure}) along with Lemma~\ref{lemm.closure.monotone}(\ref{closure.monotone}) and Corollary~\ref{corr.interior.monotone}: 
$$
\begin{array}{r@{\ }c@{\ }c@{\ }c@{\ }ll}
\closure{O}
&\stackrel{L\ref{lemm.closure.interior}(\ref{item.closure.interior.open})\&L\ref{lemm.closure.monotone}(\ref{closure.monotone})}\subseteq&
\closure{\interior(\closure{O})}
&\stackrel{L\ref{lemm.closure.interior}(\ref{item.closure.interior.closed})}\subseteq&
\interior(\closure{O})
\\
\interior(C)
&\stackrel{L\ref{lemm.closure.interior}(\ref{item.closure.interior.open})}\subseteq&
\interior(\closure{\interior(C)})
&\stackrel{L\ref{lemm.closure.interior}(\ref{item.closure.interior.closed})\&C\ref{corr.interior.monotone}}\subseteq&
\interior(C)
\end{array}
$$
\end{proof}

The terminology `regular open/closed set' is from the topological literature.
It is not directly related to terminology `regular point' from Definition~\ref{defn.tn}(\ref{item.regular.point}), which comes from semitopologies.
However, it turns out that a mathematical connection does exist between these two notions:
\begin{corr}
\label{corr.ic.ci.regular}
Suppose $(\ns P,\opens)$ is a semitopology and $O\in\opens$ and $C\in\closed$.
Then:
\begin{enumerate*}
\item\label{item.ic.ci.regular.open}
$\interior(C)$ is a regular open set.
\item\label{item.ic.ci.regular.closed}
$\closure{O}$ is a regular closed set.
\end{enumerate*}
\end{corr}
\begin{proof}
Direct from Definition~\ref{defn.regular.open.set} and Lemma~\ref{lemm.ic.ci}.
\end{proof}

An elementary observation about open sets will be useful:
\begin{lemm}
\label{lemm.clint.between}
Suppose $(\ns P,\opens)$ is a semitopology and $O,O'\in\opens$.
Then the following are equivalent:
\begin{enumerate*}
\item\label{item.client.between.1} 
$O\between O'$.
\item\label{item.client.between.2} 
$O\between\interior(\closure{O'})$.
\item\label{item.client.between.3} 
$\interior(\closure{O})\between\interior(\closure{O'})$.
\end{enumerate*}
\end{lemm}
\begin{proof}
First we prove the equivalence of parts~\ref{item.client.between.1} and~\ref{item.client.between.2}:
\begin{enumerate}
\item
Suppose $O\between O'$.
By Lemma~\ref{lemm.closure.interior}(\ref{item.closure.interior.open}) $O\between \interior(\closure{O'})$.
\item
Suppose there is some $p\in O\cap\interior(\closure{O'})$.
Then $O$ is an open neighbourhood of $p$ and $p\in\closure{O'}$, so by Definition~\ref{defn.closure}(\ref{item.closure}) $O\between O'$ as required.\footnote{Lemma~\ref{lemm.closure.using.nbhd.intersections} packages this argument up nicely with some slick notation, which we have not yet set up.}
\end{enumerate}
Equivalence of parts~\ref{item.client.between.1} and~\ref{item.client.between.3} then follows easily by two applications of the equivalence of parts~\ref{item.client.between.1} and~\ref{item.client.between.2}.
\end{proof}

\begin{rmrk}
\label{rmrk.pi-base}
Lemma~\ref{lemm.clint.between} is true in topologies as well, but it is not prominent in the literature.
Two standard reference works~\cite{engelking:gent,willard:gent} do not seem to mention it.
It appears as equation~10 in Theorem~1.37 of~\cite{koppelberg:hanba1}, and as a lemma in $\pi$-base~\cite{pi-base:lemma} (thanks to the mathematics StackExchange community for the pointers).  
Interestingly, this result is as true in topologies as it is in semitopologies, but somehow, it \emph{matters} more in the latter than the former.
\end{rmrk}

\jamiesubsection{Hypertransitivity}

We are now ready to define hypertransitivity.

\begin{nttn}
\label{nttn.between.nbhd}
Suppose $(\ns P,\opens)$ is a semitopology and $O'\in\opens$ and $\mathcal O\subseteq\opens$.
\begin{enumerate*}
\item\label{item.between.nbhd.1}
Write $O'\between\mathcal O$, or equivalently $\mathcal O\between O'$, when $O'\between O$ for every $O\in\mathcal O$.
In symbols:
$$
O'\between\mathcal O
\quad\text{when}\quad
\Forall{O{\in}\mathcal O}O'\between O .
$$
\item\label{item.between.nbhd}
As a special case of part~\ref{item.between.nbhd.1} above taking $\mathcal O=\nbhd(p)$ (Definition~\ref{defn.nbhd.system}), if $p\in\ns P$ then write $O'\between\nbhd(p)$, or equivalently $\nbhd(p)\between O'$, when $O'\between O$ for every $O\in\opens$ such that $p\in O$. 
\end{enumerate*}
\end{nttn}

\begin{lemm}
\label{lemm.closure.using.nbhd.intersections}
Suppose $(\ns P,\opens)$ is a semitopology and $p\in\ns P$ and $O'\in\opens$.
Then 
$$
p\in\closure{O'}
\quad\text{if and only if}\quad 
O'\between\nbhd(p) .
$$
\end{lemm}
\begin{proof}
This just rephrases Definition~\ref{defn.closure}(\ref{item.closure}). 
\end{proof}

\begin{defn}
\label{defn.sc}
Suppose $(\ns P,\opens)$ is a semitopology.
Call $p\in\ns P$ a \deffont{hypertransitive point} when for every $O',O''\in\opens$, 
$$
O'\between\nbhd(p)\between O''
\quad\text{implies}\quad O'\between O''.
$$
Call $(\ns P,\opens)$ a \deffont{hypertransitive semitopology} when every $p\in\ns P$ is hypertransitive.
\end{defn}

Lemma~\ref{lemm.sc.op.reg.op} notes some equivalent formulations of hypertransitivity: 
\begin{lemm}
\label{lemm.sc.op.reg.op}
Suppose $(\ns P,\opens)$ is a semitopology and $p\in\ns P$.
Then the following are equivalent:
\begin{enumerate*}
\item\label{item.sc.op.reg.op.1}
$p$ is hypertransitive.
\item\label{item.sc.op.reg.op.2}
For every pair of open sets $O',O''\in\opens$, $p\in \closure{O'}\cap \closure{O''}$ implies $O'\between O''$.
\item\label{item.sc.op.reg.op.3}
For every pair of \emph{regular} open sets $O',O''\in\regularOpens$, $p\in \closure{O'}\cap \closure{O''}$ implies $O'\between O''$. 
\end{enumerate*}
\end{lemm}
\begin{proof}
For the equivalence of parts~\ref{item.sc.op.reg.op.1} and~\ref{item.sc.op.reg.op.2} we reason as follows:
\begin{itemize*}
\item
Suppose $p$ is hypertransitive and suppose $p\in\closure{O'}$ and $p\in\closure{O''}$.
By Lemma~\ref{lemm.closure.using.nbhd.intersections} it follows that $O'\between\nbhd(p)\between O''$.
By hypertransitivity, $O'\between O''$ as required.
\item
Suppose for every $O,O'\in\opens$, $p\in\closure{O}\cap\closure{O'}$ implies $O'\between O''$, and suppose $O'\between\nbhd(p)\between O''$.
By Lemma~\ref{lemm.closure.using.nbhd.intersections} $p\in\closure{O}\cap\closure{O'}$ and therefore $O'\between O''$.
\end{itemize*}
For the equivalence of parts~\ref{item.sc.op.reg.op.2} and~\ref{item.sc.op.reg.op.3} we reason as follows: 
\begin{itemize*}
\item
Part~\ref{item.sc.op.reg.op.2} implies part~\ref{item.sc.op.reg.op.3} follows since every open regular set is also an open set.
\item
To show part~\ref{item.sc.op.reg.op.3} implies part~\ref{item.sc.op.reg.op.2}, suppose for every pair of regular opens $O',O''\in\regularOpens$, $p\in \closure{O'}\cap \closure{O''}$ implies $O'\between O''$, and suppose $O',O''\in\opens$ are two open sets that are not necessarily regular, and suppose $p\in\closure{O'}\cap\closure{O''}$.
We must show that $O'\between O''$.

Write $P'=\interior(\closure{O'})$ and $P''=\interior(\closure{O''})$ and note by Corollary~\ref{corr.ic.ci.regular} and Lemma~\ref{lemm.closure.closed} that $P'$ and $P''$ are regular open sets and $\closure{P'}=\closure{O'}$ and $\closure{P''}=\closure{O''}$.
Then $\closure{P'}\between\closure{P''}$, so $P'\between P''$, and $O'\between O''$ follows from Lemma~\ref{lemm.clint.between}
\qedhere\end{itemize*}
\end{proof}

\jamiesubsection{Regular = quasiregular + hypertransitive}

\begin{lemm}
\label{lemm.regular.sc}
Suppose $(\ns P,\opens)$ is a semitopology and $p\in\ns P$.
Then:
\begin{enumerate*}
\item\label{item.r.implies.sc}
If $p$ is regular then it is hypertransitive.
\item\label{item.sc.implies.uc}
If $p$ is hypertransitive then it is unconflicted.
\item
The reverse implication need not hold: it is possible for $p$ to be unconflicted but not hypertransitive.
\item
It is possible for $p$ to be hypertransitive (and unconflicted), but not quasiregular (and thus not weakly regular or regular).
\end{enumerate*}
\end{lemm}
\begin{proof}
We consider each part:
\begin{enumerate}
\item
Suppose $p$ is regular and $O,O'\in\opens$ and $O\between\nbhd(p)\between O'$.
By Definition~\ref{defn.tn}(\ref{item.regular.point}) (since $p$ is regular) $\community(p)$ is a topen (= open and transitive) neighbourhood of $p$.
Therefore by transitivity $O\between O'$ as required. 
\item
Suppose $p$ is hypertransitive and suppose $p',p''\in\ns P$ and $p'\intertwinedwith p\intertwinedwith p''$.
Now consider $p'\in O'\in\opens$ and $p''\in O''\in\opens$.
By our intertwinedness assumptions we have that $O'\between\nbhd(p)\between O''$.
But $p$ is hypertransitive, so $O'\between O''$ as required.
\item
It suffices to provide a counterexample.
Consider the bottom right semitopology in Figure~\ref{fig.012}, and take $p=\ast$ and $O'=\{1\}$ and $O''=\{0,2\}$.
Note that:
\begin{itemize*}
\item
$\ast$ is unconflicted, since it is intertwined only with itself and $1$.
\item
$O'$ and $O'$ intersect every open neighbourhood of $\ast$, but $O'\notbetween O''$, so $\ast$ is not hypertransitive.
\end{itemize*} 
\item
It suffices to provide an example.
Consider the semitopology illustrated in Figure~\ref{fig.012}, top-right diagram; so $\ns P=\{0,1,2\}$ and $\opens=\{\varnothing,\{0\},\{2\},\{1,2\},\{0,1\},\{0,1,2\}\}$.
The reader can check that $p=1$ is hypertransitive, but $\intertwined{1}=\{1\}$ and $\community(1)=\varnothing$ so $p$ is not quasiregular.
\qedhere\end{enumerate}
\end{proof}

(Yet) another characterisation of being quasiregular will be helpful:
\begin{lemm}
\label{lemm.quasiregular.iff.between}
Suppose $(\ns P,\opens)$ is a semitopology and $p\in\ns P$.
Then the following conditions are equivalent:
\begin{enumerate*}
\item\label{item.quasiregular.iff.between.1}
$p$ is quasiregular (meaning by Definition~\ref{defn.tn}(\ref{item.quasiregular.point}) that $\community(p)\neq\varnothing$).
\item\label{item.quasiregular.iff.between.2}
$\community(p)\between\nbhd(p)$ (meaning by Notation~\ref{nttn.between.nbhd}(\ref{item.between.nbhd}) that $\community(p)\between O$ for every $O\in\nbhd(p)$).
\item\label{item.quasiregular.iff.between.3}
$p\in\closure{\community(p)}$.
\end{enumerate*}
\end{lemm}
\begin{proof}
Equivalence of parts~\ref{item.quasiregular.iff.between.2} and~\ref{item.quasiregular.iff.between.3} is immediate from Lemma~\ref{lemm.closure.using.nbhd.intersections}.

For equivalence of parts~\ref{item.quasiregular.iff.between.1} and~\ref{item.quasiregular.iff.between.2}, we prove two implications:
\begin{itemize}
\item
Suppose $p$ is quasiregular, meaning by Definition~\ref{defn.tn}(\ref{item.quasiregular.point}) that $\community(p)\neq\varnothing$.
Pick some $p'\in\community(p)$ (it does not matter which).
It follows by construction in Definitions~\ref{defn.intertwined.points}(\ref{intertwined.defn}) and~\ref{defn.tn}(\ref{item.tn}) and Lemma~\ref{lemm.interior.open} that $p'\intertwinedwith p$, so that $p'\in\community(p)$. 
Using Definition~\ref{defn.intertwined.points}(\ref{item.p.intertwinedwith.p'}) it follows that $\community(p)\between O$ for every $O\in\nbhd(p)$, as required.
\item
Suppose $\community(p)\between\nbhd(p)$.
Then in particular $\community(p)\between\ns P$ (because $p\in\ns P\in\opens$), and by Notation~\ref{nttn.between}(\ref{item.between}) it follows that $\community(p)\neq\varnothing$.
\qedhere\end{itemize}
\end{proof}

Compare and contrast Theorem~\ref{thrm.regular=qr+sc} with Theorem~\ref{thrm.r=wr+uc}:
\begin{thrm}
\label{thrm.regular=qr+sc}
Suppose $(\ns P,\opens)$ is a semitopology and $p\in\ns P$.
Then the following are equivalent:
\begin{enumerate*}
\item
$p$ is regular.
\item
$p$ is quasiregular and hypertransitive.
\end{enumerate*}
\end{thrm}
\begin{proof}
We consider two implications:
\begin{itemize}
\item
\emph{Suppose $p$ is regular.}\quad

Then $p$ is quasiregular by Lemma~\ref{lemm.wr.r}(\ref{item.r.implies.wr}\&\ref{item.wr.implies.qr}), and hypertransitive by Lemma~\ref{lemm.regular.sc}(\ref{item.r.implies.sc}). 
\item
\emph{Suppose $p$ is quasiregular and hypertransitive.}\quad

By Lemma~\ref{lemm.regular.sc}(\ref{item.sc.implies.uc}) $p$ is unconflicted.
If we can prove that $p$ is weakly regular (meaning by Definition~\ref{defn.tn}(\ref{item.weakly.regular.point}) that $p\in\community(p)$), then by Theorem~\ref{thrm.r=wr+uc} it would follow that $p$ is regular as required.
Thus, it would suffice to show that $p\in\community(p)$, thus that there is an open neighbourhood of points with which $p$ is intertwined.

Write $O''=\interior(\ns P\setminus\community(p))$.
We have two subcases to consider:
\begin{itemize*}
\item
\emph{Suppose $\nbhd(p)\between O''$.}\quad

By Lemma~\ref{lemm.quasiregular.iff.between} (since $p$ is quasiregular) we have that $\community(p)\between\nbhd(p)$.
Thus $\community(p)\between\nbhd(p)\between O''$, and by hypertransitivity of $p$ it follows that $\community(p)\between O''$.
But this contradicts the construction of $O''$ as being a subset of $\ns P\setminus\community(p)$, so this case is impossible.
\item
\emph{Suppose $\nbhd(p)\notbetween O''$.}\quad
Then there exists some $O\in\nbhd(p)$ such that $O\notbetween O''$, and it follows that $O\subseteq\community(p)$ so that $p\in\community(p)$ as required.
\end{itemize*}
Thus $p$ is weakly regular, as required.
\qedhere\end{itemize}
\end{proof}

\begin{rmrk}
\label{rmrk.two.char.r}
Recall from Definition~\ref{defn.tn}(\ref{item.regular.point}) the notion of a \emph{regular point}, which in the theory of semitopologies is a canonical well-behavedness property because it ensures that the point is part of a community of other regular points; all in a single maximal intertwined actionable coalition (i.e. open set).
We have obtained two nice characterisations of regularity: 
\begin{enumerate*}
\item
Regular = weakly regular + unconflicted, by Theorem~\ref{thrm.r=wr+uc}. 
\item
Regular = quasiregular + hypertransitive, by Theorem~\ref{thrm.regular=qr+sc}. 
\end{enumerate*}
We are now ready to dualise everything --- including semitopologies, regularity, weak regularity, quasiregularity, and being unconflicted and hypertransitive --- and thus we will investigate the algebraic structures which correspond to the sets structures above:
\end{rmrk}

\section{Semiframes: compatible complete semilattices} 
\label{sect.semiframes}

\subsection{Complete join-semilattices, and morphisms between them} 

We recall some (mostly standard) definitions and facts:
\begin{defn}
\label{defn.complete.semilattice}
\leavevmode
\begin{enumerate*}
\item
A \deffont{poset} $(\ns X,\leq)$ is a set $\ns X$ of \deffont[elements (of a poset)]{elements} and a relation ${\leq}\subseteq\ns X\times\ns X$ that is transitive, reflexive, and antisymmetric.
\item
A poset $(\ns X,\leq)$ is a \deffont{complete join-semilattice} when every $X\subseteq\ns X$ ($X$ may be empty or equal to all of $\ns X$) has a least upper bound --- or \deffont[join (of elements in a poset)]{join} --- $\bigvee X\in\ns X$.

All the semilattices we consider will be join (rather than meet) semilattices, so we may omit the word `join' and just call this a \emph{complete semilattice} henceforth. 
\item\label{item.bot.X}
If $(\ns X,\leq)$ is a complete semilattice then we may write 
$$
\tbot_{\ns X}=\bigvee\varnothing.
$$
By the least upper bound property, $\tbot_{\ns X}\leq x$ for every $x\in\ns X$.
\item\label{item.top.X}
If $(\ns X,\leq)$ is a complete semilattice then we may write 
$$
\ttop_{\ns X}=\bigvee\ns X.
$$
By the least upper bound property, $x\leq \ttop_{\ns X}$ for every $x\in\ns X$.
\end{enumerate*}
\end{defn}

\begin{defn}
\label{defn.complete.semilattice.morphism}
Suppose $(\ns X',\leq')$ and $(\ns X,\leq)$ are complete join-semilattices.
Define a \deffont{morphism} $g:(\ns X',\leq')\to(\ns X,\leq)$ to be a function $\ns X'\to\ns X$ that commutes with joins, and sends $\ttop_{\ns X'}$ to $\ttop_{\ns X}$.
That is:
\begin{enumerate*}
\item\label{item.semilattice.morphism.join}
If $X'\subseteq\ns X'$ then $g(\bigvee X')=\bigvee_{x'{\in}\ns X'}g(x')$.
\item\label{item.semilattice.morphism.top}
$g(\ttop_{\ns X'})=\ttop_{\ns X}$.
\end{enumerate*} 
\end{defn}

\begin{rmrk}
In Definition~\ref{defn.complete.semilattice}(\ref{item.semilattice.morphism.top}) we insist that $g(\ttop_{\ns X'})=\ttop_{\ns X}$; i.e. we want our notion of morphism to preserve the top element.

This does not follow from Definition~\ref{defn.complete.semilattice}(\ref{item.semilattice.morphism.join}), because $g$ need not be surjective onto $\ns X$, so we need to add this as a separate condition.
Contrast with $g(\tbot_{\ns X})=\tbot_{\ns X'}$, which does follow from Definition~\ref{defn.complete.semilattice}(\ref{item.semilattice.morphism.join}), because $\tbot_{\ns X}$ is the least upper bound of $\varnothing$.

We want $g(\ttop_{\ns X'})=\ttop_{\ns X}$ because our intended model is that $(\ns X,\leq)=(\opens,\subseteq)$ is the semilattice of open sets of a semitopology $(\ns P,\opens)$, and similarly for $(\ns X',\leq')$, and $g$ is equal to $f^\mone$ where $f:(\ns P,\opens)\to(\ns P',\opens')$ is a continuous function. 
\end{rmrk}

We recall a standard result:
\begin{lemm}
\label{lemm.semi.hom.mon}
Suppose $(\ns X,\leq)$ is a complete join-semilattice. 
Then:
\begin{enumerate*}
\item\label{item.semi.hom.mon.1}
If $x_1,x_2\in\ns X$ then $x_1\leq x_2$ if and only if $x_1\tor x_2 = x_2$.
\item\label{item.semi.hom.mon.2}
If $f:(\ns X,\leq)\to(\ns X',\leq')$ is a semilattice morphism 
(Definition~\ref{defn.complete.semilattice.morphism})
then $f$ is a \deffont{monotone morphism}: if $x_1\leq x_2$ then $f(x_1)\leq f(x_2)$, for every $x_1,x_2\in\ns X$.
\end{enumerate*}
\end{lemm}
\begin{proof}
We consider each part in turn:
\begin{enumerate}
\item
Suppose $x_1\leq x_2$.
By the definition of a least upper bound, this means precisely that $x_2$ is a least upper bound for $\{x_1,x_2\}$.
It follows that $x_1\tor x_2=x_2$.
The converse implication follows just by reversing this reasoning.
\item
Suppose $x_1\leq x_2$.
By part~\ref{item.semi.hom.mon.1} of this result $x_1\tor x_2=x_2$, so $f(x_1\tor x_2)=f(x_2)$.
By Definition~\ref{defn.complete.semilattice.morphism} $f(x_1)\tor f(x_2)=f(x_2)$.
By part~\ref{item.semi.hom.mon.1} of this result $f(x_1)\leq f(x_2)$.
\qedhere\end{enumerate}
\end{proof}

\begin{rmrk}
As the reader may know, \emph{frames} and \emph{locales} are the same thing: the category of locales is just the categorical opposite of the category of frames.
So every time we write `semiframe', the reader can safely read `semilocale'; these are two names for essentially the same structure up to reversing arrows.
The literature on frames and locales is huge: the interested reader can consult two classic texts~\cite{johnstone:stos,maclane:sheglf}; more recent (and very readable) presentations include~\cite{picado:fraltw,picado:seppft}.
\end{rmrk}

\jamiesubsection{The compatibility relation} 
\label{subsect.compatibility.relation}

Definition~\ref{defn.compatibility.relation} is a simple idea, but so far as we are aware it is novel:
\begin{defn}
\label{defn.compatibility.relation}
Suppose $(\ns X,\leq)$ is a complete semilattice.
A \deffont{compatibility relation ${\ast}\subseteq\ns X\times\ns X$}\index{$x\ast x'$ (compatibility relation on elements)} is a relation on $\ns X$ such that: 
\begin{enumerate*}
\item\label{item.compatible.symmetric}
$\ast$ is \emph{symmetric}, so if $x,x'\in\ns X$ then 
$$
x\ast x' 
\quad\text{if and only if}\quad
x'\ast x. 
$$
\item\label{item.compatible.reflexive}
$\ast$ is a \deffont{properly reflexive relation},\footnote{`Properly reflexive' is a loose riff on terminologies like `proper subset of' or `proper ideal of a ring'.  We might also call this `non-$\tbot$ reflexive', which is descriptive, but perhaps a bit of a mouthful.}
 by which we mean 
$$
\Forall{x\in\ns X{\setminus}\{\tbot_{\ns X}\}} x\ast x .
$$
Note that it will follow from the axioms of a compatibility relation that $x\ast x\liff x\neq\tbot_{\ns X}$; see Lemma~\ref{lemm.tbot.incompatible}(\ref{item.properly.reflexive.iff}).
\item\label{item.compatible.distributive}
$\ast$ satisfies a \deffont[distributive law (for compatibility relation)]{distributive law}, that 
if $x\in \ns X$ and $X'\subseteq\ns X$ then
$$
x\ast\bigvee X' \liff \Exists{x'{\in}X'} x\ast x' .
$$
\end{enumerate*}
Thus we can say:
\begin{quoting}
a compatibility relation ${\ast}\subseteq\ns X\times\ns X$ is a symmetric properly reflexive completely distributive relation on $\ns X$. 
\end{quoting}
When $x\ast x'$ holds, we may call $x$ and $x'$ \deffont{compatible elements}.
\end{defn}

\begin{rmrk}
\label{rmrk.compatibility.intuition}
The compatibility relation $\ast$ is what it is, but we take a moment to discuss some intuitions, and to put it in the context of some natural generalisations:
\begin{enumerate}
\item\label{item.abstract.intersection}
We can think of $\ast$ as an \emph{abstract intersection}.

It lets us observe whether $x$ and $x'$ intersect --- but without having to explicitly represent this intersection as a meet $x\tand x'$ in the semilattice itself.

We call $\ast$ a \emph{compatibility relation} following an intuition of $x,x'\in\ns X$ as observations, and $x\ast x'$ holds when there is some possible world at which it is possible to observe $x$ and $x'$ together.
More on this in Example~\ref{xmpl.simple.concrete.model}.
\item
We can think of $\ast$ as a \emph{generalised intersection}; so our notion of semiframe in Definition~\ref{defn.semiframe}
is an instance of a frame with a \emph{generalised} meet.

We will concentrate on the case where $x \ast x'$ measures whether $x$ and $x'$ intersect, but there are other possibilities.
Here are some natural ways to proceed:
\begin{enumerate}
\item
$(\ns X,\cti)$ is a complete join-semilattice and ${\ast} : (\ns X\times\ns X)\to \ns X$ is any commutative distributive map.
For concreteness, we can set $x\ast x' \in\{\tbot_{\ns X},\ttop_{\ns X}\}\subseteq\ns X$.
\item 
$(\ns X,\cti)$ is a complete join-semilattice and ${\ast} : (\ns X\times\ns X)\to \mathbb N$ is any commutative distributive map.
We think of $x\ast x'$ as returning the \emph{size} of the intersection of $x$ and $x'$.
\item
Any complete join-semilattice $(\ns X,\cti)$ is of course a (generalised) semiframe by taking $x\ast x' = \bigvee\{x'' \mid x''\cti x,\ x''\cti x'\}$.
\item
We can generalise further, in more than one direction.
We would take $(\ns X,\cti)$ and $(\ns X',\cti')$ to be complete join-semilattices and ${\ast} : (\ns X\times\ns X)\to \ns X'$ to be any commutative distributive map (which generalises the above).
We could also take $\ns X$ to be a cocomplete symmetric monoidal category~\cite[Section~VII]{maclane:catwm}: a category with all colimits and with a (symmetric) monoid action $\ast$ that distributes over (commutes with) colimits. 
\item
The compatibility relation $\ast$ is binary: as noted in part~\ref{item.abstract.intersection} of this Remark
it abstracts two open sets having a nonempty open intersection.

However, the literature on distributed systems is rich in well-behavedness conditions based on intersections of sets: e.g. ones corresponding to nonempty ternary intersections between open sets, and on intersections between open and closed sets, and on other types of intersections.
These are not thought of in the literature in topological/algebraic terms --- that way of looking at things is a contribution of this paper and the recent~\cite{gabbay:semtad} --- but they could be.
 
This suggests that $\ast$ may be a (canonical simplest) representative of a design space of `nonempty $n$-ary intersection' operators, which remains to be explored.
\end{enumerate}
\end{enumerate}
More details, with some references, are in Remark~\ref{rmrk.generalising.ast}.
\end{rmrk}

\begin{lemm}
\label{lemm.compatibility.monotone}
Suppose $(\ns X,\leq)$ is a complete semilattice and suppose ${\ast}\subseteq\ns X\times\ns X$ is a compatibility relation on $\ns X$.
Then:
\begin{enumerate*}
\item\label{item.ast.monotone}
$\ast$ is monotone on both arguments.

That is: if $x_1\ast x_2$ and $x_1\cti x_1'$ and $x_2\cti x_2'$, then $x_1'\ast x_2'$. 
\item\label{item.ast.lower.bound}
If $x_1,x_2\in\ns X$ have a non-$\tbot$ lower bound $\tbot_{\ns X}\lneq x\leq x_1,x_2$, then $x_1\ast x_2$.

In words we can write: $\ast$ reflects non-$\tbot$ lower bounds.
\item
The converse implication to part~\ref{item.ast.lower.bound} need not hold: it may be that $x_1\ast x_2$ ($x_1$ and $x_2$ are compatible) but the greatest lower bound of $\{x_1,x_2\}$ is $\tbot$.
\end{enumerate*}
\end{lemm}
\begin{proof}
We consider each part in turn:
\begin{enumerate}
\item
We argue much as for Lemma~\ref{lemm.semi.hom.mon}(\ref{item.semi.hom.mon.1}).
Suppose $x_1\ast x_2$ and $x_1\cti x_1'$ and $x_2\cti x_2'$.
By Lemma~\ref{lemm.semi.hom.mon} $x_1\tor x_1'=x_1'$ and $x_2\tor x_2'=x_2'$.
It follows using distributivity and commutativity (Definition~\ref{defn.compatibility.relation}(\ref{item.compatible.distributive}\&\ref{item.compatible.symmetric})) that $x_1\ast x_2$ implies that $(x_1\tor x_1')\ast (x_2\ast x_2')$, and thus that $x_1'\ast x_2'$ as required.
\item
Suppose $\tbot_{\ns X}\neq x \leq x_1,x_2$, so $x$ is a non-$\tbot_{\ns X}$ lower bound.
By assumption $\ast$ is properly reflexive (Definition~\ref{defn.compatibility.relation}(\ref{item.compatible.reflexive})) so (since $x\neq\tbot_{\ns X}$) $x\ast x$.
By part~\ref{item.ast.monotone} of this result it follows that $x_1\ast x_2$ as required.
\item
It suffices to provide a counterexample.
Define $(\ns X,\cti,\ast)$ by: 
\begin{itemize*}
\item
$\ns X = \{\tbot,0,1,\ttop\}$.
\item
$\tbot\cti 0,1 \cti \ttop$ and $\neg(0\cti 1)$ and $\neg(1\cti 0)$.
\item
$x\ast x'$ for every $\tbot\neq x,x'\in\ns X$.
\end{itemize*}
We note that $0\ast 1$ but the greatest lower bound of $\{0,1\}$ is $\tbot$.
We will revisit a slightly more elaborate version of this counterexample in Figure~\ref{fig.ast.no.tand}.
\qedhere\end{enumerate}
\end{proof}

\jamiesubsection{The definition of a semiframe}

\begin{defn}
\label{defn.semiframe}
A \deffont{semiframe} is a tuple $(\ns X,\cti,\ast)$ such that 
\begin{enumerate*}
\item
$(\ns X,\cti)$ is a complete semilattice (Definition~\ref{defn.complete.semilattice}), and 
\item
$\ast$ is a compatibility relation on it (Definition~\ref{defn.compatibility.relation}).
\end{enumerate*}
Slightly abusing terminology, we can say that 
\begin{quoting}
semiframe = \emph{compatible complete semilattice}.
\end{quoting}
\end{defn}

Semiframes are new, so far as we know, but they are a natural idea.
We consider some elementary ways to generate examples, starting with arguably the simplest possible instance:
\begin{xmpl}[The empty semiframe]
\label{xmpl.empty.semiframe}
Suppose $(\ns X,\cti,\ast)$ is a semiframe.
\begin{enumerate*}
\item
If $\ns X$ is a singleton set, so that $\ns X=\{\bullet\}$ for some element $\bullet$, then we call $(\ns X,\cti,\ast)$ the \deffont{empty semiframe} or \deffont{singleton semiframe}.

Then necessarily $\bullet=\tbot_{\ns X}=\ttop_{\ns X}$ and $\bullet\cti\bullet$ and $\neg(\bullet\ast\bullet)$.
\item\label{item.nonempty.semiframe}
If $\ns X$ has more than one element then we call $(\ns X,\cti,\ast)$ a \deffont{nonempty semiframe}.
Then necessarily $\tbot_{\ns X}\neq\ttop_{\ns X}$.
\end{enumerate*}
Thus, $(\ns X,\cti,\ast)$ is nonempty if and only if $\tbot_{\ns X}\neq\ttop_{\ns X}$.
We call a singleton semiframe \emph{empty}, because this corresponds to the semiframe of open sets of the empty topology, which has no points and one open set, $\varnothing$.
\end{xmpl}

Example~\ref{xmpl.simple.concrete.model} continues Remark~\ref{rmrk.compatibility.intuition}:
\begin{xmpl}
\label{xmpl.simple.concrete.model}
\leavevmode
\begin{enumerate*}
\item\label{item.open.semiframe}
Suppose $(\ns P,\opens)$ is a semitopology.
Then the reader can check that the \emph{semiframe of open sets} $(\ns P,\subseteq,\between)$ is a semiframe.
We will study this example in detail; see Definition~\ref{defn.semi.to.dg} and Lemma~\ref{lemm.Fr.semiframe}.
\item
Suppose $(\ns X,\leq,\tbot,\ttop)$ is a frame (a complete lattice such that meets distribute over arbitrary joins).
Then $(\ns X,\leq,\ast)$ is a semiframe, where $x\ast x'$ when $x\tand x'\neq\tbot$.\footnote{Just being a complete lattice is not enough; it has to be distributive as well.  Consider $\omega\plus 1=\mathbb N\cup\{\omega\}$ with its usual ordering, augmented with an element $d$ such that $0\leq d\leq\omega$.  Then $\omega=\bigvee\mathbb N$ and $d\ast \omega$, but $\neg(d\ast n)$ for every $n\in\mathbb N$.} 
\item
Take $\ns X=\{\tbot,0,1,\ttop\}$ with $\tbot\cti 0\cti\ttop$ and $\tbot\cti 1\cti \ttop$ (so $0$ and $1$ are incomparable).
There are two possible semiframe structures on this, characterised by choosing $0\ast 1$ or $\neg(0\ast 1)$.
\item
See also the semiframes used in Lemmas~\ref{lemm.no.converge}.
\end{enumerate*}
\end{xmpl}

Definition~\ref{defn.semi.to.dg} is just an example of semiframes for now, though we will see much more of it later: 
\begin{defn}[{\bf Semitopology $\to$ semiframe}]
\label{defn.semi.to.dg}
Suppose $(\ns P,\opens)$ is a semitopology.
Define the \deffont{semiframe of open sets $\tf{Fr}(\ns P,\opens)$} (cf. Example~\ref{xmpl.simple.concrete.model}(\ref{item.open.semiframe}))
by:
\begin{enumerate*}
\item
$\tf{Fr}(\ns P,\opens)$ has elements open sets $O\in\opens$.
\item
$\cti$ is subset inclusion.
\item\label{item.semiframe.ast}
$\ast$ is $\between$ (sets intersection).
\end{enumerate*}
We may write 
$$
(\opens,\subseteq,\between)
\quad\text{as a synonym for}\quad
\tf{Fr}(\ns P,\opens) .
$$
\end{defn}

\begin{lemm}
\label{lemm.Fr.semiframe}
Suppose $(\ns P,\opens)$ is a semitopology.
Then $(\opens,\subseteq,\between)$ is indeed a semiframe.
\end{lemm}
\begin{proof}
As per Definition~\ref{defn.semiframe} we must show that $\opens$ is a complete semilattice (Definition~\ref{defn.complete.semilattice})
and $\between$ is a compatibility relation (Definition~\ref{defn.compatibility.relation})
--- symmetric, properly reflexive, and distributive and satisfies a distributive law that if $O\between\bigcup\mathcal O'$ then $O\between O'$ for some $O'\in\mathcal O'$. 
These are all facts of sets.
\end{proof}

\begin{rmrk}
\label{rmrk.setting.the.scene.semiframes}
Definition~\ref{defn.semi.to.dg} and Lemma~\ref{lemm.Fr.semiframe} are the start of our development.
Once we have built more machinery, we will have a pair of translations:
\begin{itemize*}
\item
Definition~\ref{defn.semi.to.dg} and Lemma~\ref{lemm.Fr.semiframe} go from semitopologies to semiframes.
\item
Definition~\ref{defn.st.g} and Lemma~\ref{lemm.St.semitop} go from semiframes to semitopologies.
\end{itemize*}
These translations are part of a dual pair of functors between categories of semitopologies and semiframes, as described in Definitions~\ref{defn.morphism.semitopologies} and~\ref{defn.category.of.spatial.graphs} and Proposition~\ref{prop.semitop.adjunction}.

Semitopologies are (relatively) concrete: we have concrete points and open sets that are sets of points.
Semiframes are more abstract: we have a join-complete semilattice, and a compatibility relation.
The duality we are about to build will show how these two worlds interact and reflect each other.
\end{rmrk}

We conclude with a simple technical lemma:
\begin{lemm}
\label{lemm.tbot.incompatible}
Suppose $(\ns X,\leq,\ast)$ is a semiframe (a complete semilattice with a compatibility relation) and $x\in\ns X$.
Then:
\begin{enumerate*}
\item\label{item.tbot.incompatible.tbot}
$\neg(x\ast\tbot_{\ns X})$ and in particular $\neg(\tbot_{\ns X}\ast\tbot_{\ns X})$.
\item\label{item.properly.reflexive.iff}
$x\ast x$ if and only if $x\neq\tbot_{\ns X}$.
\item
$x\ast\ttop_{\ns X}$ if and only if $x\neq\tbot_{\ns X}$.
\item
$\ttop_{\ns X}\ast\ttop_{\ns X}$ holds precisely if $\ns X$ is nonempty (Example~\ref{xmpl.empty.semiframe}).
\end{enumerate*}
\end{lemm}
\begin{proof}
We consider each part in turn:
\begin{enumerate}
\item
Recall from Definition~\ref{defn.complete.semilattice}(\ref{item.bot.X}) that $\tbot_{\ns X}=\bigvee\varnothing$.
By distributivity (Definition~\ref{defn.compatibility.relation}(\ref{item.compatible.distributive}))
$$
x\ast\tbot_{\ns X}\liff \Exists{x'\in \varnothing}x\ast x' \liff \bot.
$$
\item
We just combine part~\ref{item.tbot.incompatible.tbot} of this result with 
Definition~\ref{defn.compatibility.relation}(\ref{item.compatible.reflexive}).
\item
Suppose $x\neq\tbot_{\ns X}$.
Then $\tbot_{\ns X}\lneq x\leq x\leq\ttop_{\ns X}$, and by Lemma~\ref{lemm.compatibility.monotone}(\ref{item.ast.lower.bound}) $x\ast\ttop_{\ns X}$.

Suppose $x=\tbot_{\ns X}$.
Then $\neg(x\ast\ttop_{\ns X})$ by combining commutativity of $\ast$ (Definition~\ref{defn.compatibility.relation}(\ref{item.compatible.symmetric})) with part~\ref{item.tbot.incompatible.tbot} of this result.
\item
If $\ns X$ is nonempty then by Example~\ref{xmpl.empty.semiframe} $\tbot_{\ns X}\neq\ttop_{\ns X}$ and so $\ttop_{\ns X}\ast\ttop_{\ns X}$ holds by part~\ref{item.properly.reflexive.iff} of this result.
However, in the degenerate case that $\ns X$ has one element then $\ttop_{\ns X}=\tbot_{\ns X}$ and $\ttop_{\ns X}\ast\ttop_{\ns X}$ does not hold. 
\qedhere\end{enumerate}
\end{proof}

\begin{rmrk}
Recall from \cite[Definition~5.22, page~128]{priestley:intlo} that if $\ns X$ is a lattice, then the \deffont{pseudocomplement} to $x\in\ns X$ is $x^*=\bigvee\{x'\in\ns X\mid x'\wedge x=\tbot\}$.
A semiframe $(\ns X,\cti,\ast)$ naturally supports a notion of pseudocomplement for $x\in\ns X$, given by 
$$
x^c=\bigvee\{x'\in\ns X\mid \neg(x'\ast x)\}.
$$
It is easy to prove that $\neg(x^c\ast x)$, arguing by contradiction: if $x^c\ast x$ then $\bigvee\{x'\mid \neg(x'\ast x)\}\ast x$, and by distributivity (Definition~\ref{defn.compatibility.relation}(\ref{item.compatible.distributive})) there exists $x'\in\ns X$ such that $x'\ast x$ and $\neg(x'\ast x)$, a contradiction.

Note that it may be that $(x^c)^c\lneq x$.
For example, in the semiframe illustrated in Figure~\ref{fig.ast.no.tand}, $0^c=\bigvee\{1,2,3\}=\ttop$ and $(0^c)^c=\tbot\lneq 0$ (this behaviour will be familiar to the reader who has seen, for example, double negation in intuitionistic logic).

$x^c$ and related constructions will be useful later, in Definition~\ref{defn.cast} and Lemma~\ref{lemm.cast.comp}. 
\end{rmrk}

\jamiesection{Semifilters \& abstract points}
\label{sect.semifilters.and.points}

\jamiesubsection{The basic definition, and discussion}

\begin{defn}
\label{defn.point}
Suppose $(\ns X,\cti,\ast)$ is a semiframe and suppose $\afilter\subseteq\ns X$.
Then: 
\begin{enumerate*}
\item\label{item.prime}
Call $\afilter$ \deffont[prime subset of a semiframe]{prime} when for every $x,x'\in\ns X$,
$$
x\tor x'\in \afilter 
\quad\text{implies}\quad
x\in \afilter\lor x'\in \afilter .
$$
\item\label{item.completely.prime}
Call $\afilter$ \deffont[completely prime subset of a semiframe]{completely prime} when for every (possibly empty) $X\subseteq\ns X$,
$$
\bigvee X\in \afilter
\quad\text{implies}\quad
\Exists{x{\in}X}x\in \afilter.
$$
(This condition is used in Lemma~\ref{lemm.op.commutes.with.joins}, which is needed for Lemma~\ref{lemm.Op.unions}.) 
\item\label{item.up-closed}
Call $\afilter$ \deffont[up-closed subset of a semiframe]{up-closed} when $x\in \afilter$ and $x\cti x'$ implies $x'\in \afilter$.
\item\label{item.weak.clique}
Call $\afilter$ \deffont[compatible subset of a semiframe]{compatible} when its elements are \deffont[pairwise compatible subset of a semiframe]{pairwise compatible}, by which we mean that $x\ast x'$ for every $x, x'\in \afilter$.
\item\label{item.semifilter}
A \deffont{semifilter} is a nonempty, up-closed, compatible subset $\afilter\subseteq\ns X$.
\item\label{item.maximal.semifilter}
Call $\afilter\subseteq\ns X$ a \deffont{maximal semifilter} when it is a semifilter and is contained in no strictly greater semifilter.
\item\label{item.abstract.point}
An \deffont[abstract point (completely prime semifilter)]{abstract point} is a completely prime semifilter.
\item
Write 
$$
\tf{Points}(\ns X,\cti,\ast)
$$
for the set of abstract points of $(\ns X,\cti,\ast)$.
\end{enumerate*}
\end{defn}

\begin{nttn}
We will generally write $\afilter\subseteq\ns X$ for a subset of $\ns X$ that is intended to be a semifilter, or for which in most cases of interest $\afilter$ is a semifilter.
We will generally write $\apoint\subseteq\ns X$ when the subset is intended to be an abstract point, or when in most cases of interest $\apoint$ is an abstract point.
\end{nttn}

\begin{rmrk}
\label{rmrk.no.meet}
\emph{Note on design:}
The notion of semifilter from Definition~\ref{defn.point} is, obviously, based on the standard notion of filter~\cite[I.2.2, page~12]{johnstone:stos}.
We just replace the \emph{closure under binary meets} condition 
\begin{quoting}
`if $x,x'\in \afilter$ then $x\tand x'\in\afilter$' 
\end{quoting}
with a weaker \emph{compatibility condition}
\begin{quoting}
`if $x,x'\in\afilter$ then $x\ast x'$'.
\end{quoting}
This is in keeping with our move from frames to semiframes, which weakens from $\tand$ to the compatibility relation $\ast$.

Note that a semifilter or abstract point need not be directed:
\begin{enumerate*}
\item
Consider $\nbhd(0)$ in the (semiframes of open sets of the) semitopologies in the left-hand and middle examples in Figure~\ref{fig.nbhd}.
In both cases, $\{0,1\},\{0,2\}\in\nbhd(0)$ but $\{0\}\notin\nbhd(0)$ because $\{0\}$ is not an open set.
\item
Consider $\{0,1,2\}$ with the discrete semitopology (so every set is open).
Then the set of all two- or three-element subsets $\{\{0,1\},\{1,2\},\{2,0\},\{0,1,2\}\}$ is a semifilter, but it is not closed under sets intersections because it does not contain $\{0\}$, $\{1\}$, or $\{2\}$.
\end{enumerate*}
This second example is particularly interesting.
As the reader may know, the intuition of a filter in topology is a set of \emph{approximations}.
But this example is clearly not approximating anything --- after all, we are in the discrete semitopology and there is no need to approximate anything since we can just take a singleton set!
This suggests that a better intuition for semiframe is a set of \emph{collaborations}; in this case, of $0$ with $1$, $1$ with $2$, and $2$ with $0$.

Thus in particular, the standard result in frames that a proper finite filter\footnote{Recall that a proper filter is a filter that does not contain $\tbot$.} has a non-$\tbot$ least element (obtained as the finite meet of all the elements in the filter), does not hold for semifilters in semiframes.
See also Remark~\ref{rmrk.other.properties} and Proposition~\ref{prop.non.tbot.lower.bound}(\ref{item.tbot.lower.bound.semifilter}).
\end{rmrk}

\begin{xmpl}
\label{xmpl.abstract.point}
Suppose $(\ns X,\cti,\ast)$ is a semiframe.
We recall some (standard) facts about abstract points, which carry over from topologies and frames:
\begin{enumerate*}
\item\label{item.xmpl.nbhd.abstract.point}
Suppose $(\ns P,\opens)$ is a semitopology and $(\ns X,\cti,\ast)=(\opens,\subseteq,\between)$. 
By Lemma~\ref{lemm.Fr.semiframe}, $(\ns X,\cti,\ast)$ is a semiframe. 

If $p\in\ns P$ then 
$$
\nbhd(p)=\{O\in\opens \mid p\in O\}
$$ 
from Definition~\ref{defn.nbhd} is an abstract point: see Proposition~\ref{prop.nbhd.iff}.
Intuitively, $\nbhd(p)$ abstractly represents $p$ as the set of all of its open approximations in $\opens$.
\item
Suppose $(\ns P,\opens)$ is a semitopology.
Then $(\opens,\subseteq,\between)$ could contain an abstract point that is not the neighbourhood semifilter $\nbhd(p)$ of a point $p\in\ns P$.
 
Set $\ns X=\{\mathbb Q\}\cup\{(\pi\minus q,\pi\plus q)\subseteq\mathbb Q \mid q\in\mathbb Q_{\geq 0}\}$ (the set of all symmetric open intervals around $\pi$ in the rational numbers $\mathbb Q$), and set ${\cti}={\subseteq}$ and ${\ast}={\between}$. 

Set $P=\ns X\setminus\{\varnothing\}$ to be the set of all \emph{nonempty} symmetric open intervals around $\pi$.
Note that $\pi\notin\mathbb Q$, but $P$ is a set of open sets `approximating' $\pi$. 
\item
We mention one more (standard) example.
Consider $\mathbb N$ with the \deffont{final segment semitopology} such that opens are either $\varnothing$ or sets $n_\geq = \{n'\in\mathbb N \mid n'\geq n\}$.
Then $\{n_\geq \mid n\in\mathbb N\}$ is an abstract point.
Intuitively, this approximates a point at infinity, which we can understand as $\omega$. 
\end{enumerate*}
\end{xmpl}

\begin{figure}
\centering
\includegraphics[align=c,width=0.3\columnwidth,trim={50 20 50 20},clip]{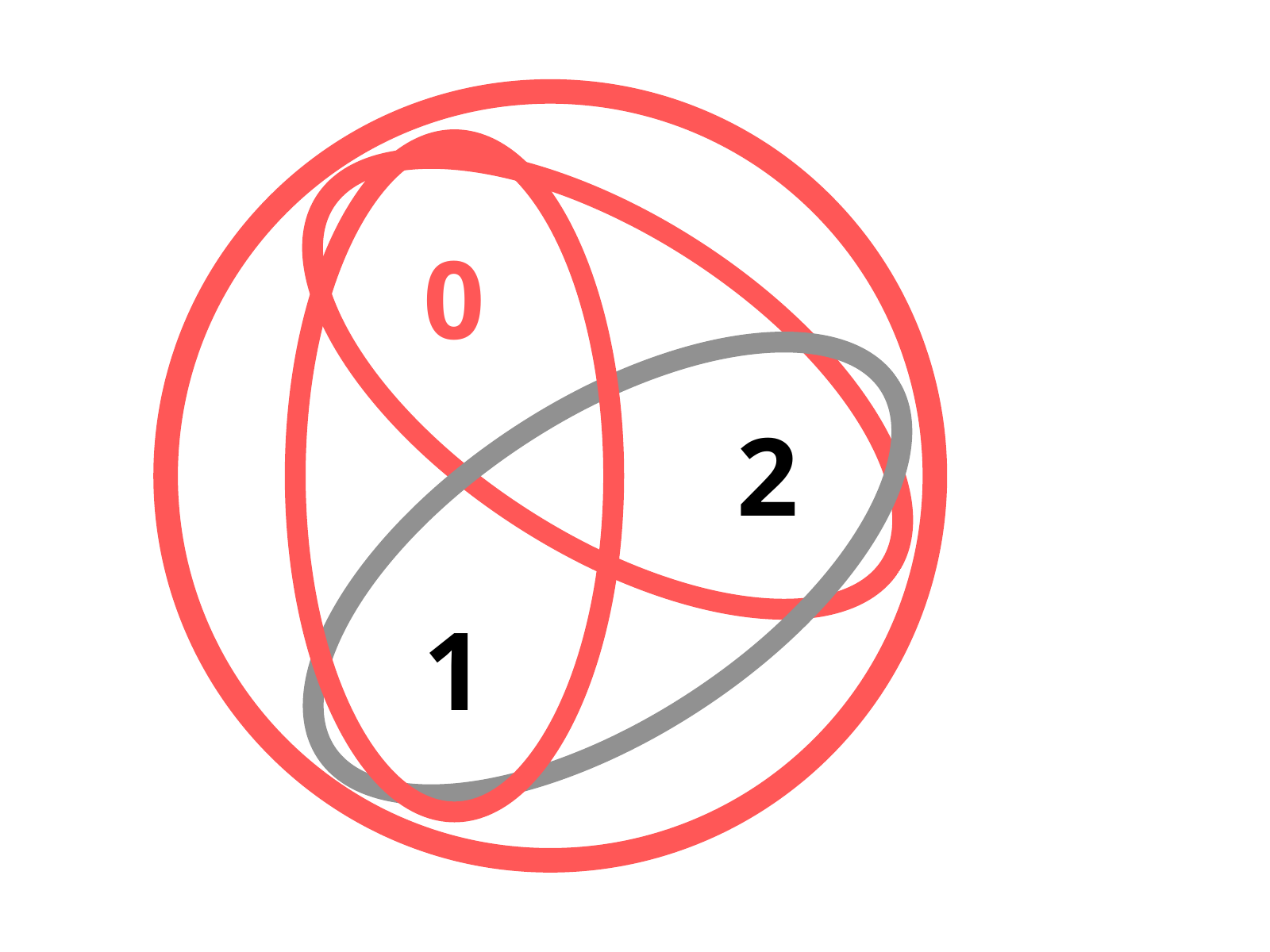}
\quad
\includegraphics[align=c,width=0.3\columnwidth,trim={50 20 50 80},clip]{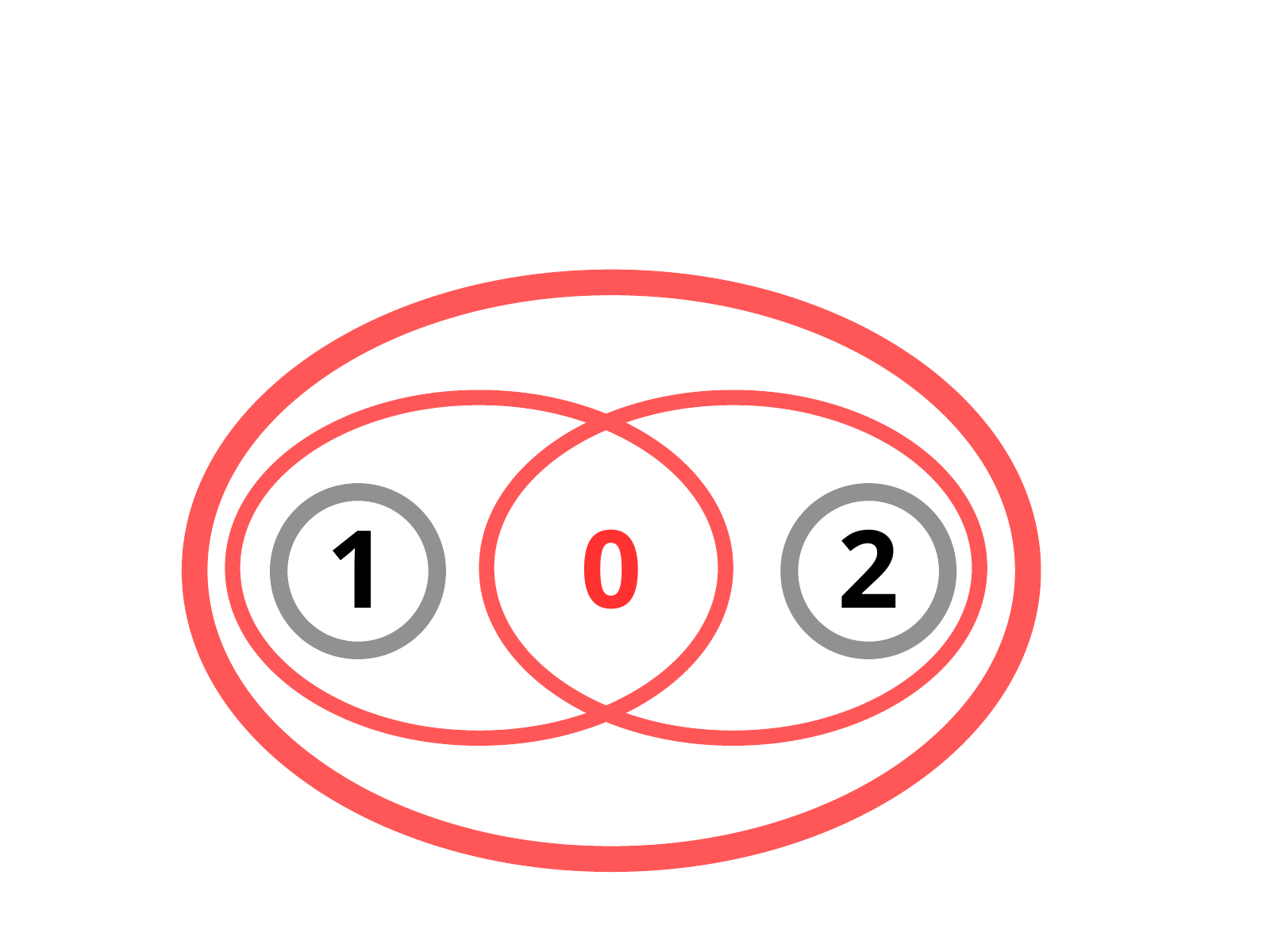}
\quad
\includegraphics[align=c,width=0.3\columnwidth,trim={50 20 50 20},clip]{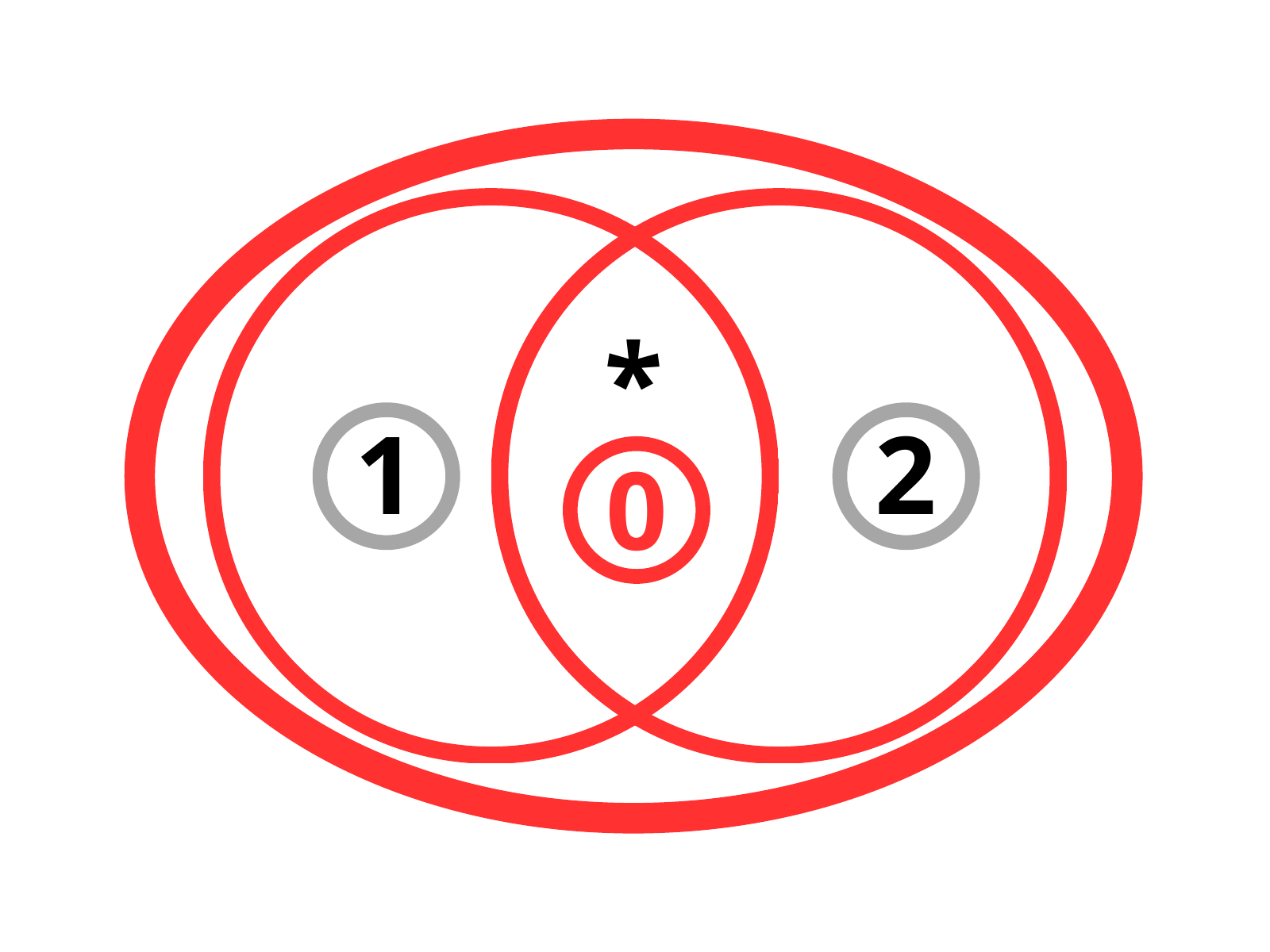}
\\[4ex]
\caption{Examples of open neighbourhoods (Remarks~\ref{rmrk.no.meet} and~\ref{rmrk.nbhd.filter})}
\label{fig.nbhd}
\end{figure}

\begin{lemm}
\label{lemm.compatible.not.tbot}
Suppose $(\ns X,\cti,\ast)$ is a semiframe and suppose $\afilter\subseteq\ns X$ is compatible.
Then $\tbot_{\ns X}\notin\afilter$.
\end{lemm}
\begin{proof}
By compatibility, $x\ast x$ for every $x\in\afilter$.
We use Lemma~\ref{lemm.tbot.incompatible}(\ref{item.tbot.incompatible.tbot}).
\end{proof}

\begin{rmrk}
\label{rmrk.other.properties}
We continue Remark~\ref{rmrk.no.meet}.

As the reader may know, a semiframe still has greatest lower bounds, because we can build them as $x\tand x' = \bigvee \{x'' \mid x''\cti x,\ x''\cti x'\}$.
It is just that this greatest lower bound may be unhelpful.
To see why, consider again the examples in Figure~\ref{fig.nbhd}.
In the left-hand and middle examples in Figure~\ref{fig.nbhd}, the greatest lower bound of $\{0,1\}$ and $\{0,2\}$ exists in the semiframe of open sets: but it is $\varnothing$ the empty set in the left-hand and middle example, not $\{0\}$.
In the right-hand example, the greatest lower bound of $\{0,\ast,1\}$ and $\{0,\ast,2\}$ is $\{0\}$, not $\{0,\ast\}$.

So the reader could ask whether perhaps we should add the following weakened meet-closure condition to the definition of semifilters (and thus to abstract points):
\begin{quoting}
\emph{If $x, x'\in \afilter$ and $x\tand x'\neq\tbot$ then $x\tand x'\in \afilter$.}
\end{quoting}
Intuitively, this insists that semifilters are closed under \emph{non-$\tbot$} greatest lower bounds. 
However, there are two problems with this:
\begin{itemize*}
\item
It would break our categorical duality proof in the construction of $g^\circ$ in Lemma~\ref{lemm.gcirc.well-defined}; see the discussion in Remark~\ref{rmrk.further.restrictions.on.points}.
This technical difficulty may be superable, but \dots
\item
\dots the condition is probably not what we want anyway.
It would mean that the set of open neighbourhoods of $\ast$ in the right-hand example of Figure~\ref{fig.nbhd}, would not be a semifilter, because it contains $\{0,\ast,1\}$ and $\{0,\ast,2\}$ but not its (non-$\varnothing$) greatest lower bound $\{0\}$.
\end{itemize*}
\end{rmrk}

\jamiesubsection{Properties of semifilters}

\jamiesubsubsection{Things that are familiar from filters}
\label{subsect.familiar}

\begin{lemm}
\label{lemm.P.top}
Suppose $(\ns X,\cti,\ast)$ is a semiframe and $\afilter\subseteq\ns X$ is a semifilter.
Then:
\begin{enumerate*}
\item\label{item.P.yes.top}
$\ttop_{\ns X}\in \afilter$.
\item\label{item.P.no.bot}
$\tbot_{\ns X}\notin \afilter$.
\end{enumerate*}
\end{lemm}
\begin{proof}
We consider each part in turn:
\begin{enumerate}
\item
By nonemptiness (Definition~\ref{defn.point}(\ref{item.abstract.point})) $\afilter$ is nonempty, so there exists some $x\in \afilter$.
By definition $x\cti \ttop_{\ns X}$.
By up-closure (Definition~\ref{defn.point}(\ref{item.up-closed})) $\ttop_{\ns X}\in \afilter$ follows. 
\item
By assumption in Definition~\ref{defn.point}(\ref{item.weak.clique}) elements in $\afilter$ are pairwise compatible (so $x\ast x$ for every $x\in\afilter$).
We use Lemma~\ref{lemm.compatible.not.tbot}.
\qedhere\end{enumerate} 
\end{proof}

\begin{lemm}
\label{lemm.maxfilter.vs.primefilter}
Suppose $(\ns X,\cti,\ast)$ is a semiframe.
It is possible for a semifilter $\afilter\subseteq\ns X$ to be completely prime but not maximal.
\end{lemm}
\begin{proof}
We give a standard example (which also works for frames and filters).
Take $\ns P=\{0,1\}$ and $\opens=\{\varnothing, \{0\}, \{0,1\}\}$.
Then $\apoint'=\{\{0,1\}\}$ is an abstract point, but it is not a maximal semifilter (it is not even a maximal abstract point) since $\apoint'$ is contained in the strictly larger semifilter $\{\{0\},\{0,1\}\}$ (which is itself also a strictly larger abstract point).
\end{proof}

\begin{lemm}
\label{lemm.prime.completely.prime.finite}
If $(\ns X,\cti,\ast)$ is a finite semiframe (meaning that $\ns X$ is finite) then the properties of 
\begin{itemize*}
\item
being a prime semifilter (Definition~\ref{defn.point}(\ref{item.prime})) and 
\item
being a completely prime semifilter (Definition~\ref{defn.point}(\ref{item.completely.prime})), 
\end{itemize*}
coincide.
\end{lemm}
\begin{proof}
This is almost trivial, except that if $X=\varnothing$ in the condition for being completely prime then we get that $\tbot_{\ns X}\notin P$ --- but we know that anyway from Lemma~\ref{lemm.P.top}(\ref{item.P.no.bot}), from the compatibility condition on semifilters.
\end{proof}

\begin{lemm}
\label{lemm.zorn.maximal.semifilter}
Suppose $(\ns X,\cti,\ast)$ is a semiframe.
Then:
\begin{enumerate*}
\item
The union of an ascending chain of semifilters in $\ns X$, is a semifilter in $\ns X$.
\item
As a corollary, every semifilter $F\subseteq\ns X$ is contained in some maximal semifilter $F'\subseteq\ns X$ (assuming Zorn's lemma).
\end{enumerate*}
\end{lemm} 
\begin{proof}
We consider each part in turn:
\begin{enumerate}
\item
By a straightforward verification of the conditions of being a semifilter from Definition~\ref{defn.point}(\ref{item.semifilter}).
\item
Direct application of Zorn's lemma.
\qedhere\end{enumerate}
\end{proof}

\begin{rmrk}
\label{rmrk.tbot.not.in}
\leavevmode
\begin{enumerate*}
\item\label{item.tbot.not.in.1}
Lemma~\ref{lemm.P.top}(\ref{item.P.no.bot}) has a small twist to it.
In the theory of \emph{filters}, it does not follow from the property of being nonempty, up-closed, and closed under finite meets, that $\tbot_{\ns X}\notin\afilter$; this must be added as a distinct condition if required.

In contrast, we see in the proof of Lemma~\ref{lemm.P.top}(\ref{item.P.no.bot}) that for semifilters, $\tbot_{\ns X}\notin\afilter$ follows from the compatibility condition.
\item\label{item.tbot.not.in.2}
Lemma~\ref{lemm.prime.completely.prime.finite} matters in particular to us here, because we are particularly interested in abstracting the behaviours of finite semitopologies, because our original motivation for looking at both of these structures comes from looking at real networks, which are finite.\footnote{This is carefully worded.  We care about \emph{abstracting} properties of finite semitopologies, but we should not restrict to considering \emph{only} semitopologies and semiframes that are actually finite!  See Remark~\ref{rmrk.why.infinite}.}
\end{enumerate*}
\end{rmrk}

\jamiesubsubsection{Things that are different from filters}
\label{subsect.things.that.are.different}

\begin{rmrk}
Obviously, by definition semifilters are necessarily compatible but not necessarily closed under meets.
But aside from this fact, we have so far seen semiframes and semifilters behave more-or-less like frames and filters, modulo small details like that mentioned in Remark~\ref{rmrk.tbot.not.in}(\ref{item.tbot.not.in.1}).

But there are also differences, as we will now briefly explore.
In the theory of (finite) frames, the following facts hold:
\begin{enumerate*}
\item
\emph{Every proper filter $\afilter$ has a greatest lower bound $x$, and $\afilter=x^\cti=\{x'\mid x\cti x'\}$.}

Just take $x=\bigwedge\afilter$ the meet of all of its (finitely many) elements.
This is not $\tbot$, by the filter's finite intersection property.
\item
\emph{Every proper filter can be extended to a maximal filter.}\footnote{A proper filter is a filter that does not contain $\tbot$.  A maximal filter is a filter that is maximal amongst proper filters.}

Just extend using Zorn's lemma (as in Lemma~\ref{lemm.zorn.maximal.semifilter}).
\item
\emph{Every maximal filter is completely prime.}

It is a fact of finite frames that a maximal filter is prime,\footnote{A succinct proof is in Wikipedia~\cite{wiki:Idealordertheory}.} and since we assume the frame is finite, it is also completely prime.
\item
\emph{Every non-$\tbot$ element $x\neq\tbot_{\ns X}$ in a finite frame is contained in some abstract point.}

Just form $\{x'\mid x\cti x'\}$, observe it is a filter, form a maximal filter above it, and we get an abstract point. 
\item
\emph{As a corollary, if the frame is nonempty (so $\tbot\neq\ttop$; see Example~\ref{xmpl.empty.semiframe}) then it has at least one abstract point.} 
\end{enumerate*}
In Lemma~\ref{lemm.no.converge} and Proposition~\ref{prop.non.tbot.lower.bound} we consider some corresponding \emph{non-properties} of (finite) semiframes.
\end{rmrk}

\begin{lemm}
\label{lemm.no.converge}
Suppose $(\ns X,\cti,\ast)$ is a semiframe.
It is possible for $\tf{Points}(\ns X,\cti,\ast)$ to be empty, even if $(\ns X,\cti,\ast)$ is nonempty (Example~\ref{xmpl.empty.semiframe}(\ref{item.nonempty.semiframe})).
This is possible even if $\ns X$ is finite, and even if $\ns X$ is infinite. 
\end{lemm}
\begin{proof}
It suffices to provide an example.
We define a semiframe as below, and as illustrated in Figure~\ref{fig.ast.no.tand}:
\begin{itemize*}
\item
$\ns X=\{\tbot,0,1,2,3,\ttop\}$.
\item
Let $x\cti x'$ when $x=x'$ or $x=\tbot$ or $x'=\ttop$.
\item
Let $x\ast x'$ when $x\tand x'\neq\tbot$.\footnote{Unpacking what that means, we obtain this: $x\neq\tbot \land x=x'$ or $x\neq\tbot \land x'=\ttop$ or $x'\neq\tbot \land x=\ttop$.

This definition for $\ast$ is what we need for our counterexample, but other choices for $\ast$ also yield valid semiframes.  For example, we can set $x\ast x'$ when $x,x'\neq\tbot$.
} 
\end{itemize*}
Then $(\ns X,\cti,\ast)$ has no abstract points.

For suppose $P$ is one such. 
By Lemma~\ref{lemm.P.top} $\ttop\in P$.
Note that $\ttop=0\tor 1=2\tor 3$.
By assumption $P$ is completely prime, we know that $0\in P \lor 1\in P$, and also $2\in P\lor 3\in P$.
But this is impossible because $0$, $1$, $2$, and $3$ are not compatible.

For the infinite case, we just increase the width of the semiframe by taking $\ns X=\{\tbot\}\cup\mathbb N\cup \{\ttop\}$.
\end{proof}

\begin{prop}
\label{prop.non.tbot.lower.bound}
Suppose $(\ns X,\cti,\ast)$ is a semiframe and $\afilter\subseteq\ns X$ is a semifilter.
Then:
\begin{enumerate*}
\item\label{item.tbot.lower.bound.semifilter}
It is not necessarily the case that $\afilter$ has a non-$\tbot$ greatest lower bound (even if $\ns X$ is finite).
\item
Every semifilter can be extended to a maximal semifilter, but \dots
\item\label{item.max.not.prime}
\dots this maximal semifilter is not necessarily prime (even if $\ns X$ is finite).
\item
There may exist a non-$\tbot$ element $x\neq\tbot_{\ns X}$ that is contained in no abstract point.
\end{enumerate*}
\end{prop}
\begin{proof}
We consider each part in turn:
\begin{enumerate}
\item
Consider $(\powerset(\{0,1,2\}),\subseteq,\between)$ and take 
$$
\afilter = \{\{0,1\},\ \{1,2\},\ \{0,2\},\ \{0,1,2\}\} .
$$
The greatest lower bound of $\afilter$ is $\varnothing$.
\item
This is Lemma~\ref{lemm.zorn.maximal.semifilter}.
\item
$\afilter$ from part~\ref{item.tbot.lower.bound.semifilter} of this result is maximal, and it cannot be extended to a point $P\supseteq\afilter$.

Figure~\ref{fig.ast.no.tand} gives another counterexample, and in a rather interesting way: the semitopology has four maximal semifilters $\{i,\ttop\}$ for $i\in\{0,1,2,3\}$, but by Lemma~\ref{lemm.no.converge} it has no prime semifilters at all.\footnote{See also a discussion of the design of the notion of semifilter in Remarks~\ref{rmrk.other.properties} and~\ref{rmrk.further.restrictions.on.points}.} 
\item
We just take $x=0\in\ns X$ from the example in 
Lemma~\ref{lemm.no.converge}
(see Figure~\ref{fig.ast.no.tand}).
Since this semiframe has no abstract points at all, there is no abstract point that contains $x$.
\qedhere\end{enumerate}
\end{proof}

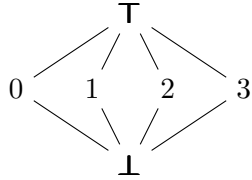
\begin{figure}
\begin{center}
\begin{tikzpicture}[scale=.55]
  \node (one) at (0,2) {$\ttop$};
  \node (a) at (-3,0) {$0$};
  \node (b) at (-1,0) {$1$};
  \node (c) at (1,0) {$2$};
  \node (d) at (3,0) {$3$};
  \node (zero) at (0,-2) {$\tbot$};
  \draw (zero) -- (a) -- (one) -- (b) -- (zero) -- (c) -- (one) -- (d) -- (zero);
\end{tikzpicture}
\end{center}
\caption{A semiframe with no abstract points (Lemma~\ref{lemm.no.converge})}
\label{fig.ast.no.tand}
\end{figure}

\begin{rmrk}
For now, we will just read Proposition~\ref{prop.non.tbot.lower.bound} as a caution not to assume that semiframes and semifilters behave like frames and filters.
Sometimes they do, and sometimes they don't; we have to check.

We now proceed to build our categorical duality, culminating with Theorem~\ref{thrm.categorical.duality.semiframes}.
Once that machinery is constructed, 
we will continue our study of the fine structure of semifilters in Section~\ref{sect.closer.look.at.semifilters}.
\end{rmrk}

\jamiesubsection{Sets of abstract points}

\begin{defn}
\label{defn.Op}
Suppose $(\ns X,\cti,\ast)$ is a semiframe and recall $\tf{Points}(\ns X,\cti,\ast)$ from Definition~\ref{defn.point}(\ref{item.abstract.point}).
Define a map $\f{Op}:\ns X\to\powerset(\tf{Points}(\ns X,\cti,\ast))$ by
$$
\f{Op}(x) = \{P\in\tf{Points}(\ns X,\cti,\ast) \mid x\in P\} . 
$$
\end{defn}

\begin{lemm}
\label{lemm.op.commutes.with.joins}
Suppose $(\ns X,\cti,\ast)$ is a semiframe and $X\subseteq\ns X$.
Then
$$
\f{Op}(\bigvee X)=\bigcup_{x\in X}\f{Op}(x) .
$$
In words: we can say that $\f{Op}$ commutes with joins, and that $\f{Op}$ commutes with taking least upper bounds. 
\end{lemm}
\begin{proof}
Suppose $P\in\tf{Points}(\ns X,\cti,\ast)$.
We reason as follows:
$$
\begin{array}[b]{r@{\ }l@{\qquad}l}
P\in\f{Op}(\bigvee X)
\liff& 
\bigvee X\in P
&\text{Definition~\ref{defn.Op}}
\\
\liff&
\Exists{x{\in}X}x\in P
&\text{Definition~\ref{defn.point}(\ref{item.completely.prime})}
\\
\liff&
P\in\bigcup_{x\in X}\f{Op}(x)
&\text{Definition~\ref{defn.Op}}
\end{array}
\qedhere
$$
\end{proof}

\begin{prop}
\label{prop.semiframe.to.Op}
Suppose $(\ns X,\cti,\ast)$ is a semiframe and $x,x'\in\ns X$. 
Then:
\begin{enumerate*}
\item\label{item.semiframe.to.Op.subset}
If $x\cti x'$ then $\f{Op}(x)\subseteq \f{Op}(x')$. 
\item\label{item.semiframe.to.Op.between}
If $\f{Op}(x)\between \f{Op}(x')$ then $x\ast x'$. 
\item\label{item.semiframe.to.Op.top}
$\f{Op}(\ttop_{\ns X})=\tf{Points}(\ns X,\cti,\ast)$ 
and
$\f{Op}(\tbot_{\ns X})=\varnothing$.
\item\label{item.semiframe.to.Op.bigvee}
$\f{Op}(\bigvee X)=\bigcup_{x\in X}\f{Op}(x)$ for $X\subseteq\ns X$.
\end{enumerate*}
\end{prop}
\begin{proof}
We consider each part in turn:
\begin{enumerate}
\item
\emph{We prove that $x\cti x'$ implies $\f{Op}(x)\subseteq \f{Op}(x')$.}
 
Suppose $x\cti x'$, and consider some abstract point $P\in\f{Op}(x)$.
By Definition~\ref{defn.Op} $x\in P$, and by up-closure of $P$ (Definition~\ref{defn.point}(\ref{item.up-closed})) $x'\in P$, so by Definition~\ref{defn.Op} $P\in\f{Op}(x')$.
$P$ was arbitrary, and it follows that $\f{Op}(x)\subseteq\f{Op}(x')$. 
\item
\emph{We prove that $\f{Op}(x)\between \f{Op}(x')$ implies $x\ast x'$.}

Suppose there exists an abstract point $P\in\f{Op}(x)\cap\f{Op}(x')$.
By Definition~\ref{defn.Op} $x,x'\in P$, and by compatibility of $P$ (Definition~\ref{defn.point}(\ref{item.weak.clique})) $x\ast x'$. 
\item
Unpacking Definition~\ref{defn.Op}, it suffices to show that $\ttop_{\ns X}\in P$ and $\tbot_{\ns X}\notin P$, for every abstract point $P\in\tf{Points}(\ns X,\cti,\ast)$.
This is from Lemma~\ref{lemm.P.top}(\ref{item.P.yes.top}).
\item
This is just Lemma~\ref{lemm.op.commutes.with.joins}.
\qedhere\end{enumerate}
\end{proof}

\begin{rmrk}
\label{rmrk.not.enough.points}
Proposition~\ref{prop.semiframe.to.Op} carries a clear suggestion that 
$(\{\f{Op}(x) \mid x\in\ns X\},\subseteq,\between)$ is trying, in some sense, to be an isomorphic copy of $(\ns X,\cti,\ast)$.
Lemma~\ref{lemm.Op.may.fail} notes that it may not quite manage this, because there may not be enough points (indeed, there may not be any abstract points at all).
This will (just as for topologies and frames) lead us to the notion of a \emph{spatial} semiframe in Definition~\ref{defn.spatial.graph} and Proposition~\ref{prop.Op.subseteq}.
\end{rmrk}

\begin{lemm}
\label{lemm.Op.may.fail}
The converse implications in Proposition~\ref{prop.semiframe.to.Op}(\ref{item.semiframe.to.Op.subset}\&\ref{item.semiframe.to.Op.between}) need not hold.
That is:
\begin{enumerate*}
\item
There exists a semiframe $(\ns X,\cti,\ast)$ and $x,x'\in\ns X$ such that $\f{Op}(x)\subseteq\f{Op}(x')$ yet $x\not\cti x'$.
\item 
There exists a semiframe $(\ns X,\cti,\ast)$ and $x,x'\in\ns X$ such that $x\ast x'$ yet $\f{Op}(x)\notbetween\f{Op}(x')$. 
\end{enumerate*}
\end{lemm}
\begin{proof}
The example from Lemma~\ref{lemm.no.converge} (as illustrated in Figure~\ref{fig.ast.no.tand}) is a counterexample for both cases: 
\begin{itemize*}
\item
$\f{Op}(0)\subseteq\f{Op}(1)$ because both are equal to the empty set, yet $0\not\cti 1$; and 
\item
$\ttop\ast\ttop$ yet $\f{Op}(\ttop)\notbetween\f{Op}(\ttop)$.
\qedhere\end{itemize*}
\end{proof}

\jamiesubsection{The semitopology of abstract points}

Recall from Definition~\ref{defn.point}(\ref{item.abstract.point}) that an abstract point in a semiframe $(\ns X,\cti,\ast)$ is a nonempty up-closed compatible completely prime subset of $\ns X$, and recall from Definition~\ref{defn.Op} that 
$$
\f{Op}(x) = \{P\in \tf{Points}(\ns X,\cti,\ast) \mid x\in P\},
$$
or in words: $\f{Op}(x)$ is the set of abstract points that contain $x$.
\begin{defn}
\label{defn.OpX}
Suppose $(\ns X,\cti,\ast)$ is a semiframe.
Then define $\tf{Op}(\ns X,\cti,\ast)$ by
$$
\tf{Op}(\ns X,\cti,\ast) = \{\f{Op}(x) \mid x\in\ns X\} .
$$ 
\end{defn}

\begin{lemm}
\label{lemm.Op.unions}
Suppose $(\ns X,\cti,\ast)$ is a semiframe.
Then:
\begin{enumerate*}
\item\label{item.Op.unions.1}
$\tf{Op}(\ns X,\cti,\ast)$ from Definition~\ref{defn.OpX} is closed under arbitrary sets union.
\item\label{item.Op.unions.2}
As a corollary, $(\tf{Op}(\ns X,\cti,\ast),\subseteq)$ (in words: $\tf{Op}(\ns X,\cti,\ast)$ ordered by subset inclusion) is a complete join-semilattice.
\end{enumerate*}
\end{lemm}
\begin{proof}
Part~\ref{item.Op.unions.1} is just Lemma~\ref{lemm.op.commutes.with.joins}.
The corollary part~\ref{item.Op.unions.2} is just a fact, since $\tf{Op}(\ns X,\cti,\ast)\subseteq\powerset(\tf{Points}(\ns X,\cti,\ast))$, and sets union is the join (least upper bound) in the powerset lattice. 
\end{proof}

Recall from Definition~\ref{defn.semi.to.dg} and Lemma~\ref{lemm.Fr.semiframe} that we showed how to go from a semitopology $(\ns P,\opens)$ to a semiframe $(\opens,\subseteq,\between)$.
We now show how to go in the other direction: 
\begin{defn}[{\bf Semiframe $\to$ semitopology}]
\label{defn.st.g}
Suppose $(\ns X,\cti,\ast)$ is a semiframe. 
Define the \deffont{semitopology of abstract points $\tf{St}(\ns X,\cti,\ast)$}\index{$\tf{St}(\ns X,\cti,\ast)$ (semitopology of abstract points)} by
$$
\tf{St}(\ns X,\cti,\ast) = \bigl(\tf{Points}(\ns X,\cti,\ast), \tf{Op}(\ns X,\cti,\ast)\bigr) .
$$
Unpacking this a little:
\begin{enumerate*}
\item\label{item.st.pt}
The set of points of $\tf{St}(\ns X,\cti,\ast)$ is the set of abstract points $\tf{Points}(\ns X,\cti,\ast)$ from Definition~\ref{defn.point}(\ref{item.abstract.point}) --- namely, the completely prime nonempty up-closed compatible subsets of $\ns X$.\footnote{There are no guarantees in general about \emph{how many} abstract points exist; e.g. Lemma~\ref{lemm.no.converge} gives an example of a semiframe that has no abstract points at all and so maps to the empty semitopology.  Later on in Definition~\ref{defn.spatial.graph} we consider conditions to ensure the existence of abstract points.}
\item\label{item.st.op}
Open sets $\tf{Opens}(\ns X,\cti,\ast)$ are the $\f{Op}(x)$ from Definition~\ref{defn.Op}:
$$
\f{Op}(x)=\{P\in \tf{Points}(\ns X,\cti,\ast) \mid x\in P\} . 
$$  
\end{enumerate*}
\end{defn}

\begin{lemm}
\label{lemm.St.semitop}
Suppose $(\ns X,\cti,\ast)$ is a semiframe. 
Then $\tf{St}(\ns X,\cti,\ast)$ from Definition~\ref{defn.st.g} is indeed a semitopology.
\end{lemm}
\begin{proof}
From conditions~\ref{semitopology.empty.and.universe} and~\ref{semitopology.unions} of Definition~\ref{defn.semitopology}, we need to check that 
$\tf{Op}(\ns X,\cti,\ast)$ contains $\varnothing$ and $\tf{Points}(\ns X,\cti,\ast)$ and is closed under arbitrary unions.
This is from Proposition~\ref{prop.semiframe.to.Op}(\ref{item.semiframe.to.Op.top}\&\ref{item.semiframe.to.Op.bigvee}).
\end{proof}

Recall from Definitions~\ref{defn.st.g} and~\ref{defn.semi.to.dg} that $\tf{St}(\ns X,\cti,\ast)$ is a semitopology, and $\tf{Fr}\,\tf{St}(\ns X,\cti,\ast)$ is a semiframe each of whose elements is the set of abstract points of $(\ns X,\cti,\ast)$ that contain some $x\in\ns X$: 
\begin{lemm}
\label{lemm.st.opens.generator}
Suppose $(\ns X,\cti,\ast)$ is a semiframe. 
Then
$\f{Op}:(\ns X,\cti,\ast) \to \tf{Fr}\,\tf{St}(\ns X,\cti,\ast)$ is surjective.
\end{lemm}
\begin{proof}
Direct from Definition~\ref{defn.st.g}(\ref{item.st.op}).
\end{proof}

We conclude with Definition~\ref{defn.top.ind} and Proposition~\ref{prop.top.ind.eq}, which are standard properties of the construction in Definition~\ref{defn.st.g}.

\begin{defn}
\label{defn.topind}
\label{defn.top.ind}
Suppose $(\ns P,\opens)$ is a semitopology and $p,p'\in\ns P$.
Define $p\topind p'$\index{$p\topind p'$ (topologically indistinguishable points)} by
$$
p\topind p'
\quad\text{when}\quad
\Forall{O\in\opens} p\in O\liff p'\in O .
$$
We recall some standard terminology from topology: 
\begin{enumerate*}
\item\label{item.top.ind.1}
Call $p$ and $p'$ \deffont[topologically indistinguishable points]{topologically indistinguishable} when $p\topind p'$. 
\item\label{item.top.ind.2}
Call $p$ and $p'$ \deffont[topologically distinguishable]{topologically distinguishable} when $\neg(p\topind p')$ (so there exists some $O\in\opens$ such that $p\in O\land p'\notin O$ or $p\notin O\land p'\in O$).
\item\label{item.T0.space}
Call $(\ns P,\opens)$ a \deffont{$T_0$ space} when points are topologically indistinguishable precisely when they are equal, or in symbols: ${\topind} = {=}$.
\end{enumerate*}
\end{defn}

\begin{prop}
\label{prop.top.ind.eq}
Suppose $(\ns X,\cti,\ast)$ is a semiframe.
Then $\tf{St}(\ns X,\cti,\ast)$ (Definition~\ref{defn.st.g}) is a $T_0$ space.
\end{prop}
\begin{proof}
Suppose $P,P'\in\tf{Points}(\ns X,\cti,\ast)$.
Unpacking Definition~\ref{defn.point}(\ref{item.abstract.point}), this means that $P$ and $P'$ are completely prime nonempty up-closed compatible subsets of $\ns X$.

It is immediate that $P=P'$ implies $P\topind P'$. 

Now suppose $P\topind P'$ in $\tf{St}(\ns X,\cti,\ast)$; to prove $P=P'$ it would suffice to show that $x\in P \liff x\in P'$, for arbitrary $x\in\ns X$.
By Definition~\ref{defn.st.g}(\ref{item.st.op}), every open set in $\tf{St}(\ns X,\cti,\ast)$ has the form $\f{Op}(x)$ for some $x\in\ns X$.
We reason as follows: 
$$
\begin{array}{r@{\ }l@{\qquad}l}
x\in P \liff&
P\in\f{Op}(x) 
&\text{Definition~\ref{defn.Op}}
\\
\liff&
P'\in\f{Op}(x)
&\text{$P$, $P'$ top. indisting.}
\\
\liff&
x\in P' 
&\text{Definition~\ref{defn.Op}}
\end{array}
$$
Since $x$ was arbitrary and $P,P'\subseteq\opens$, it follows that $P=P'$ as required.
\end{proof}

\jamiesection{Spatial semiframes \& sober semitopologies}
\label{sect.spatial.and.sober}

\jamiesubsection{Definition of spatial semiframes}

\begin{rmrk}
We continue Remark~\ref{rmrk.not.enough.points}.
We saw in Example~\ref{xmpl.abstract.point}(2\&3) that there may be \emph{more} abstract points than there are concrete points, and in Remark~\ref{rmrk.not.enough.points} that there may also be \emph{fewer}.

In the theory of frames, the condition of being \emph{spatial} means that the abstract points and concrete points correspond.
We imitate this terminology for a corresponding definition on semiframes: 
\end{rmrk}

\begin{defn}[{\bf Spatial semiframe}]
\label{defn.spatial.graph}
Call a semiframe $(\ns X,\cti,\ast)$ \deffont[spatial semiframe]{spatial} when:
\begin{enumerate*}
\item\label{item.spatial.iff}
$\f{Op}(x)\subseteq \f{Op}(x')$ implies $x\cti x'$, for every $x,x'\in\ns X$. 
\item\label{item.spatial.ast.iff}
$x\ast x'$ implies $\f{Op}(x)\between \f{Op}(x')$, for every $x,x'\in\ns X$. 
\end{enumerate*}
\end{defn}

\begin{rmrk}
Not every semiframe is spatial, just as not every frame is spatial.
Lemma~\ref{lemm.no.converge} gives an example of a semiframe that is not spatial because it has no points at all, as illustrated in Figure~\ref{fig.ast.no.tand}.
\end{rmrk}

We check that the conditions in Definition~\ref{defn.spatial.graph} correctly strengthen the implications in Proposition~\ref{prop.semiframe.to.Op} to become logical equivalences:
\begin{prop}
\label{prop.Op.subseteq}
Suppose $(\ns X,\cti,\ast)$ is a spatial semiframe and $x,x'\in\ns X$.
Then:
\begin{enumerate*}
\item\label{item.Op.spatial.cti}
$x\cti x'$ if and only if $\f{Op}(x)\subseteq \f{Op}(x')$. 
\item\label{item.Op.spatial.ast}
$x\ast x'$ if and only if $\f{Op}(x)\between \f{Op}(x')$.
\item\label{item.Op.spatial.inj}
$x= x'$ if and only if $\f{Op}(x)=\f{Op}(x')$. 
\item\label{item.Op.top.all.points}
$\f{Op}(\ttop_{\ns X})=\tf{Points}(\ns X,\cti,\ast)$ 
and
$\f{Op}(\tbot_{\ns X})=\varnothing$.
\item\label{item.Op.vee}
$\f{Op}(\bigvee X)=\bigcup_{x\in X}\f{Op}(x)$ for $X\subseteq\ns X$.
\end{enumerate*}
\end{prop}
\begin{proof}
We consider each part in turn:
\begin{enumerate}
\item
\emph{We prove that $x\cti x'$ if and only if $\f{Op}(x)\subseteq \f{Op}(x')$.}
 
The right-to-left implication is direct from Definition~\ref{defn.spatial.graph}(\ref{item.spatial.iff}).
The left-to-right implication is Proposition~\ref{prop.semiframe.to.Op}(\ref{item.semiframe.to.Op.subset}).
\item
\emph{We prove that $x\ast x'$ if and only if $\f{Op}(x)\between \f{Op}(x')$.}

The left-to-right implication is direct from Definition~\ref{defn.spatial.graph}(\ref{item.spatial.ast.iff}).
The right-to-left implication is Proposition~\ref{prop.semiframe.to.Op}(\ref{item.semiframe.to.Op.between}).
\item
\emph{We prove that $x= x'$ if and only if $\f{Op}(x)= \f{Op}(x')$.}

If $x=x'$ then $\f{Op}(x)=\f{Op}(x')$ is immediate.
If $\f{Op}(x)=\f{Op}(x')$ then $\f{Op}(x)\subseteq\f{Op}(x')$ and $\f{Op}(x')\subseteq\f{Op}(x)$.
By part~\ref{item.Op.spatial.cti} of this result (or direct from Definition~\ref{defn.spatial.graph}(\ref{item.spatial.iff})) $x\cti x'$ and $x'\cti x$.
By antisymmetry of $\cti$ it follows that $x=x'$.
\item
This is just Proposition~\ref{prop.semiframe.to.Op}(\ref{item.semiframe.to.Op.top})
\item
This is just Lemma~\ref{lemm.op.commutes.with.joins}.
\qedhere\end{enumerate}
\end{proof}

Definition~\ref{defn.semiframe.iso} will be useful in a moment:\footnote{More on this topic later on in Definition~\ref{defn.category.of.spatial.graphs}, when we build the category of semiframes.}
\begin{defn}
\label{defn.semiframe.iso}
Suppose $(\ns X,\cti,\ast)$ and $(\ns X',\cti',\ast')$ are semiframes.
Then an \deffont[isomorphism between semiframes]{isomorphism} between them is a function $g:\ns X\to\ns X'$ such that:
\begin{enumerate*}
\item\label{item.semiframe.iso.1}
$g$ is a bijection between $\ns X$ and $\ns X'$.
\item\label{item.semiframe.iso.2}
$x_1\cti x_2$ if and only if $g(x_1)\cti g(x_2)$.
\item\label{item.semiframe.iso.3}
$x_1\ast x_2$ if and only if $g(x_1)\ast g(x_2)$.
\end{enumerate*}
\end{defn}

\begin{lemm}
\label{lemm.iso.semiframe.top}
Suppose $(\ns X,\cti,\ast)$ and $(\ns X',\cti',\ast')$ are semiframes and $g:\ns X\to\ns X'$ is an isomorphism between them.
Then $g(\tbot_{\ns X})=g(\tbot_{\ns X'})$ and $g(\ttop_{\ns X})=\ttop_{\ns X'}$.
\end{lemm}
\begin{proof}
By construction $\tbot_{\ns X}\leq x$ for every $x\in\ns X$.
It follows from Definition~\ref{defn.semiframe.iso}(\ref{item.semiframe.iso.2}) that $g(\tbot_{\ns X})\leq g(x)$ for every $x\in\ns X$; but $g$ is a bijection, so $g(\tbot_{\ns X})\leq x'$ for every $x'\in\ns X'$.
It follows that $g(\tbot_{\ns X})=\tbot_{\ns X'}$.

By similar reasoning we conclude that $g(\ttop_{\ns X})=\ttop_{\ns X'}$.
\end{proof}

\begin{rmrk}
Suppose $(\ns X,\cti,\ast)$ is a semiframe and recall from Definition~\ref{defn.OpX} that
$\tf{Op}(\ns X,\cti,\ast)=\{\f{Op}(x) \mid x\in\ns X\}$.
Then the intuitive content of Proposition~\ref{prop.Op.subseteq} is that a semiframe $(\ns X,\cti,\ast)$ is spatial when  
$(\ns X,\cti,\ast)$ is isomorphic (in the sense made formal by Definition~\ref{defn.semiframe.iso}) to
$(\tf{Op}(\ns X,\cti,\ast),\subseteq,\between)$.

And, because $\f{Op}(\ttop_{\ns X})=\tf{Points}(\ns X,\cti,\ast)$ we can write a slogan:
\begin{quoting}
\emph{A semiframe is spatial when it is (up to isomorphism) generated by its abstract points.}
\end{quoting}
We will go on to prove in Proposition~\ref{prop.Gr.P.spatial} that every semitopology generates a spatial semiframe --- and in Theorem~\ref{thrm.categorical.duality.semiframes} we will tighten and extend the slogan above to a full categorical duality. 
\end{rmrk}

\subsection{The neighbourhood semifilter $\nbhd(p)$} 

\subsubsection{The definition and basic lemma}

Recall from Definition~\ref{defn.nbhd} that we define $\nbhd(p)=\{O\in\opens\mid p\in O\}$.

\begin{rmrk}
\label{rmrk.nbhd.filter}
If $(\ns P,\opens)$ is a topology, then $\nbhd(p)$ is a filter (a nonempty up-closed down-directed set) and is often called the \emph{neighbourhood filter} of $p$.

We are working with semitopologies, so $\opens$ is not necessarily closed under intersections, and $\nbhd(p)$ is not necessarily a filter (it is still a compatible set, because every $O\in\nbhd(p)$ contains $p$).
Figure~\ref{fig.nbhd} illustrates examples of this: e.g. in the left-hand example $\{0,1\},\{0,2\}\in \nbhd(0)$ but $\{0\}\not\in\nbhd(0)$ (because this is not an open set).
\end{rmrk}

Recall from Definition~\ref{defn.nbhd} that $\nbhd(p)=\{O\in\opens \mid p\in O\}$.
\begin{prop}
\label{prop.nbhd.iff}
Suppose $(\ns P,\opens)$ is a semitopology and $p\in\ns P$ and $O\in\opens$.
Then:
\begin{enumerate*}
\item\label{item.nbhd.point}
$\nbhd(p)$ (Definition~\ref{defn.nbhd}) is an abstract point (a completely prime semifilter) in the semiframe $\tf{Fr}(\ns P,\opens)$ (Definition~\ref{defn.semi.to.dg}).
In symbols:
$$
\nbhd:\ns P\to \tf{Points}(\tf{Fr}(\ns P,\opens)) .
$$
\item\label{item.nbhd.iff}
The following are equivalent:
$$
\nbhd(p)\in\f{Op}(O)
\quad\liff\quad
O\in \nbhd(p)
\quad\liff\quad
p\in O.
$$
\item\label{item.nbhd.mone}
We have an equality:
$$
\nbhd^\mone(\f{Op}(O))=O.
$$
\end{enumerate*}
\end{prop}
\begin{proof}
We consider each part in turn:
\begin{enumerate}
\item
From Definition~\ref{defn.point}(\ref{item.abstract.point}), we must check that $\nbhd(p)$ is a nonempty, completely prime, up-closed, and compatible subset of $\opens$ when considered as a semiframe as per Definition~\ref{defn.semi.to.dg}.
All properties are by facts of sets; we give brief details:
\begin{itemize*}
\item
$\nbhd(p)$ is nonempty because $p\in \ns P\in\opens$.
\item
$\nbhd(p)$ is completely prime because it is a fact of sets that if $P\subseteq\opens$ and $p\in \bigcup P$ then $p\in O$ for some $O\in P$.
\item
$\nbhd(p)$ is up-closed because it is a fact of sets that if $p\in O$ and $O\subseteq O'$ then $p\in O'$.
\item
$\nbhd(p)$ is compatible because if $p\in O$ and $p\in O'$ then $O\between O'$.
\end{itemize*}
\item
By Definition~\ref{defn.Op}, $\f{Op}(O)$ is precisely the set of abstract points $P$ that contain $O$, and by part~\ref{item.nbhd.point} of this result $\nbhd(p)$ is one of those points.
By Definition~\ref{defn.nbhd}, $\nbhd(p)$ is precisely the set of open sets that contain $p$.
The equivalence follows.
\item
We reason as follows:
$$
\begin{array}[b]{r@{\ }l@{\quad}l}
p\in\nbhd^\mone(\f{Op}(O))
\liff&
\nbhd(p)\in\f{Op}(O)
&\text{Fact of function inverse}
\\
\liff&
p\in O
&\text{Part~2 of this result}
\end{array}
\qedhere$$
\end{enumerate}
\end{proof}

\begin{corr}
\label{corr.op.sub.between}
Suppose $(\ns P,\opens)$ is a semitopology and $O,O'\in\opens$.
Then:
\begin{enumerate*}
\item
$\f{Op}(O)\subseteq\f{Op}(O')$ if and only if $O\subseteq O'$.
\item
$\f{Op}(O)\between \f{Op}(O')$ if and only if $O\between O'$.
\item
As a corollary, $\nbhd^\mone(\varnothing)=\varnothing$ and $\nbhd^\mone(\tf{Points}(\opens,\subseteq,\between))=\ns P$; i.e. $\nbhd^\mone$ maps the bottom/top element to the bottom/top element.
\end{enumerate*}
\end{corr}
\begin{proof}
We consider each part in turn:
\begin{enumerate}
\item
If $\f{Op}(O)\subseteq\f{Op}(O')$ then $\nbhd^\mone(\f{Op}(O))\subseteq\nbhd^\mone(\f{Op}(O'))$ by facts of inverse images, and $O\subseteq O'$ follows by Proposition~\ref{prop.nbhd.iff}(\ref{item.nbhd.mone}).

If $O\subseteq O'$ then $\f{Op}(O)\subseteq\f{Op}(O')$ by Proposition~\ref{prop.semiframe.to.Op}(\ref{item.semiframe.to.Op.subset}).
\item 
If $O\between O'$ then there exists some point $p\in\ns P$ with $p\in O\cap O'$.
By Proposition~\ref{prop.nbhd.iff}(\ref{item.nbhd.point}) $\nbhd(p)$ is an abstract point, and by Proposition~\ref{prop.nbhd.iff}(\ref{item.nbhd.iff}) $\nbhd(p)\in\f{Op}(O)\cap\f{Op}(O')$; thus $\f{Op}(O)\between\f{Op}(O')$.

If $\f{Op}(O)\between\f{Op}(O')$ then $O\between O'$ by Proposition~\ref{prop.semiframe.to.Op}(\ref{item.semiframe.to.Op.between}).
\item
Routine from Proposition~\ref{prop.semiframe.to.Op}(\ref{item.semiframe.to.Op.top}) (or from Lemma~\ref{lemm.iso.semiframe.top}).
\qedhere\end{enumerate}
\end{proof}

\jamiesubsubsection{Application to semiframes of open sets}

\begin{prop}
\label{prop.nbhd.mone.bijects}
Suppose $(\ns P,\opens)$ is a semitopology.
Then:
\begin{enumerate*}
\item\label{item.nbhd.mone.bijects.1}
$\nbhd^\mone$ bijects open sets of $\tf{St}(\opens,\subseteq,\between)$ (as defined in Definition~\ref{defn.st.g}(\ref{item.st.op})), with open sets of $(\ns P,\opens)$, taking $\f{Op}(O)$ to $O$.
\item\label{item.nbhd.iso}
$\nbhd^\mone$ is an isomorphism between the semiframe of open sets of $\tf{St}(\opens,\subseteq,\between)$, and the semiframe of open sets of $(\ns P,\opens)$ (Definition~\ref{defn.semiframe.iso}).
\end{enumerate*}
\end{prop}
\begin{proof}
We consider each part in turn:
\begin{enumerate}
\item
Unpacking Definition~\ref{defn.st.g}(\ref{item.st.op}), an open set in $\tf{St}\,\tf{Fr}(\ns P,\opens)$ has the form $\f{Op}(O)$ for some $O\in\opens$.
By Proposition~\ref{prop.nbhd.iff}(\ref{item.nbhd.mone}) $\nbhd^\mone(\f{Op}(O))=O$, and so $\nbhd^\mone$ is surjective and injective.
\item
Unpacking Definition~\ref{defn.semiframe.iso} it suffices to check that:
\begin{itemize*}
\item
$\nbhd^\mone$ is a bijection, and maps $\f{Op}(O)$ to $O$. 
\item
$\f{Op}(O)\subseteq\f{Op}(O')$ if and only if $O\subseteq O'$.
\item
$\f{Op}(O)\between \f{Op}(O')$ if and only if $O\between O'$.
\end{itemize*}
The first condition is part~\ref{item.nbhd.mone.bijects.1} of this result; the second and third are from Corollary~\ref{corr.op.sub.between}.
\qedhere\end{enumerate}
\end{proof}

\begin{prop}
\label{prop.Gr.P.spatial}
Suppose $(\ns P,\opens)$ is a semitopology.
Then the semiframe $\tf{Fr}(\ns P,\opens)=(\opens,\subseteq,\between)$ from Definition~\ref{defn.semi.to.dg} is spatial.
\end{prop}
\begin{proof}
The properties required by Definition~\ref{defn.spatial.graph} are that $\f{Op}(O)\subseteq\f{Op}(O')$ implies $O\subseteq O'$, and $O\between O'$ implies $\f{Op}(O)\between\f{Op}(O')$.
Both of these are immediate from Proposition~\ref{prop.nbhd.mone.bijects}(\ref{item.nbhd.iso}).
\end{proof}

\jamiesubsubsection{Application to characterise $T_0$ spaces}

\begin{lemm}
\label{lemm.nbhd.top.ind}
Suppose $(\ns P,\opens)$ is a semitopology and $p,p'\in\ns P$.
Then the following are equivalent:
\begin{enumerate*}
\item\label{item.nbhd.top.ind.1}
$\nbhd(p)=\nbhd(p')$\ (cf. also Lemma~\ref{lemm.intertwined.sober})
\item\label{item.nbhd.top.ind.2}
$\Forall{O{\in}\opens}p\in O\liff p\in O'$
\item\label{item.nbhd.top.ind.3}
$p\topind p'$\ (Definition~\ref{defn.topind}: $p$ and $p'$ are topologically indistinguishable in $(\ns P,\opens)$).
\item\label{item.nbhd.top.ind.4}
$\nbhd(p)\topind \nbhd(p')$\ ($\nbhd(p)$ and $\nbhd(p')$ are topologically indistinguishable as --- by Proposition~\ref{prop.nbhd.iff}(\ref{item.nbhd.point}) --- abstract points in $\tf{St}\,\tf{Fr}(\ns P,\opens)$). 
\end{enumerate*}
\end{lemm}
\begin{proof}
Equivalence of parts~\ref{item.nbhd.top.ind.1} and~\ref{item.nbhd.top.ind.2} is direct from Definition~\ref{defn.nbhd}.
Equivalence of parts~\ref{item.nbhd.top.ind.2} and~\ref{item.nbhd.top.ind.3} is just Definition~\ref{defn.top.ind}(\ref{item.top.ind.1}).
Equivalence of parts~\ref{item.nbhd.top.ind.4} and~\ref{item.nbhd.top.ind.1} is from Proposition~\ref{prop.top.ind.eq}.
\end{proof}

\begin{corr}
\label{corr.T0.nbhd.inj}
Suppose $(\ns P,\opens)$ is a semitopology.
Then the following are equivalent:
\begin{enumerate*}
\item
$(\ns P,\opens)$ is $T_0$ (Definition~\ref{defn.top.ind}(\ref{item.T0.space})).
\item
$\nbhd:\ns P\to \tf{Points}(\opens,\subseteq,\between)$ is injective.
\end{enumerate*}
\end{corr}
\begin{proof}
Suppose $(\ns P,\opens)$ is $T_0$, and suppose $\nbhd(p)=\nbhd(p')$.
By Lemma~\ref{lemm.nbhd.top.ind}(\ref{item.nbhd.top.ind.1}\&\ref{item.nbhd.top.ind.3}) $p\topind p'$. 
By Definition~\ref{defn.top.ind}(\ref{item.top.ind.2}) $p=p'$.
Since $p$ and $p'$ were arbitrary, $\nbhd$ is injective.

Suppose $\nbhd$ is injective.
Reversing the reasoning of the previous paragraph, we deduce that $(\ns P,\opens)$ is $T_0$.
\end{proof}

\jamiesubsection{Sober semitopologies}
\label{subsect.sober.semitopologies}

Recall from Proposition~\ref{prop.Gr.P.spatial} that if we go from a semitopology $(\ns P,\opens)$ to a semiframe $(\opens,\subseteq,\between)$, then the result is not just any old semiframe --- it is a \emph{spatial} one. 

We now investigate what happens when we go from a semiframe to a semitopology using Definition~\ref{defn.st.g}.

\jamiesubsubsection{The definition and a key result}
 
\begin{defn}
\label{defn.sober.semitopology}
Call a semitopology $(\ns P,\opens)$ \deffont[sober semitopology]{sober} when every abstract point $P$ of $\tf{Fr}(\ns P,\opens)$ --- i.e. every completely prime nonempty up-closed compatible set of open sets --- is equal to the neighbourhood semifilter $\nbhd(p)$ of some unique $p\in\ns P$. 

Equivalently, $(\ns P,\opens)$ is sober when $\nbhd:\ns P\to\tf{Points}(\tf{Fr}(\ns P,\opens))$ (Definition~\ref{defn.point}(\ref{item.abstract.point})) is a bijection. 
\end{defn}

\begin{rmrk}
\label{rmrk.enough.points}
A bijection is a map that is injective and a surjective.
We noted in Corollary~\ref{corr.T0.nbhd.inj} that a space is $T_0$ when $\nbhd$ is injective.
So the sobriety condition can be thought of as having two parts: 
\begin{itemize*}
\item
$\nbhd$ is injective and the space is $T_0$, so it intuitively contains no `unnecessary' duplicates of points;
\item
$\nbhd$ is surjective, so the space contains `enough' points that there is (precisely) one concrete point for every abstract point.\footnote{`Unnecessary' and `enough' are in scare quotes here because these are subjective terms.  For example, if points represent computer servers on a network then we might consider it a \emph{feature} to not be $T_0$ by having multiple points that are topologically indistinguishable --- e.g. for backup, or to reduce latency --- and likewise, we might consider it a feature to not have one concrete point for every abstract point, if this avoids redundancies.  There is no contradiction here: a computer network based on a small non-sober space with multiple backups of what it has, may be a more efficient and reliable system than one based on a larger non-sober space that does not back up its servers but is full of redundant points.
And, this smaller non-sober space may present itself to the user abstractly as the larger, sober space. 

Users may even forget about the computation that goes on under the hood of this abstraction, as illustrated by the following \emph{true story:} I had a paper presenting an efficient algorithm rejected because it `lacked motivation'.  Why?  Because the algorithm was unnecessary: the reviewer claimed, apparently with a straight face, that guessing the answer until you got it right was computationally equivalent. 
}
\end{itemize*} 
\end{rmrk}

We start with a very simple example of sober semitopologies:
\begin{lemm}
\label{lemm.discrete.sober}
Suppose $\ns P$ is any set.
Then the discrete semitopology $(\ns P,\powerset(\opens))$ is sober.
\end{lemm}
\begin{proof}
Consider an abstract point $P\subseteq\opens$ (completely prime nonempty up-closed and compatible, as per Definition~\ref{defn.point}(\ref{item.abstract.point})).
Then $\ns P\in P$ and $\ns P=\bigvee\{\{p\}\mid p\in\ns P\}$.
Since $P$ is completely prime, $\{p\}\in P$ for some $p\in\ns P$.
It follows easily that $P=\nbhd(p)$.
\end{proof}

\begin{xmpl}
\label{xmpl.sober.non-sober}
We give some more examples of sober and non-sober semitopologies.
\begin{enumerate}
\item
Take $\ns P=\{0,1\}$ and $\opens=\{\varnothing,\{0,1\}\}$.
This has one abstract point $P=\{\{0,1\}\}$ but two concrete points $0$ and $1$.
It is therefore not sober.
\item
Take $\ns P=\{0,1\}$ and $\opens=\{\varnothing,\{1\},\{0,1\}\}$.
This has two abstract points 
$$
\{\{1\},\{0,1\}\}
\quad\text{and}\quad \{\{0,1\}\}
$$ 
corresponding to two concrete points $0$ and $1$.
It is sober.
\item
Take $\ns P=\mathbb N$ with the final topology; so $O\in\opens$ when $O=\varnothing$ or $O=n_\geq$ for some $n\in\mathbb N$, where $n_\geq = \{n'\in\mathbb N \mid n'\geq n\}$. 
Take $P=\{n_\geq \mid n\in\mathbb N\}$.
The reader can check that this is an abstract point (up-closed, completely prime, compatible); however $P$ is not the neighbourhood semifilter of any $n\in\mathbb N$.
Thus this space is not sober.
\item\label{item.sober.R}
$\mathbb R$ with its usual topology (which is also a semitopology) is sober.

This is a known result for topologies, but Remark~\ref{rmrk.no.meet} (and also the later Remark~\ref{rmrk.sober.not.hausdorff}) caution us that we cannot take this for granted, so we sketch the proof. 
 
Suppose $P$ is an abstract point; we wish to exhibit a unique $p\in\mathbb R$ such that $P=\nbhd(p)$.

We cover $\mathbb R$ with open intervals of radius $1$ by writing $\mathbb R=\bigcup\{\openinterval{r\minus 0.5,r\plus 0.5} \mid r\in\mathbb R\}$, and we use the completely prime property to find (at least one) such open interval that is in $P$; write it $O_1\in P$.
We then cover $O_1$ with open intervals of radius at most $1/2$ by writing $O_1=\bigcup\{O_1\cap\openinterval{r\minus 0.25,r\plus 0.25}\mid r\in O_1\}$, and we iterate to obtain a sequence $(O_i\mid i\in\mathbb N)\subseteq P$.
This converges to some unique $p\in\mathbb R$.
We check that $P=\nbhd(p)$:
\begin{itemize*}
\item
Suppose $O\in\opens$ is such that $p\in O$.
Because $p\in O$, there exists some $\epsilon$ such that $\openinterval{p\minus\epsilon,p\plus\epsilon}\subseteq O$.
It follows that for $i>1/\epsilon$ we have $O_i\subseteq O$ and thus $O\in P$ by up-closedness. 
\item
Suppose $O'\in\opens$ is such that $p\notin O'$.
For a sufficiently large $i$ we have $O_i\notintersectswith O'$, so by compatibility it follows that $O'\notin P$.
\end{itemize*}
\item
$\mathbb Q$ is sober.
The argument is much as for $\mathbb R$ above.
We have to work just a little harder because the $p$ we obtain need not be rational, but the arguments on open intervals remain valid.
\end{enumerate}
\end{xmpl}

\begin{prop}
\label{prop.spatial.gives.sober}
Suppose $(\ns X,\cti,\ast)$ is a semiframe. 
Then $\tf{St}(\ns X,\cti,\ast)$ from Definition~\ref{defn.st.g} is a sober semitopology.
\end{prop}
\begin{proof}
We know from Lemma~\ref{lemm.St.semitop} that $\tf{St}(\ns X,\cti,\ast)$ is a semitopology.
The issue is whether it is sober; thus by Definition~\ref{defn.sober.semitopology} we wish to exhibit every abstract point $P$ of $\tf{Fr}\,\tf{St}(\ns X,\cti,\ast)$ as a neighbourhood semifilter $\nbhd(p)$ for some unique abstract point $p$ of $(\ns X,\cti,\ast)$.
The calculations to do so are routine, but we give details.

Fix some abstract point $P$ of $\tf{Fr}\,\tf{St}(\ns X,\cti,\ast)$.
By Definition~\ref{defn.point}(\ref{item.abstract.point}), 
$P$ 
is a completely prime nonempty up-closed set of intersecting open sets in the semitopology $\tf{St}(\ns X,\cti,\ast)$, and by Definition~\ref{defn.st.g}(\ref{item.st.op}) each open set in $\tf{St}(\ns X,\cti,\ast)$ has the form $\f{Op}(x)=\{p\in\tf{Points}(\ns X,\cti,\ast) \mid x\in p\}$ for some $x\in\ns X$. 

We define $p\subseteq\ns X$ as follows:
$$
p=\{x\in\ns X \mid \f{Op}(x)\in P\} \subseteq\ns X .
$$
By construction we have that $x\in p$ if and only if $\f{Op}(x)\in P$, and so
$$
\begin{array}{r@{\ }l@{\qquad}l}
\nbhd(p) 
=&
\{\f{Op}(x) \mid p\in \f{Op}(x)\}
&\text{Definition~\ref{defn.nbhd}}
\\
=&
\{\f{Op}(x) \mid x\in p\}
&\text{Definition~\ref{defn.Op}}
\\
=&
\{\f{Op}(x) \mid \f{Op}(x)\in P\}
&\text{Construction of $p$}
\\
=&
P 
&\text{Fact}
.
\end{array}
$$
Now $P$ is completely prime, nonempty, up-closed, and compatible and it follows by elementary calculations using Proposition~\ref{prop.Op.subseteq} that $p$ is also completely prime, nonempty, up-closed, and compatible --- so $p$ is an abstract point of $(\ns X,\cti,\ast)$.

So we have that
$$
p\in\tf{Point}(\ns X,\cti,\ast)
\quad\text{and}\quad
P = \f{nbhd}(p) .
$$
To prove uniqueness of $p$, suppose $p'$ is any other abstract point such that $P=\nbhd(p')$.
We follow the definitions: $\f{Op}(x)\in \nbhd(p') \liff \f{Op}(x)\in \nbhd(p)$, and thus by Definition~\ref{defn.nbhd} $p'\in\f{Op}(x) \liff p\in\f{Op}(x)$, and thus by Definition~\ref{defn.Op} $x\in p'\liff x\in p$, and thus $p'=p$. 
\end{proof}

\jamiesubsubsection{Sober topologies contrasted with sober semitopologies}
\label{subsect.sober.top.contrast}

\begin{figure}
\centering
\includegraphics[width=0.35\columnwidth,trim={50 20 50 20},clip]{diagrams/012_triangle.pdf}
\qquad
\includegraphics[width=0.4\columnwidth,trim={50 20 0 20},clip]{diagrams/square-diagram.pdf}
\\[4ex]
\caption{Two counterexamples for sobriety: (a) finite $T_0$ (and also $T_1$) semitopology that is not sober (Lemma~\ref{lemm.T0.not.sober});\ (b) Hausdorff semitopology that is not sober (Lemma~\ref{lemm.hausdorff.not.sober})}
\label{fig.012-triangle}
\end{figure}

\begin{figure}
\centering
\includegraphics[width=0.35\columnwidth,trim={50 20 50 20},clip]{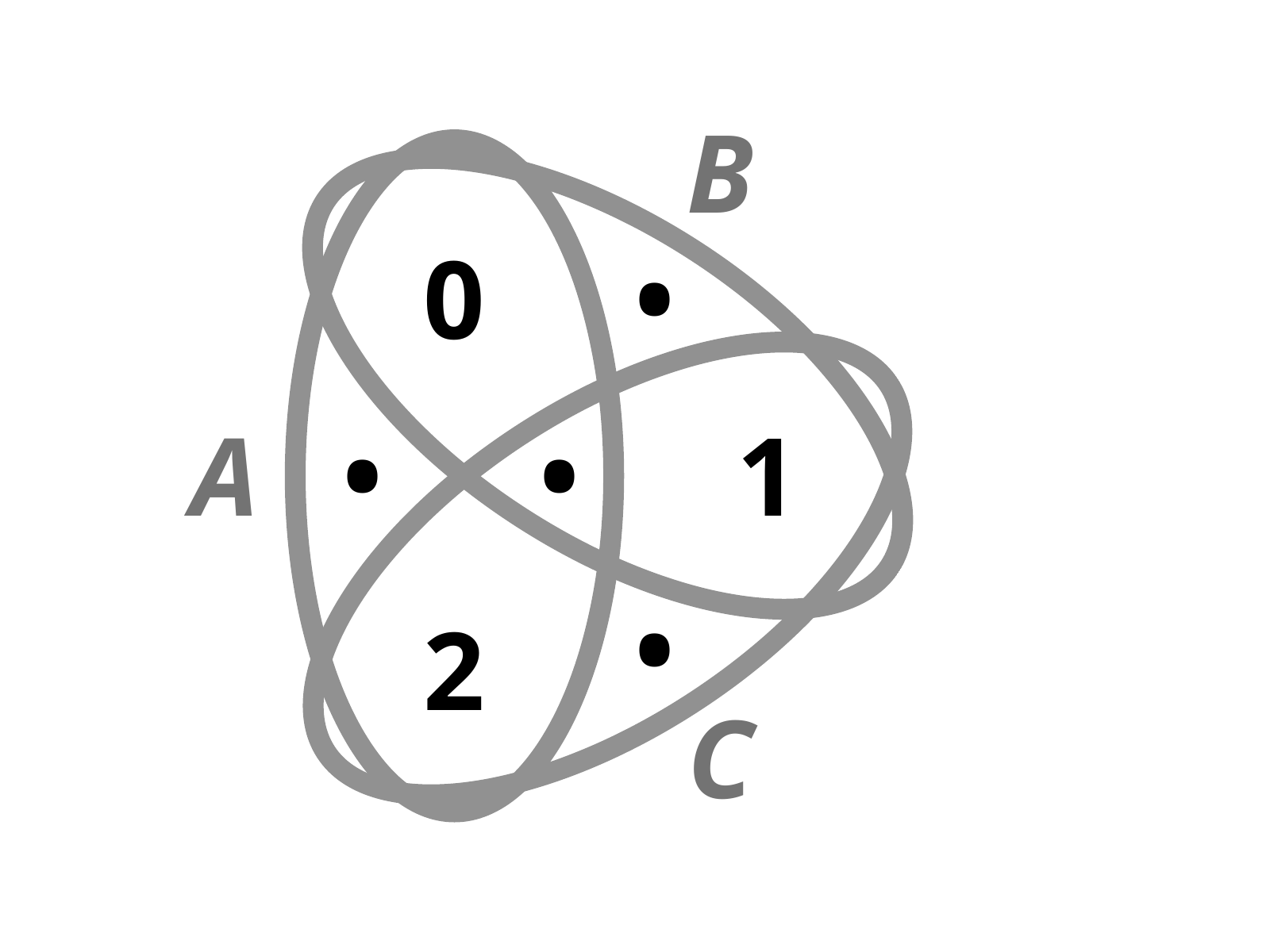}
\qquad
\includegraphics[width=0.4\columnwidth,trim={50 20 0 20},clip]{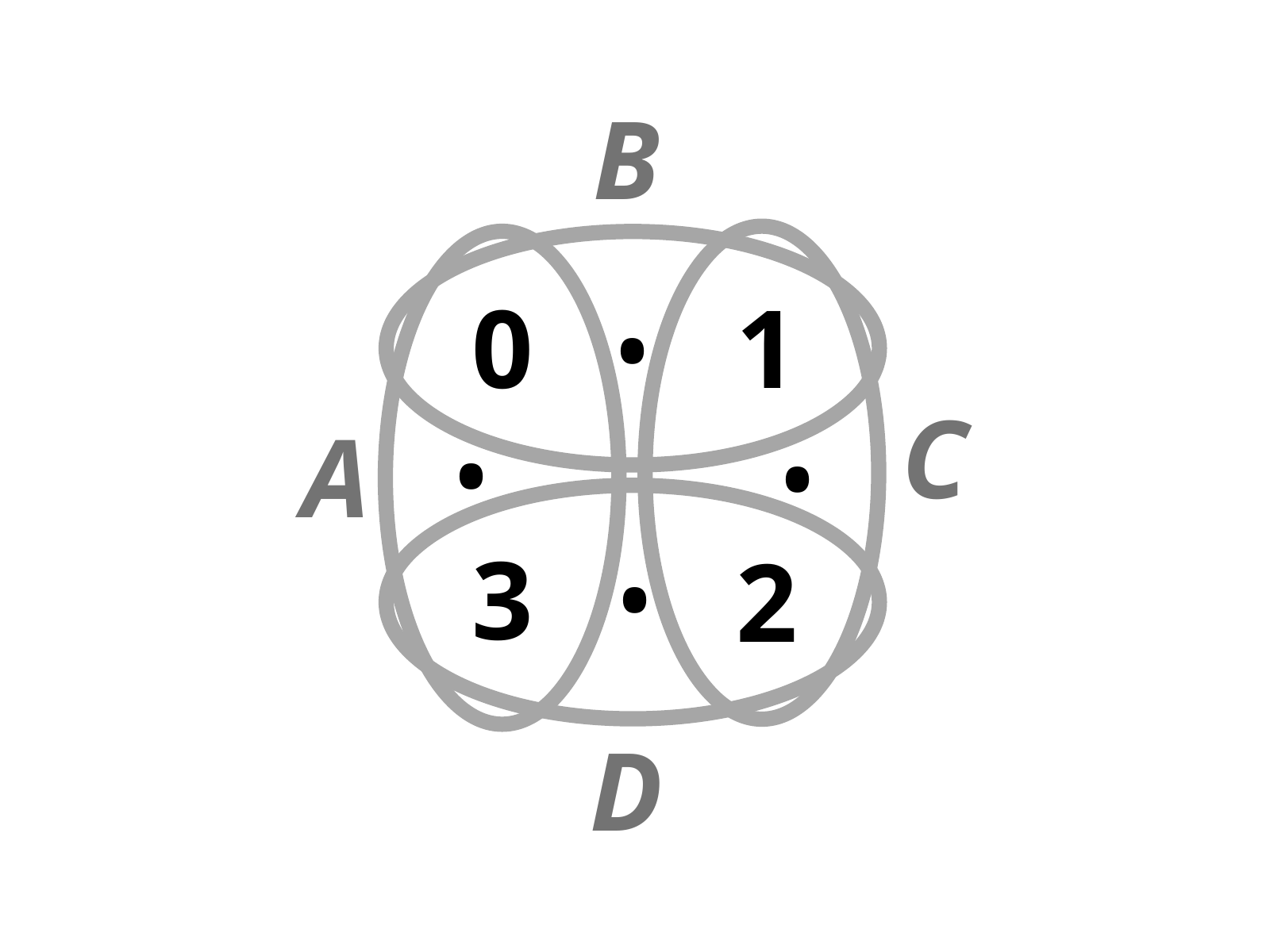}
\\[3ex]
\caption{Soberifications of the examples in Figure~\ref{fig.012-triangle} (Remark~\ref{rmrk.soberified.examples})}
\label{fig.012-triangle-sober}
\end{figure}

We will need Notation~\ref{nttn.irreducible.closed.set} for Remark~\ref{rmrk.topology.separation.axioms}:
\begin{nttn}
\label{nttn.irreducible.closed.set}
Call a closed set \deffont[irreducible closed set]{irreducible} when it cannot be written as the union of two proper closed subsets.
\end{nttn}

\begin{rmrk}
\label{rmrk.topology.separation.axioms}
Topology has a wealth of separation actions.
Three of them are: $T_0$ (distinct points have distinct neighbourhood (semi)filters); $T_1$ (distinct points have distinct open neighbourhoods); and $T_2$, also known as the Hausdorff condition (distinct points have disjoint open neighbourhoods) --- see Remark~\ref{rmrk.T0-T2} for formal statements.
For topologies, the following is known of sobriety:
\begin{enumerate*}
\item
Every finite $T_0$ (and thus every finite $T_1$) topological space is sober. 
\item
Every $T_2$/Hausdorff space (including infinite ones) is sober~\cite[page~475, Theorem~3]{maclane:sheglf}. 
\item
A topological space is sober if and only if every nonempty irreducible closed set is the closure of a unique point~\cite[page~475]{maclane:sheglf}. 
\end{enumerate*}
The situation for semitopologies is different, as we explore in the rest of this Subsection.
\end{rmrk}

\begin{lemm}
\label{lemm.T0.not.sober}
\leavevmode
\begin{enumerate*}
\item
It is not necessarily the case that a finite $T_0$ semitopology (or even a finite $T_1$ semitopology) is sober (Definition~\ref{defn.sober.semitopology}).
\item
It is not necessarily the case that if every nonempty irreducible closed set is the closure of a unique point, then a semitopology is sober.
\end{enumerate*}
These non-implications hold even if the semitopology is regular (so $p\in\community(p)\in\topens$ for every $p$; see Definition~\ref{defn.tn}(\ref{item.regular.point})).
\end{lemm}
\begin{proof}
We provide a semitopology that is a counterexample for both parts.

Consider the left-hand semitopology illustrated in Figure~\ref{fig.012-triangle}, so that:
\begin{itemize*}
\item
$\ns P=\{0,1,2\}$, and 
\item
$\opens=\{\varnothing,\{0,1\},\{1,2\},\{0,2\},\{0,1,2\}\}$.
\end{itemize*}
We note that:
\begin{itemize*}
\item
$(\ns P,\opens)$ is $T_0$ and $T_1$.
\item
$(\ns P,\opens)$ is regular because all points are intertwined, so that $\community(p)=\ns P$ for every $p\in\ns P$.
\item
The nonempty irreducible closed sets are $\{0\}$ (which is the complement of $\{1,2\}$), $\{1\}$, and $\{2\}$.
Since these are singleton sets, they are certainly the closures of unique points.
\end{itemize*}
So $(\ns P,\opens)$ is $T_0$, regular, and irreducible closed sets are the closures of unique points.

We take as our semifilter $P=\opens\setminus\{\varnothing\}$.
The reader can check that $P$ is completely prime (Definition~\ref{defn.point}(\ref{item.completely.prime})), nonempty, up-closed, and compatible ($P$ is also the greatest semifilter); but, $P$ is not the neighbourhood semifilter of $0$, $1$, or $2$ in $\ns P$.
Thus, $(\ns P,\opens)$ is not sober. 
\end{proof}

\begin{rmrk}
The counterexample used in Lemma~\ref{lemm.T0.not.sober} generalises, as follows: the reader can check that the \emph{all-but-one} semitopology from Example~\ref{xmpl.semitopologies}(\ref{item.counterexample.X-x}) on three or more points (so open sets are generated by $\ns P\setminus\{p\}$ for every $p\in\ns P$) has similar behaviour.
\end{rmrk}

In topology, every Hausdorff space is sober.
In semitopologies, this implication does not hold, and in a rather strong sense:
\begin{lemm}
\label{lemm.hausdorff.not.sober}
\leavevmode
\begin{enumerate*}
\item\label{item.hausdorff.not.sober}
It is not necessarily the case that if a semitopology is Hausdorff, then it is sober. 
\item
Every quasiregular Hausdorff semitopology is (discrete and therefore) sober.
\end{enumerate*}
\end{lemm}
\begin{proof}
We consider each part in turn:
\begin{enumerate}
\item
It suffices to give a counterexample. 
Consider the right-hand semitopology illustrated in Figure~\ref{fig.012-triangle} (which we also used, for different purposes, in Figure~\ref{fig.square.diagram}), so that:
\begin{itemize*}
\item
$\ns P=\{0,1,2,3\}$, and 
\item
$\opens$ is generated by $X=\{\{3,0\},\{0,1\},\{1,2\},\{2,3\}\}$.
\end{itemize*}
This is Hausdorff, but it is not sober: the reader can check that the up-closure $\{3,0\}^\leq\subseteq\opens$ 
is nonempty, up-closed, compatible, and completely prime, but it is not the neighbourhood filter of any $p\in\ns P$.
\item
From Lemmas~\ref{lemm.quasiregular.hausdorff} (quasiregular Hausdorff is discrete) and~\ref{lemm.discrete.sober} (discrete is sober).
\qedhere\end{enumerate}
\end{proof}

\begin{rmrk}
\label{rmrk.sober.not.hausdorff}
A bit more discussion of Lemma~\ref{lemm.hausdorff.not.sober}.
\begin{enumerate}
\item
The space used in the counterexample for part~\ref{item.hausdorff.not.sober} is Hausdorff, $T_1$, and unconflicted (Definition~\ref{defn.conflicted}(\ref{item.unconflicted})).
It is not quasiregular (Definition~\ref{defn.tn}(\ref{item.quasiregular.point})) because the community of every point is empty; see Proposition~\ref{prop.unconflicted.irregular}.
\item
The implication holds if we add quasiregularity as a condition: every quasiregular Hausdorff space is sober.
But, this holds for very bad reasons, because by Lemma~\ref{lemm.quasiregular.hausdorff} every quasiregular Hausdorff space is discrete.
\item
Thus, the non-implication discussed in Lemma~\ref{lemm.hausdorff.not.sober} is informative and tells us something interesting about semitopological sobriety.
Semitopological sobriety is not just a weak form of topological sobriety, and it has its own distinct personality.
\end{enumerate}
\end{rmrk}

\begin{rmrk}
\label{rmrk.soberified.examples}
We can inject the examples illustrated in Figure~\ref{fig.012-triangle} (used in Lemmas~\ref{lemm.T0.not.sober} and~\ref{lemm.hausdorff.not.sober}) into \emph{soberified} versions of the spaces that are sober and have an isomorphic lattice of open sets, as illustrated in Figure~\ref{fig.012-triangle-sober}:
\begin{enumerate*}
\item
The left-hand semitopology has abstract points (completely prime semifilters; see Definition~\ref{defn.point}(\ref{item.completely.prime})) generated as the $\subseteq$-up-closures of the following sets: $\{A\}$, $\{B\}$, $\{C\}$, $\{A,B\}$, $\{B,C\}$, $\{C,A\}$, and $\{A,B,C\}$.
Of these, $\{A,B\}^\subseteq=\nbhd(0)$, $\{B,C\}^\subseteq=\nbhd(1)$, and $\{C,A\}^\subseteq=\nbhd(2)$.
The other completely prime semifilters are not generated as the neighbourhood semifilters of any point in the original space, so we add points as illustrated using $\bullet$ in 
the left-hand diagram in Figure~\ref{fig.012-triangle-sober}.
This semitopology is sober, and has the same semiframe of open sets.
\item
For the right-hand example, we again add a $\bullet$ point for every abstract point in the original space that is not already the neighbourhood semifilter of a point in the original space.
These abstract points are generated as the $\subseteq$-up-closures of $\{A\}$, $\{B\}$, $\{C\}$, and $\{D\}$.

There is no need to add a $\bullet$ for the abstract point generated as the $\subseteq$-up-closure of $\{A,B\}$, because $\{A,B\}^\subseteq=\nbhd(0)$.
Similarly $\{B,C\}^\subseteq=\nbhd(1)$, $\{C,D\}^\subseteq=\nbhd(2)$, and $\{D,A\}^\subseteq=\nbhd(3)$.
Note that $\{A,B,C\}$ does not generate an abstract point because it is not compatible: $A\notbetween C$.
Similarly for $\{B,C,D\}$, $\{C,D,A\}$, $\{D,A,C\}$, and $\{A,B,C,D\}$.
\end{enumerate*}
These soberified spaces are instances of a general construction described in Theorem~\ref{thrm.nbhd.morphism}.
And, continuing the observation made in Remark~\ref{rmrk.sober.not.hausdorff}, note that neither of these spaces, with their extra points, are Hausdorff.
\end{rmrk}

\jamiesection{Four categories \& functors between them}
\label{sect.duality}

\jamiesubsection{The categories $\tf{SemiTop}/\tf{Sober}$ of semitopologies/sober semitopologies}

\begin{defn}
\label{defn.morphism.semitopologies}
\leavevmode
\begin{enumerate*}
\item\label{item.morphism.st}
Suppose $(\ns P,\opens)$ and $(\ns P',\opens')$ are semitopologies and $f:\ns P\to\ns P'$ is any function.
Then call $f$ a \deffont{morphism of semitopologies} when
$f$ is continuous, by which we mean (as standard) that 
$$
O'\in\opens'
\quad\text{implies}\quad
f^\mone(O')\in\opens .
$$ 
\item
Define $\tf{SemiTop}$ the \deffont[category of semitopologies $\tf{SemiTop}$]{category of semitopologies} such that: 
\begin{itemize*}
\item
objects are semitopologies, and 
\item
arrows are morphisms of semitopologies (continuous maps on points).\footnote{A discussion of possible alternatives, for future work, is in Remark~\ref{rmrk.more.conditions}.  See also Remarks~\ref{rmrk.no.meet} and~\ref{rmrk.other.properties}.}
\end{itemize*}
\item\label{item.cat.sober}
Write $\tf{Sober}$ for the \deffont[category of sober semitopologies $\tf{Sober}$]{category of sober semitopologies} and continuous functions between them.
By construction, $\tf{Sober}$ is the full subcategory of $\tf{SemiTop}$, on its sober semitopologies. 
\end{enumerate*}
\end{defn}

\begin{rmrk}
\label{rmrk.reading.nbhd.morphism}
For convenience reading Theorem~\ref{thrm.nbhd.morphism} we recall some facts:
\begin{enumerate*}
\item
The \emph{semiframe} 
$$
\tf{Fr}(\ns P,\opens)=(\opens,\subseteq,\between)
$$ 
from Definition~\ref{defn.semi.to.dg} has elements open sets $O\in\opens$, preordered by subset inclusion and with a compatibility relation given by sets intersection.

It is spatial, by Proposition~\ref{prop.Gr.P.spatial}.
\item
An abstract point $P$ in $\tf{Points}(\tf{Fr}(\ns P,\opens))$ is a completely prime nonempty up-closed compatible subset of $\opens$.
\item\label{item.describe.StFr}
$\tf{St}\,\tf{Fr}(\ns P,\opens)$ is by Definition~\ref{defn.st.g} a semitopology such that:
\begin{enumerate*}
\item
Its set of points is $\tf{Points}(\tf{Fr}(\ns P,\opens))$; the set of abstract points in $\tf{Fr}(\ns P,\opens)=(\opens,\subseteq,\between)$, the semilattice of open sets of $(\ns P,\opens)$, and 
\item
Its open sets are given by the $\f{Op}(O)$, for $O\in\opens$. 
\end{enumerate*}
It is sober, by Proposition~\ref{prop.spatial.gives.sober}. 
\end{enumerate*}
\end{rmrk}

Continuing Remark~\ref{rmrk.reading.nbhd.morphism}, a notation will be useful:
\begin{nttn}
\label{nttn.soberify}
Suppose $(\ns P,\opens)$ is a semitopology.
Then define
$$
\tf{Soberify}(\ns P,\opens) = \tf{St}\,\tf{Fr}(\ns P,\opens).
$$
We may use $\tf{Soberify}(\ns P,\opens)$ and $\tf{St}\,\tf{Fr}(\ns P,\opens)$ interchangeably, depending on whether we want to emphasise ``this is a sober semitopology obtained from $(\ns P,\opens)$'' or ``this is $\tf{St}$ acting on $\tf{Fr}(\ns P,\opens)=(\opens,\subseteq,\between)$''.
\end{nttn}

\begin{thrm}
\label{thrm.nbhd.morphism}
Suppose $(\ns P,\opens)$ is a semitopology.
Then 
\begin{enumerate*}
\item
$\nbhd:\ns P\to \tf{Points}(\tf{Fr}(\ns P,\opens))$ is a morphism of semitopologies from 
$(\ns P,\opens)$ to $\tf{St}\,\tf{Fr}(\ns P,\opens)=\tf{Soberify}(\ns P,\opens)$
\item
taking $(\ns P,\opens)$ to a sober semitopology $\tf{Soberify}(\ns P,\opens)$, such that
\item\label{item.nbhd.morphism.is.iso}
$\nbhd^\mone$ induces a bijection on open sets by mapping $\f{Op}(O)$ to $O$, and furthermore this is an isomorphism of the semiframes of open sets, in the sense of Definition~\ref{defn.semiframe.iso}.
\end{enumerate*}
\end{thrm}
\begin{proof}
We consider each part in turn:
\begin{enumerate}
\item
Following Definition~\ref{defn.morphism.semitopologies} 
we must show that $\nbhd$ is continuous (inverse images of open sets are open) from $(\ns P,\opens)$ to $\tf{Soberify}(\ns P,\opens)$.
So following Definition~\ref{defn.st.g}(\ref{item.st.op}), consider $\f{Op}(O)\in\opens(\tf{Soberify}(\ns P,\opens))$.
By Proposition~\ref{prop.nbhd.iff}(\ref{item.nbhd.mone}) 
$$
\nbhd^\mone(\f{Op}(O))=O\in\opens.
$$
Continuity follows.
\item
$\tf{Soberify}(\ns P,\opens)$ is sober by Proposition~\ref{prop.spatial.gives.sober}.
\item
This is Proposition~\ref{prop.nbhd.mone.bijects}.
\qedhere\end{enumerate}
\end{proof}

\begin{rmrk}
\label{rmrk.nbhd.summary}
We can summarise Theorem~\ref{thrm.nbhd.morphism} as follows: 
\begin{enumerate*}
\item
By construction, the kernel of the $\nbhd$ function (the relation determined by which points it maps to equal elements) is topological indistinguishability $\topind$. 
\item
We can think of $\tf{St}\,\tf{Fr}(\ns P,\opens)$ as being obtained from $(\ns P,\opens)$ by 
\begin{enumerate*}
\item
quotienting topologically equivalent points to obtain a $T_0$ space, and then
\item
adding extra points to make it sober.
\end{enumerate*}
See also the discussion in Remark~\ref{rmrk.enough.points} about what it means to have `enough' points.
\item
This is done without affecting the semiframe of open sets (up to isomorphism), with the semiframe bijection given by $\nbhd^\mone$.
\end{enumerate*}
In this sense, we can view $\tf{St}\,\tf{Fr}(\ns P,\opens)$ as a \deffont[soberification $\tf{St}\,\tf{Fr}(\ns P,\opens)$]{soberification}\index{$\tf{St}\,\tf{Fr}(\ns P,\opens)$ (soberification)} of $(\ns P,\opens)$.
\end{rmrk}

\jamiesubsection{The categories $\tf{SemiFrame}$/$\tf{Spatial}$ of semiframes/spatial semiframes}

\begin{defn}
\label{defn.category.of.spatial.graphs}
\leavevmode
\begin{enumerate*}
\item\label{item.category.spatial.morphism}
Suppose $(\ns X,\cti,\ast)$ and $(\ns X',\cti',\ast')$ are semiframes (Definition~\ref{defn.semiframe}) and $g:\ns X\to\ns X'$ is any function. 
Then call $g$ a \deffont{morphism of semiframes} when:
\begin{enumerate*}
\item\label{item.g.semilattice.morphism}
$g$ is a morphism of complete semilattices (Definition~\ref{defn.complete.semilattice.morphism}).
\item\label{item.g.compatible}
$g$ is \deffont[compatibility (of morphism of semiframes)]{compatible}, 
by which we mean that $g(x')\ast g(x'')$ implies $x'\ast x''$ for every $x',x''\in\ns X'$. 
\end{enumerate*}
\item\label{item.semilat}
We define $\tf{SemiFrame}$ the \deffont[category of semiframes $\tf{SemiFrame}$]{category of semiframes} such that:
\begin{itemize*}
\item
objects are semiframes, and 
\item
arrows are morphisms of semiframes.
\end{itemize*}
\item\label{item.spatial}
Write $\tf{Spatial}$ for the \deffont[category of spatial semiframes $\tf{Spatial}$]{category of spatial semiframes} and semiframe morphisms between them.
By construction, $\tf{Spatial}$ is the full subcategory of $\tf{SemiFrame}$, on its spatial semiframes (Definition~\ref{defn.spatial.graph}). 
\end{enumerate*}
\end{defn}

\begin{lemm}
\label{lemm.Op.morphism}
Suppose $(\ns X,\cti,\ast)$ is a semiframe. 
Then $\f{Op}:(\ns X,\cti,\ast)\to\tf{Fr}\,\tf{St}(\ns X,\cti,\ast)$ is a morphism of semiframes and is surjective on underlying sets.
\end{lemm}
\begin{proof}
Following Definition~\ref{defn.category.of.spatial.graphs}(\ref{item.category.spatial.morphism})
we must show that 
\begin{itemize*}
\item
$\f{Op}$ is a semilattice morphism (Definition~\ref{defn.complete.semilattice.morphism}) --- commutes with joins and maps $\ttop_{\ns X}$ to $\tf{Points}(\ns X,\cti,\ast)$) --- and 
\item
is compatible with the compatibility relation $\ast$, and
\item
we must show that $\f{Op}$ is surjective.
\end{itemize*}
We consider each property in turn:
\begin{itemize*}
\item
\emph{$\f{Op}$ is a semilattice morphism.}

$\f{Op}(\bigvee X)=\bigvee_{x\in X} \f{Op}(x)$ by Lemma~\ref{lemm.op.commutes.with.joins}, and $\f{Op}(\ttop_{\ns X})=\tf{Points}(\ns X,\cti,\ast)$ by Proposition~\ref{prop.semiframe.to.Op}(\ref{item.semiframe.to.Op.top}).
\item
\emph{$\f{Op}$ is compatible with $\ast$.}

Unpacking Definition~\ref{defn.category.of.spatial.graphs}(\ref{item.g.compatible}), we must show that $\f{Op}(x)\between \f{Op}(x')$ implies $x\ast x'$.
We use Proposition~\ref{prop.semiframe.to.Op}(\ref{item.semiframe.to.Op.between}).
\item
\emph{$\f{Op}$ is surjective}
\ \dots \ by Lemma~\ref{lemm.st.opens.generator}.
\qedhere
\end{itemize*}
\end{proof}

\jamiesubsection{Functoriality of the maps}

\begin{defn}
\label{defn.g.circ}
Suppose $g:(\ns X',\cti',\ast')\to(\ns X,\cti,\ast)\oldin \tf{SemiFrame}$.
Define a mapping $g^\circ:\tf{St}(\ns X,\cti,\ast)\to\tf{St}(\ns X',\cti',\ast')$ by 
$$
\begin{array}{r@{\ }l@{\quad}c@{\quad}l}
g^\circ :&\tf{Points}(\ns X,\cti,\ast) &\longrightarrow& \tf{Points}(\ns X',\cti',\ast')
\\
& P &\longmapsto& P'=\{x'\in\ns X' \mid g(x')\in P\}.
\end{array}
$$
\end{defn}

\begin{rmrk}
We will show that $g^\circ$ from Definition~\ref{defn.g.circ} is an arrow in $\tf{SemiTop}$.
We will need to prove the following:
\begin{itemize*}
\item
If $P\in\tf{Points}(\ns X,\cti,\ast)$ then $g^\circ(P)\in\tf{Points}(\ns X',\cti',\ast')$.
\item
$g^\circ$ is a morphism of semitopologies.
\end{itemize*}
We do this in Lemmas~\ref{lemm.gcirc.well-defined} and~\ref{lemm.g.circ.continuous} respectively.
\end{rmrk}

\begin{lemm}[$g^\circ$ well-defined]
\label{lemm.gcirc.well-defined}
Suppose $g:(\ns X',\cti',\ast')\to(\ns X,\cti,\ast) \oldin \tf{SemiFrame}$ and suppose $P\in\tf{Points}(\ns X,\cti,\ast)$.
Then $g^\circ(P)$ from Definition~\ref{defn.g.circ} is indeed in $\tf{Points}(\ns X',\cti',\ast')$ --- and thus $g^\circ$ is well-defined function from $\tf{Points}(\ns X,\cti,\ast)$ to $\tf{Points}(\ns X',\cti',\ast')$.
\end{lemm}
\begin{proof}
For brevity write
$$
P'=\{x'\in\ns X'\mid g(x')\in P\} .
$$ 
We must check that $P'$ is a completely prime nonempty up-closed compatible subset of $\ns X'$.
We consider each property in turn:
\begin{enumerate}
\item
\emph{$P'$ is completely prime.}\quad

Consider some $X'\subseteq P'$ and suppose $g(\bigvee X')\in P$.
By Definition~\ref{defn.category.of.spatial.graphs}(\ref{item.g.semilattice.morphism}) $g$ is a semilattice homomorphism, so by Definition~\ref{defn.complete.semilattice}(\ref{item.semilattice.morphism.top}) $g(\bigvee X')=\bigvee_{x'\in X'}g(x')$.
Thus $\bigvee_{x'\in X'} g(x')\in P$.
By assumption $P$ is completely prime, so $g(x')\in P$ for some $x'\in X'$.
Thus $x'\in P'$ for that $x'$.
Since $X'$ was arbitrary, it follows that $P'$ is completely prime.
\item
\emph{$P'$ is nonempty.}\quad

By assumption $g$ is an arrow in $\tf{SemiFrame}$ (i.e. a semiframe morphism) and unpacking Definition~\ref{defn.category.of.spatial.graphs}(\ref{item.g.semilattice.morphism}) it follows that it is a semilattice homomorphism.
In particular by Definition~\ref{defn.complete.semilattice}(\ref{item.semilattice.morphism.top}) $g(\ttop_{\ns X'})=\ttop_{\ns X}$, and 
by Lemma~\ref{lemm.P.top}(\ref{item.P.yes.top}) $\ttop_{\ns X}\in P$.
Thus $\ttop_{\ns X'}\in P'$, so $P'$ is nonempty.
\item
\emph{$P'$ is up-closed.}\quad

Suppose $x'\in P'$ and $x'\leq x''$.
By construction $g(x')\in P$.
By Lemma~\ref{lemm.semi.hom.mon} (because $g$ is a semilattice morphism by Definition~\ref{defn.category.of.spatial.graphs}(\ref{item.g.semilattice.morphism})) $g$ is monotone, so $g(x')\leq g(x'')$.
By assumption in Definition~\ref{defn.point}(\ref{item.up-closed}) $P$ is up-closed, so that $g(x'')\in P$ and thus $x''\in P'$ as required.
\item
\emph{$P'$ is compatible.}\quad

Suppose $x',x''\in P'$.
Thus $g(x'),g(x'')\in P$.
By assumption in Definition~\ref{defn.point}(\ref{item.weak.clique}) $P$ is compatible, so $g(x')\ast g(x'')$.
By compatibility of $g$ (Definition~\ref{defn.category.of.spatial.graphs}(\ref{item.g.compatible})) it follows that $x'\ast x''$.
Thus $P'$ is compatible.
\qedhere\end{enumerate}
\end{proof}

\begin{rmrk}
\label{rmrk.further.restrictions.on.points}
\emph{Note on design:}
If we want to impose further conditions on being an abstract point (such as those discussed in Remark~\ref{rmrk.other.properties}) then 
Lemma~\ref{lemm.gcirc.well-defined} would need to be extended to show that these further conditions are preserved by the $g^\circ$ operation, so that for $P\in\tf{Points}(\ns X,\cti,\ast)$ an abstract point in $(\ns X,\cti,\ast)$, $g^\circ(P)=\{x'\in\ns X'\mid g(x')\in P\}$ is an abstract point in $(\ns X',\cti',\ast')$. 

For example: consider what would happen if we add the extra condition on semifilters from Remark~\ref{rmrk.other.properties}.
Then the $P'$ defined in the proof of Lemma~\ref{lemm.gcirc.well-defined} above might not be closed under this additional condition (it will be if $g$ is surjective).
This could be mended by closing $P'$ under greatest lower bounds that are not $\tbot$, but that in turn might compromise the property of being completely prime.
These comments are not a proof that the problems would be insuperable; but they suggest that complexity would be added.
For now, we prefer to keep things simple!
\end{rmrk}

\begin{lemm}
\label{lemm.g.circ.inv}
Suppose $g:(\ns X',\cti',\ast')\to(\ns X,\cti,\ast) \oldin \tf{SemiFrame}$, and suppose $x'\in\ns X'$.
Then 
$$
(g^\circ)^\mone(\f{Op}(x'))=\f{Op}(g(x')).
$$
\end{lemm}
\begin{proof}
Consider an abstract point $P\in\tf{Point}(\tf{Gr}(\ns X',\cti',\ast'))$.
We just chase definitions:
$$
\begin{array}{r@{\ }l@{\quad}l}
P\in (g^\circ)^\mone(\f{Op}(x'))
\liff& 
g^\circ(P)\in \f{Op}(x')
&\text{Fact of inverse image}
\\
\liff&
x'\in g^\circ(P)
&\text{Definition~\ref{defn.Op}}
\\
\liff&
g(x')\in P 
&\text{Definition~\ref{defn.g.circ}}
\\
\liff&
P\in \f{Op}(g(x')).
&\text{Definition~\ref{defn.Op}}
\end{array}
$$
The choice of $P$ was arbitrary, so $(g^\circ)^\mone(\f{Op}(x'))=\f{Op}(g(x'))$ as required.
\end{proof}

\begin{lemm}[$g^\circ$ continuous]
\label{lemm.g.circ.continuous}
Suppose $g:(\ns X',\cti',\ast')\to(\ns X,\cti,\ast) \oldin \tf{SemiFrame}$. 
Then $g^\circ:\tf{St}(\ns X,\cti,\ast)\to\tf{St}(\ns X',\cti',\ast')$ is continuous:
$$
(g^\circ)^\mone(\mathcal O')\in\opens(\tf{St}(\ns X,\cti,\ast))
$$
for every $\mathcal O'\in\opens(\tf{St}(\ns X',\cti',\ast'))$. 
\end{lemm}
\begin{proof}
By Definition~\ref{defn.st.g}, $\mathcal O'=\f{Op}(x')$ for some $x'\in\ns X'$.
By Lemma~\ref{lemm.g.circ.inv} $(g^\circ)^\mone(\f{Op}(x'))=\f{Op}(g(x'))$.
By Definition~\ref{defn.st.g}(\ref{item.st.op}) $\f{Op}(g(x'))\in\opens(\tf{St}(\ns X,\cti,\ast))$.
\end{proof}

\begin{prop}[Functoriality]
\label{prop.semitop.adjunction}
\leavevmode
\begin{enumerate*}
\item
Suppose $f:(\ns P,\opens)\to(\ns P',\opens') \oldin \tf{SemiTop}$ (so $f$ is a continuous map on underlying points).

Then $f^\mone$ is an arrow $\tf{Fr}(\ns P',\opens')\to\tf{Fr}(\ns P,\opens)$ in $\tf{SemiFrame}$.
\item
Suppose $g:(\ns X',\cti',\ast')\to(\ns X,\cti,\ast) \oldin \tf{SemiFrame}$.

Then $g^\circ$ from Definition~\ref{defn.g.circ} is an arrow $\tf{St}(\ns X,\cti,\ast)\to\tf{St}(\ns X',\cti',\ast')$ in $\tf{SemiTop}$. 
\item
The assignments $f\mapsto f^\mone$ and $g\mapsto g^\circ$ are \deffont[functorial map]{functorial} --- they map identity maps to identity maps, and commute with function composition.
\end{enumerate*}
\end{prop}
\begin{proof}
We consider each part in turn:
\begin{enumerate}
\item
Following Definition~\ref{defn.category.of.spatial.graphs}, we must check that $f^\mone$ is a morphism of semiframes.
We just unpack what this means and see that the required properties are just facts of taking inverse images:
\begin{itemize*}
\item
\emph{$f^\mone$ commutes with joins, i.e. with $\bigcup$.}\quad

This is a fact of inverse images.
\item
\emph{$f^\mone$ maps $\ttop_{\tf{Fr}(\ns P',\opens')}=\ns P'$ to $\ttop_{\tf{Fr}(\ns P,\opens)}=\ns P$.}\quad

This is a fact of inverse images.
\item
\emph{$f^\mone$ is compatible, meaning that $f^\mone(O')\between f^\mone(O'')$ implies $O'\between O''$.}\quad 

This is a fact of inverse images.
\end{itemize*}
\item
We must check that $g^\circ$ is continuous.
This is Lemma~\ref{lemm.g.circ.continuous}.
\item
Checking functoriality is routine, but we sketch the reasoning anyway:
\begin{itemize*}
\item
Consider the identity function $\id$ on some semitopology $(\ns P,\opens)$.
Then $\id^\mone$ should be the identity function on $(\opens,\subseteq,\between)$. 
It is.
\item
Consider $f:(\ns P,\opens)\to(\ns P',\opens')$ and $f':(\ns P',\opens')\to(\ns P'',\opens'')$.
Then $(f'\circ f)^\mone$ should be equal to $f^\mone\circ (f')^\mone$.
It is.
\item
Consider the identity function $\id$ on $(\ns X,\cti,\ast)$.
Then $\id^\circ$ should be the identity function on $\tf{Points}(\ns X,\cti,\ast)$.
We look at Definition~\ref{defn.g.circ} and see that this amounts to checking that $P=\{x\in\ns X \mid \id(x)\in P\}$.
It is.
\item
Consider $g:(\ns X,\cti,\ast)\to(\ns X',\cti',\ast')$ and $g':(\ns X',\cti',\ast')\to(\ns X'',\cti'',\ast'')$ and consider some $P''\in\tf{Points}(\ns X'',\cti'',\ast'')$.
Then $(g'\circ g)^\circ(P'')$ should be equal to $(g^\circ\circ (g')^\circ)(P'')$.
We look at Definition~\ref{defn.g.circ} and see that this amounts to checking that 
$\{x\in\ns X \mid g'(g(x))\in P''\}=\{x\in\ns X \mid g(x)\in P'\}$ 
where
$P'=\{x'\in\ns X' \mid g'(x')\in P''\}$.
Unpacking these definitions, we see that the equality does indeed hold.
\qedhere\end{itemize*}
\end{enumerate}
\end{proof}

\jamiesubsection{Sober semitopologies are dual to spatial semiframes}

\begin{rmrk}
\label{rmrk.categorical.duality}
A \emph{categorical duality} between two categories $\mathbb C$ and $\mathbb D$ is a categorical equivalence between $\mathbb C$ and $\mathbb D^{op}$.
See~\cite[IV.4]{maclane:catwm} or~\cite[Subsection~1.5 \& Exercise~4.2.i]{riehl:cattic}.\footnote{The Wikipedia page is also exceptionally clear~\cite{wiki:Equivalence_of_categories}.}
It suffices to provide an adjoint pair of functors whose unit and counit are natural isomorphisms.\footnote{This is not a necessary condition; \emph{non-adjoint equivalences} are possible.  But the duality in this paper comes from an adjoint equivalence.}

There are many duality results in the literature.
The duality between topologies and frames is described (for example) in \cite[page~479, Corollary~4]{maclane:sheglf}.
A duality between distributive lattices and coherent spaces is in \cite[page~66]{johnstone:stos}.
There is the classic duality by Stone between Boolean algebras and compact Hausdorff spaces with a basis of clopen sets~\cite{stone:therba,johnstone:stos}. 
An encyclopedic treatment is in \cite{caramello:toptas}, with a rather good overview in Example~2.9 on page~17.

Theorem~\ref{thrm.categorical.duality.semiframes} states a duality result between $\tf{Sober}$ and $\tf{Spatial}$, 
thus appending another item to this extensive canon.

It also constructively moves us forward in studying semitopologies, because it gives us an algebraic treatment of semitopologies, and a formal framework for studying morphisms between semitopologies.
For instance: taking morphisms to be continuous functions is sensible not just because this is also how things work for topologies, but also because this is what is categorically dual to the ${\cti}/{\ast}$-homomorphisms between semiframes (Definition~\ref{defn.category.of.spatial.graphs}). 
And of course, if we become interested in different notions of semitopology morphism (a flavour of these is in Remark~\ref{rmrk.more.conditions}) then the algebraic framework gives us a distinct mathematical light with which to inspect and evaluate them. 
\end{rmrk}

\begin{thrm}
\label{thrm.categorical.duality.semiframes}
The maps $\tf{St}$ (Definition~\ref{defn.st.g}) and $\tf{Fr}$ (Definition~\ref{defn.semi.to.dg}), with actions on arrows as described in Proposition~\ref{prop.semitop.adjunction}, form a categorical duality between the categories
\begin{itemize*}
\item
$\tf{Sober}$ of sober semitopologies (Definition~\ref{defn.sober.semitopology}) and continuous compatible morphisms between them; and 
\item
$\tf{Spatial}$ of spatial semiframes and morphisms between them (Definition~\ref{defn.category.of.spatial.graphs}(\ref{item.spatial})).
\end{itemize*}
\end{thrm}
\begin{proof}
There are various things to check:
\begin{itemize}
\item
Proposition~\ref{prop.spatial.gives.sober} shows that $\tf{St}$ maps spatial semiframes to sober semitopologies.
\item
Proposition~\ref{prop.Gr.P.spatial} shows that $\tf{Fr}$ maps sober semitopologies to spatial semiframes.
\item
By Proposition~\ref{prop.semitop.adjunction} the maps $f\mapsto f^\mone$ (inverse image) and $g\mapsto g^\circ$ (Definition~\ref{defn.g.circ}) are functorial.
\item
The equivalence morphisms are given by the bijections $\f{Op}$ and $\nbhd$:
\begin{itemize*}
\item
$\f{Op}$ is from Definition~\ref{defn.Op}.
By Lemma~\ref{lemm.Op.morphism} $\f{Op}$ is a morphism $(\ns X,\cti,\ast)\to\tf{Fr}\,\tf{St}(\ns X,\cti,\ast)$ in $\tf{Spatial}$ that is surjective on underlying sets.
Injectivity is from Proposition~\ref{prop.Op.subseteq}(\ref{item.Op.spatial.inj}).
\item
$\nbhd$ is from Definition~\ref{defn.nbhd}.
By Theorem~\ref{thrm.nbhd.morphism} $\nbhd$ is a morphism $(\ns P,\opens)\to \tf{St}\,\tf{Fr}(\ns P,\opens)$ in $\tf{Sober}$.
It is a bijection on underlying sets by the sobriety condition in Definition~\ref{defn.sober.semitopology}.
\end{itemize*}
\end{itemize}
Finally, we must check naturality of $\f{Op}$ and $\nbhd$, which means (as standard) checking commutativity of the following diagrams:
$$
\begin{tikzcd}
(\ns P,\opens) \arrow{r}{\nbhd} \arrow{d}{f} 
& 
\tf{St}\,\tf{Fr}(\ns P,\opens) \arrow{d}{(f^\mone)^\circ}
\\
(\ns P',\opens') \arrow{r}{\nbhd} 
&
\tf{St}\,\tf{Fr}(\ns P',\opens') 
\end{tikzcd}
\qquad
\begin{tikzcd}
(\ns X,\cti,\ast) \arrow{r}{\f{Op}} \arrow{d}{g} 
& 
\tf{Fr}\,\tf{St}(\ns X,\cti,\ast) \arrow{d}{(g^\circ)^\mone}
\\
(\ns X',\cti',\ast') \arrow{r}{\f{Op}} 
&
\tf{Fr}\,\tf{St}(\ns X',\cti',\ast') 
\end{tikzcd}
$$
We proceed as follows:
\begin{itemize}
\item
Suppose $g:(\ns X',\cti',\ast')\to(\ns X,\cti,\ast)$ in $\tf{Spatial}$, so that 
$g^\circ:\tf{St}(\ns X,\cti,\ast)\to\tf{St}(\ns X',\cti',\ast')$ in $\tf{Sober}$ and
$(g^\circ)^\mone:\tf{Fr}\,\tf{St}(\ns X',\cti',\ast')\to \tf{Fr}\,\tf{St}(\ns X,\cti,\ast)$ in $\tf{Spatial}$.
To prove naturality we must check that
$$
(g^\circ)^\mone(\f{Op}(x)) = \f{Op}(g(x)) 
$$
for every $x\in\ns X$.
This is just Lemma~\ref{lemm.g.circ.inv}.
\item
Suppose $f:(\ns P,\opens)\to(\ns P',\opens')$ in $\tf{SemiTop}$, so that 
$f^\mone:\tf{Fr}(\ns P',\opens')\to\tf{Fr}(\ns P,\opens)$ in $\tf{Spatial}$ and
$(f^\mone)^\circ:\tf{St}\,\tf{Fr}(\ns P,\opens)\to \tf{St}\,\tf{Fr}(\ns P',\opens')$ in $\tf{SemiTop}$.
To prove naturality we must check that
$$
(f^\mone)^\circ(\nbhd(p)) = \nbhd(f(p)) .
$$
We just chase definitions, for an open set $O'\in\opens'$:
$$
\begin{array}[b]{r@{\ }l@{\qquad}l}
O'\in (f^\mone)^\circ(\nbhd(p))
\liff&
f^\mone(O') \in \nbhd(p)
&\text{Definition~\ref{defn.g.circ}}
\\
\liff&
p\in f^\mone(O')
&\text{Definition~\ref{defn.nbhd}}
\\
\liff&
f(p)\in O'
&\text{Inverse image}
\\
\liff&
O'\in\nbhd(f(p)) 
&\text{Definition~\ref{defn.nbhd}}
.
\end{array}
\qedhere$$
\end{itemize}
\end{proof}

\begin{rmrk}
We review the background to Theorem~\ref{thrm.categorical.duality.semiframes}:
\begin{enumerate*}
\item
A semitopology $(\ns P,\opens)$ is a set of points $\ns P$ and a set of open sets $\opens\subseteq\powerset(\ns P)$ that contains $\ns P$ and is closed under arbitrary (possibly empty) unions (Definition~\ref{defn.semitopology}).
\item
A morphism between semitopologies is a continuous function, just as for topologies (Definition~\ref{defn.morphism.semitopologies}(\ref{item.morphism.st})).
\item
A semiframe $(\ns X,\cti,\ast)$ is a complete join-semilattice $(\ns X,\cti)$ with a properly reflexive distributive \emph{compatibility relation} $\ast$ (Definition~\ref{defn.semiframe}).
\item
A morphism between semiframes is a morphism of complete join-semilattices with $\ttop$ that is compatible with the compatibility relation (Definition~\ref{defn.category.of.spatial.graphs}(\ref{item.category.spatial.morphism})).
\item
An \emph{abstract point} of a semitopology $(\ns P,\opens)$ is a completely prime nonempty up-closed compatible subset $P\subseteq\opens$ (Definition~\ref{defn.point}(\ref{item.abstract.point})).
\item
A semitopology is \emph{sober} when the neighbourhood semifilter map $p\in\ns P\mapsto \nbhd(p)=\{O\in\opens \mid p\in O\}$ is injective and surjective between the points of $\ns P$ and the abstract points of $\ns P$ (Definition~\ref{defn.sober.semitopology}).
\item
By Theorem~\ref{thrm.nbhd.morphism}, and as discussed in Remark~\ref{rmrk.nbhd.summary}, every (possibly non-sober) semitopology $(\ns P,\opens)$ maps into its \emph{soberification} $\tf{St}\,\tf{Fr}(\ns P,\opens)$, which has an isomorphic semiframe of open sets.
So even if our semitopology $(\ns P,\opens)$ is not sober, there is a standard recipe to make it so. 
\item
A semiframe is \emph{spatial} when $x\in\ns X \mapsto \f{Op}(x)=\{P\in\tf{Point} \mid x\in P\}$ respects $\cti$ and $\ast$ in senses make formal in
Definition~\ref{defn.spatial.graph} and Proposition~\ref{prop.Op.subseteq}.
\item 
Sober semitopologies and continuous functions between them, and spatial semiframes and semiframe morphisms between them, are categorically dual (Theorem~\ref{thrm.categorical.duality.semiframes}).
\end{enumerate*}
\end{rmrk}

\begin{rmrk}
Note what Theorem~\ref{thrm.categorical.duality.semiframes} does \emph{not} do: it does not give a duality between all semitopologies and all semiframes; it gives a duality between sober semitopologies and spatial semiframes.
This in itself is nothing new --- the topological duality is just the same ---
but what is interesting is that our motivation for studying semitopologies comes from practical network systems.
These tend to be (finite) non-sober semitopologies --- non-sober, because a guarantee of sobriety cannot be enforced, and anyway it is precisely the point of the exercise to achieve coordination, \emph{without} explicitly requiring every possible constellation of cooperating agents to be explicitly represented by a point.

It is true that by Theorem~\ref{thrm.nbhd.morphism} every non-sober $T_0$ semitopology can be embedded into a sober one without affecting the semiframe of open sets, but this makes the system to which it corresponds larger, by adding points.
Thus the duality in Theorem~\ref{thrm.categorical.duality.semiframes} is a mathematical statement, but not necessarily a practical one --- and this is as expected, because we knew that this is an abstract result.
$\nbhd$ maps a point to a set of (open) sets; and $\f{Op}$ maps an (open) set of points to a set of sets of (open) sets.
Of course these might not be computationally optimal.
\end{rmrk}

\jamiesection{Well-behavedness conditions, dually}
\label{sect.closer.look.at.semifilters}

We now study how properties of semifilters and abstract points correspond to the well-behavedness properties which we found useful in studying semitopologies --- for example \emph{topens}, \emph{regularity}, and being \emph{unconflicted} (Definitions~\ref{defn.transitive},\ \ref{defn.tn} and~\ref{defn.conflicted}). 

\jamiesubsection{(Maximal) semifilters and transitive elements}

\begin{rmrk}[Semifilters are not filters]
We know that semifilters do not necessarily behave like filters.
For instance:
\begin{enumerate*}
\item
It is possible for a finite semifilter to have more than one minimal element, because the closure under binary meets condition of filters is replaced by a weaker compatibility condition (see also Remarks~\ref{rmrk.no.meet} and~\ref{rmrk.other.properties}).
\item
There are more semifilters than proper filters --- even if the underlying space is a topology.
For example, the discrete semitopology on $\{0,1,2\}$ (whose open sets are all subsets of the space) is a topology.
Every proper filter in this space is a semifilter, but it also has a semifilter $\bigl\{\{0,1\},\{1,2\},\{2,0\},\{0,1,2\}\bigr\}$ (see the top-left diagram in Figure~\ref{fig.012-triangle}) and this is not a filter.
\end{enumerate*}
More on this in Subsection~\ref{subsect.things.that.are.different}.

In summary: semifilters are different and we cannot necessarily take their behaviour for granted without checking it.
We now examine them in more detail.
\end{rmrk}

We start with some easy definitions and results:
\begin{nttn}
\label{nttn.X.ast.Y}
Suppose $(\ns X,\cti,\ast)$ is a semiframe and $X,Y\subseteq\ns X$ and $x\in\ns X$.
Then we generalise $x\ast y$ to $x\ast Y$, $X\ast y$, and $X\ast Y$ as follows:
\begin{enumerate*}
\item\label{item.x.ast.Y}
Write $x\ast Y$ for $\Forall{y{\in}Y}x\ast y$. 
\item\label{item.X.ast.y}
Write $X\ast y$ for $\Forall{x{\in}X}x\ast y$. 
\item\label{item.X.ast.Y}
Write $X\ast Y$ for $\Forall{x{\in}X}\Forall{y{\in}Y}x\ast y$. 
\end{enumerate*}
We read $x\ast Y$ as `$x$ is \deffont[compatibility (of $x$ with $Y$:\ $x\ast Y$)]{compatible} with $Y$', and similarly for $X\ast y$ and $X\ast Y$.
\end{nttn}

\begin{rmrk}
\label{rmrk.promise.ast.int.char}
We will see later on in Lemma~\ref{lemm.intertwined.sober} that
$X\ast X'$ generalises $p\intertwinedwith p'$, in the sense that if $X=\nbhd(p)$ and $X'=\nbhd(p')$, then $p\intertwinedwith p'$ if and only if $\nbhd(p)\ast\nbhd(p')$.
\end{rmrk}

\begin{lemm}[Characterisation of maximal semifilters]
\label{lemm.char.maxfilter}
Suppose $(\ns X,\cti,\ast)$ is a semiframe and $\afilter\subseteq\ns X$ is a semifilter.
Then the following conditions are equivalent:
\begin{enumerate*}
\item
$\afilter$ is maximal.
\item
For every $x\in\ns X$, $x\ast\afilter$ if and only if $x\in\afilter$. 
\end{enumerate*}
\end{lemm}
\begin{proof}
We prove two implications:
\begin{itemize}
\item
\emph{Suppose $\afilter$ is a maximal semifilter.}

Suppose $x\in\afilter$.
Then $x\ast\afilter$ is immediate from 
Notation~\ref{nttn.X.ast.Y}(\ref{item.x.ast.Y}) and semifilter compatibility (Definition~\ref{defn.point}(\ref{item.weak.clique})).

Suppose $x\ast\afilter$; thus by Notation~\ref{nttn.X.ast.Y}(\ref{item.x.ast.Y}) $x$ is compatible with (every element of) $\afilter$.
We note that the $\cti$-up-closure of $\{x\}\cup\afilter$ is a semifilter (nonempty, up-closed, compatible).
By maximality, $x\in\afilter$.
\item
\emph{Suppose $x\ast\afilter$ if and only if $x\in\afilter$, for every $x\in\ns X$.}

Suppose $\afilter'$ is a semifilter and $\afilter\subseteq\afilter'$.
Consider $x'\in\afilter'$.
Then $x\ast\afilter$ by compatibility of $\afilter'$, and so $x\in\afilter$.
Thus, $\afilter'\subseteq\afilter$. 
\qedhere\end{itemize}
\end{proof}

\begin{defn}
\label{defn.x.transitive}
Suppose $(\ns X,\cti,\ast)$ is a semiframe and $x\in\ns X$.
Call $x$ \deffont[transitive element (in a semiframe)]{transitive} when:
\begin{enumerate*}
\item
$x\neq\tbot_{\ns X}$.
\item
$x'\ast x\ast x''$ implies $x'\ast x''$, for every $x',x''\in\ns X$.
\end{enumerate*}
\end{defn}

`Being topen' in semitopologies (Definition~\ref{defn.transitive}(\ref{transitive.cc})) corresponds to `being transitive' in semiframes (Definition~\ref{defn.x.transitive}): 
\begin{lemm}[Characterisation of topen sets]
\label{lemm.topen.transitive}
Suppose $(\ns P,\opens)$ is a semitopology and $O\in\opens$.
Then the following are equivalent:
\begin{enumerate*}
\item
$O$ is topen in $(\ns P,\opens)$ in the sense of Definition~\ref{defn.transitive}(\ref{transitive.cc}).
\item
$O$ is transitive in $(\opens,\subseteq,\between)$ in the sense of Definition~\ref{defn.x.transitive}.\footnote{\emph{Confusing terminology alert:}  Definition~\ref{defn.transitive}(\ref{transitive.transitive}) also has a notion of \emph{transitive set}.  The notion of transitive set is well-defined for a set that may not be open.  In the world of semiframes, we just have elements of the semiframe (which correspond, intuitively, to open sets).  Thus \emph{transitive} semiframe elements correspond to (nonempty) transitive open sets of a semitopology, which are called \emph{topens}.}
\end{enumerate*}
\end{lemm}
\begin{proof}
We unpack the definitions and note that the condition for being topen --- being a nonempty open set that is transitive for $\between$ --- is identical to the condition for being transitive in $(\opens,\subseteq,\between)$ --- being a non-$\tbot_{\opens}$ element that is transitive for ${\ast}={\between}$.
\end{proof}

\jamiesubsection{The compatibility system $x^\ast$}

\begin{defn}
\label{defn.x.ast}
Suppose $(\ns X,\cti,\ast)$ is a semiframe and $x\in\ns X$.
Then define $x^\ast$ the \deffont[compatibility system (of an element:\ $x^\ast$)]{compatibility system}\index{$x^\ast$ (compatibility system of a semiframe element)} of $x$ by
$$
x^\ast=\{x'\in\ns X\mid x'\ast x\} .
$$
\end{defn}

\begin{lemm}
\label{lemm.bigvee.ast.union}
Suppose $(\ns X,\cti,\ast)$ is a semiframe and $X\subseteq\ns X$.
Then $(\bigvee X)^\ast = \bigcup_{x{\in}X} x^\ast$.
\end{lemm}
\begin{proof}
We just follow the definitions:
$$
\begin{array}[b]{r@{\ }l@{\qquad}l}
y\in(\bigvee X)^\ast
\liff&
y\ast \bigvee X
&\text{Definition~\ref{defn.x.ast}}
\\
\liff&
\Exists{x{\in}X}y\ast x
&\text{Definition~\ref{defn.compatibility.relation}(\ref{item.compatible.distributive})}
\\
\liff&
\Exists{x{\in}X}y\in x^\ast
&\text{Definition~\ref{defn.x.ast}}
\\
\liff&
y\in\bigcup_{x{\in}X} x^\ast 
&\text{Fact of sets}
\end{array}
\qedhere$$
\end{proof}

\begin{lemm}
\label{lemm.x.ast.cycle}
Suppose $(\ns X,\cti,\ast)$ is a semiframe and $x\in\ns X$ is transitive.
Then the following are equivalent for every $y\in\ns X$:
$$
y\ast x
\quad\liff\quad
y\in x^\ast
\quad\liff\quad
y\ast x^\ast .
$$
\end{lemm}
\begin{proof}
We prove a cycle of implications:
\begin{itemize*}
\item
Suppose $y\ast x$.
Then $y\in x^\ast$ is direct from Definition~\ref{defn.x.ast}.
\item
Suppose $y\in x^\ast$.
Then $y\ast x^\ast$ --- meaning by Notation~\ref{nttn.X.ast.Y}(\ref{item.x.ast.Y}) that $y\ast x'$ for every $x'\in x^\ast$ --- follows by transitivity of $x$.
\item
Suppose $y\ast x^\ast$.
By proper reflexivity of $\ast$ (Definition~\ref{defn.compatibility.relation}(\ref{item.compatible.reflexive}); since $x\neq\tbot_{\ns X}$) $x\in x^\ast$, and $y\ast x$ follows.
\qedhere\end{itemize*} 
\end{proof}

\begin{prop}
\label{prop.trans.cps}
Suppose $(\ns X,\cti,\ast)$ is a semiframe and suppose $\tbot_{\ns X}\neq x\in\ns X$.
Then the following are equivalent:
\begin{enumerate*}
\item\label{item.cps.transitive}
$x$ is transitive. 
\item\label{item.cps.point}
$x^\ast$ is a completely prime semifilter (i.e. an abstract point).
\item\label{item.cps.semifilter}
$x^\ast$ is a semifilter.
\item\label{item.cps.compatible}
$x^\ast$ is compatible.
\item\label{item.cps.maximal}
$x^\ast$ is a maximal semifilter.
\end{enumerate*}
\end{prop}
\begin{proof}
We first prove a cycle of implications between parts~\ref{item.cps.transitive}, \ref{item.cps.point}, \ref{item.cps.semifilter}, and~\ref{item.cps.compatible}:
\begin{enumerate}
\item
Suppose $x$ is transitive.
We need to check that $x^\ast$ is nonempty, up-closed, compatible, and completely prime.
We consider each property in turn:
\begin{itemize*}
\item
$x\ast x$ by proper reflexivity of $\ast$ (Definition~\ref{defn.compatibility.relation}(\ref{item.compatible.reflexive}); since $x\neq\tbot_{\ns X}$), so $x\in x^\ast$.
\item
It follows from monotonicity of $\ast$ (Lemma~\ref{lemm.compatibility.monotone}(\ref{item.ast.monotone})) that if $x'\cti x''$ and $x\ast x'$ then $x\ast x''$.
\item
Suppose $x'\ast x\ast x''$.
By transitivity of $x$ (Definition~\ref{defn.x.transitive}), $x'\ast x''$.
\item
Suppose $x\ast\bigvee X'$; then by distributivity of $\ast$ (Definition~\ref{defn.compatibility.relation}(\ref{item.compatible.distributive}))
$x\ast x'$ for some $x'\in X'$.
\end{itemize*}
\item
If $x^\ast$ is a completely prime semifilter, then it is certainly a semifilter.
\item
If $x^\ast$ is a semifilter, then it is compatible (Definition~\ref{defn.point}(\ref{item.semifilter})\&\ref{item.weak.clique}).
\item
Suppose $x^\ast$ is compatible (Definition~\ref{defn.point}(\ref{item.weak.clique})) and suppose $x'\ast x\ast x''$.
By Lemma~\ref{lemm.x.ast.cycle} $x',x''\in x^\ast$, and by compatibility of $x^\ast$ we have $x'\ast x''$.
Thus, $x$ is transitive.
\end{enumerate}
To conclude, we prove two implications between parts~\ref{item.cps.compatible} and~\ref{item.cps.maximal}:
\begin{itemize}
\item
Suppose $x^\ast$ is a semifilter.
By equivalence of parts~\ref{item.cps.semifilter} and~\ref{item.cps.transitive} of this result, $x$ is transitive, and so 
using Lemma~\ref{lemm.x.ast.cycle} 
$x'\ast x^\ast$ if and only if $x'\in x^\ast$.
By Lemma~\ref{lemm.char.maxfilter}, $x^\ast$ is maximal.
\item
Clearly, if $x^\ast$ is a maximal semifilter then it is a semifilter.
\qedhere\end{itemize}
\end{proof}

\jamiesubsection{The compatibility system $\afilter^\ast$}

\jamiesubsubsection{Basic definitions and results}

\begin{defn}
\label{defn.X.ast}
Suppose $(\ns X,\cti,\ast)$ is a semiframe and suppose $\afilter\subseteq\ns X$ ($\afilter$ may be a semifilter, but the definition does not depend on this).
Define $\afilter^\ast$ the \deffont[compatibility system (of a set;\ $\afilter^\ast$)]{compatibility system}\index{$\afilter^\ast$ (compatibility system of a set)} of $\afilter$ by
$$
\afilter^\ast  = \{x'\in\ns X \mid x'\ast\afilter\}
$$
Unpacking Notation~\ref{nttn.X.ast.Y}(\ref{item.x.ast.Y}), and combining with Definition~\ref{defn.x.ast}, we can write:
$$
\afilter^\ast  
=
\{x'\in\ns X \mid x'\ast\afilter\}
=
\{x'\in\ns X \mid \Forall{x{\in}\afilter}x'\ast x\}
=
\bigcap\{x^\ast \mid x\in \afilter \} . 
$$
\end{defn}

Lemma~\ref{lemm.nbhd.ast.char} presents one easy and useful example of Definition~\ref{defn.X.ast}:
\begin{lemm}
\label{lemm.nbhd.ast.char}
Suppose $(\ns P,\opens)$ is a semitopology and suppose $p\in\ns P$ and $O'\in\opens$.
Then: 
$$
\begin{array}{l@{\ \liff\ }l}
O'\in\nbhd(p)^\ast
&
\Forall{O{\in}\opens}p\in O\limp O'\between O
\\
O'\notin\nbhd(p)^\ast
&
\Exists{O{\in}\opens}p\in O\land O'\notbetween O .
\end{array}
$$
\end{lemm}
\begin{proof}
We just unpack Definitions~\ref{defn.nbhd} and~\ref{defn.X.ast}.
\end{proof}

\begin{lemm}
\label{lemm.X.ast.up-closed}
Suppose $(\ns X,\cti,\ast)$ is a semiframe and $\afilter\subseteq\ns X$.
Then
$\afilter^\ast$ is up-closed.
\end{lemm}
\begin{proof}
This is just from Definition~\ref{defn.X.ast} and monotonicity of $\ast$ (Lemma~\ref{lemm.compatibility.monotone}(\ref{item.ast.monotone})).
\end{proof}

\begin{lemm}
\label{lemm.afilter.subset.afilter.ast}
Suppose $(\ns X,\cti,\ast)$ is a semiframe and $\afilter\subseteq\ns X$ is a semifilter.
Then:
\begin{enumerate*}
\item\label{item.afilter.subset.afilter.ast.1}
If $x\in \afilter$ then $\afilter \subseteq x^\ast$.
\item\label{item.afilter.subset.afilter.ast.2}
As a corollary, $\afilter\subseteq\afilter^\ast$.
\end{enumerate*}
\end{lemm}
\begin{proof}
Suppose $x\in\afilter$.
By compatibility of $\afilter$ (Definition~\ref{defn.point}(\ref{item.weak.clique})), $x'\ast x$ for every $x'\in\afilter$.
It follows from Definition~\ref{defn.x.ast} that $\afilter\subseteq x^\ast$.
The corollary is immediate from Definition~\ref{defn.X.ast}.
\end{proof}

We can use Lemma~\ref{lemm.afilter.subset.afilter.ast} and Definition~\ref{defn.X.ast} to give a more succinct rendering of Lemma~\ref{lemm.char.maxfilter}:
\begin{corr}
\label{corr.new.char.maxfilter}
Suppose $(\ns X,\cti,\ast)$ is a semiframe and $\afilter\subseteq\ns X$ is a semifilter.
Then the following are equivalent:
\begin{enumerate*}
\item\label{item.new.char.maxfilter.1}
$\afilter$ is maximal.
\item\label{item.new.char.maxfilter.2}
$\afilter^\ast=\afilter$.
\item\label{item.new.char.maxfilter.3}
$\afilter^\ast\subseteq\afilter$.
\end{enumerate*}
\end{corr}
\begin{proof}
Equivalence of parts~\ref{item.new.char.maxfilter.1} and~\ref{item.new.char.maxfilter.2} just repeats Lemma~\ref{lemm.char.maxfilter} using Definition~\ref{defn.X.ast}.
To prove equivalence of parts~\ref{item.new.char.maxfilter.2} and~\ref{item.new.char.maxfilter.3} we use use Lemma~\ref{lemm.afilter.subset.afilter.ast}(\ref{item.afilter.subset.afilter.ast.2}).
\end{proof}

\jamiesubsubsection{Strong compatibility: when $\afilter^\ast$ is a semifilter}

Proposition~\ref{prop.trans.cps} relates good properties of $x$ (transitivity) to good properties of its compatibility system $x^\ast$ (e.g. being compatible).
It will be helpful to ask similar questions of $\afilter^\ast$.
What good properties are of interest for $\afilter^\ast$, and what conditions can we impose on $\afilter$ to guarantee them?

\begin{defn}
\label{defn.F.strongly.compatible}
Suppose $(\ns X,\cti,\ast)$ is a semiframe.
Then:
\begin{enumerate*}
\item\label{item.strongly.compatible.filter}
Call $\afilter\subseteq\ns X$ \deffont[strong compatibility (of a set)]{strongly compatible} when $\afilter^\ast$ is nonempty and compatible.
\item\label{item.strongly.compatible.filter.space}
Call $(\ns X,\cti,\ast)$ \deffont[strongly compatible semiframe]{strongly compatible} when every abstract point (completely prime semifilter) $\apoint\subseteq\ns X$ is strongly compatible.
\end{enumerate*}
\end{defn}

\begin{rmrk}
\label{rmrk.what.does.strongly.compatible.mean}
For the reader's convenience we unpack Definition~\ref{defn.F.strongly.compatible}.
\begin{enumerate}
\item
By Definition~\ref{defn.point}(\ref{item.weak.clique}), $\afilter^\ast$ is compatible when $x\ast x'$ for every $x,x'\in\afilter^\ast$.
Combining this with Definition~\ref{defn.X.ast} and Notation~\ref{nttn.X.ast.Y}, $\afilter^\ast$ is compatible when $x\ast \afilter\ast x'$ implies $x\ast x'$, for every $x,x'\in\ns X$.
Thus, $\afilter$ is strongly compatible when 
$$
\Forall{x,x'{\in}\ns X}\ x\ast\afilter\ast x' \limp x\ast x'.
$$
\item
$(\ns X,\cti,\ast)$ is strongly compatible when every abstract point $\apoint\in\tf{Point}(\ns X,\cti,\ast)$ is strongly compatible in the sense just given above.
\end{enumerate}
\end{rmrk}

\begin{lemm}
\label{lemm.ht.sc.eq}
Suppose $(\ns P,\opens)$ is a semitopology and $p\in\ns P$.
Recall from Definition~\ref{defn.semi.to.dg}(\ref{item.semiframe.ast}) and Lemma~\ref{lemm.Fr.semiframe} that $(\opens,\subseteq,\between)$ is a semiframe.
Then the following are equivalent:
\begin{enumerate*}
\item
The point $p\in\ns P$ is hypertransitive in the sense of Definition~\ref{defn.sc}.
\item
The semifilter $\nbhd(p)\subseteq\opens$ is strongly compatible in the sense of Definition~\ref{defn.F.strongly.compatible}.
\end{enumerate*} 
\end{lemm}
\begin{proof}
Remark~\ref{rmrk.what.does.strongly.compatible.mean} notes that the condition in Definition~\ref{defn.sc} is precisely the condition for $\nbhd(p)$ to be strongly compatible. 
\end{proof}

\begin{rmrk}
\label{rmrk.semiframes.caution}
Given Lemma~\ref{lemm.ht.sc.eq}, the reader might ask why we do not just call a strongly compatible semifilter `hypertransitive'.

There is a case for doing so, but caution is required:
strong compatibility of semiframes is not \emph{quite} the same thing as hypertransitivity of points.
Every point $p$ generates a semifilter $\nbhd(p)$, but there may be more semifilters than there are points, and this makes the strong compatibility condition subtly different from the hypertransitivity condition. 
We shall see the effects of this in Lemma~\ref{lemm.r=wr+sc}(\ref{item.wr.kt.1b}), and in Theorem~\ref{thrm.r=wr+sc} (see Remark~\ref{rmrk.subtly.different} for a brief discussion), and then again in Definition~\ref{defn.strongly.compatible.semitopology} where we define a notion of \emph{strongly compatible semitopology} (essentially: all of its semifilters are strongly compatible), which is not the same thing as the space being hypertransitive (essentially: all of its points are hypertransitive). 

Therefore, we maintain a terminological distinction: \emph{points} are hypertransitive, \emph{semiframes} are strongly compatible.
The notions are related, but not quite the same thing.
\end{rmrk}

\begin{lemm}
\label{lemm.ast.semifilter.compatible}
Suppose $(\ns X,\cti,\ast)$ is a semiframe and suppose $\afilter\subseteq\ns X$ is nonempty.
Then the following are equivalent:
\begin{enumerate*}
\item\label{item.ast.semifilter.compatible.1}
$\afilter^\ast$ is a semifilter.
\item\label{item.ast.semifilter.compatible.2}
$\afilter^\ast$ is compatible.
\item\label{item.ast.semifilter.compatible.3}
$\afilter$ is strongly compatible.
\end{enumerate*}
\end{lemm}
\begin{proof}
Equivalence of parts~\ref{item.ast.semifilter.compatible.2} and~\ref{item.ast.semifilter.compatible.3} is just Definition~\ref{defn.F.strongly.compatible}.
For equivalence of parts~\ref{item.ast.semifilter.compatible.1} and~\ref{item.ast.semifilter.compatible.2} we prove two implications:
\begin{itemize}
\item
Suppose $\afilter^\ast$ is a semifilter.
Then $\afilter^\ast$ is compatible by assumption in Definition~\ref{defn.point}(\ref{item.semifilter}).
\item
Suppose $\afilter^\ast$ is compatible.
Then $\afilter^\ast$ is up-closed by Lemma~\ref{lemm.X.ast.up-closed}, and nonempty by Lemma~\ref{lemm.afilter.subset.afilter.ast}(\ref{item.afilter.subset.afilter.ast.2}) (since $\afilter$ is nonempty).
Thus, by Definition~\ref{defn.point}(\ref{item.semifilter}) $\afilter^\ast$ is a semifilter.
\qedhere\end{itemize}
\end{proof}

\begin{lemm}
\label{lemm.not.necessarily.strongly.compatible}
Suppose $(\ns X,\cti,\ast)$ is a semiframe and suppose $\afilter\subseteq\ns X$.
Then it is not necessarily the case that $\afilter^\ast$ is a semifilter.

This non-implication holds even in strong well-behavedness conditions:  that $(\ns X,\cti,\ast)$ is spatial and $\afilter$ is an abstract point (a completely prime semifilter).
\end{lemm}
\begin{proof}
It suffices to provide a counterexample.
Let $(\ns P,\opens)=(\{0,1,2\},\{\varnothing,\{0\},\{2\},\ns P\})$, as illustrated in the top-left semitopology in Figure~\ref{fig.012}.
Take $(\ns X,\cti,\ast)=(\opens,\subseteq,\between)$ (which is spatial by Proposition~\ref{prop.Gr.P.spatial}) and set $\afilter=\nbhd(1)=\{0,1,2\}$.
Then $\nbhd(1)^\ast=\{\{0\},\{2\},\{0,1,2\}\}$, and this is not compatible because $\{0\}\notbetween\{2\}$.\footnote{$1$ is also a \emph{conflicted} point; see Example~\ref{xmpl.conflicted.points}(\ref{item.example.of.conflicted.point}).  This is no accident: by Lemma~\ref{lemm.regular.sc}(\ref{item.sc.implies.uc}) if $p$ is conflicted then it is not hypertransitive, and by Lemma~\ref{lemm.ht.sc.eq} it follows that $\nbhd(p)^\ast$ is not compatible.}
\end{proof}

\begin{rmrk}
Lemma~\ref{lemm.not.necessarily.strongly.compatible} gives an example of a semifilter $\afilter$ that is not strongly compatible (i.e. such that $\afilter^\ast$ is not a semifilter).
Note that in this example both the space and $\afilter$ are well-behaved.
This raises the question of finding sufficient (though perhaps not necessary) criteria for strong compatibility.
We conclude with Proposition~\ref{prop.x.transitive.afilter.ast} which provides one such criterion; it will be useful later in Lemma~\ref{lemm.r=wr+sc} and Theorem~\ref{thrm.r=wr+sc.st}.
\end{rmrk}

\begin{figure}
\centering
\includegraphics[width=0.4\columnwidth,trim={0 100 0 100},clip]{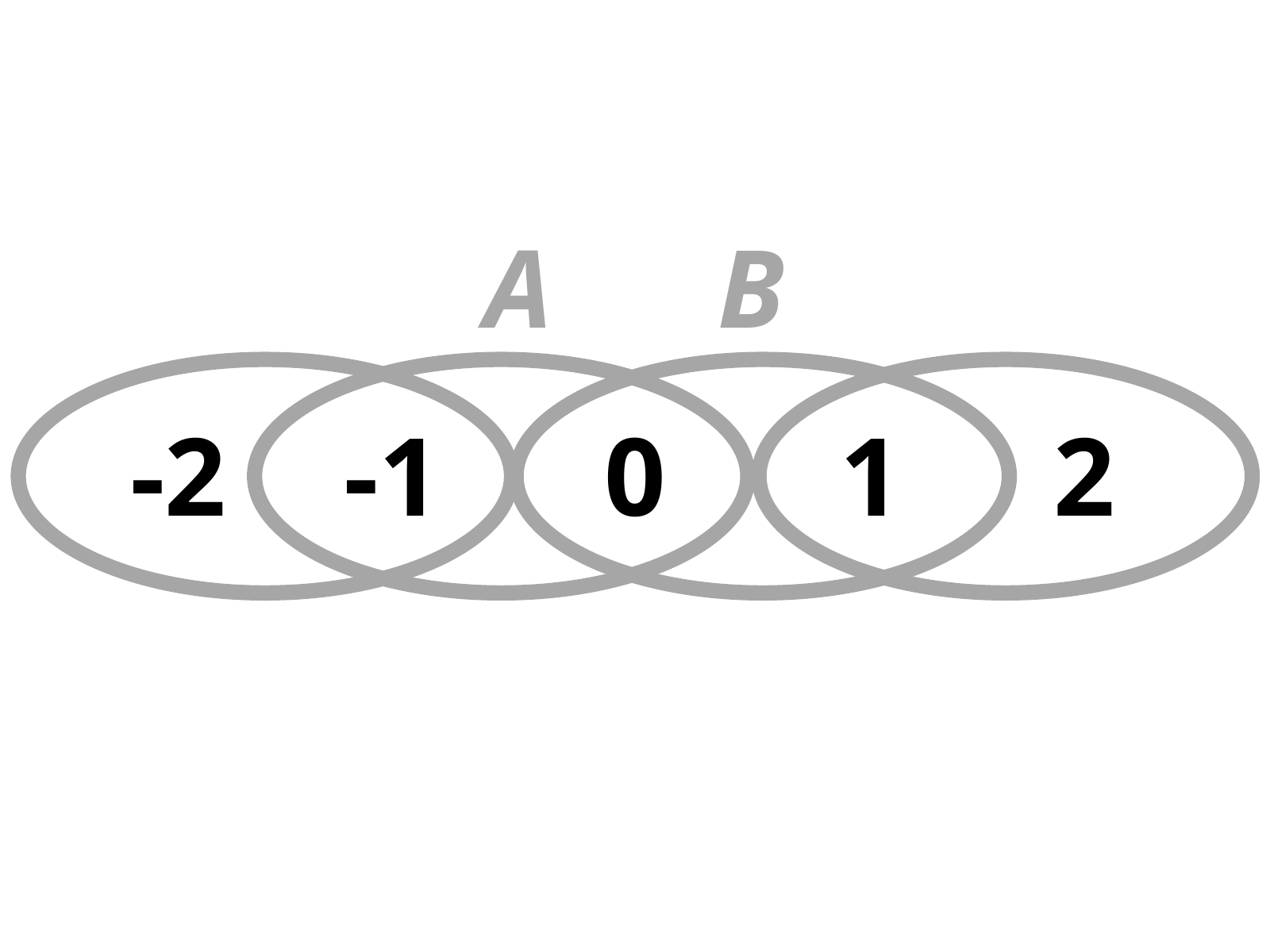}
\caption{Strongly compatible filter that contains no transitive element}
\label{fig.strong.compat.no.transitive}
\end{figure}

Proposition~\ref{prop.x.transitive.afilter.ast} bears a family resemblance to Theorem~\ref{thrm.max.cc.char} (if a point has a topen neighbourhood then it is regular):
\begin{prop}
\label{prop.x.transitive.afilter.ast}
Suppose $(\ns X,\cti,\ast)$ is a semiframe and $\afilter\subseteq\ns X$ is a semifilter.
Then:
\begin{enumerate*}
\item
If $\afilter$ contains a transitive element then $\afilter$ is strongly compatible.
\item
The converse implication need not hold: it may be that $\afilter$ is strongly compatible yet $\afilter$ contains no transitive element.
\end{enumerate*}
\end{prop}
\begin{proof}
We consider each part in turn:
\begin{enumerate}
\item
Suppose $x\in\afilter$ is transitive.
By Lemma~\ref{lemm.ast.semifilter.compatible} it would suffice to show that $\afilter^\ast$ is compatible (Definition~\ref{defn.point}(\ref{item.weak.clique})).
So consider $y\ast\afilter\ast y'$.
Then $y\ast x\ast y'$ and by transitivity $y\ast y'$.
Thus $\afilter^\ast$ is compatible.
\item
It suffices to provide a counterexample.
We take, as illustrated in Figure~\ref{fig.strong.compat.no.transitive},
\begin{itemize*}
\item
$\ns P=\{\minus 2,\minus 1,0,1,2\}$ and 
\item
we let $\opens$ be generated by $\{i,i\plus 1\}$ for $\minus 2\leq i\leq 1$ (unordered pairs of adjacent numbers).
\end{itemize*}
Write $A=\{\minus 1,0\}$ and $B=\{0,1\}$ and let $\afilter$ be the up-closure of $\{A,B\}$. %
Note that $A$ and $B$ are not transitive (i.e. not topen).
The reader can check that $\afilter^\ast=\afilter$ (e.g. $\{1,2\}\notin\afilter^\ast$ because $\{1,2\}\notbetween \{\minus 1,0\}\in\afilter$), but $\afilter$ contains no transitive element.
\qedhere\end{enumerate}
\end{proof}

\jamiesubsection{Semiframe characterisation of community}

\begin{rmrk}
We saw the notion of $\community(p)$ the \emph{community} of a point in Definition~\ref{defn.tn}(\ref{item.tn}).
In this Subsection we construct an analogue to it in semiframes.
We will give two characterisations: one in Definition~\ref{defn.abstract.community}, and another in Proposition~\ref{prop.framecommunity.universal}.
\end{rmrk}

We will mostly be interested in Definition~\ref{defn.cast} when $\afilter$ is a semifilter, but the definition does not require this: 
\begin{defn}
\label{defn.cast}
Suppose $(\ns X,\cti,\ast)$ is a semiframe and $\afilter\subseteq\ns X$ and $x\in\ns X$.
Then define $\cclo{\afilter}\in\ns X$, $\cast{\afilter}\in\ns X$, and $\cast{x}\in\ns X$ by
$$
\cclo{\afilter} = \bigvee \{y\in\ns X \mid y\notin \afilter\},
\qquad 
\cast{\afilter} = \cclo{(\afilter^\ast)}, 
\qquad\text{and}\qquad
\cast{x} = \cclo{(x^\ast)} .
$$ 
\end{defn}

\begin{rmrk}
\label{rmrk.cast.simpler}
We unpack the definitions of $\cast{\afilter}$ and $\cast{x}$: 
$$
\begin{array}{r@{\ }l@{\qquad}l}
\cast{\afilter} 
=& \cclo{(\afilter^\ast)} 
&\text{Definition~\ref{defn.cast}}
\\
=& \bigvee \{y\in\ns X \mid y\notin \afilter^\ast\} 
&\text{Definition~\ref{defn.cast}}
\\
=& \bigvee \{y\in\ns X \mid \neg(y\ast \afilter)\} 
&\text{Definition~\ref{defn.X.ast}}
\\[2ex]
\cast{x} 
=& \cclo{(x^\ast)}
&\text{Definition~\ref{defn.cast}}
\\
=& \bigvee \{y\in\ns X \mid y\notin x^\ast \} 
&\text{Definition~\ref{defn.cast}}
\\
=& \bigvee \{y\in\ns X \mid \neg(y\ast x) \} .
&\text{Definition~\ref{defn.x.ast}}
\end{array}
$$ 
\end{rmrk}

Lemma~\ref{lemm.cast.comp} will be useful, and gives some intuition for $\cclo{(\text{-})}$ and $\cast{(\text{-})}$ by unpacking their concrete meaning in the special case of a semiframe of open sets of a semitopology:
\begin{lemm}
\label{lemm.cast.comp}
Suppose $(\ns P,\opens)$ is a semitopology and $p\in\ns P$ and $O\in\opens$.
Then:
\begin{enumerate*}
\item\label{item.cast.comp.cclo}
$\cclo{\nbhd(p)}=\ns P\setminus\closure{p}$.
\item\label{item.cast.comp.nbhd}
$\cast{\nbhd(p)}=\ns P\setminus\intertwined{p}$.
\item\label{item.cast.comp.O}
$\cast{O}=\ns P\setminus\closure{O}=\interior(\ns P\setminus O)$.
\end{enumerate*}
\end{lemm}
\begin{proof}
We consider each part in turn:
\begin{enumerate}
\item
It is a fact of Definition~\ref{defn.closure} that $\ns P\setminus\closure{p}=\bigcup\{O'\in\opens \mid p\notin O'\}$.
By Proposition~\ref{prop.nbhd.iff}(\ref{item.nbhd.iff}) $p\notin O'$ if and only if $O'\notin\nbhd(p)$.
\item
It is a fact of Definition~\ref{defn.intertwined.points}, which is spelled out in Lemma~\ref{lemm.char.not.intertwined}(\ref{item.intertwined.open.avoid}), that $\ns P\setminus\intertwined{p}=\bigcup\{O'\in\opens \mid \Exists{O{\in}\opens} p\in O \land O'\notbetween O\}$.
By Lemma~\ref{lemm.nbhd.ast.char} $\Exists{O{\in}\opens} p\in O\land O'\notbetween O$ precisely when $O'\notin\nbhd(p)^\ast$. 
\item
By Definitions~\ref{defn.cast} and~\ref{defn.X.ast} we have
$$
\cclo{O}=\cast{(O^\ast)}=\bigcup\{O'{\in}\opens \mid O'\notin O^\ast\} =\bigcup\{O'{\in}\opens \mid O'\notbetween O\} .
$$
The result then follows by routine reasoning on closures (Definition~\ref{defn.closure}).
\qedhere\end{enumerate}
\end{proof}

\begin{defn}
\label{defn.abstract.community}
Suppose $(\ns X,\cti,\ast)$ is a semiframe and $\afilter\subseteq\ns X$.
Then define $\framecommunity(\afilter)\in\ns X$ the \deffont[abstract community (of a set:\ $\framecommunity(\afilter)$)]{abstract community}\index{$\framecommunity(\afilter)$ (abstract community of a set)} of $\afilter$ by
$$
\framecommunity(\afilter)=\cast{(\cast{\afilter})} \in \ns X.
$$
(For a more direct characterisation, see Proposition~\ref{prop.framecommunity.universal}.)
\end{defn}

\begin{prop}
\label{prop.framecommunity.nbhd.community}
Suppose $(\ns P,\opens)$ is a semitopology and $p\in\ns P$.
Then 
$$
\framecommunity(\nbhd(p))=\community(p) .
$$
In words: the abstract community of the abstract point $\nbhd(p)$ in $(\opens,\subseteq,\between)$, is identical to the community of $p$.
\end{prop}
\begin{proof}
We reason as follows:
$$
\begin{array}[b]{r@{\ }l@{\qquad}l}
\framecommunity(\nbhd(p))
=&
\cast{(\cast{\nbhd(p)})}
&\text{Definition~\ref{defn.abstract.community}}
\\
=&
\cast{(\ns P\setminus\intertwined{p})}
&\text{Lemma~\ref{lemm.cast.comp}(\ref{item.cast.comp.nbhd})}
\\
=&
\interior(\ns P\setminus (\ns P\setminus\intertwined{p}))
&\text{Lemma~\ref{lemm.cast.comp}(\ref{item.cast.comp.O})}
\\
=&
\interior(\intertwined{p})
&\text{Fact of sets}
\\
=&
\community(p)
&\text{Definition~\ref{defn.tn}(\ref{item.tn})}
\end{array}
\qedhere$$
\end{proof}

We can also give a more direct characterisation of the abstract community from Definition~\ref{defn.abstract.community}:
\begin{prop}
\label{prop.framecommunity.universal}
Suppose $(\ns X,\cti,\ast)$ is a semiframe and $\afilter\subseteq\ns X$.
Then 
$$
\framecommunity(\afilter)=\bigvee\{x\in\ns X \mid x^\ast\subseteq\afilter^\ast\} ,
$$
and $\framecommunity(\afilter)$ is the greatest element in $\ns X$ such that $\framecommunity(\afilter)^\ast\subseteq\afilter^\ast$.
\end{prop}
\begin{proof}
We follow the definitions: 
$$
\begin{array}{r@{\ }l@{\qquad}l}
\cast{(\cast{\afilter})}
=&
\bigvee\{x\in\ns X\mid \neg (x\ast \cast{\afilter})\} 
&\text{Remark~\ref{rmrk.cast.simpler}}
\\
=&
\bigvee\{x\in\ns X\mid \neg (x\ast \bigvee\{y \mid \neg (y\ast\afilter)\})\}
&\text{Remark~\ref{rmrk.cast.simpler}}
\\
=&
\bigvee\{x\in\ns X\mid \neg \Exists{y{\in}\ns X} (x\ast y \land \neg (y\ast\afilter))\}
&\text{Definition~\ref{defn.compatibility.relation}(\ref{item.compatible.distributive})}
\\
=&
\bigvee\{x\in\ns X \mid \Forall{y{\in}\ns X} y\ast x \limp y\ast\afilter\}
&\text{Fact of logic}
\\
=&
\bigvee\{x\in\ns X \mid x^\ast\subseteq \afilter^\ast \}
&\text{Definitions~\ref{defn.x.ast} \&~\ref{defn.X.ast}}
\end{array}
$$
To see that $\framecommunity(\afilter)$ is the greatest element such that $\framecommunity(\afilter)^\ast\subseteq\afilter^\ast$, we note from Lemma~\ref{lemm.bigvee.ast.union} that
$$
\framecommunity(\afilter)^\ast = \bigcup \{x^\ast \mid x{\in}\ns X,\ x^\ast\subseteq\afilter^\ast\} .
\qedhere$$
\end{proof}

\jamiesubsection{Semiframe characterisation of regularity}
\label{subsect.semiframe.regularity}

We now have enough to generalise the notions of quasiregularity, weak regularity, and regularity from semitopologies (Definition~\ref{defn.tn} parts~\ref{item.quasiregular.point}, \ref{item.weakly.regular.point}, and~\ref{item.regular.point}) to semiframes:
\begin{defn}
\label{defn.afilter.regular}
Suppose $(\ns X,\cti,\ast)$ is a semiframe and $\afilter\subseteq\ns X$ is a semifilter.
\begin{enumerate}
\item\label{item.afilter.quasiregular}
Call $\afilter$ \deffont[quasiregular semifilter]{quasiregular} when $\framecommunity(\afilter)\neq\tbot_{\ns X}$.

Thus, there exists some $x\in\ns X$ such that $x^\ast\subseteq \afilter^\ast$.
\item\label{item.afilter.weakly.regular}
Call $\afilter$ \deffont[weakly regular semifilter]{weakly regular} when $\framecommunity(\afilter)\in\afilter$.
\item\label{item.afilter.regular}
Call $\afilter$ \deffont[regular semifilter]{regular} when $\framecommunity(\afilter)\in\afilter$ and $\framecommunity(\afilter)$ is transitive.
\end{enumerate}
\end{defn}

Lemma~\ref{lemm.afilter.regular.imp} does for semiframes what Lemma~\ref{lemm.wr.r} does for semitopologies:
\begin{lemm}
\label{lemm.afilter.regular.imp}
Suppose $(\ns X,\cti,\ast)$ is a semiframe and $\afilter\subseteq\ns X$ is a semifilter.
Then:
\begin{enumerate*}
\item
If $\afilter$ is regular then it is weakly regular.
\item
If $\afilter$ is weakly regular then it is quasiregular.
\end{enumerate*}
(The converse implications need not hold, and it is possible for $\afilter$ to not be quasiregular:
it is convenient to defer the proofs to Corollary~\ref{corr.quasiregular.no.converse}.)
\end{lemm}
\begin{proof}
The proofs are easy:
If $\framecommunity(\afilter)\in\afilter$ and $\framecommunity(\afilter)$ is transitive, then certainly $\framecommunity(\afilter)\in\afilter$.
If $\framecommunity(\afilter)\in\afilter$ then by Lemma~\ref{lemm.P.top}(\ref{item.P.no.bot}) $\framecommunity(\afilter)\neq\tbot_{\ns X}$.
\end{proof}

\begin{lemm}
\label{lemm.r=wr+sc}
Suppose $(\ns X,\cti,\ast)$ is a semiframe and $\afilter\subseteq\ns X$ is a semifilter.
Then:
\begin{enumerate*}
\item\label{item.wr.kt.1}
If $\afilter$ is quasiregular and strongly compatible then $\framecommunity(\afilter)$ is transitive.
\item\label{item.wr.kt.1b}
The converse implication need not hold: it is possible for $\afilter$ to be quasiregular and $\framecommunity(\afilter)$ to be transitive, yet $\afilter$ is not strongly compatible.
\item\label{item.wr.kt.2}
If $\afilter$ is weakly regular and $\framecommunity(\afilter)$ is transitive then $\afilter$ is strongly compatible.
\item\label{item.wr.kt.iff}
If $\afilter$ is weakly regular, then $\framecommunity(\afilter)$ is transitive if and only if $\afilter$ is strongly compatible.
\end{enumerate*}
\end{lemm}
\begin{proof}
We consider each part in turn:
\begin{enumerate}
\item
Suppose $\afilter$ is quasiregular and strongly compatible.
 
By quasiregularity $\tbot_{\ns X}\neq\framecommunity(\afilter)$.
By Proposition~\ref{prop.framecommunity.universal} $\framecommunity(\afilter)^\ast\subseteq\afilter^\ast$.
By strong compatibility $\afilter^\ast$ is a semifilter and so in particular $\afilter^\ast$ is compatible.
It follows from Proposition~\ref{prop.trans.cps}(\ref{item.cps.transitive}\&\ref{item.cps.compatible}) that $\framecommunity(\afilter)$ is transitive, as required.
\item
It suffices to provide a counterexample.
Let $(\mathbb R,\opens)$ be the real numbers with their usual topology, and let $(\mathbb R,\opens')$ be the topology generated by $\opens\cup\{\{0\}\}$ --- in words: we add $\{0\}$ as an open set.
 
Let $\afilter$ be the semifilter of all $\opens$-open neighbourhoods of $0$.
$\afilter^\ast$ is the set of $\opens'$-open sets that intersect every $\opens$-open neighbourhood of $0$.
This is not compatible, because it contains $\openinterval{0,}$ (the set of numbers strictly less than $0$) and $\openinterval{,0}$ (the set of numbers strictly greater than $0$), and these do not intersect.
Using Proposition~\ref{prop.framecommunity.universal}, we calculate that $\framecommunity(\afilter)=\{0\}$; this is transitive because it is a singleton set.

So $\afilter$ is quasiregular, $\framecommunity(\afilter)$ is transitive, yet $\afilter$ is not strongly compatible.
\item
Suppose $\framecommunity(\afilter)$ is transitive and suppose $\afilter$ is weakly regular, so $\framecommunity(\afilter)\in\afilter$.
By Proposition~\ref{prop.x.transitive.afilter.ast} $\afilter$ is strongly compatible.
\item
From parts~\ref{item.wr.kt.1} and~\ref{item.wr.kt.2} of this result, noting from Lemma~\ref{lemm.afilter.regular.imp} that if $\afilter$ is weakly regular then it is quasiregular.
\qedhere\end{enumerate}
\end{proof}

\begin{thrm}
\label{thrm.r=wr+sc}
Suppose $(\ns X,\cti,\ast)$ is a semiframe and $\afilter\subseteq\ns X$ is a semifilter.
Then $\afilter$ is regular if and only if $\afilter$ is weakly regular and strongly compatible.
We can write this succinctly as follows:
\begin{quoting}
Regular = weakly regular + strongly compatible.
\end{quoting}
(Compare this slogan with the version for semitopologies in Theorem~\ref{thrm.r=wr+uc}.)
\end{thrm}
\begin{proof}
Suppose $\afilter$ is weakly regular and strongly compatible.
By Lemma~\ref{lemm.r=wr+sc}(\ref{item.wr.kt.iff}) $\framecommunity(\afilter)$ is transitive, and by Definition~\ref{defn.afilter.regular}(\ref{item.afilter.regular}) $\afilter$ is regular.

For the converse implication we just reverse the reasoning above.
\end{proof}

\begin{rmrk}
\label{rmrk.subtly.different}
In Theorem~\ref{thrm.regular=qr+sc} we characterised regularity of points in terms of quasiregularity and being hypertransitive.
In view of Lemma~\ref{lemm.ht.sc.eq} we might expect Theorem~\ref{thrm.r=wr+sc} to read `regular = quasiregular + strongly compatible'.
But this is false, as per the discussion in Remark~\ref{rmrk.semiframes.caution} and the counterexample in Lemma~\ref{lemm.r=wr+sc}(\ref{item.wr.kt.1b}).
Thus, the semiframes results are subtly different from those governing point-set semitopologies. 
\end{rmrk}

\jamiesubsection{Semiframe characterisation of (quasi/weak)regularity}
	
The direct translation in Definition~\ref{defn.afilter.regular} of parts~\ref{item.quasiregular.point}, \ref{item.weakly.regular.point}, and~\ref{item.regular.point} of Definition~\ref{defn.tn}, along with the machinery we have now built, makes Lemma~\ref{lemm.match.up} easy to prove:
\begin{lemm}
\label{lemm.match.up}
Suppose $(\ns P,\opens)$ is a semitopology and $p\in\ns P$.
Recall from Definition~\ref{defn.nbhd} and Proposition~\ref{prop.nbhd.iff}(\ref{item.nbhd.point})
that $\nbhd(p)=\{O\in\opens \mid p\in O\}$ is a (completely prime) semifilter.
Then:
\begin{enumerate*}
\item
$p$ is quasiregular in the sense of Definition~\ref{defn.tn}(\ref{item.quasiregular.point})
if and only if 
$\nbhd(p)$ is quasiregular in the sense of Definition~\ref{defn.afilter.regular}(\ref{item.afilter.quasiregular}).
\item
$p$ is weakly regular in the sense of Definition~\ref{defn.tn}(\ref{item.weakly.regular.point}) if and only if $\nbhd(p)$ is weakly regular in the sense of Definition~\ref{defn.afilter.regular}(\ref{item.afilter.weakly.regular}).
\item
$p$ is regular in the sense of Definition~\ref{defn.tn}(\ref{item.regular.point}) if and only if $\nbhd(p)$ is regular in the sense of Definition~\ref{defn.afilter.regular}(\ref{item.afilter.regular}).
\end{enumerate*}
\end{lemm}
\begin{proof}
We consider each part in turn:
\begin{enumerate}
\item
Suppose $p$ is quasiregular.
By Definition~\ref{defn.tn}(\ref{item.quasiregular.point}) $\community(p)\neq\varnothing$.
By Proposition~\ref{prop.framecommunity.nbhd.community} $\framecommunity(\nbhd(p))\neq\varnothing=\tbot_{\opens}$.
By Definition~\ref{defn.afilter.regular}(\ref{item.afilter.quasiregular}) $\nbhd(p)$ is quasiregular.

The reverse implication follows just reversing the reasoning above.
\item
Suppose $p$ is weakly regular.
By Definition~\ref{defn.tn}(\ref{item.weakly.regular.point}) $p\in\community(p)$.
By Definition~\ref{defn.nbhd} $\community(p)\in\nbhd(p)$.
By Proposition~\ref{prop.framecommunity.nbhd.community} $\framecommunity(\nbhd(p))\in\nbhd(p)$ as required.

The reverse implication follows just reversing the reasoning above.
\item
Suppose $p$ is regular.
By Definition~\ref{defn.tn}(\ref{item.regular.point}) $p\in\community(p)\in\topens$.
By Definition~\ref{defn.nbhd} and Proposition~\ref{prop.framecommunity.nbhd.community} $\framecommunity(\nbhd(p))\in\nbhd(p)$.
By Proposition~\ref{prop.framecommunity.nbhd.community} and Lemma~\ref{lemm.topen.transitive} $\framecommunity(\nbhd(p))$ is transitive. 

The reverse implication follows just reversing the reasoning above.
\qedhere\end{enumerate}
\end{proof}

\begin{prop}
\label{prop.regular.match.up}
Suppose $(\ns P,\opens)$ is a semitopology and $p\in\ns P$.
Then 
\begin{itemize*}
\item
$p$ is quasiregular / weakly regular / regular in $(\ns P,\opens)$ in the sense of Definition~\ref{defn.tn} 
\\
if and only if 
\item
$\nbhd(p)$ is quasiregular / weakly regular / regular in $\tf{Soberify}(\ns P,\opens)$ in the sense of Definition~\ref{defn.afilter.regular}.
\end{itemize*}
\end{prop} 
\begin{proof}
We consider just the case of regularity; quasiregularity and weak regularity are no different.

Suppose $p$ is regular.
By Definition~\ref{defn.tn}(\ref{item.regular.point}) $p\in \community(p)\in\topens$.
It follows from Lemma~\ref{lemm.topen.transitive} that $\community(p)$ is transitive in $(\opens,\subseteq,\between)$, and from Proposition~\ref{prop.nbhd.iff}(\ref{item.nbhd.iff}) that $\community(p)\in\nbhd(p)$.
It follows from Proposition~\ref{prop.framecommunity.nbhd.community} that $\nbhd(p)$ is regular in the sense of Definition~\ref{defn.afilter.regular}(\ref{item.afilter.regular}).
\end{proof}

\begin{corr}
\label{corr.quasiregular.no.converse}
Suppose $(\ns X,\cti,\ast)$ is a semiframe and $\afilter\subseteq\ns X$ is a semifilter.
Then the converse implications in Lemma~\ref{lemm.afilter.regular.imp} need not hold: $\afilter$ may be quasiregular but not regular, and it may be weakly regular but not regular, and it may not even be quasiregular.
\end{corr}
\begin{proof}
It suffices to provide counterexamples.
We easily obtain these by using Proposition~\ref{prop.regular.match.up} to consider $\nbhd(p)$ for $p\in\ns P$ as used in Lemma~\ref{lemm.wr.r}.
\end{proof}

\jamiesubsection{Characterisation of being intertwined}

This Subsection continues Remark~\ref{rmrk.promise.ast.int.char}.

The notion of points being intertwined from Definition~\ref{defn.intertwined.points}(\ref{item.p.intertwinedwith.p'}) generalises in semiframes to the notion of semifilters being compatible:
\begin{lemm}
\label{lemm.intertwined.sober}
Suppose $(\ns P,\opens)$ is a semitopology and $p,p'\in\ns P$.
Then 
$$
p\intertwinedwith p'
\quad\liff\quad
\nbhd(p)\ast\nbhd(p')
\quad\liff\quad
\nbhd(p)\intertwinedwith\nbhd(p') .
$$
For clarity and precision we unpack this.
The following are equivalent: 
\begin{enumerate*}
\item\label{item.intertwined.sober.pp'.intertwinedwith}
$p\intertwinedwith p'$ in the semitopology $(\ns P,\opens)$ (Definition~\ref{defn.intertwined.points}(\ref{item.p.intertwinedwith.p'})).

In words: the point $p$ is intertwined with the point $p'$.
\item\label{item.intertwined.sober.nbhd.ast}
$\nbhd(p)\ast\nbhd(p')$ in the semiframe $(\opens,\subseteq,\between)$ (Notation~\ref{nttn.X.ast.Y}(\ref{item.X.ast.Y})).

In words: the abstract point $\nbhd(p)$ is compatible with the abstract point $\nbhd(p')$.
\item\label{item.intertwined.sober.3}
$\nbhd(p)\intertwinedwith\nbhd(p')$ in the semitopology $\tf{St}(\opens,\subseteq,\between)$ (Definition~\ref{defn.intertwined.points}(\ref{item.p.intertwinedwith.p'})). 

In words: the point $\nbhd(p)$ is intertwined with the point $\nbhd(p')$.
\end{enumerate*}
\end{lemm}
\begin{proof}
We unpack definitions:
\begin{itemize*}
\item
By Definition~\ref{defn.intertwined.points}(\ref{item.p.intertwinedwith.p'}) $p\intertwinedwith p'$ when for every pair of open neighbourhoods $p\in O$ and $p'\in O'$ we have $O\between O'$.
\item
By Notation~\ref{nttn.X.ast.Y}(\ref{item.X.ast.Y}) $\nbhd(p)\ast\nbhd(p')$ when for every $O\in\nbhd(p)$ and $O'\in\nbhd(p')$ we have $O\ast O'$.

By Proposition~\ref{prop.nbhd.iff}(\ref{item.nbhd.iff}) we can simplify this to: $p\in O$ and $p'\in O'$ implies $O\ast O'$.
\item
By Definition~\ref{defn.intertwined.points}(\ref{item.p.intertwinedwith.p'}) and Theorem~\ref{thrm.nbhd.morphism}, $\nbhd(p)\intertwinedwith \nbhd(p')$ when:
for every pair of open neighbourhoods $\nbhd(p)\in \f{Op}(O)$ and $\nbhd(p')\in \f{Op}(O')$ we have $\f{Op}(O)\between \f{Op}(O')$.

By Proposition~\ref{prop.nbhd.iff}(\ref{item.nbhd.iff}) we can simplify this to: $p\in O$ and $p'\in O'$ implies $\f{Op}(O)\between \f{Op}(O')$.

By Proposition~\ref{prop.semiframe.to.Op}(\ref{item.semiframe.to.Op.between}) we can simplify this further to: $p\in O$ and $p'\in O'$ implies $O\ast O'$.
\end{itemize*}
But by definition, the compatibility relation $\ast$ of $(\opens,\subseteq,\between)$ is $\between$, so $O\ast O'$ and $O\between O'$ are the same assertion.
The equivalences follow. 
\end{proof}

The property of being intertwined is preserved and reflected when we use $\nbhd$ to map to the soberified space:
\begin{corr}
\label{corr.intertwined.sober}
Suppose $(\ns P,\opens)$ is a semitopology and $p,p'\in\ns P$.
Then $p\intertwinedwith p'$ in $(\ns P,\opens)$ if and only if $\nbhd(p)\intertwinedwith\nbhd(p')$ in $\tf{Soberify}(\ns P,\opens)$. 
\end{corr}
\begin{proof}
This just reiterates the equivalence of parts~\ref{item.intertwined.sober.pp'.intertwinedwith} and~\ref{item.intertwined.sober.3} in Lemma~\ref{lemm.intertwined.sober}.
\end{proof}

\begin{prop}
\label{prop.conflicted.sober}
Suppose $(\ns P,\opens)$ is a semitopology.
Then:
\begin{enumerate*}
\item\label{item.conflicted.sober.1}
It may be that $(\ns P,\opens)$ is unconflicted (meaning that it contains no conflicted points), but the semitopology $\tf{Soberify}(\ns P,\opens)$ contains a conflicted point.
\item\label{item.conflicted.sober.2}
It may further be that $(\ns P,\opens)$ is unconflicted and $p\in\ns P$ is such that $\nbhd(p)$ is conflicted in the semitopology $\tf{Soberify}(\ns P,\opens)$.
\end{enumerate*}
We can summarise the two assertions above as follows:
\begin{enumerate*}
\item
Soberifying a space might introduce a conflicted point, even if none was originally present.
\item
Soberifying a space can make a point that was unconflicted, into a point that is conflicted.\footnote{If we stretch the English language, we might say that soberifying a space can conflictify one of its points.}
\end{enumerate*}
\end{prop}
\begin{proof}
It suffices to provide counterexamples.
\begin{enumerate}
\item
Consider the right-hand semitopology in Figure~\ref{fig.012-triangle}; this is unconflicted because every point is intertwined only with itself.
The soberification of this space is illustrated in the right-hand semitopology in Figure~\ref{fig.012-triangle-sober}.
Each of the extra points is intertwined with the two numbered points next to it; e.g. the extra point in the open set $A$ --- write it $\bullet_A$ (in-between $3$ and $0$) --- is intertwined with $0$ and $3$; so $3\intertwinedwith \bullet_A\intertwinedwith 0$.
However, the reader can check that $3\notintertwinedwith 0$.
Thus, $\bullet_A$ is conflicted.
\item
We define $(\ns P,\opens)$ by:
\begin{itemize}
\item
$\ns P = \openinterval{\minus 1,1}$ (real numbers between $\minus 1$ and $1$ exclusive).
\item
$\opens$ is generated by:
\begin{itemize*}
\item 
All open intervals that do not contain $0$; so this is open intervals $\openinterval{r_1,r_2}$ where $\minus 1\leq r_1<r_2\leq 0$ or $0\leq r_1<r_2\leq 1$.
\item
All of the open intervals $\openinterval{\minus 1/n,1/n}$, for $n\geq 2$. 
\end{itemize*}
\end{itemize}
The reader can check that:
\begin{itemize*}
\item
Points in this semitopology are intertwined only with themselves.
\item
The soberification includes four additional points, corresponding to completely prime semifilters $\minus 1/0$ generated by $\{\openinterval{\minus 1/n,0} \mid n\geq 2\}$ and $\plus 1/0$ generated by $\{\openinterval{0,1/n} \mid n\geq 2\}$, and to the endpoints $\minus 1$ and $i\plus 1$.
\item
$\minus 1/0$ and $\plus 1/0$ are intertwined with $0$, but are not intertwined with one another.
\end{itemize*}
Thus, $0$ is conflicted in $\tf{Soberify}(\ns P,\opens)$ but not in $(\ns P,\opens)$.
\qedhere\end{enumerate}
\end{proof}

\begin{rmrk}
Proposition~\ref{prop.conflicted.sober} may seem surprising in view of Corollary~\ref{corr.intertwined.sober}, but the key observation is that the soberified space may add points to the original space.
These points can add conflicting behaviour that is `hidden' in the completely prime semifilters of the original space.

Thus, Proposition~\ref{prop.conflicted.sober} shows that the property of `being unconflicted' \emph{cannot} be characterised purely in terms of the semiframe of open sets --- if it could be, then soberification would make no difference, by Theorem~\ref{thrm.nbhd.morphism}(\ref{item.nbhd.morphism.is.iso}).

There is nothing wrong with that, except that we are interested in well-behavedness conditions on semiframes.
We can now look for some other condition --- but one having to do purely with open sets --- that might play a similar role in the theory of (weak/quasi)regularity of semiframes, as being unconflicted does in theory of (weak/quasi)regularity of semitopologies.

We already saw a candidate for this in Theorem~\ref{thrm.r=wr+sc}: \emph{strong compatibility}.
We examine this next.
\end{rmrk}

\jamiesubsection{Strong compatibility in semitopologies}

\begin{rmrk}
Note that:
\begin{enumerate*}
\item
Theorem~\ref{thrm.r=wr+sc} characterises `regular' for semiframes as `weakly regular + strongly compatible'. 
\item
Theorem~\ref{thrm.r=wr+uc} characterises `regular' for semitopologies as `weakly regular + unconflicted'.
\end{enumerate*}
We know from results like Lemma~\ref{lemm.match.up} and Corollary~\ref{corr.intertwined.sober} that there are accurate correspondences between notions of regularity in semiframes and semitopologies.
This is by design, e.g. in Definition~\ref{defn.afilter.regular}; we designed the semiframe definitions so that semiframe regularity and semitopological regularity would match up closely.

Yet there are differences too, since Theorem~\ref{thrm.r=wr+sc} uses strong compatibility, and Theorem~\ref{thrm.r=wr+uc} uses being unconflicted.
What is the difference here, and why does it arise?

One answer is given by Proposition~\ref{prop.conflicted.sober}, which illustrates that the condition of `unconflicted' (which comes from semitopologies) does not sit comfortably with the `pointless' semiframe definitions. 
This raises the question of how strong compatibility (which comes from semiframes) translates into the context of semitopologies; and how this relates to being (un)conflicted?

We look into this now; see Remark~\ref{rmrk.summary.of.sc} for a summary.
\end{rmrk}

We can translate the notion of \emph{strongly compatible filter} (Definition~\ref{defn.F.strongly.compatible}(\ref{item.strongly.compatible.filter})) to semitopologies in the natural way, just applying it to the neighbourhood semifilter $\nbhd(p)$ of a point (Definition~\ref{defn.nbhd}):
\begin{defn}
\label{defn.sc.point}
Suppose $(\ns P,\opens)$ is a semitopology.
Then call $p\in\ns P$ \deffont{strongly compatible} when the (by Example~\ref{xmpl.abstract.point}(\ref{item.xmpl.nbhd.abstract.point})) abstract point $\nbhd(p)$ is strongly compatible (Definition~\ref{defn.F.strongly.compatible}) as a semifilter in $(\opens,\subseteq,\between)$.
\end{defn}

We unpack what Definition~\ref{defn.sc.point} means concretely:
\begin{lemm}
\label{lemm.what.sc.point.means}
Suppose $(\ns P,\opens)$ is a semitopology and $p\in\ns P$.
Then the following are equivalent:
\begin{enumerate*}
\item\label{item.what.sc.point.means.1}
$p$ is hypertransitive (Definition~\ref{defn.sc}).
\item\label{item.what.sc.point.means.2}
$\nbhd(p)$ is strongly compatible.
\item\label{item.what.sc.point.means.3}
$\nbhd(p)^\ast$ is compatible.
\item\label{item.what.sc.point.means.4}
For every $O',O''\in\opens$, if $O'\ast \nbhd(p)\ast O''$ then $O'\between O''$.
\end{enumerate*}
(Above, $O'\ast\nbhd(p)$ follows Notation~\ref{nttn.X.ast.Y}(\ref{item.x.ast.Y}) and means that $O'\between O$ for every $p\in O\in\opens$, and similarly for $\nbhd(p)\ast O''$.)
\end{lemm}
\begin{proof}
Equivalence of parts~\ref{item.what.sc.point.means.1} and~\ref{item.what.sc.point.means.2} is just Lemma~\ref{lemm.ht.sc.eq}.
Equivalence of parts~\ref{item.what.sc.point.means.2} and~\ref{item.what.sc.point.means.3} is Definition~\ref{defn.F.strongly.compatible}(\ref{item.strongly.compatible.filter}).
For the equivalence of parts~\ref{item.what.sc.point.means.3} and~\ref{item.what.sc.point.means.4}, we just unpack what it means for $\nbhd(p)^\ast$ to be compatible (see Remark~\ref{rmrk.what.does.strongly.compatible.mean}).
\end{proof}

`$p\in\ns P$ is strongly compatible' is a strictly stronger condition than `$p\in\ns P$ is unconflicted':
\begin{lemm}
\label{lemm.sc.stronger.than.unconflicted}
Suppose $(\ns P,\opens)$ is a semitopology and $p\in\ns P$. 
Then:
\begin{enumerate*}
\item\label{item.sc.uc}
If $p$ is strongly compatible then it is unconflicted.
\item\label{item.sc.uc.contra}
If $p$ is conflicted then: $p$ is not strongly compatible, $\nbhd(p)$ is not strongly compatible, and $\nbhd(p)^\ast$ is not compatible.
\item\label{item.sc.stronger.than.uc.cx}
The reverse implication need not hold, even if $(\ns P,\opens)$ is sober:%
\footnote{\dots meaning that every abstract point in $(\opens,\subseteq,\between)$ is the neighbourhood semifilter of a unique concrete point in $\ns P$.} 
it is possible for $p$ to be unconflicted but not strongly compatible.
\end{enumerate*}
\end{lemm}
\begin{proof}
We consider each part in turn:
\begin{enumerate}
\item
Suppose $p$ is strongly compatible and suppose $p'\intertwinedwith p\intertwinedwith p''$; we must show that $p'\intertwinedwith p''$.
Consider open neighbourhoods $p'\in O'$ and $p''\in O''$.
By assumption $p'\intertwinedwith p$ and so by 
Lemma~\ref{lemm.intertwined.sober}(\ref{item.intertwined.sober.pp'.intertwinedwith}\&\ref{item.intertwined.sober.nbhd.ast}) $\nbhd(p')\ast\nbhd(p)$.
Since $O'\in\nbhd(p')$, it follows that $O'\ast\nbhd(p)$, and similarly it follows that $\nbhd(p)\ast O''$.
Then by strong compatibility, $O'\between O''$ as required. 
\item
We take the contrapositive of part~\ref{item.sc.uc} of this result, and use Lemma~\ref{lemm.what.sc.point.means}.
\item
It suffices to provide a counterexample.
Consider the bottom right semitopology in Figure~\ref{fig.012}, and take $p=\ast$ and $O'=\{1\}$ and $O''=\{0,2\}$.
Note that:
\begin{itemize*}
\item
$\ast$ is unconflicted, since it is intertwined only with itself and $1$.
\item
$O'$ and $O'$ intersect every open neighbourhood of $\ast$, but $O'\notbetween O''$, so $\ast$ is not strongly compatible.
\end{itemize*} 
This space is sober: the only completely prime filters are the neighbourhood semifilters of $\ast$, $0$, $1$, and $2$.
\qedhere\end{enumerate}
\end{proof}

\begin{xmpl}
Continuing Lemma~\ref{lemm.sc.stronger.than.unconflicted}(\ref{item.sc.stronger.than.uc.cx}), it is possible for a point to be strongly compatible (Definition~\ref{defn.sc.point}) but not regular, or even quasiregular (Definition~\ref{defn.tn}(\ref{item.regular.point}, \ref{item.quasiregular.point})).
Consider the right-hand semitopology illustrated in Figure~\ref{fig.012-triangle} and take $p=0$.
The reader can check that $p$ is strongly compatible, but it is not quasiregular (i.e. $\community(p)=\varnothing$) and thus also not regular. 
\end{xmpl}

Lemma~\ref{lemm.sc.point.iff.soberified} shows that the situation outlined in Proposition~\ref{prop.conflicted.sober}(\ref{item.conflicted.sober.2}) cannot arise if we work with a strongly compatible point instead of an unconflicted one \dots
\begin{lemm}
\label{lemm.sc.point.iff.soberified}
Suppose $(\ns P,\opens)$ is a semitopology and $p,p'\in\ns P$.
Then the following are equivalent:
\begin{enumerate*}
\item
$p$ is hypertransitive in $(\ns P,\opens)$.
\item
$\nbhd(p)$ is hypertransitive in $\tf{Soberify}(\ns P,\opens)$ (Notation~\ref{nttn.soberify}).
\end{enumerate*}
\end{lemm}
\begin{proof}
Note that from Lemma~\ref{lemm.what.sc.point.means}, $p$ is hypertransitive in $(\ns P,\opens)$ when 
$$
(\Forall{O{\in}\opens} p\in O \limp O'\between O\between O'') \quad\text{implies}\quad  O'\between O''
$$
for every $O',O''\in\opens$. 
Also, from Definition~\ref{defn.st.g}(\ref{item.st.op}) and Lemma~\ref{lemm.what.sc.point.means}, $\nbhd(p)$ is hypertransitive in $\tf{Soberify}(\ns P,\opens)$ when 
\begin{multline*}
(\Forall{O{\in}\opens} \nbhd(p)\in\f{Op}(O) \limp \f{Op}(O')\between \f{Op}(O)\between \f{Op}(O''))
\\
\text{implies}\quad
\f{Op}(O')\between \f{Op}(O'') 
\end{multline*}
for every $\f{Op}(O'),\f{Op}(O'')\oldin\tf{Opens}(\tf{Soberify}(\ns P,\opens))$.

Now by Proposition~\ref{prop.nbhd.iff}(\ref{item.nbhd.iff}), $\nbhd(p)\in\f{Op}(O)$ if and only if $p\in O$, and 
by Corollary~\ref{corr.op.sub.between} $\f{Op}(O')\between\f{Op}(O)$ if and only if $O'\between O$, and $\f{Op}(O)\between\f{Op}(O'')$ if and only if $O\between O''$.
The result follows.
\end{proof}

\noindent\dots but, the situation outlined in Proposition~\ref{prop.conflicted.sober}(\ref{item.conflicted.sober.1}) \emph{can} arise, indeed we use the same counterexample:
\begin{lemm}
\label{lemm.sc.point.still.not.quite.right}
It may be that every point in $(\ns P,\opens)$ is hypertransitive, yet $\tf{Soberify}(\ns P,\opens)$ contains a point that is not hypertransitive.
\end{lemm}
\begin{proof}
The same counterexample as used in Proposition~\ref{prop.conflicted.sober}(\ref{item.conflicted.sober.1}) illustrates a space $(\ns P,\opens)$ such that every point in $(\ns P,\opens)$ is hypertransitive, but $\tf{Soberify}(\ns P,\opens)$ contains a point that is not hypertransitive.
We note that $\bullet_A$ (the extra point in-between $3$ and $0$) is not hypertransitive, because both $B$ and $D$ intersect with every open neighbourhood of $\bullet_A$, but $B$ does not intersect with $D$.
\end{proof}

The development above suggests that we define:
\begin{defn}
\label{defn.strongly.compatible.semitopology}
Call a semitopology $(\ns P,\opens)$ \deffont[strongly compatible semitopology]{strongly compatible} when $(\opens,\subseteq,\between)$ is strongly compatible in the sense of Definition~\ref{defn.F.strongly.compatible}(\ref{item.strongly.compatible.filter.space}).
\end{defn}

The proof of Proposition~\ref{prop.sc.semitop.robust} is then very easy:
\begin{prop}
\label{prop.sc.semitop.robust}
Suppose $(\ns P,\opens)$ is a semitopology.
Then the following are equivalent:
\begin{enumerate*}
\item\label{item.sc.semitop.robust.1}
$(\ns P,\opens)$ is strongly compatible in the sense of Definition~\ref{defn.strongly.compatible.semitopology}.
\item\label{item.sc.semitop.robust.2}
$\tf{Soberify}(\ns P,\opens)$ is strongly compatible in the sense of Definition~\ref{defn.strongly.compatible.semitopology}.
\item\label{item.sc.semitop.robust.3}
$\tf{Soberify}(\ns P,\opens)$ is strongly compatible in the sense of Definition~\ref{defn.F.strongly.compatible}(\ref{item.strongly.compatible.filter.space}).
\item\label{item.sc.semitop.robust.4}
$\tf{Soberify}(\ns P,\opens)$ is hypertransitive in the sense of Definition~\ref{defn.sc}.
\end{enumerate*}
\end{prop}
\begin{proof}
We unpack Definition~\ref{defn.strongly.compatible.semitopology} and note that strong compatibility of $(\ns P,\opens)$ is expressed purely as a property of its semiframe of open sets $(\opens,\subseteq,\between)$.
By Theorem~\ref{thrm.nbhd.morphism}(\ref{item.nbhd.morphism.is.iso}), the semiframe of open sets of $\tf{Soberify}(\ns P,\opens)$ is isomorphic to $(\opens,\subseteq,\between)$, via $\nbhd^\mone$.
Equivalence of parts~\ref{item.sc.semitop.robust.1} and~\ref{item.sc.semitop.robust.2} follows.

By Notation~\ref{nttn.soberify} and Remark~\ref{rmrk.reading.nbhd.morphism}(\ref{item.describe.StFr}), the points of $\tf{Soberify}(\ns P,\opens)$ are just abstract points of $(\opens,\subseteq,\between)$.
Equivalence of parts~\ref{item.sc.semitop.robust.2} and~\ref{item.sc.semitop.robust.3} follows. 

Equivalence of part~\ref{item.sc.semitop.robust.4} with the other parts follows using Lemmas~\ref{lemm.what.sc.point.means} and~\ref{lemm.sc.point.iff.soberified}.
\end{proof}

Recall from Definition~\ref{defn.tn}(\ref{item.regular.S}) that $(\ns P,\opens)$ being (weakly) regular means that every point in $(\ns P,\opens)$ is (weakly) regular.
Recall from Definition~\ref{defn.strongly.compatible.semitopology} that $(\ns P,\opens)$ being strongly compatible means that $(\opens,\subseteq,\between)=\tf{Fr}(\ns P,\opens)$ is strongly compatible in the sense of Definition~\ref{defn.F.strongly.compatible}(\ref{item.strongly.compatible.filter.space}).
We can now prove an analogue of Theorems~\ref{thrm.r=wr+uc} and~\ref{thrm.r=wr+sc}:
\begin{thrm}
\label{thrm.r=wr+sc.st}
Suppose $(\ns P,\opens)$ is a semitopology and $p\in\ns P$.
Then the following are equivalent:
\begin{enumerate*}
\item
$(\ns P,\opens)$ is regular.
\item
$(\ns P,\opens)$ is weakly regular and strongly compatible.
\end{enumerate*}
\end{thrm}
\begin{proof}
Suppose $(\ns P,\opens)$ is regular, meaning that every $p\in\ns P$ is regular.

By Theorems~\ref{thrm.r=wr+uc} and~\ref{thrm.regular=qr+sc} every $p\in\ns P$ is weakly regular and hypertransitive.
(So by Lemma~\ref{lemm.what.sc.point.means} every $\nbhd(p)$ is strongly compatible, and by Lemma~\ref{lemm.sc.point.iff.soberified} also hypertransitive.)
The definition of weak regularity for a space in Definition~\ref{defn.tn}(\ref{item.regular.S}) is pointwise, so it follows immediately that $(\ns P,\opens)$ is weakly regular.  

But, the definition of strong compatibility for a space in Definition~\ref{defn.strongly.compatible.semitopology} is on its semiframe of open sets, which may include abstract points not only of the form $\nbhd(p)$.  
It therefore does not follow immediately that $(\ns P,\opens)$ is strongly compatible; Lemma~\ref{lemm.sc.point.still.not.quite.right} contains a counterexample.

We can still prove that $(\ns P,\opens)$ is strongly compatible --- but we need to do a bit more work.

Unpacking Definition~\ref{defn.strongly.compatible.semitopology}, we must show that $(\opens,\subseteq,\between)$ is strongly compatible.
Unpacking Definition~\ref{defn.F.strongly.compatible}(\ref{item.strongly.compatible.filter.space}), we must show that every abstract point in $(\opens,\subseteq,\between)$ is strongly compatible.

So consider an abstract point $\apoint\subseteq\opens$.
By Corollary~\ref{corr.topen.partition.char} $\ns P$ has a topen partition $\mathcal T$, which means that: every $\atopen\in\mathcal T$ is topen; the elements of $\mathcal T$ are disjoint; and $\bigcup\mathcal T=\ns P$.

Now $\bigcup\mathcal T=\ns P\in\atopen$ by Definition~\ref{defn.point}(\ref{item.abstract.point}) and Lemma~\ref{lemm.P.top}(\ref{item.P.yes.top}), so by Definition~\ref{defn.point}(\ref{item.completely.prime}) there exists at least one (and in fact precisely one) $\atopen\in\mathcal T$ such that $\atopen\in\apoint$.
Now $\atopen$ is a transitive element in $\opens$, so by Proposition~\ref{prop.x.transitive.afilter.ast} $\apoint\subseteq\opens$ is strongly compatible as required.
\end{proof}

\begin{rmrk}
\label{rmrk.summary.of.sc}
We summarise what we have seen:
\begin{enumerate*}
\item
The notions of (quasi/weak)regularity match up nicely between a semitopology, and its semiframe soberification (Proposition~\ref{prop.regular.match.up}).
\item
We saw in Proposition~\ref{prop.conflicted.sober} that the notions of (un)conflicted point and unconflicted space from Definition~\ref{defn.conflicted}(\ref{item.unconflicted}) are not robust under forming soberification (Notation~\ref{nttn.soberify}).
From the point of view of a pointless methodology in semitopologies --- in which we seek to understand a semitopology $(\ns P,\opens)$ starting from its semiframe structure $(\opens,\subseteq,\between)$ --- this is a defect. 
\item
A pointwise notion of strong compatibility exists; by Lemma~\ref{lemm.what.sc.point.means} it is actually hypertransitivity from Definition~\ref{defn.sc}.  
This is preserved pointwise by soberification (Lemma~\ref{lemm.sc.point.iff.soberified}),
but soberification can still introduce \emph{extra} points, and it turns out that the property of a space being pointwise hypertransitive is still not robust under soberification because the extra points need not necessarily be hypertransitive; see Lemma~\ref{lemm.sc.point.still.not.quite.right}. 
\item\label{item.r=wr+sc.natural}
This motivates the notion of a \emph{strongly compatible} semitopology from Definition~\ref{defn.strongly.compatible.semitopology}; and then Proposition~\ref{prop.sc.semitop.robust} becomes easy.

Our larger point (no pun intended) is that the Definition and its corresponding Proposition are natural, \emph{and also} that the other design decisions are \emph{less} natural, as noted above.

Perhaps somewhat unexpectedly, `regular = weakly regular + strongly compatible' then works pointwise \emph{and} for the entire space; see Theorem~\ref{thrm.r=wr+sc.st}.
Thus Definition~\ref{defn.strongly.compatible.semitopology} has good properties and is natural from a pointless/semiframe/open sets perspective.
\end{enumerate*}
\end{rmrk}

\section{Conclusions and related and future work}
\label{sect.conclusions}

\subsection{Topology vs. semitopology}
\label{subsect.vs}

We briefly compare and contrast topology/frames and semitopology/semiframes.
This list is not exhaustive but we hope it will give a feel for how the two differ: 
\begin{enumerate}
\item
\emph{Topology:}\ 
We are typically interested in spaces with separation axioms.\footnote{The Wikipedia page on separation axioms~\cite{wiki:Separation_axiom} includes an excellent overview with over a dozen separation axioms.  No anti-separation axioms are discussed.} 

\emph{Semitopology:}\ 
We are interested in guaranteeing agreement between participants of a distributed system, and this is all about \emph{anti-separation} properties of their actionable coalitions.
Semifilters have a compatibility condition (Definition~\ref{defn.point}(\ref{item.weak.clique})); and 
regularity, being intertwined, and being unconflicted or strongly compatible are anti-separation properties (see Definitions~\ref{defn.intertwined.points}, \ref{defn.tn}, \ref{defn.conflicted}, and~\ref{defn.F.strongly.compatible}, and Remark~\ref{rmrk.not.hausdorff}).\footnote{An extra word on this:  Our theory of semitopologies admits spaces whose points partition into distinct communities.
Surely it \emph{must be bad} if not all points need be in consensus in a final state?  

Not at all: for example, most blockchains have a \emph{mainnet} and several \emph{testnets} and it is understood that each should be coherent within itself, but different nets \emph{need not} be in consensus with one another --- indeed, if the mainnet had to agree with a testnet then this would likely be a bug, not a feature.  So the idea of a single space with multiple partitions of consensus is not a new idea; it is an old idea, which we frame in a new, fruitful, and more general way.}
\item
\emph{Topology:}\quad 
If a minimal open neighbourhood of a point exists then it is least, because we can intersect two minimal neighbourhoods to get a smaller one which by minimality is equal to both.
A finite filter has a least element.

\emph{Semitopology:}\quad 
A point may have multiple minimal open neighbourhoods --- examples are very easy to generate, see e.g. the top-right example in Figure~\ref{fig.012}.
A finite semifilter need not have a least element (see Remark~\ref{rmrk.other.properties}).
\item
\emph{Topology:}\quad
Every finite $T_0$ topology is sober.
A topology is sober if and only if every nonempty irreducible closed set is the closure of a unique point.

\emph{Semitopology:}\quad
Neither property holds.  See Lemma~\ref{lemm.T0.not.sober}.
\item
Semitopological questions such as \emph{`is this a topen set'} or \emph{`are these two points intertwined'} or \emph{`does this point have a topen neighbourhood'} --- and many other definitions in this paper, such as our taxonomy of points into \emph{regular}, \emph{weakly regular}, \emph{quasiregular}, \emph{conflicted}, and \emph{strongly compatible} are novel and/or play a larger role in the theory than they typically do in topology.
\end{enumerate}

\subsection{Related work}
\label{subsect.related.work}

\paragraph*{Dualities}

We discussed duality results in detail in Remark~\ref{rmrk.categorical.duality}.
The reader may know that there are a many such results, starting with Stone's classic duality between Boolean algebras and compact Hausdorff spaces with a basis of clopen sets~\cite{stone:therba,johnstone:stos}.
The duality between frames and topologies is described in \cite[page~479, Corollary~4]{maclane:sheglf}.
See also the encyclopaedic treatment in \cite{caramello:toptas}, with an overview in Example~2.9 on page~17.
Our duality between semiframes and semitopologies fits into this canon.

\paragraph*{Union sets, closure spaces, and minimal structures}

There is a thread of research into \emph{union-closed families}; these are subsets of a finite powerset closed under unions, so that a union-closed families is precisely just a finite semitopology. 
The motivation is to study the combinatorics of finite subsemilattices of a powerset.
Some progress has been made in this~\cite{poonen:unicf}; the canonical reference for the relevant combinatorial conjectures is the `problem session' on page~525 (conjectures 1.9, 1.9', and 1.9") of~\cite{rival:grao}.
See also recent progress in a conjecture about union-closed families.\footnote{\url{https://web.archive.org/web/20230330170701/https://en.wikipedia.org/wiki/Union-closed_sets_conjecture\#Partial_results}.}
There is no direct connection to semitopologies, and certainly no consideration of duality results.
Perhaps the duality in this paper may be of some interest in that community. 

A \emph{closure space} is a subset of a powerset that is closed under intersections~\cite[page~173]{erne:clo}.
Up to taking sets complements, a closure space is a semitopology, and likewise a finite closure space is, up to taking sets complements, a union-closed family.
The motivation for closure spaces is to study closure operations in a topology-flavoured style, so closure spaces (unlike union-closed sets) share a topological flavour with semitopologies. 
However the applications are completely different, and algebra reveals how the underlying structures are also mathematically distinct, as is made clear by a lattice-based presentation of closure spaces and morphisms~\cite[Subsection~2.4]{erne:clo}, which uses a structure of `based lattices' (essentially: lattices plus a generating basis): this differs from semitopologies with their semiframe presentation as compatible complete semilattices.%
\footnote{There are other differences too.  For example: in a semitopology, the closure of the empty set is empty; in a closure space it need not be.  This is because the `closure' in a closure space is intended in a subtly different sense, which is intended to model things that include `deductive closure of', and the deductive closure of the empty set is the set of tautologies, which might be nonempty.} 

A \emph{minimal structure} on a set $X$ is a subset of $\powerset(X)$ that contains $\varnothing$ and $X$.
Thus a semitopology is a minimal structure that is also closed under arbitrary unions.
There is a thread of research into minimal structures, studying how notions familiar from topology (such as continuity) fare in weak (minimal) settings~\cite{noiri:defsgf} and how this changes as axioms (such as closure under unions) are added or removed.
An accessible discussion is in~\cite{szaz:minsgt}, and see the brief but comprehensive references in Remark~3.7 of that paper.
Of course our focus is on properties of semitopologies 
which are not considered in that literature; but we share an observation with minimal structures that it is useful to study topology-like constructs, in the absence of closure under intersections.

\paragraph*{Algebraic topology as applied to distributed computing tasks}

The reader may know that solvability results about distributed computing tasks have been obtained from algebraic topology, starting with the impossibility of $k$-set consensus and the Asynchronous Computability Theorem~\cite{herlihy_asynchronous_1993,borowsky_generalized_1993,saks_wait-free_1993} in 1993.
See~\cite{herlihy_distributed_2013} for numerous such results.
 
The basic observation is that states of a distributed algorithm form a simplicial complex, called its \emph{protocol complex}, and topological properties of this complex, like connectivity, are constrained by the underlying communication and fault model.
These topological properties in turn can determine what tasks are solvable. 
For example: every algorithm in the wait-free model with atomic read-write registers has a connected protocol complex, and because the consensus task's output complex is disconnected, consensus in this model is not solvable~\cite[Chapter~4]{herlihy_distributed_2013}.

This paper is also topological, but in a different way: we use (semi)topologies to study consensus in and of itself, rather than the solvability of consensus or other tasks in particular computation models.
Put another way: the papers cited above use topology to study the solvability of distributed tasks, but this paper shows how the very idea of `distribution' can be viewed as having a semitopological foundation.

Of course we can imagine that these might be combined --- that in future work we may find interesting and useful things to say about the topologies of distributed algorithms when viewed as algorithms \emph{on} and \emph{in} a semitopology.

\paragraph*{Fail-prone systems and quorum systems}

Given a set of processes $\ns P$, a \emph{fail-prone} system~\cite{malkhi_byzantine_1998}  (or \emph{adversary structure}~\cite{hirt_player_2000}) is a set of \emph{fail-prone sets} $\mathcal{F}=\{F_1,...,F_n\}$ where, for every $1\leq i\leq n$, $F_i\subseteq \ns P$.
$\mathcal{F}$ denotes the assumptions that the set of processes that will fail (potentially maliciously) is a subset of one of the fail-prone sets.
A \emph{dissemination quorum system} for $\mathcal{F}$ is a set  $\{Q_1,..., Q_m\}$ of quorums where, for every $1\leq i\leq m$, $Q_i\subseteq \ns P$, and such that 
\begin{itemize*}
\item
for every two quorums $Q$ and $Q'$ and for every fail-prone set $F$, $\left(Q\cap Q'\right)\setminus F\neq\emptyset$ and 
\item
for every fail-prone set $F$, there exists a quorum disjoint from $F$.
\end{itemize*}
Several distributed algorithms, such as Bracha Broadcast~\cite{bracha_asynchronous_1987} and PBFT~\cite{castro_practical_2002}, rely on a quorum system for a fail-prone system $\mathcal{F}$ in order to solve problems such as reliable broadcast and consensus assuming (at least) that the assumptions denoted by $\mathcal{F}$ are satisfied.

Several recent works generalise the fail-prone system model to heterogeneous systems.
Under the failure assumptions of a traditional fail-prone system, Bezerra et al.~\cite{bezerra_relaxed_2022} study reliable broadcast when participants each have their own set of quorums.
Asymmetric Fail-Prone Systems~\cite{DBLP:journals/dc/AlposCTZ24} generalise fail-prone systems to allow participants to make different failure assumption and have different quorums.
In Permissionless Fail-Prone Systems~\cite{cachin_quorum_2023}, participants not only make assumptions about failures, but also make assumptions about the assumptions of other processes;
the resulting structure seems closely related to witness semitopologies, but the exact relationship still needs to be elucidated.

Federated Byzantine Agreement Systems~\cite{mazieres2015stellar} are an instance of semitopologies.
García-Pérez and Gotsman~\cite{garcia2018federated} rigorously prove the correctness of broadcast abstractions in Stellar's Federated Byzantine Agreement model and investigate the model's relationship to dissemination quorum systems.
The Personal Byzantine Quorum System model~\cite{losa:stecbi} is an abstraction of Stellar's Federated Byzantine Agreement System model and accounts for the existence of disjoint consensus clusters (in the terminology of the paper) which can each stay in agreement internally but may disagree between each other.
Consensus clusters are closely related to the notion of topen in Definition~\ref{defn.transitive}(\ref{transitive.cc}).

Sheff et al. study heterogeneous consensus in a model called Learner Graphs~\cite{sheff_heterogeneous_2021} and propose a consensus algorithm called Heterogeneous Paxos.

Cobalt, the Stellar Consensus Protocol, Heterogeneous Paxos, and the Ripple Consensus Algorithm~\cite{macbrough_cobalt_2018,mazieres2015stellar,sheff_heterogeneous_2021,schwartz_ripple_2014} are consensus algorithms that rely on heterogeneous quorums or variants thereof.
The Stellar network~\cite{lokhafa:fassgp} and the XRP Ledger~\cite{schwartz_ripple_2014} are two global payment networks that use heterogeneous quorums to achieve consensus among an open set of participants; the Stellar network is an instance of a witness semitopology.

The literature on fail-prone systems and quorum systems is most interested in synchronisation algorithms for distributed systems and has been less concerned with their deeper mathematical structure.
Some work by the second author and others~\cite{losa:stecbi} gets as far as proving an analogue to Proposition~\ref{prop.cc.unions} (though we think it is fair to say that the presentation in this paper is simpler and clearer), but it fails to notice the connection with topology and the subsequent results which we present in this paper, and there is no consideration of algebra as used in this paper.

\paragraph*{(Semi)lattices with extra structure}

I am not aware of semiframes having been studied in the literature, but they are in excellent company, in the sense that things have been studied that are structurally similar.
We mention two examples to give a flavour of this extensive literature:
\begin{enumerate}
\item
A \deffont{quantale} is a complete lattice $(\mathsf Q,\bigvee)$ with an associative \emph{multiplication} operation $\ast : (\ns Q \times \ns Q) \to \ns Q$ that distributes over $\bigvee$ in both arguments~\cite{rosenthal:quaata}.
A commutative quantale whose multiplication is restricted to map to either the top or bottom element in $\mathsf Q$ is close being a semiframe.\footnote{But not quite! We also need proper reflexivity (Definition~\ref{defn.compatibility.relation}(\ref{item.compatible.reflexive})), and quantale morphisms do not necessarily map the top element to the top element like semiframe morphisms should (Definitions~\ref{defn.complete.semilattice.morphism} and~\ref{defn.category.of.spatial.graphs}(\ref{item.category.spatial.morphism})).}
For reference, a pleasingly simple representation result for quantales is given in~\cite{brown:reptq}.
\item
An \deffont{overlap algebra} is a complete Heyting algebra $\ns X$ with an \emph{overlap relation} $\overlaps\subseteq\ns X\times\ns X$ whose intuition is that $x\overlaps y$ when $x\tand y$ is \emph{inhabited}.
The motivation for this comes from constructive logic, in which $\Exists{p}(p\in x\land p\in y)$ is a different and stronger statement than $\neg\Forall{p}\neg(p\in x\land p\in y)$.  
Accordingly, overlap algebras are described as `a constructive look at Boolean algebras'~\cite{ciraulo:oveacl}.

Overlap algebras are not semiframes, but they share an idea with semiframes in making a structural distinction between `intersect' and `have a non-empty join'.
\end{enumerate}

\subsection{Future work}
\label{subsect.future.work}

The list of future work below is not exhaustive, of course.
If this list inspires, through some omission that is obvious to the reader, a new idea, then that is best of all: finding good questions is the first and arguably most important step in research. 

\begin{rmrk}[Other notions of morphism]
\label{rmrk.more.conditions}
In Definition~\ref{defn.morphism.semitopologies}(\ref{item.morphism.st}) we take a morphism of semitopologies $f:(\ns P,\opens)\to(\ns P',\opens')$ to be a continuous function $f:\ns P\to\ns P'$.
Correspondingly, in Definition~\ref{defn.category.of.spatial.graphs}(\ref{item.category.spatial.morphism}) we take a morphism of semiframes $g:(\ns X',\cti',\ast')\to(\ns X,\cti,\ast)$ to be a compatible morphism of complete semilattices.

The reader may be familiar with conditions on maps between topologies other than continuity, such as being \emph{open} ($f$ maps open sets to open sets) and \emph{closed} ($f$ maps closed sets to closed sets).
These also make sense in semitopologies.

A further natural design space in the semitopological case is to include conditions on sets intersections and strict inclusions.
We briefly list some conditions that we could impose on $f:\ns P\to\ns P'$:
\begin{enumerate*}
\item
If $O\between O'$ then $f^\mone(O)\between f^\mone(O')$ (it is automatic that if $f^\mone(O)\between f^\mone(O')$ then $O\between O'$, but the reverse implication is a distinct condition).
\item 
If $O\subsetneq O'$ then $f^\mone(O)\subsetneq f^\mone(O')$.
\item
$O\between O'$ implies $f^\mone(O)\subseteq f^\mone(O')$ implies $O\subseteq O'$, or an equivalent contrapositive: \\
$O\between O'$ implies $O\not\subseteq O'$ implies $f^\mone(O)\not\subseteq f^\mone(O')$.\footnote{This condition is motivated by a study of graph representations of semitopologies.  See~\cite[Chapter~17]{gabbay:semdca}.}
\end{enumerate*} 
\end{rmrk}

\begin{rmrk}[Representations]
\label{rmrk.representations}
The Sierpi\'nski space $\f{Sk}=(\ns P,\opens)$ sets $\ns P=\{0,1\}$ and $\opens=\{\varnothing,\{1\},\{0,1\}\}$.
This is both a topology and a semitopology, and it is a \emph{classifying space} for open sets, in the sense that $\f{Hom}(\text{-},\f{Sk}):\tf{SemiTop}\to\tf{Set}$ is naturally isomorphic to $\f{Opns}:\tf{SemiTop}\to\tf{Set}$ which maps $(\ns P,\opens)$ to $\opens$.

$\f{Sk}$ does not classify semitopologies, because semiframes suggest that we should view the set of open sets $\opens$ of a semitopology as a semiframe structure, having a subset inclusion (of course) and \emph{also} a \emph{generalised intersection} $\between$.
A likely classifying space for this is the the top-left example in Figure~\ref{fig.012}, such that 
\begin{itemize*}
\item
$\ns P=\{0,1,2\}$ and 
\item
$\opens=\{\varnothing,\{0\},\{2\},\{0,2\},\{0,1,2\}\}$.
\end{itemize*}
With $\f{Sk}$ in mind, this space looks like two copies of $\f{Sk}$ glued end-to-end --- i.e. like two open sets --- where the $1\in\ns P$ represents where they might intersect.
Call this space $\THREE$.

$\THREE$ suggests that a logic for semiframes might be naturally \emph{three-valued}, with values $\mathbf{t}$, $\mathbf{f}$, and also $\mathbf{b}$ for `is in the intersection', and that the spaces $\f{Hom}(\text{-},\THREE)$ of continuous mappings to $\THREE$ might play a useful role.
This might also have some impact on suitable notions of exponential space (see Remark~\ref{rmrk.exponential.spaces}). 
Current research~\cite{gabbay:decasd} considers applying three-valued modal logics based on semitopologies to reason explicitly about decentralised algorithms, like Paxos~\cite{lamport2001paxos}. 
\end{rmrk}

\begin{rmrk}[Exponential spaces]
\label{rmrk.exponential.spaces}
It remains to check whether the category $\tf{SemiFrame}$ of semiframes is closed~\cite[page~180, Section~VII.7]{maclane:catwm}, or Cartesian.\footnote{It would be surprising if it were not Cartesian.  The category of semitopologies is Cartesian.}

It remains to look into the \emph{Vietoris} (also called the \emph{exponential}) semitopologies~\cite[Exercise~2.7.20, page~120]{engelking:gent}, though in view of Remark~\ref{rmrk.representations} it is also an open question what the most suitable notion of exponential would be --- we speculate that exponentials over $\THREE$ would be suitable.

More generally, it remains to study functors of the from $\f{Hom}(\text{-},B)$ and $\f{Hom}(A,\text{-})$, for different values of $A$ and $B$.
\end{rmrk}

\begin{rmrk}[Finiteness and compactness]
\label{rmrk.why.infinite}
The relation of semitopologies to finiteness is interesting.
On the one hand, our motivating examples --- distributed networks --- are finite because they exist in the real world.
On the other hand, in distributed networks, precisely because they are distributed, participants may not be able to depend on an exhaustive search of the full network being practical (or even permitted --- this could be interpreted as a waste of resources or even as hostile or dangerous).

This requires mathematical models and algorithms that \emph{make sense} on at least countably infinitely many points.\footnote{This is no different than a programming language including a datatype of arbitrary precision integers: the program must eventually terminate, but because we do not know when, we need the \emph{idea} of an infinity in the language.}

In fact, arguably even `countably large' is not quite right.
The natural cardinality for semitopologies may be \emph{uncountable}, since network latency means that we cannot even enumerate the network: no matter how carefully we count, we could always in principle discover new participants who have joined in the past (but we just had not heard of them yet).

This motivates future work in which we consider algebraic conditions on a semiframe $(\ns X,\cti,\ast)$ that mimic some of the properties of open sets of finite semitopologies (without necessarily insisting on finiteness itself).
For instance:
\begin{enumerate*}
\item
We could insist that a $\cti$-descending chain of non-$\tbot_{\ns X}$ elements in $\ns X$ have a non-$\tbot_{\ns X}$ greatest lower bound in $\ns X$.
\item 
We could insist that a $\cti$-descending chain of elements strictly $\cti$-greater than some $x\in \ns X$ have a greatest lower bound that is strictly $\cti$-greater than $x$.
\item
We could insist that if $(x_i\mid i\geq 0)$ and $(y_i\mid i\geq 0)$ are two $\cti$-descending chains of elements, and $x_i\ast y_i$ for every $i\geq 0$ --- in words: $x_i$ is compatible with $y_i$ --- then the greatest lower bounds of the two chains are compatible.
\end{enumerate*}
The reader may notice how these conditions are reminiscent of compactness conditions from topology: e.g. a metric space is compact if and only if every descending chain of open sets has a nonempty intersection.
This is no coincidence, since one of the uses of compactness in topology is precisely to recover some of the characteristics of finite topologies. 

Considering semiframes (and indeed semitopologies) with compactness/finiteness flavoured conditions, is future work.
\end{rmrk}

\begin{rmrk}[Computational/logical behaviour]
\label{rmrk.computation}
The original motivation for semiframes comes from semitopologies, which are motivated by blockchains and other decentralised systems, which are often real software artefacts that do real computation.
It might therefore be useful to think about `computable' semiframes, whatever this should mean.
In~\cite{gabbay:semdca} we consider \emph{witness semitopologies} as a model of computably tractable semitopologies (being a witness semitopology is a concrete condition on sets), and we propose an algebraic abstraction of these as \emph{(strongly) chain-complete} semitopologies; a semitopology is (strongly) chain-complete when any infinite descending chain of (nonempty) open sets has a (nonempty) open intersection.
we argue that being chain-complete is an appropriate generalisation of the familiar Alexandrov condition on topologies.

This observation, along with Remark~\ref{rmrk.why.infinite}, suggests that it would be interesting to look at semiframes with additional conditions having to do with limits of possibly infinite descending chains. 
\end{rmrk}

\begin{rmrk}[Generalising $\ast$]
\label{rmrk.generalising.ast}
In Remark~\ref{rmrk.compatibility.intuition} we mentioned that we can think of semitopologies not as \emph{`topologies without intersections'} so much as \emph{`topologies with a generalised intersection'}.
In this paper we have studied a relation called $\between$ (for point-set semitopologies) and $\ast$ (for semiframes), which intuitively measure whether two elements intersect.
But really, this is just a notion of generalised meet, which suggests:
\begin{itemize*}
\item
We would take $(\ns X,\cti)$ and $(\ns X',\cti')$ to be complete join-semilattices and the generalised meet ${\ast} : (\ns X\times\ns X)\to \ns X'$ is any commutative distributive map. 
\item
We could generalise in a different direction and consider (for example) cocomplete symmetric monoidal categories: $\ast$ becomes the (symmetric) monoid action.
\item
We could generalise from binary intersections to $n$-ary intersections.
This generalisation deserves its own brief discussion:
\end{itemize*}
The literature on decentralised systems is rich in intersection conditions.
A $Q^2$ property from~\cite[Theorem~1]{DBLP:conf/asiacrypt/DamgardDFN07} asserts (in semitopological language) that all pairs of nonempty open sets intersect --- but there is more, e.g. a $Q^3$ property~\cite[Definition 2]{DBLP:journals/dc/AlposCTZ24} asserts that all \emph{triples} of nonempty open sets intersect.
Other unpublished work considers situations where up to \emph{five} nonempty open sets should intersect.
From the point of view of this paper, it suggests a sequence of $n$-ary generalisations of $\between$ and $\ast$.

Furthermore, a $B^3$ condition~\cite[Theorem 2]{DBLP:conf/asiacrypt/DamgardDFN07}\cite[Definition 5]{DBLP:journals/dc/AlposCTZ24}, considers a situation of even greater generality (and worse behaviour!), where open neighbourhoods are local to points in the sense that if $O\in\nbhd(p)$ and $p'\in O$ then it does not necessarily follow that $O\in\nbhd(p')$.
This suggests `heterogeneous' generalisations of semitopology and semiframe whose mathematical behaviour is at time of writing unexplored. 
\end{rmrk}

\begin{rmrk}[Homotopy and convergence]
\label{rmrk.homotopy.and.convergence}
We have not looked in any detail at notions of \emph{path} and \emph{convergence} in semitopologies and semiframes.
Preliminary considerations suggest that, in the finite case, a path converging to a point can be identified with a minimal open set containing that point.
In the topological case this is trivial, since minimal open sets are least; in the semitopological case it is not, since there may be many minimal open sets (Lemma~\ref{lemm.two.min}).
Thus, even in a semitopology with no `holes', it can be possible to approach a single point from multiple directions. 
Developing this part of the theory is future work.
\end{rmrk}

\begin{rmrk}[Constructive mathematics]
We have not considered what semiframes would look like in a constructive setting.
Much of the interest in frames and locales (versus point-set topologies) comes from working in a constructive setting; e.g. in the topos of sheaves over a base space, locales give a good fibrewise topology of bundles.
To what extent similar structures might be built using semiframes, or what other structures might emerge instead, are currently entirely open questions.
\end{rmrk}

\begin{rmrk}[Semiframes, logic, and analysis of real systems]
\label{rmrk.sl}
A natural application of semiframes is to help build logics for consensus, and in particular logics for consensus \emph{algorithms} (like Paxos~\cite{lamport2001paxos}).
A declarative approach to specifying consensus algorithmms, based on the semitopological ideas in this paper, is current research recently submitted for publication~\cite{gabbay:decasd}.
In separate work current at the time of writing, these techniques have helped find errors in a real consensus algorithm; details to be published in due course.

That work %
illustrates how the ideas in this paper can be made actionable to analyse properties of real consensus algorithms.
\end{rmrk}

\subsection{Final comments} 
\label{subsect.final.comments}

Innovation in decentralised systems design is currently very active, and in particular there is much work on notions of consensus and quorum~\cite{DBLP:journals/dc/AlposCTZ24,sheff_heterogeneous_2021,cachin_quorum_2023,li_quorum_2023,bezerra_relaxed_2022,garcia2018federated,lokhafa:fassgp,losa:stecbi,florian_sum_2022,li_open_2023}, and many new systems~\cite{macbrough_cobalt_2018,mazieres2015stellar,sheff_heterogeneous_2021,schwartz_ripple_2014} (and this is just a sample of an active literature).
This explosion of practical design is also an explosion of suggestions to the theory community, of new structures to consider.
One way to do this is to apply a classic technique which has been usefully applied many times before: \emph{topologise, then dualise}.

Of course, practitioners are also invited to learn more topology and algebra and consider phrasing some aspects of their systems in these terms.%
\footnote{A start on this is in~\cite{gabbay:decasd}.  True story: I attended a talk on a novel distributed consensus algorithm which featured an interesting five-way quorum intersection condition.  I said to the speaker after the talk \emph{``So consider each quorum as a proposition: what are the five properties that your algorithm needs to guarantee, in the sense that they are always true at some point?''}  --- a natural question, given my background in logic.  The speaker, who is no fool and is expert in practical systems, had not thought of things in this way and did not know to answer.  I suggested that until this question was answered, he didn't really understood his algorithm. Of course I tried to be polite, tactful and constructive when I said this, but it's true. }
If there is one thing holding back practical progress in decentralised algorithm design, it is that it is exceptionally hard.
If (semi)topology, algebra, and logic can help, then that will be a useful contribution.

\bibliographystyle{plain}

\begin{thebibliography}{10}

\bibitem{DBLP:journals/dc/AlposCTZ24}
Orestis Alpos, Christian Cachin, Bj{\"{o}}rn Tackmann, and Luca Zanolini.
\newblock Asymmetric distributed trust.
\newblock {\em Distributed Computing}, 37(3):247--277, 2024, Springer-Verlag.
\newblock \url{https://doi.org/10.1007/s00446-024-00469-1}.

\bibitem{bezerra_relaxed_2022}
João~Paulo Bezerra, Petr Kuznetsov, and Alice Koroleva.
\newblock Relaxed reliable broadcast for decentralized trust.
\newblock In Mohammed-Amine Koulali and Mira Mezini, editors, {\em Networked
  Systems}, Lecture Notes in Computer Science, pages 104--118. Springer Nature,
  2022.
\newblock \url{https://doi.org/10.1007/978-3-031-17436-0_8}.

\bibitem{borowsky_generalized_1993}
Elizabeth Borowsky and Eli Gafni.
\newblock Generalized {FLP} {Impossibility} {Result} for {T}-resilient
  {Asynchronous} {Computations}.
\newblock In {\em Proceedings of the {Twenty}-fifth {Annual} {ACM} {Symposium}
  on {Theory} of {Computing}}, {STOC} '93, pages 91--100, New York, NY, USA,
  1993. ACM.

\bibitem{bourbaki:gent1}
Nicolas Bourbaki.
\newblock {\em General Topology: Chapters 1-4}.
\newblock Elements of mathematics. Springer-Verlag, Berlin, 1998.

\bibitem{bracha_asynchronous_1987}
Gabriel Bracha.
\newblock Asynchronous {Byzantine} agreement protocols.
\newblock {\em Information and Computation}, 75(2):130--143, November 1987,
  Elsevier.
\newblock \url{https://doi.org/10.1016/0890-5401(87)90054-X}.

\bibitem{brown:reptq}
Carolyn Brown and Doug Gurr.
\newblock A representation theorem for quantales.
\newblock {\em Journal of Pure and Applied Algebra}, 85(1):27--42, 1993,
  Elsevier.
\newblock \url{https://doi.org/10.1016/0022-4049(93)90169-T}.

\bibitem{cachinbook}
Christian Cachin, Rachid Guerraoui, and Lu{\'{\i}}s E.~T. Rodrigues.
\newblock {\em Introduction to Reliable and Secure Distributed Programming
  {(2.} ed.)}.
\newblock Springer, 2011.
\newblock \url{https://doi.org/10.1007/978-3-642-15260-3}.

\bibitem{cachin_quorum_2023}
Christian Cachin, Giuliano Losa, and Luca Zanolini.
\newblock Quorum systems in permissionless networks.
\newblock In Eshcar Hillel, Roberto Palmieri, and Etienne Rivière, editors,
  {\em 26th International Conference on Principles of Distributed Systems
  ({OPODIS} 2022)}, volume 253 of {\em Leibniz International Proceedings in
  Informatics ({LIPIcs})}, pages 17:1--17:22. Schloss Dagstuhl –
  Leibniz-Zentrum für Informatik, 2023.
\newblock \url{https://doi.org/10.4230/LIPIcs.OPODIS.2022.17}.

\bibitem{caramello:toptas}
Olivia Caramello.
\newblock A topos-theoretic approach to {S}tone-type dualities, 2011.
\newblock \url{https://doi.org/10.48550/arXiv.1103.3493}.

\bibitem{castro_practical_2002}
Miguel Castro and Barbara Liskov.
\newblock Practical {Byzantine} fault tolerance and proactive recovery.
\newblock {\em ACM Transactions on Computer Systems (TOCS)}, 20(4):398--461,
  2002, ACM, New York, USA.
\newblock \url{http://dl.acm.org/citation.cfm?id=571640}.

\bibitem{ciraulo:oveacl}
Francesco Ciraulo and Michele Contente.
\newblock {Overlap Algebras: a Constructive Look at Complete Boolean Algebras}.
\newblock {\em {Logical Methods in Computer Science}}, 16, February 2020.
\newblock \url{https://doi.org/10.23638/LMCS-16(1:13)2020}.

\bibitem{DBLP:conf/asiacrypt/DamgardDFN07}
Ivan Damg{\aa}rd, Yvo Desmedt, Matthias Fitzi, and Jesper~Buus Nielsen.
\newblock Secure protocols with asymmetric trust.
\newblock In Kaoru Kurosawa, editor, {\em Advances in Cryptology - {ASIACRYPT}
  2007, 13th International Conference on the Theory and Application of
  Cryptology and Information Security, Kuching, Malaysia, December 2-6, 2007,
  Proceedings}, volume 4833 of {\em Lecture Notes in Computer Science}, pages
  357--375. Springer, 2007.

\bibitem{priestley:intlo}
B.~A. Davey and Hilary~A. Priestley.
\newblock {\em Introduction to Lattices and Order}.
\newblock Cambridge University Press, Cambridge, UK, 2 edition, 2002.

\bibitem{engelking:gent}
Ryszard Engelking.
\newblock {\em General Topology}.
\newblock Sigma Series in Pure Mathematics. Heldermann Verlag, Berlin, 1989.

\bibitem{erne:clo}
Marcel Erné.
\newblock Closure.
\newblock In {\em Beyond topology}, volume 486, pages 163--238. American
  Mathematical Society, Rhode Island, USA, 2009.

\bibitem{florian_sum_2022}
Martin Florian, Sebastian Henningsen, Charmaine Ndolo, and Björn Scheuermann.
\newblock The sum of its parts: {Analysis} of federated byzantine agreement
  systems.
\newblock {\em Distributed Computing}, 35(5):399--417, 10 2022,
  Springer-Verlag.
\newblock \url{https://doi.org/10.1007/s00446-022-00430-0}.

\bibitem{gabbay:semdca}
Murdoch~J. Gabbay.
\newblock {\em Semitopology: decentralised collaborative action via topology,
  algebra, and logic}.
\newblock College Publications, London, UK, August 2024.
\newblock \url{https://www.collegepublications.co.uk/logic/?00056}.

\bibitem{gabbay:semtad}
Murdoch~J Gabbay.
\newblock Semitopology: a topological approach to decentralized collaborative
  action.
\newblock {\em Journal of Logic and Computation}, 35(5), January 2025, Oxford
  University Press, Oxford, UK.
\newblock \url{https://doi.org/10.1093/logcom/exae050}.

\bibitem{gabbay:decasd}
Murdoch~J. Gabbay and Luca Zanolini.
\newblock A declarative approach to specifying distributed algorithms using
  three-valued modal logic, 2025.
\newblock (Submitted for publication.) \url{https://arxiv.org/abs/2502.00892}.

\bibitem{garcia2018federated}
{\'A}lvaro Garc{\'\i}a-P{\'e}rez and Alexey Gotsman.
\newblock Federated byzantine quorum systems.
\newblock In {\em 22nd International Conference on Principles of Distributed
  Systems (OPODIS 2018)}. Schloss Dagstuhl-Leibniz-Zentrum f\"ur Informatik,
  Germany, 2018.
\newblock \url{https://doi.org/10.4230/LIPIcs.OPODIS.2018.17}.

\bibitem{herlihy_distributed_2013}
Maurice Herlihy, Dmitry Kozlov, and Sergio Rajsbaum.
\newblock {\em Distributed computing through combinatorial topology}.
\newblock Morgan Kaufmann, Burlington, Massachusetts, USA, 2013.
\newblock \url{https://doi.org/10.1016/C2011-0-07032-1}.

\bibitem{herlihy_asynchronous_1993}
Maurice Herlihy and Nir Shavit.
\newblock The asynchronous computability theorem for t-resilient tasks.
\newblock In {\em Proceedings of the twenty-fifth annual {ACM} symposium on
  {Theory} of computing}, pages 111--120. Springer, Berlin, Heidelberg, 1993.
\newblock \url{https://doi.org/10.1007/978-3-662-53426-7_31}.

\bibitem{hirt_player_2000}
Martin Hirt and Ueli Maurer.
\newblock Player {Simulation} and {General} {Adversary} {Structures} in
  {Perfect} {Multiparty} {Computation}.
\newblock {\em Journal of Cryptology}, 13(1):31--60, January 2000, Springer,
  Berlin.
\newblock \url{https://doi.org/10.1007/s001459910003}.

\bibitem{johnstone:stos}
Peter~T. Johnstone.
\newblock {\em Stone spaces}, volume~3.
\newblock Cambridge University Press, 1986.

\bibitem{koppelberg:hanba1}
Sabine Koppelberg.
\newblock {\em Handbook of Boolean algebras, volume 1}.
\newblock North-Holland, Amsterdam, 1989.
\newblock Series editors Robert Bonnet and James Donald Monk.

\bibitem{lamport_part-time_1998}
Leslie Lamport.
\newblock The part-time parliament.
\newblock {\em {ACM} Transactions on Computer Systems}, 16(2):133--169, 05
  1998, ACM, New York, NY, USA.
\newblock \url{http://doi.acm.org/10.1145/279227.279229}.

\bibitem{lamport2001paxos}
Leslie Lamport.
\newblock {P}axos made simple.
\newblock {\em ACM SIGACT News (Distributed Computing Column)}, 32(4):51--58,
  2001, ACM.
\newblock
  \url{http://www.cs.utexas.edu/users/lorenzo/corsi/cs380d/past/03F/notes/paxos-simple.pdf}.
  Permalink:
  \url{https://web.archive.org/web/20250320162715/https://www.cs.utexas.edu/~lorenzo/corsi/cs380d/past/03F/notes/paxos-simple.pdf}.

\bibitem{li_quorum_2023}
Xiao Li, Eric Chan, and Mohsen Lesani.
\newblock {Quorum Subsumption for Heterogeneous Quorum Systems}.
\newblock In Rotem Oshman, editor, {\em 37th International Symposium on
  Distributed Computing (DISC 2023)}, volume 281 of {\em Leibniz International
  Proceedings in Informatics (LIPIcs)}, pages 28:1--28:19, Dagstuhl, Germany,
  2023. Schloss Dagstuhl -- Leibniz-Zentrum f{\"u}r Informatik.
\newblock \url{https://doi.org/10.4230/LIPIcs.DISC.2023.28}.

\bibitem{li_open_2023}
Xiao Li and Mohsen Lesani.
\newblock Open {Heterogeneous} {Quorum} {Systems}, April 2023.
\newblock \url{https://arxiv.org/abs/2304.02156v1}.

\bibitem{lokhafa:fassgp}
Marta Lokhava, Giuliano Losa, David Mazi\`eres, Graydon Hoare, Nicolas Barry,
  Eli Gafni, Jonathan Jove, Rafa\l{} Malinowsky, and Jed McCaleb.
\newblock Fast and secure global payments with {S}tellar.
\newblock In {\em Proceedings of the 27th ACM Symposium on Operating Systems
  Principles}, SOSP '19, page 80–96, New York, NY, USA, 2019. Association for
  Computing Machinery.

\bibitem{losa:stecbi}
Giuliano Losa, Eli Gafni, and David Mazi{\`e}res.
\newblock Stellar consensus by instantiation.
\newblock In Jukka Suomela, editor, {\em 33rd International Symposium on
  Distributed Computing (DISC 2019)}, volume 146 of {\em Leibniz International
  Proceedings in Informatics (LIPIcs)}, pages 27:1--27:15, Dagstuhl, Germany,
  2019. Schloss Dagstuhl--Leibniz-Zentrum f\"ur Informatik.

\bibitem{maclane:catwm}
Saunders {Mac Lane}.
\newblock {\em Categories for the Working Mathematician}, volume~5 of {\em
  Graduate Texts in Mathematics}.
\newblock Springer-Verlag, New York, 1971.
\newblock \url{https://doi.org/10.1007/978-1-4757-4721-8}.

\bibitem{maclane:sheglf}
Saunders {Mac Lane} and Ieke Moerdijk.
\newblock {\em Sheaves in Geometry and Logic: A First Introduction to Topos
  Theory}.
\newblock Universitext. Springer-Verlag, New York, 1992.
\newblock \url{https://doi.org/10.1007/978-1-4612-0927-0}.

\bibitem{macbrough_cobalt_2018}
Ethan MacBrough.
\newblock Cobalt: {BFT} {Governance} in {Open} {Networks}.
\newblock 2 2018.
\newblock Available online at \url{https://doi.org/10.48550/arXiv.1802.07240}.

\bibitem{malkhi_byzantine_1998}
Dahlia Malkhi and Michael Reiter.
\newblock Byzantine quorum systems.
\newblock {\em Distributed computing}, 11(4):203–213, 10 1998,
  Springer-Verlag.
\newblock \url{https://doi.org/10.1007/s004460050050}.

\bibitem{mazieres2015stellar}
David Mazi{\`e}res.
\newblock The {S}tellar consensus protocol: a federated model for
  {I}nternet-level consensus.
\newblock Technical report, {Stellar Development Foundation}, 2015.
\newblock \url{https://www.stellar.org/papers/stellar-consensus-protocol.pdf}.
  Permalink:
  \url{https://web.archive.org/web/20240629063518/https://stellar.org/learn/stellar-consensus-protocol}.

\bibitem{naor:loacaq}
Moni Naor and Avishai Wool.
\newblock The load, capacity and availability of quorum systems.
\newblock In {\em Proceedings 35th Annual Symposium on Foundations of Computer
  Science}, volume~27, pages 214--225. IEEE, New York, USA, 1994.
\newblock \url{https://doi.org/10.1109/SFCS.1994.365692}.

\bibitem{pi-base:lemma}
{$\pi$-Base}.
\newblock \url{https://topology.pi-base.org/theorems/T000420}.
\newblock A community database of topological counterexamples. Permalink:
  \url{https://web.archive.org/web/20240108192930/https://topology.pi-base.org/theorems/T000420}.

\bibitem{picado:fraltw}
Jorge Picado and Aleš Pultr.
\newblock {\em Frames and Locales: Topology Without Points}.
\newblock Frontiers in Mathematics. Birkhäuser, Basel Switzerland, 1st
  edition, 2012.
\newblock \url{https://doi.org/10.1007/978-3-0348-0154-6}.

\bibitem{picado:seppft}
Jorge Picado and Aleš Pultr.
\newblock {\em Separation in Point-Free Topology}.
\newblock Birkhäuser, Basel Switzerland, 1st edition, 2021.
\newblock \url{https://doi.org/10.1007/978-3-030-53479-0}.

\bibitem{poonen:unicf}
Bjorn Poonen.
\newblock Union-closed families.
\newblock {\em Journal of Combinatorial Theory, Series A}, 59(2):253--268,
  1992, Elsevier.
\newblock \url{https://doi.org/10.1016/0097-3165(92)90068-6}.

\bibitem{noiri:defsgf}
Valeriu Popa and Takashi Noiri.
\newblock On the definitions of some generalized forms of continuity under
  minimal conditions.
\newblock {\em Memoirs of the Faculty of Science. Series A. Mathematics},
  22:9--18, 01 2001, Kyushu University, Fukuoka, Japan.

\bibitem{riehl:cattic}
Emily Riehl.
\newblock {\em Category theory in context}.
\newblock Modern Math Originals. Dover, 2016.
\newblock Available online at
  \url{https://emilyriehl.github.io/files/context.pdf}. Permalink:
  \url{https://web.archive.org/web/20250115022343/https://emilyriehl.github.io/files/context.pdf}.

\bibitem{riker:thepc}
William~H. Riker.
\newblock {\em The Theory of Political Coalitions}.
\newblock Yale University Press, 1962.

\bibitem{rival:grao}
Ivan Rival, editor.
\newblock {\em Graphs and Order. The Role of Graphs in the Theory of Ordered
  Sets and its Applications}, volume 147 of {\em Proceedings of NATO ASI,
  Series C (ASIC) in Banff, Canada}.
\newblock Springer, Dordrecht, 1985.
\newblock \url{https://doi.org/10.1007/978-94-009-5315-4}.

\bibitem{rosenthal:quaata}
Kimmo~I. Rosenthal.
\newblock {\em Quantales and their applications}.
\newblock Number 234 in Pitman Research Notes in Mathematics. Longman
  Scientific \& Technical, UK, 1990.

\bibitem{saks_wait-free_1993}
Michael Saks and Fotios Zaharoglou.
\newblock Wait-free k-set agreement is impossible: {The} topology of public
  knowledge.
\newblock In {\em Proceedings of the twenty-fifth annual {ACM} symposium on
  {Theory} of computing}, volume~29, pages 1449--1483, Philadelphia, USA, 2000.
\newblock \url{https://doi.org/10.1137/S0097539796307698}.

\bibitem{schwartz_ripple_2014}
David Schwartz, Noah Youngs, and Arthur Britto.
\newblock The {R}ipple {P}rotocol {C}onsensus {A}lgorithm.
\newblock {\em Ripple Labs Inc White Paper}, 5(8):151, 2014.

\bibitem{sheff_heterogeneous_2021}
Isaac Sheff, Xinwen Wang, Robbert~van Renesse, and Andrew~C. Myers.
\newblock Heterogeneous {Paxos}.
\newblock In Quentin Bramas, Rotem Oshman, and Paolo Romano, editors, {\em 24th
  {International} {Conference} on {Principles} of {Distributed} {Systems}
  ({OPODIS} 2020)}, volume 184 of {\em Leibniz {International} {Proceedings} in
  {Informatics} ({LIPIcs})}, pages 5:1--5:17, Dagstuhl, Germany, 2021. Schloss
  Dagstuhl–Leibniz-Zentrum für Informatik.
\newblock \url{https://doi.org/10.4230/LIPIcs.OPODIS.2020.5}.

\bibitem{stone:therba}
Marshall~H. Stone.
\newblock The theory of representation for boolean algebras.
\newblock {\em Transactions of the American Mathematical Society},
  40(1):37--111, July 1936, American Mathematical Society, Rhode Island, USA.
\newblock \url{https://doi.org/10.2307/1989664}.

\bibitem{szaz:minsgt}
{\'A}rp{\'a}d Sz{\'a}z.
\newblock Minimal structures, generalized topologies, and ascending systems
  should not be studied without generalized uniformities.
\newblock {\em Filomat (Nis)}, 21(1):87--97, 2007, University of Nis, Faculty
  of Sciences and Mathematics, Serbia.
\newblock \url{http://www.jstor.org/stable/26194889}.

\bibitem{vickers:topvl}
Steven Vickers.
\newblock {\em Topology via Logic}.
\newblock Cambridge University Press, USA, 1989.

\bibitem{wiki:Equivalence_of_categories}
Wikipedia.
\newblock {Equivalence of categories}.
\newblock
  \url{http://en.wikipedia.org/w/index.php?title=Equivalence\%20of\%20categories&oldid=1227082771},
  2024.
\newblock Permalink:
  \url{https://web.archive.org/web/20230316075107/https://en.wikipedia.org/wiki/Equivalence_of_categories}.

\bibitem{wiki:Idealordertheory}
Wikipedia.
\newblock {Ideal (order theory); maximal ideals}.
\newblock
  \url{http://en.wikipedia.org/w/index.php?title=Ideal\%20(order\%20theory)&oldid=1200832357},
  2024.
\newblock Permalink:
  \url{https://web.archive.org/web/20230724184908/https://en.wikipedia.org/wiki/Ideal\_(order\_theory)\#Maximal\_ideals}.

\bibitem{wiki:Separation_axiom}
Wikipedia.
\newblock {Separation axiom}; main definitions.
\newblock
  \url{https://en.wikipedia.org/w/index.php?title=Separation\%20axiom&oldid=1230514922#Main_definitions},
  2024.
\newblock Permalink:
  \url{https://web.archive.org/web/20221103233631/https://en.wikipedia.org/wiki/Separation_axiom#Main_definitions}.

\bibitem{willard:gent}
Stephen Willard.
\newblock {\em General Topology}.
\newblock Addison-Wesley, 1970.
\newblock Reprinted by Dover Publications, New York, USA.

\end{thebibliography}
\hyphenation{Mathe-ma-ti-sche}

\section*{Acknowledgements}

I dedicate this paper to the memory of Phil Scott: colleague, mentor, gentleman, friend, and inspiration to us all. 

I am extremely grateful to the anonymous referees for feedback on this work.
Thanks to Giuliano Losa and the Stellar Development Foundation for their generous support and funding in pursuing this research.

\end{document}